\documentclass[12pt]{article} 
\usepackage[paper=letterpaper,margin=.8in]{geometry}
\usepackage[utf8]{inputenc}

    \usepackage{amsmath}
    \usepackage{amssymb}
    \usepackage{amsthm}
    \usepackage{thmtools}

    \usepackage{titlesec}
    \usepackage[colorlinks = true, linkcolor = blue, citecolor = red, urlcolor  = blue, anchorcolor = blue]{hyperref}
    \usepackage{url}

    \usepackage{dsfont}
    \usepackage{mathrsfs}
    \usepackage{mathtools}
    \usepackage{upgreek}
    
    \usepackage{tikz}
    \usepackage{pgfplots}
    \pgfplotsset{compat=1.18}
    \usepackage{graphicx}
    \usepackage{wrapfig}
    \usepackage{subcaption}

    \usepackage{enumitem}
    \usepackage{setspace}

    \usepackage{booktabs}
    \usepackage{multirow}

    \usepackage[normalem]{ulem}
    \usepackage{tensor}
    \usepackage{braket}
    \usepackage{xcolor}
    
    \usepackage[numbers,sort&compress]{natbib}

\numberwithin{equation}{section}

\theoremstyle{plain}
\newtheorem*{theorem*}{Theorem}
\newtheorem*{definition*}{Definition}
\newtheorem*{lemma*}{Lemma}
\newtheorem*{proposition*}{Proposition}
\newtheorem*{corollary*}{Corollary}
\newtheorem*{conjecture*}{Conjecture}

\newtheorem{theorem}{Theorem}

\newtheorem{lemma}{Lemma}
\newtheorem{proposition}{Proposition}
\newtheorem{corollary}{Corollary}
\newtheorem{conjecture}{Conjecture}

\newcommand{\penrosediagram}[4][jagged]{%
\begin{tikzpicture}[scale=#2,
    decoration={zigzag,segment length=4pt,amplitude=1.5pt}]
    \fill[blue!8]  (0,0) -- (2,2) -- (0,4) -- cycle;  
    \fill[red!8]   (4,0) -- (2,2) -- (4,4) -- cycle;  
    \ifstrequal{#1}{jagged}{%
        \draw[thick, decorate] (0,4) -- (4,4);
        \draw[thick, decorate] (0,0) -- (4,0);
    }{%
        \draw[thick] (0,4) -- (4,4);
        \draw[thick] (0,0) -- (4,0);
    }%
    \draw[thick] (0,0) -- (0,4);
    \draw[thick] (4,0) -- (4,4);
    \draw[thick] (0,0) -- (4,4);
    \draw[thick] (0,4) -- (4,0);
    \node at (3.2,2) {\(#3\)};
    \node at (0.8,2) {\(#4\)};
\end{tikzpicture}%
}

    \newcommand{\calo}{\mathcal{O}}
    \newcommand{\cals}{\mathcal{S}}
    
    \newcommand{\calh}{\mathcal{H}}
    
    \newcommand{\cala}{\mathcal{A}}

    \newcommand{\dbar}{\mathrm{d}\hspace*{-0.16em}\bar{}\hspace*{0.2em}}

\begin{document}
	
\begin{flushright}
    MIT-CTP/6067
\end{flushright}

\thispagestyle{empty}
\vspace*{2.5cm}
\begin{center}
{\bf {\LARGE  A de Sitter Anti-Scrambling Algebra }}\\
\begin{center}
\vspace{1cm}
{\bf Wentao Cui$^{1}$ and David K. Kolchmeyer$^{2}$}\\
 \bigskip \rm
\bigskip
$^{1}$Center for Theoretical Physics -- a Leinweber Institute,\\
Massachusetts Institute of Technology, Cambridge, MA 02139, USA
\bigskip

$^{2}$School of Natural Sciences, Institute for Advanced Study,\\
Princeton, NJ 08540, USA
\rm
  \end{center}
\vspace{2.5cm}
{\bf Abstract}
\end{center}
\begin{quotation}
\noindent

We study the algebra of observables of a semiclassical observer in de Sitter space. When operators are time-evolved by a scrambling time, out-of-time-ordered correlation functions (OTOCs) receive corrections due to shockwave scattering near the cosmological horizon. When the observer's clock is treated classically, we show that the time advance effect implies the breakdown of the KMS condition for vacuum correlators of matter operators. When the observer's clock is quantized, this implies that the Hartle-Hawking state is not a trace on the resulting crossed-product algebra. We conjecture that the observer's algebra has a trivial commutant. When the time separation between operators exceeds the scrambling time, the OTOC decays to zero and a free product algebra emerges. We illustrate this using classical solutions of de Sitter JT gravity with shocks. We comment on how our results pose a challenge for observer-centric static patch holography.

\end{quotation}

	\pagebreak

	{
		\hypersetup{linkcolor=black}
		\tableofcontents
	}


	\pagebreak
	
\section{Introduction}

Although the holographic principle can be formulated in both Anti-de Sitter (AdS) and de Sitter (dS) spacetimes, much less is known about dS holography, in part due to the lack of a top-down construction.\footnote{Progress toward a top-down understanding of dS has been made in \cite{Anninos:2011ui,Anninos:2017eib,Anninos:2026hia,Silverstein:2024xnr,Banihashemi:2026mje,McAllister:2024lnt,Marini:2026zjk,Narovlansky:2023lfz,Narovlansky:2025tpb,Susskind:2022bia,Susskind:2021esx,Miyashita:2026zyl,Hoback:2026yqj}, among other works.
} Because aspects of the AdS/CFT correspondence can be studied from the bottom up, it is reasonable to expect that semiclassical gravity, and in particular the Euclidean path integral, can also teach us important lessons about dS holography. Static patch holography, which proposes that a static patch in dS is dual to a quantum-mechanical system with a finite-dimensional Hilbert space \cite{Banks:2000fe}, provides a conceptual basis for interpreting the results of bottom-up studies in dS. It has also been proposed that a successful model of dS holography should explicitly include an observer as part of the holographic map \cite{Anninos:2011af, Anninos:2011zn, Nakayama:2011qh,Maldacena:2024spf,Chen:2025jqm,Witten:2023xze,Tietto:2025oxn,Narovlansky:2023lfz,Blommaert:2026ofx,Goto:2026ipq}. The observer's worldline can be used to define gauge-invariant observables \cite{CLPW, Sivaramakrishnan:2024ydy,Cheung:2026euf}, and their internal degrees of freedom may be necessary to resolve puzzles surrounding the Hilbert space dimension of a closed universe \cite{Usatyuk:2024isz,Abdalla:2025gzn,Harlow:2025pvj,Wei:2025guh,Usatyuk:2024mzs,Akers:2025ahe,Engelhardt:2025azi,Zhao:2026mpl,Harlow:2026hky}.

In this paper, we explore dS holography by studying an observer's algebra in Lorentzian semiclassical gravity. To motivate our approach, we review similar algebraic studies in AdS. A classic result in AdS/CFT is that the thermofield double (TFD) state above the Hawking-Page temperature is dual to an eternal two-sided black hole \cite{Maldacena:2001kr}. This duality can be recast in an algebraic language developed by Leutheusser and Liu \cite{Leutheusser:2021frk,Leutheusser:2021qhd}. They carefully defined the large-$N$ limit of the algebra of single-trace operators in one CFT and showed that the result, an emergent type III$_1$ algebra, is equivalent to the bulk quantum field theory algebra of the black hole exterior, which we take to be the right wedge in Figure \ref{fig:penroseintro}.\footnote{In a follow-up work \cite{Leutheusser:2022bgi}, they argued that generic causally complete subregions in the bulk should also emerge from the large-$N$ limit of suitable boundary subalgebras, including subalgebras that are not canonically associated with boundary subregions. The general idea, which extends beyond AdS holography, is that any subalgebra in semiclassical gravity should emerge from the large-$N$ limit of a suitable subalgebra of the putative quantum-mechanical dual.} We refer to this algebra as $\cala_R$. The commutant algebra,\footnote{The commutant $\cala^\prime$ of an algebra $\cala$ is defined to be the algebra of all operators on the Hilbert space which commute with $\cala$. A von Neumann algebra $\cala$ is, by definition, equal to its double-commutant: $\cala = \cala^{\prime \prime}.$} $\cala_R^\prime$, is equivalent to the QFT algebra $\cala_L$ of the left wedge in Figure \ref{fig:penroseintro}, which is dual to the other CFT in the same sense. The algebras obey the following properties:
\begin{align}
    &\cala_L^\prime = \cala_R\,, \quad \quad \cala_R^\prime = \cala_L\,, \label{eq:fac1}
    \\
    \label{eq:fac2}
    &\cala_L \cap \cala_R = \mathbb{C} \,.
\end{align}
The second property above states that the only operators that belong to both $\cala_R$ and $\cala_L$ are multiples of the identity.  When a pair of algebras obeys both \eqref{eq:fac1} and \eqref{eq:fac2}, we say that they {\it factorize}. This generalizes the notion of Hilbert space factorization.\footnote{That is, given a factorized Hilbert space $\calh = \calh_L \otimes \calh_R$, we can obtain a pair of factorizing algebras by defining $\cala_L$ to be the algebra of all operators on $\calh_L$, and likewise for $\cala_R$. A general pair of factorizing algebras cannot necessarily be characterized in this way.} These properties are physically equivalent to the fact that the left and right wedges in Figure \ref{fig:penroseintro} are causal complements of each other.

\begin{figure}[h]
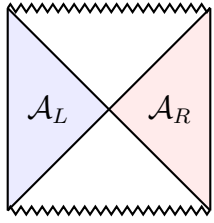

  \centering
                \penrosediagram{0.67}{\mathcal{A}_R}{\mathcal{A}_L}
  \caption{An eternal two-sided AdS black hole.}
  \label{fig:penroseintro}
\end{figure}

The reason why it is useful to explicitly discuss the wedge algebras $\cala_R$ and $\cala_L$ as well as the abstract algebraic properties \eqref{eq:fac1} and \eqref{eq:fac2} is that they continue to make sense away from the strict semiclassical limit, such as when the metric is subject to large quantum fluctuations and the geometric notion of a causally complete subregion no longer exists. That is, von Neumann algebras provide a natural language to describe subsystems in quantum gravity away from the semiclassical regime. A very interesting application of this idea appeared in the recent work of Penington and Tabor \cite{Penington:2025hrc}, which studied the eternal black hole algebras $\cala_L$ and $\cala_R$ in a regime where the single-trace operators may have large time separations, on the order of the scrambling time $\sim \log \frac{1}{G_N}$ (in units where $\ell_{AdS} = 1$). The single-trace operators create high energy shockwaves that backreact on the near-horizon geometry. Despite these quantum fluctuations of the metric, it is shown in \cite{Penington:2025hrc} that the algebras $\cala_R$ and $\cala_L$ continue to exist and factorize. Another example may be found in AdS JT gravity \cite{Kolchmeyer:2023gwa, Penington:2023dql}, which admits an exact quantization at finite $G_N$. In particular, \cite{Penington:2023dql} argued that algebraic factorization still holds when wormhole effects are incorporated. The fact that \eqref{eq:fac1} and \eqref{eq:fac2} are robust under both perturbative and non-perturbative corrections strongly points to the existence of an exact, non-perturbative theory in which the Hilbert space also factorizes.\footnote{In JT gravity, a mechanism for this was explained in \cite{Boruch:2024kvv}.} Thus, gravitational algebras are useful for motivating AdS/CFT from the bottom up.

Our starting point for the study of algebras in dS is the work of Chandrasekaran, Longo, Penington, and Witten (CLPW) \cite{CLPW}. They considered a pair of observers at the North and South poles of global dS. Each observer is naturally associated with a static patch. If we quantize a QFT on this background, we may define $\cala_R$ and $\cala_L$ to be the QFT algebras associated with each static patch. These algebras are generated by the quantum fields evaluated near the observer worldlines, and they factorize. Moreover, if we add the observer Hamiltonians to $\cala_L$ and $\cala_R$, then \cite{CLPW} showed that the algebras each admit a trace, and thus a well-defined notion of entropy. Up to a state-independent additive constant, this von Neumann entropy agrees with the generalized Gibbons-Hawking entropy of the observer's horizon. The trace is given by the expectation value in the Hartle-Hawking state, which has maximum entropy. This result may be interpreted as evidence for static patch holography. In particular, one may speculate that $\cala_R$ and $\cala_L$ are large-$N$ limits of a suitable single-trace algebra in a putative quantum-mechanical dual description. The formal ``type II'' trace on $\cala_R$ or $\cala_L$ would be dual to the ``type I'' trace in the dual theory.\footnote{A type I trace is defined by summing over the diagonal matrix elements of a given operator. A type II trace is a linear functional on an algebra that obeys standard properties such as cyclicity and positivity, but is not equivalent to any type I trace.} Thus, the dual would furnish the explicit microstates that account for the Gibbons-Hawking entropy of the observer's cosmological horizon.

CLPW \cite{CLPW} worked in the strict semiclassical limit, $G_N \rightarrow 0$. In this paper, we generalize their algebra by allowing the quantum field operators to be time-separated by a scrambling time, $\sim \log \frac{1}{G_N}$ (in units where $\ell_{dS} = 1$).\footnote{There are different notions of the scrambling time. The earliest scrambling time would be of order $\log M$ in dS units, where $M$ is the mass of the observer. This timescale is associated with the observer recoil effect which was studied in \cite{Kolchmeyer:2024fly}. In this paper, we will mostly focus on a two-dimensional model where the scrambling time is $\sim \log \frac{1}{G_N}$.} This is the simplest setting in which one can study corrections to the results of \cite{CLPW}. Our approach directly follows the analogous work of Penington and Tabor \cite{Penington:2025hrc} in AdS. We will focus on the algebra of just one of the observers, which we henceforth refer to as $\cala$. In most of this paper, $\cala$ will not include the observer's Hamiltonian, because we will not quantize the observer's clock.\footnote{After quantizing the observer's clock, we may construct a crossed product algebra $\cala_{\mathrm{cr}}$, which we discuss in Section \ref{sec:KMS-crossedprod}.} The operators in $\cala$ may be classified into early-time and late-time operators, where the time separation between an early-time and late-time operator is given by a large parameter $T$ (up to an order-one additive constant), while other time separations are order one. We work in the limit where $G_N \rightarrow 0$, but $G_N e^{\frac{2 \pi T}{\beta_c}}$ is held fixed, where $\beta_c$ is the inverse temperature of the observer's cosmological horizon. In this setting, the only QFT correlation functions that are affected are those with out-of-time-ordered insertions of early- and late-time operators. These operators produce shockwaves along the cosmological horizon, and their gravitational interactions are governed by an S-matrix that may be computed in the eikonal approximation. In AdS, the eikonal phase is positive, which physically means that shockwaves always experience a time delay \cite{Gao:2000ga}. In dS, the eikonal phase may be negative, such that time advances are permitted \cite{Anninos:2018svg, Aalsma:2020aib}. Due to this effect, two observers that are initialized at opposite poles of dS may communicate with each other. This has significant implications for the observer's algebra $\cala$. In particular, although we lack a rigorous proof, we will provide evidence that the commutant of $\cala$ is trivial (i.e., that $\cala^\prime = \mathbb{C}$). In other words, we conjecture that the properties \eqref{eq:fac1} and \eqref{eq:fac2} completely break down in dS. Moreover, if we follow CLPW and include the observer's Hamiltonian in their algebra, we find that the Hartle-Hawking state is no longer tracial because it does not obey the cyclic property $\mathrm{Tr}\, ab = \mathrm{Tr}\, ba$.\footnote{The use of the word cyclic here is not to be confused with the notion of a cyclic state in the context of Tomita-Takesaki theory or the Reeh–Schlieder theorem.} Taken at face value, this appears to be evidence against a version of static patch holography where the observer plays a fundamental role.\footnote{This ``observer-centric'' version of static-patch holography has been adopted in various works \cite{Anninos:2011af, Anninos:2011zn, Nakayama:2011qh,Maldacena:2024spf,Chen:2025jqm,Witten:2023xze,Tietto:2025oxn,Narovlansky:2023lfz,Blommaert:2026ofx,Goto:2026ipq,Narovlansky:2025tpb}.} In particular, the Hartle-Hawking state is expected to be a universal maximum entropy state \cite{Witten:2023xze, Blommaert:2025bgd}. In a finite-dimensional Hilbert space, such a state is tracial because the density matrix is proportional to the identity. Thus, to reconcile our results with the general expectations of static patch holography, new insights are needed.

The results quoted above crucially rely on the eikonal approximation for computing shockwave scattering amplitudes \cite{Kabat:1992tb}. We will mostly work with a two-dimensional toy model in which the S-matrix is simply
\begin{equation}
    S = e^{i x P_A P_B}\,,
    \label{eq:smatrixintro}
\end{equation}
where $P_A$ and $P_B$ are the energies of the two scattering shockwaves which are each understood to be of order $e^{\frac{ \pi T}{\beta_c}} $, and $x$ is a negative parameter proportional to $- G_N$.\footnote{We will actually assume that $P_A$ and $P_B$ are order one and interpret $x$ as $- G_N e^{\frac{2 \pi T}{\beta_c}}$. The difference is unimportant for the present discussion.} If we instead set $x$ to be positive, then our toy model is equivalent to semiclassical AdS JT gravity, where $x$ is proportional to $G_N$.\footnote{The $x > 0$ case was studied in \cite{Penington:2025hrc}.} Although correlation functions may be computed exactly in $x$, it is instructive to expand \eqref{eq:smatrixintro} as a power series in $x$, and work order-by-order in perturbation theory.\footnote{When we use the terminology ``perturbation theory,'' we are always referring to the power series expansion of the eikonal S-matrix. We do not discuss perturbation theory away from the eikonal approximation in this paper.} At a given order in perturbation theory, the results of \cite{CLPW} hold. That is, the observer's algebra has a nontrivial commutant, and the Hartle-Hawking state has the properties of a trace.\footnote{There is a technical subtlety: when computing an expectation value in the Hartle-Hawking state to a fixed order in perturbation theory, one may encounter a divergent integral. The correct statement is that when the relevant integrals converge, the resulting power series is consistent with a tracial Hartle-Hawking state. We discuss the details in Section \ref{sec:KMS-crossedprod}.} This is natural to expect because when $x > 0$, the algebra admits an interpretation as a boundary algebra of a two-sided AdS black hole in JT gravity. The Hartle-Hawking state must obey the cyclic property ($\mathrm{Tr}\, ab = \mathrm{Tr}\, ba$) at any fixed order in perturbation theory because it is exactly a trace when $x \geq 0$. In a more realistic model, such as the four-dimensional model that we introduce in Section \ref{sec:4dmodel}, one may also see that the Hartle-Hawking state obeys the cyclic property to all orders in perturbation theory.

We will first work in a physical setting where the observer has a perfect clock which is taken to be a part of the classical background. The observer's algebra $\cala$ is generated by QFT operators near their worldline at both early and late times, and correlation functions contain a parameter $x$. Let $\calh$ be the Hilbert space that $\cala$ acts on, and let $\ket{\Omega} \in \calh$ be the canonical vacuum state on $\cala$ (such as the Bunch-Davies vacuum of global dS), which is invariant under time translations.\footnote{The observer's clock provides a precise notion of time.} One may quantize the observer's clock by appending to $\calh$ a factor of $L^2(\mathbb{R})$ capturing the clock degrees of freedom. The observer's Hamiltonian acts on the extended Hilbert space $\calh \otimes L^2(\mathbb{R})$, and may be added to $\cala$ to obtain the crossed-product algebra $\cala_{\mathrm{cr}}$ \cite{Witten:2021unn}. In order for the Hartle-Hawking state to be a trace on $\cala_{\mathrm{cr}}$, the vacuum state $\ket{\Omega}$ must be KMS with respect to time translations of the observer's clock. More informally, we may say that the vacuum must be thermal from the observer's point of view. An important result of our work is that in our dS setup, $\ket{\Omega}$ does not obey the KMS condition. The KMS condition is closely related to Tomita-Takesaki theory, which we review in Appendix \ref{sec:aQFT-review}. An important lesson is that for every \emph{cyclic and separating} state on a von Neumann algebra $\mathfrak{A}$, there is a time-translation symmetry of $\mathfrak{A}$ for which the state is KMS.\footnote{This symmetry need not have a geometric interpretation. We refer to it as ``time-translation'' as a matter of convenience.} A state $\ket{\Psi}$ is \emph{cyclic} with respect to $\mathfrak{A}$ when $\mathfrak{A} \ket{\Psi}$ is dense in the Hilbert space, and $\ket{\Psi}$ is \emph{separating} with respect to $\mathfrak{A}$ when $\ket{\Psi}$ is cyclic with respect to the commutant $\mathfrak{A}^\prime$. In our setup, $\ket{\Omega}$ is clearly cyclic with respect to $\cala$, and the failure of the KMS condition can be used to prove by contradiction that $\ket{\Omega}$ is not separating for $\cala$. In contrast, in the setting of CLPW, the Bunch-Davies vacuum is both cyclic and separating for $\cala$. Because $\ket{\Omega}$ is not cyclic for $\cala^\prime$, this implies that $\cala^\prime$ is ``not large enough'' to generate a dense subspace of the Hilbert space when acting on $\ket{\Omega}$. Our conjecture that $\cala^\prime$ is trivial is based on this and other considerations that we describe in Section \ref{sec:properties}.

A key question is whether the correlation functions of operators in $\cala$ can be captured by a quantum-mechanical dual description. In particular, are gravitational time advances consistent with unitary quantum mechanics? The traversable wormhole \cite{Gao:2016bin,Maldacena:2017axo} provides a concrete quantum-mechanical dual of gravitational time advance. Given that scrambling is the quantum-mechanical mechanism that underlies gravitational time delay in AdS, we let ``anti-scrambling'' refer to the analogous mechanism dual to gravitational time advance.\footnote{To our knowledge, the term ``anti-scrambling'' was coined by Daniel Harlow and Ying Zhao.} Under a certain set of physical assumptions, Maldacena, Shenker, and Stanford \cite{Maldacena:2015waa} ruled out anti-scrambling and placed an upper bound on the rate of scrambling. In dS, the KMS condition is the key assumption of \cite{Maldacena:2015waa} that fails. This rules out a putative dS hologram based upon unitary quantum mechanics in a KMS state at the dS temperature. Interesting recent works on reproducing dS dynamics from explicit microscopic constructions include \cite{Narovlansky:2023lfz, Narovlansky:2025tpb, Goto:2026ipq,Marini:2026zjk,Verlinde:2024znh,Verlinde:2024zrh,Rahman:2024iiu,Miyashita:2026zyl}.

Although much of this work focuses on out-of-time-ordered dynamics at the scrambling time, we also consider a much larger timescale at which out-of-time-ordered correlators (OTOCs) decay to zero. In particular, Penington and Tabor \cite{Penington:2025hrc} showed that at times much larger than the scrambling time, their algebra of early- and late-time boundary operators in AdS becomes a free product algebra. In Section \ref{sec:alg-free}, we show that their proof also applies to our dS setup. In particular, at timescales much longer than the scrambling time, a type III$_1$ free product algebra emerges. The vacuum state is cyclic and separating for this algebra, the KMS condition holds, and the Hartle-Hawking state is a trace on the crossed-product algebra. The free product algebra was physically interpreted in AdS \cite{Chandrasekaran:2022eqq} in terms of long wormholes supported by shocks \cite{Shenker:2013yza}. Using dS JT gravity \cite{Maldacena:2019cbz, Cotler:2019nbi,Alonso-Monsalve:2024oii} as a toy model, we provide a similar geometric interpretation of the free product algebra. In AdS/CFT, the dynamics of this algebra may be captured by modeling the time evolution operator as a random unitary \cite{Stanford:2021bhl}. We can view the free product algebra as a universal mathematical structure that describes some aspects of the late-time dynamics of both CFTs dual to AdS as well as the putative dS hologram.

We now summarize the contents of this paper. In Section \ref{sec:obsalg}, we discuss a model of an observer in four dimensions, based upon the Schwarzschild–de Sitter geometry. We define the observer's algebra $\cala$ in this context. In the near-Nariai limit, we obtain a simpler two-dimensional model, which is the focus of the remainder of this paper. We also make some general comments on scrambling and anti-scrambling. In Section \ref{sec:OTOC}, we study the OTOC in both the time and frequency domains. We discuss the analytic continuation of the OTOC, the KMS condition, and the detailed balance condition. We then define the crossed-product algebra and demonstrate that the Hartle-Hawking state is not a trace. In Section \ref{sec:properties}, we provide evidence that supports our conjecture that $\cala^\prime$ is trivial. In particular, we prove that the vacuum is not separating for $\cala$, and we discuss a method for constraining the commutant. We then review the mathematical result of \cite{Penington:2025hrc} on the emergence of the free product algebra, and emphasize that this result also applies to dS. We then briefly review dS JT gravity, and construct Penrose diagrams of spacetimes with out-of-time-ordered shocks. We explain the geometric intuition for the emergence of the free product algebra, following an analogous AdS discussion in \cite{Chandrasekaran:2022eqq}. In the Discussion, we emphasize the challenge that our results pose for observer-centric static patch holography.


\section{An Observer's Algebra in de Sitter Space} \label{sec:obsalg}

\subsection{4D Model} \label{sec:4dmodel}

The goal of this paper is to study an observer's algebra in a regime where metric backreaction effects are important. To construct a model of an observer in de Sitter space, we begin with the four-dimensional Schwarzschild-de-Sitter black hole solution with mass $M$,
\begin{align}
    \mathrm{d}s_{\mathrm{SdS}}^2 &= -f(r)\,\mathrm{d}t^2 + \frac{\mathrm{d}r^2}{f(r)} + r^2 \mathrm{d}\Omega^2 \,,
    \\
    f(r) &=  1 - \frac{\Lambda}{3} r^2 - \frac{2 G_N M}{r} = - \frac{\Lambda}{3} \frac{(r - r_0)(r - r_b)(r - r_c)}{r}\,.
    \label{eq:SdS}
    \end{align}
The black hole horizon has radius $r_b$, while the cosmological horizon has radius $r_c$. Note that $r_0$, $r_b$, and $r_c$ obey the following bounds:
\begin{align}
    r_0 < 0 < r_b < \frac{1}{\sqrt{\Lambda}} < r_c < \sqrt{\frac{3}{\Lambda}} = \ell_{\mathrm{dS}}\,.
    \label{eq:bounds}
\end{align}
When studying physics near the cosmological horizon, it is convenient to work with Kruskal coordinates, which corresponds to the metric
\begin{align}
    \mathrm{d}s_{\mathrm{SdS}}^2 = - A(UV) \, \mathrm{d}U \mathrm{d}V + B(U V)\, \mathrm{d}\Omega^2 \,.
\end{align}
The directions of increasing $U$ and $V$ are indicated in Figure \ref{fig:penrose}. Below, we will only need to know the functions $A(x)$ and $B(x)$ near $x = 0$,
\begin{align}
    A(x) &= 1 + \frac{\Lambda}{6} \left(1 - \frac{r_0 r_b}{r_c^2}\right) x + \mathcal{O}(x^2)\,,
    \\
    B(x) &= r_c^2 + \frac{1}{6}(r_c - r_0)(r_c - r_b) \Lambda x + \mathcal{O}(x^2)\,.
\end{align}
\noindent To model an observer with mass $M$, we introduce a boundary at radius $r_{\mathrm{Obs}}$ and end the spacetime there. The radius $r_{\mathrm{Obs}}$ obeys
\begin{align}
    r_b < r_{\mathrm{Obs}} < r_c \,,
\end{align}
and the physical spacetime corresponds to $r > r_{\mathrm{Obs}}$. The choice of boundary conditions for the matter fields is not important for our later analysis, as long as there is no energy flux through the boundary. The boundary condition for the metric will also not be important for our analysis, because we will work in the approximation of linearized gravity, and the metric perturbations we consider are localized on the cosmological horizon, far away from the observer.\footnote{We are referring to linearized shockwaves. There can also be graviton perturbations, but we choose to regard those as part of the matter sector. Note that dS with timelike boundaries has been considered in various other works, such as \cite{Anninos:2024wpy, Anninos:2011zn, Silverstein:2024xnr, Banihashemi:2026mje, Banihashemi:2022jys, Batra:2024kjl}. It was found that the choice of boundary conditions is important for a well-posed initial value problem \cite{Anninos:2024wpy,An:2021fcq}. Although the shockwave solutions we use are not affected by the choice of boundary conditions, we expect that the precise choice will be important for understanding corrections to our model.}

The region $r > r_{\mathrm{Obs}}$ corresponds to the exterior of the observer. Although our spacetime ends at $r = r_{\mathrm{Obs}}$, we can think of $r < r_{\mathrm{Obs}}$ as the interior of the observer, where the metric depends on the matter configuration of the observer. The form of the metric inside the observer will not be important to us.

To quantize a matter QFT on this background spacetime, it is convenient to begin with the Euclidean metric obtained from the transformation $t \rightarrow -i \tau$, where $\tau$ is compactified so that the geometry is smooth near $r = r_c$. That is, we make the identification
\begin{align}
    \tau \sim \tau + \beta_c\,, \quad \quad \beta_c := \left|\frac{4 \pi}{f^\prime(r_c)}\right| \,.
\end{align}
This Euclidean geometry has the topology of a two-sphere times a disk. Euclidean correlation functions of local operators on this geometry may be analytically continued to Lorentzian signature. The Hilbert space may be obtained from the GNS construction. The GNS vacuum looks like the Minkowski vacuum near the bifurcate horizon. In this state, the observer and the cosmological horizon are in thermal equilibrium.\footnote{The Schwarzschild-de-Sitter geometry does not represent thermal equilibrium because the black hole and cosmological horizons have different temperatures. However, our spacetime has no black hole horizon because $r_{\mathrm{Obs}} > r_b$.} Let $\calh_0$ denote this GNS Hilbert space.

The isometry group of \eqref{eq:SdS} is $\mathbb{R} \times \mathrm{SO}(3)$, which corresponds to time translations and rotations. These isometries should be gauged because they are diffeomorphisms. It is natural to equip the observer with degrees of freedom that transform nontrivially under these isometries, such as a clock and a local reference frame. In other works \cite{CLPW,Witten:2021unn,Chandrasekaran:2022eqq,Jensen:2023yxy,Kudler-Flam:2023qfl,Kudler-Flam:2024psh,DeVuyst:2024fxc,DeVuyst:2024khu,Chen:2025tbh}, these degrees of freedom are quantized, leading to a crossed-product construction of the observer's algebra. In particular, the total Hilbert space consists of $\calh_0$ tensored with other Hilbert space factors arising from the quantization of the observer degrees of freedom. In the present setup, we choose to not quantize the observer degrees of freedom, so that the total Hilbert space is simply $\calh_0$. That is, we still equip the observer with a clock and local reference frame, but we freeze these degrees of freedom in a particular configuration that is taken to be a part of the classical background.\footnote{That is, we equip the observer with a perfect clock. One might object that this is unrealistic, because a perfect clock has infinite entropy, which would exceed the entropy of the observer's static patch. We will avoid this issue by working in a strict $G_N \rightarrow 0$ limit. In particular, we only need the observer's clock to be accurate for scrambling timescales, which are only logarithmic in the cosmological Gibbons-Hawking entropy. We will consider the case of an imperfect clock in Section \ref{sec:KMS-crossedprod}.} All of the isometries are broken. Thus, local operators in the matter QFT are in fact gauge invariant, because they are implicitly dressed to the observer. Another perspective on this implicit dressing is provided below in Section \ref{sec:observerrecoil}.

Following \cite{Penington:2025hrc}, we will first construct the observer's algebra in the matter quantum field theory, ignoring all quantum gravity effects. Then, we will explain how this algebra receives gravitational corrections. Consider the matter QFT with Hilbert space $\calh_0$. Because the observer has a cosmological horizon at $r = r_c$, we will define their algebra $\cala_0$ to be the algebra of QFT operators associated with the spacetime region $r_{\mathrm{Obs}} < r < r_c$. Equivalently, we could define $\cala_0$ to be the algebra of operators in a neighborhood of $r_{\mathrm{Obs}}$ and invoke the timelike tube theorem \cite{Strohmaier:2023opz}. We consider operators in $\cala_0$ that are boosted to either the far past or the far future. Let $H$ denote the Hermitian operator that generates the time translation isometry in \eqref{eq:SdS}, and let $T$ be a large parameter, of order the scrambling time. Given $\calo \in \cala_0$, we define
\begin{align}
    \calo^E &:= e^{-i H \frac{T}{2}} \calo e^{+i H \frac{T}{2}} \,, \\
    \calo^L &:= e^{+i H \frac{T}{2}} \calo e^{-i H \frac{T}{2}} \,.
\end{align}
That is, $\calo^E$ is an operator that is inserted at early times, and $\calo^L$ is an operator that is inserted at late times. Given a set of operators $\{ \calo_n \}_{n \in \mathbb{N}} \subset \cala_0$, consider the following correlator:
\begin{equation}
\braket{ \cdots \calo_4^E \, \calo_3^L \, \calo_2^E \, \calo_1^L}_{\mathrm{QFT}} \,,
\label{eq:ELcors}
\end{equation}
where the expectation value is in the GNS vacuum of $\calh_0$. The subscript $\mathrm{QFT}$ indicates that no gravitational corrections are included. We assume that in the matter QFT, \eqref{eq:ELcors} factorizes when $T$ is large:\footnote{That is, any connected contribution to the correlator that contains both early and late operators should decay to zero at large $T$. Physically, this is because a local perturbation near the observer will fall through their horizon. From the observer's perspective, the state returns to the GNS vacuum at late times. We will ignore subtleties that may arise when massless degrees of freedom are present. See \cite{Kudler-Flam:2025pol} for a more general discussion of infrared issues in dS.}
\begin{align}
    \label{eq:limfac}
    \lim_{T \rightarrow \infty} \braket{ \cdots \calo_4^E \, \calo_3^L \, \calo_2^E \, \calo_1^L}_{\mathrm{QFT}} = \braket{ \cdots \calo_4^E \, \calo_2^E  }_{\mathrm{QFT}}
    \braket{\cdots  \calo_3^L \, \calo_1^L }_{\mathrm{QFT}} \,.
\end{align}
As we will see, gravitational corrections to \eqref{eq:ELcors} result in a non-factorized answer, so that the result decays to zero at long times.

Because we are interested in the limit of large $T$, we choose to compute correlation functions of early- and late-time operators using a doubled Hilbert space, which naturally arises by applying the GNS construction to the factorized correlators in \eqref{eq:limfac}. We consider two copies of the matter QFT, which we denote by $A$ and $B$. The early-time operators belong to QFT $A$, while the late-time operators belong to QFT $B$. The total Hilbert space consists of two copies of $\calh_0$, which we refer to as $\calh_A$ and $\calh_B$,
\begin{align}
    \label{eq:Hspace}
    \calh = \calh_A \otimes \calh_B \,.
\end{align}
Let $\cala_{0,A}$ denote the copy of $\cala_0$ that acts on $\calh_A$, and let $\cala_{0,B}$ denote the copy of $\cala_0$ that acts on $\calh_B$.  Given $a \in \cala_{0,A}$ and $b \in \cala_{0,B}$, we define
\begin{align}
    a^E &:= e^{-i H_A \frac{T}{2}} \,a\, e^{+i H_A \frac{T}{2}} \,, \label{eq:aedef} \\
    b^L &:= e^{+i H_B \frac{T}{2}} \,b\, e^{-i H_B \frac{T}{2}} \,, \label{eq:bldef}
\end{align}
where $H_A$ and $H_B$ respectively denote the time translation generators in the $A$ and $B$ theories. The gravitational interaction between the $A$ and $B$ theories is captured by an S-matrix which we will define momentarily.

\begin{figure}[ht]
  \centering
  \includegraphics[width=0.8\textwidth]{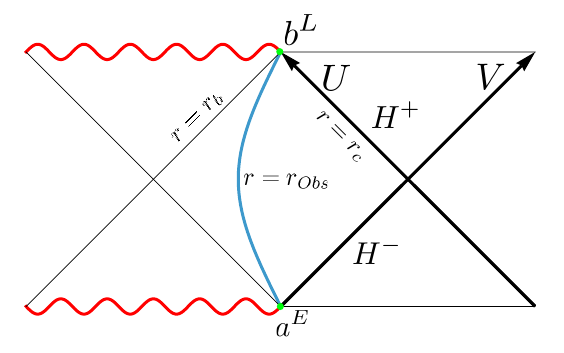}
  \caption{A Penrose diagram for the Schwarzschild-de-Sitter black hole. To construct the spacetime outside of an observer in de Sitter space, we introduce a boundary at $r = r_{\mathrm{Obs}}$, shown in blue. The observer's future cosmological horizon is given by $V = 0$, while the past horizon is at $U = 0$, where $U$ and $V$ are Kruskal coordinates that increase toward the future, as shown by the arrows. The early-time operator $a^E$ is inserted in the far past, and the late time operator $b^L$ is inserted in the far future, as indicated by the green dots. These operators create shockwaves along the past and future cosmological horizons, respectively. $H^+$ refers to the half-horizon $U > 0$, $V = 0$, while $H^-$ refers to $V < 0$, $U = 0$. }
  \label{fig:penrose}
\end{figure}

To study gravitational corrections to correlation functions of early- and late-time operators, we take the following double scaled limit:
\begin{align}
    G_N \rightarrow 0\,, \quad T \rightarrow \infty\,, \quad G_N e^{\frac{2 \pi T}{\beta_c}} \text{ fixed} \,.
    \label{eq:limit}
\end{align}
Consider operators in $\cala_{0,A}$ that are boosted to the far past. In the Penrose diagram in Figure \ref{fig:penrose}, these operators are inserted on the lower green dot. Due to the large boost, the stress-energy created by these operators is localized on the horizon at $U = 0$. We make an ansatz that the only non-zero component of the stress tensor is given by\footnote{For highly boosted states, this is the dominant component of the stress tensor. The other components may be neglected. As a consistency check, this ansatz is covariantly conserved.}
\begin{align}
    \label{eq:stresstensor}
    T_{UU} = - \frac{1}{r_c^2} P_U(\Omega) \, \delta (U) \,, 
\end{align}
where $P_U(\Omega)$, for every $\Omega \in S^2$, is an ANEC operator, defined by
\begin{equation}
P_U(\Omega) := - r_c^2  \int_{-\infty}^\infty \mathrm{d}U \, T_{UU}(U,V = 0,\Omega) \,. 
\label{eq:ANECdef}
\end{equation}
The above ANEC operators for all values of $\Omega \in S^2$ mutually commute because the stress tensors within two separate operators are spacelike separated \cite{Cordova:2018ygx}. Thus, the quantum state of a shockwave may be expressed in an eigenbasis that simultaneously diagonalizes these operators. The total null energy is given by
\begin{align}
    \label{eq:intnull}
    P_U := \int \mathrm{d}\Omega \, P_U(\Omega) \,,
\end{align}
where $\mathrm{d}\Omega$ is the measure of the round metric on $S^2$ with unit radius.

In the limit \eqref{eq:limit}, the stress tensor in \eqref{eq:stresstensor} sources a small perturbation to the metric in \eqref{eq:SdS}. The form of this perturbation was known since the work of \cite{Dray:1984ha,Aichelburg:1970dh,Penrose:1972xrn,tHooft:1987vrq,Sfetsos:1994xa}. Including the perturbation, the metric becomes
\begin{align}
    \mathrm{d}s^2 = \mathrm{d}s^2_{\mathrm{SdS}} + \delta(U) f_U(\Omega)\, \mathrm{d}U^2 \,,
\end{align}
where $f_U(\Omega)$ is a function on $S^2$. The linearized Einstein equation is
\begin{align}
    \label{eq:linEE}
    \left( - \nabla_{S^2}^2 + 1 - r_c^2 \Lambda \right)f_U(\Omega)  = - P_U(\Omega) \, 16 \pi G_N \,.
\end{align}
In the limit \eqref{eq:limit}, $P_U(\Omega)$ grows like $e^{\frac{\pi T}{\beta_c}}$. Thee right-hand side of \eqref{eq:linEE} is therefore of order $\sqrt{G_N}$, justifying the use of the linearized theory.

The above discussion for early-time operators can be repeated for late-time operators, with $U$ replaced by $V$ in equations \eqref{eq:stresstensor} through \eqref{eq:linEE}. Because the early- and late-time operators belong to different copies of a QFT, the ANEC operators $P_V(\Omega_1)$ and $P_U(\Omega_2)$ commute and may be simultaneously diagonalized. The metric with shockwaves along both the future and past horizons is given by
\begin{align}
    \mathrm{d}s^2 = \mathrm{d}s^2_{\mathrm{SdS}} + \delta(U) f_U(\Omega)\, \mathrm{d}U^2 + \delta(V) f_V(\Omega)\, \mathrm{d}V^2 \,,
\end{align}
up to higher order corrections that do not contribute to the on-shell action in the limit \eqref{eq:limit}. In this limit, the on-shell action, which receives contributions both from the Einstein-Hilbert action and the matter action, is
\begin{align}
    I = 16 \pi G_N \int \mathrm{d}\Omega \, P_V(\Omega) \frac{1}{\left(-  \nabla_{S^2}^2 + 1 - \Lambda r_c^2\right)} P_U(\Omega) \,.
    \label{eq:action}
\end{align}
As long as $r_c$ does not saturate the bounds in \eqref{eq:bounds}, the differential operator appearing in the denominator of the above equation has only non-zero eigenvalues and is thus invertible.

In the limit \eqref{eq:limit}, the metric perturbations vanish, yet $I$ is order one. Shockwave scattering is thus governed by the eikonal approximation, where in an appropriate basis, the \textit{out} state differs from the \textit{in} state by a phase. The S-matrix is simply
\begin{align}
    \cals := e^{i I} \,.
\end{align}
Consider now a sequence of operators $\{a_n \}_{n \in \mathbb{N}} \in \cala_{0,A}$ and $\{ b_n \}_{n \in \mathbb{N}} \in \cala_{0,B}$. Given a quantum field theory correlation function
\begin{align}
    \braket{\cdots a_1^E \, b_2^L \, a_3^E \, b_4^L}_{\mathrm{QFT}}
\end{align}
in the scaling limit \eqref{eq:limit}, the gravitational corrections may be incorporated by inserting $\cals$ or $\cals^{-1}$ between adjacent operators. That is, the quantum gravitational correlation functions of interest are computed as follows:
\begin{align}
    \label{eq:qgcorr}
    \braket{\cdots a_1^E \, b_2^L \, a_3^E \, b_4^L}_{\mathrm{QG}} := \lim_{\substack{
    T \rightarrow \infty \\
    G_N \rightarrow 0 \\
    G_N e^{\frac{2 \pi T}{\beta_c}} \text{ fixed }
    } }  \braket{\cdots \cals \, a_1^E \, \cals^{-1} \, b_2^L \, \cals \, a_3^E \, \cals^{-1} \, b_4^L}_{\mathrm{QFT}} \,.
\end{align}
We may think of \eqref{eq:qgcorr} as the definition of a quantum gravitational correlation function, which will be denoted using the subscript $\mathrm{QG}$. Each insertion of $\cals$ or $\cals^{-1}$ captures the gravitational interactions of the shockwaves that are created and/or annihilated by the early- and late-time operators. In particular, we should define an ``in'' Hilbert space of shockwave states, and an ``out'' Hilbert space of shockwave states, each of which are isomorphic to $\calh_A \otimes \calh_B$. The S-matrix $\cals$ is a unitary transformation from the ``in'' Hilbert space to the ``out'' Hilbert space. Early-time operators act on the ``in'' Hilbert space, and late-time operators act on the ``out'' Hilbert space. Thus, between an early- and late-time operator, there must always be an insertion of $\cals$ or $\cals^{-1}$. Note that the rightmost insertion of $\cals^{-1}$ in \eqref{eq:qgcorr} acts trivially.

Recall from \eqref{eq:aedef} and \eqref{eq:bldef} that the definitions of $a^E$ and $b^L$ involve conjugations by $e^{i H_A \frac{T}{2}}$ or $e^{-i H_B \frac{T}{2}}$. We can absorb these unitary operators into the S-matrix. Define a new S-matrix as follows:
\begin{align}
    S := e^{i H_A \frac{T}{2}} e^{-i H_B \frac{T}{2}} \,\cals \, e^{-i H_A \frac{T}{2}} e^{i H_B \frac{T}{2}} \,.
\end{align}
More explicitly, we have that
\begin{align}
    S = e^{i \delta}, \quad \quad \delta = 16 \pi G_N e^{\frac{2 \pi T}{\beta_c}} \int \mathrm{d}\Omega \, P_V(\Omega) \frac{1}{\left(-  \nabla_{S^2}^2 + 1 - \Lambda r_c^2\right)}  P_U(\Omega) \,.
    \label{eq:eikonalphase}
\end{align}
We can rewrite \eqref{eq:qgcorr} in terms of quantities that remain finite in the limit \eqref{eq:limit},
\begin{align}
    \braket{\cdots a^E_1 \, b^L_2 \, a^E_3 \, b^L_4}_{\mathrm{QG}} = \braket{\cdots S \, a_1 \, S^{-1} \, b_2 \, S \, a_3 \, S^{-1} \, b_4}_{\mathrm{QFT}} \,.
    \label{eq:qgcorr2}
\end{align}
Given $a \in \cala_{0,A}$ and $b \in \cala_{0,B}$, define
\begin{align}
\label{eq:tildeops}
    \tilde{a} := S^{1/2} \,a\, S^{-1/2}\,, \quad \quad \tilde{b} := S^{-1/2}\, b\, S^{1/2}\,.
\end{align}
Then we may simply write
\begin{align}
    \braket{\cdots a^E_1 \, b^L_2 \, a^E_3 \, b^L_4}_{\mathrm{QG}} = \braket{\cdots \tilde{a}_1 \, \tilde{b}_2 \, \tilde{a}_3 \, \tilde{b}_4}_{\mathrm{QFT}} \,.
    \label{eq:qgcorr3}
\end{align}
This shows that the quantum gravitational correlation functions of early- and late-time operators may be represented using the Hilbert space \eqref{eq:Hspace} and operators \eqref{eq:tildeops}. In the remainder of this paper, we will drop the subscripts, and all correlation functions will be in the QFT. To capture quantum gravity effects, we will simply ``dress'' operators with tildes, as in \eqref{eq:tildeops}.

We define the observer's algebra $\cala$ to be the algebra generated by $\tilde a$ and $\tilde b$ for all $a \in \cala_{0,A}$ and all $b \in \cala_{0,B}$. To better understand the physics of $\cala$, we will consider two separate limits beyond \eqref{eq:limit}. To do so, note that \eqref{eq:eikonalphase} contains two independent parameters, $G_N e^{\frac{2 \pi T}{\beta_c}}$ and $\Lambda r_c^2$, which are both held fixed in \eqref{eq:limit}. For each of the two limits, $r_c$ will approach one of its bounds in \eqref{eq:bounds}.


\subsubsection{Observer Recoil Effect} \label{sec:observerrecoil}

For the first limit, we take $G_N e^{\frac{2 \pi T}{\beta_c}} \rightarrow 0$ and $G_NM \rightarrow 0$, with their ratio held fixed. Equation \eqref{eq:eikonalphase} becomes:
\begin{equation}
\label{eq:recoilphase}
\delta = \frac{2 \ell_{\mathrm{dS}}}{ M} e^{\frac{2 \pi T}{\beta_c}} \int \mathrm{d}\Omega_1 \mathrm{d}\Omega_2 \, P_U(\Omega_1) P_V(\Omega_2) \cos \Theta_{12} \,,
\end{equation}
where $\Theta_{12}$ is the angle subtended from $\Omega_1$ to $\Omega_2$ on $S^2$. Because the Schwarzschild radius $G_NM$ goes to zero, the background geometry becomes global $\mathrm{dS}_4$ in this limit. The locus $r = 0$ is interpreted as the worldline of a pointlike observer with mass $M$, which is large in de Sitter units but small in Planck units. The eikonal phase in \eqref{eq:recoilphase} has the following interpretation. Suppose that the observer at $r = 0$ emits a shockwave at early times which travels along the $V$ axis, is localized at the North Pole of the internal $S^2$, and has null energy $P_U e^{\frac{2 \pi T}{\beta_c}}$. We will treat this shockwave as a source. Next, consider a probe shockwave along the $U$ axis. According to \eqref{eq:recoilphase}, as the probe shockwave crosses the source at $U = 0$, it is translated along the $V$ axis according to
\begin{equation}
    V \rightarrow V - \frac{2 \ell_{\mathrm{dS}}}{M} e^{\frac{2 \pi T}{\beta_c}} P_U  \cos \theta \,,
\end{equation}
where $\theta$ is the azimuthal angle of the probe. This supertranslation on the $U = 0$ horizon may be uniquely extended to an isometry of $\mathrm{dS}_4$, which we will call $\iota$. Instead of acting on the probe shock with $\iota$, we may instead work in a reference frame where the probe shock is not translated, but the worldline of the observer is transformed by $\iota^{-1}$. The new worldline represents an observer that has recoiled off the source shockwave. More details are provided in Appendix \ref{sec:recoileffect}. Thus, \eqref{eq:recoilphase} captures the physical fact that when an observer emits a shockwave, their trajectory changes. Due to the expansion of de Sitter space, a small perturbation of the observer's trajectory will grow exponentially, and this effect becomes important for out-of-time-ordered correlators on timescales of order $\log M$ in de Sitter units \cite{Kolchmeyer:2024fly}. This also explicitly demonstrates that the operators $\tilde{a}$ and $\tilde{b}$ are in fact dressed to the observer, because the momentum of the shockwaves that these operators create is compensated by an opposite change of the observer's momentum.


\subsubsection{Near-Nariai Limit} \label{sec:nearnariai}

For the second limit, we take $G_N e^{\frac{2 \pi T}{\beta_c}}\rightarrow 0$ and $r_c \rightarrow \frac{1}{\sqrt{\Lambda}}$ while holding fixed $e^{\frac{2 \pi T}{\beta_c}}/\Delta S_c$, where we define
\begin{equation}
    \Delta S_c := \frac{4 \pi r_c^2}{4  G_N} - \frac{ \pi }{ \Lambda G_N} \,.
\end{equation}
That is, $\Delta S_c$ is the difference of the entropy of the cosmic horizon and its minimum value for a Nariai black hole. The Nariai black hole geometry is $\mathrm{dS}_2 \times S_2$:
\begin{equation}
    \mathrm{d} s^2 = -\left(1 - \Lambda r^2\right) \mathrm{d}t^2 + \frac{\mathrm{d}r^2}{1 - \Lambda r^2} + \frac{1}{\Lambda} \, \mathrm{d}\Omega^2 \,.
\end{equation}
The roots of $1 - \Lambda r^2$, which correspond to the cosmic and black hole horizons, are $r_c = -r_b = \frac{1}{\sqrt{\Lambda}}$. We place the observer boundary at $r = r_{\mathrm{Obs}}$, where $r_b < r_{\mathrm{Obs}} < r_c$. Also, $\beta_c = \frac{2 \pi}{\sqrt{\Lambda}}$. The metric in Kruskal coordinates reads
\begin{equation}
    \mathrm{d}s^2 = -\frac{1}{\left(1 - \Lambda\, \frac{UV}{4}\right)^2}\,\mathrm{d}U \mathrm{d}V + \frac{1}{\Lambda} \, \mathrm{d}\Omega^2 \,,
\end{equation}
with conventions specified by Figure \ref{fig:penrose}. We define $H^+$ to be the half-horizon given by $V = 0$, $U > 0$. We define $H^-$ to be the half-horizon given by $U = 0$, $V < 0$. The eikonal phase becomes
\begin{equation}
    \delta = - \frac{4 \pi}{\Delta S_c} \frac{e^{\frac{2 \pi T}{\beta_c}}}{\Lambda} P_U P_V,
    \label{eq:2deikonal}
\end{equation}
which now depends only on the total null energies $P_U$ and $P_V$, defined in \eqref{eq:intnull}. This means that non-trivial scattering only occurs in the $s$-wave sector. In contrast, away from this limit, the S-matrix mixes sectors with different angular momenta. In the near-Nariai limit, we may KK reduce the matter theory to two dimensions, and the resulting two-dimensional theory is semiclassical JT de Sitter gravity minimally coupled to matter \cite{Maldacena:2019cbz, Cotler:2019nbi}. The phase in \eqref{eq:2deikonal} differs from that of AdS JT gravity \cite{Maldacena:2016upp} by a minus sign, which has important consequences that we will explore throughout this paper.

In the remainder of this paper, we will work in the near-Nariai limit to take advantage of the simplicity of two dimensions. We will also make an assumption about the two-dimensional matter theory. Recall that the algebra $\cala_0$ was defined to be the QFT algebra associated to the region $r_{\mathrm{Obs}} < r < r_c$. We will assume that $\cala_0$ is equivalent to both $\cala_0(H^+)$ and $\cala_0(H^-)$, where we define $\cala_0(H^+)$ to be the algebra generated by polynomials of the matter fields smeared within $H^+$, and $\cala_0(H^-)$ is defined analogously. In particular, the fields are {\it not} smeared in a neighborhood of $H^+$ or $H^-$. They are smeared strictly on $H^+$ or $H^-$. In \cite{Wall:2011hj}, this assumption is referred to as ``the existence of a null hyperspace formalism'' and was used to prove the generalized second law. This assumption was also used in \cite{Penington:2025hrc}.\footnote{See the discussion surrounding equation (2.3). Note that their horizon algebra $\hat \cala$ only contains operators that are compactly supported on the horizon. $\hat \cala$ is not a von Neumann algebra, but it is assumed to be s.o.t. dense in the von Neumann algebra of all observables in the right exterior of the black hole. To study the free product limit of the black hole scrambling algebra, the authors show that correlators of operators in $\hat \cala$ obey cluster decomposition with respect to modular translation. We will use this result in Section \ref{sec:alg-free}.} It is valid for free theories, and its general validity is discussed further in \cite{Wall:2011hj}. In Appendix \ref{sec:app3D}, we explicitly check this assumption for the case that the two-dimensional matter theory is the dimensional reduction of a three-dimensional free massive scalar field. We will moreover assume that the matter theory has a UV CFT fixed point, so that we can take $\cala_0(H^+)$ and $\cala_0(H^-)$ to be algebras in a 2D chiral CFT. In the example given in Appendix \ref{sec:app3D}, the CFT is a free chiral boson $\phi$, and the algebra is generated by the current $\partial \phi$. The motivation for this assumption is that general properties of the algebra $\cala$ may be inferred from properties of the matter fields on the horizon, which are easier to study.


\subsection{Simplified Model} \label{sec:2D-model}

By working in the near-Nariai limit and assuming the existence of a null hyperspace formalism, we have transformed our original four-dimensional setup into a two-dimensional one. We now provide more details on our two-dimensional model, beginning with some basic facts and conventions on 2D chiral CFTs.


\subsubsection{2D Chiral CFT}

Let $U$ be an affine coordinate along a horizon in two dimensions. Each primary operator\footnote{We are referring to $\mathrm{SL}(2)$ primaries, \emph{not} Virasoro primaries. Our primaries would be called ``quasi-primaries'' elsewhere. One could also use the term ``non-derivative field'' as in \cite{Mack1977convergence}.} $\phi(U)$ has a scaling dimension $h \in \mathbb{Z}_{\ge 0}$.\footnote{For simplicity, we restrict our attention to bosonic theories so we do not have to consider anticommuting fields.} The identity is the only primary with $h = 0$. There are Hermitian operators $P$, $K$, $B$ that implement the following transformations:
\begin{alignat}{2}
    &e^{i x P} \phi(U) e^{- i x P} &&= \phi(U - x)\,, \label{eq:Ptransf} \\
    &e^{i x K} \phi(U) e^{-i x K} &&= \Big(\frac{1}{1 - U x}\Big)^{2 h} \phi\Big(\frac{U}{1 - U x} \Big) \,, \label{eq:Ktransf} \\
    &e^{i x B} \phi(U) e^{-i x B} &&= e^{- h x}  \phi(e^{-x} U) \,. \label{eq:Btransf}
\end{alignat}
These satisfy the commutation relations
\begin{align}
    [B,P] = -i\,P\,, \qquad [B,K] = i\,K \,, \qquad [P,K] = -2i \,B \,.
\end{align}
The space of primary operators with a given scaling dimension is a vector space, closed under Hermitian conjugation. We can choose a basis $\phi_i(U)$ of the space of all primaries so that the basis operators are Hermitian.\footnote{The subscript $i$ refers to a basis vector, which includes the identity operator, and we will suppress this subscript when doing so does not cause confusion.} There is a normalizable vacuum vector $\ket{\Omega}$ that is invariant under the symmetries. The generator $P$ has a negative-semidefinite spectrum. Null-separated operators commute,
\begin{align}
    [ \phi_1(U_1) , \phi_2(U_2) ] = 0, \qquad \forall\, U_1 \neq U_2 \,.
\end{align}
We can choose the basis $\phi_i(U)$ such that the vacuum two-point function is given by
\begin{align}
    \braket{  \phi_i(U_1)  \phi_j(U_2) } = \delta_{ij} \left(- \frac{1}{(U_1 - U_2 - i \epsilon)^2} \right)^{h_i} \,.
\end{align}
According to the state-operator correspondence, the Hilbert space decomposes into sectors, where each sector is in one-to-one correspondence with one of the basis primaries $\phi_i$. In particular, $\phi_i(U) \ket{\Omega}$ is in the $i$th sector. We will take a complete insertion of the identity to be
\begin{align}
    \mathds{1} = \ket{\Omega} \bra{\Omega} + \sum^\prime_i \int_{-\infty}^0 \frac{\dbar P}{-P} |P\rangle_i \, {}_i\langle P|\, ,
    \label{eq:idres}
\end{align}
where $\dbar P:=\mathrm{d}P/(2\pi)$, and the prime on the sum indicates that $i$ refers to any member of the basis other than the identity. The state $|P\rangle_i$ refers to a delta-function normalizable eigenstate of $P$ within the sector labeled by $i$. If we insert \eqref{eq:idres} into the two-point function and use \eqref{eq:Ptransf}, we may derive that
\begin{align}
    \sum_j \int_{-\infty}^0 \frac{\dbar P}{-P} e^{-i x P}  \langle \phi_i(U_1) |P\rangle_j \, {}_j\langle P| \phi_i(U_2) \rangle = \left(- \frac{1}{(U_1 - U_2 - i \epsilon - x)^2} \right)^{h_i} \,,
\end{align}
and it follows that
\begin{align}
    {}_i\langle P| \phi_j(U) \rangle = \delta_{ij}  \frac{e^{-i P U}}{\sqrt{\Gamma(2 h_i)}} (-P)^{h_i} \theta(-P) \,, \label{eq:2D-PU-matrixel}
\end{align}
where $\theta(P)$ is the Heaviside step function.

\vspace{1em}

\noindent We may also define a basis that diagonalizes the $B$ generator. We define
\begin{align}
    {}_i \langle P | \lambda\rangle_j  := \delta_{ij} \sqrt{2\pi} (-P)^{-i \lambda} \theta(-P) \label{eq:2D-Plambda-matrixel}
\end{align}
such that
\begin{align}
    {}_i\langle \lambda_1|\lambda_2\rangle_j &= \delta_{ij} \, 2\pi \delta(\lambda_1 - \lambda_2) \,,
    \\
    \sum_j \int_{-\infty}^\infty \dbar\lambda \ {}_i\langle P_1|\lambda\rangle_j \, {}_j\langle\lambda|P_2\rangle_i  &= (-P_1) \theta(-P_1) \, 2\pi \delta( P_1 - P_2 ) = {}_i\langle P_1 | P_2\rangle{}_i \,,
    \\
    B |\lambda\rangle_i &= \lambda |\lambda\rangle_i \,.
\end{align}

It is possible to act on the vacuum $\ket{\Omega}$ with a primary field $\phi_i(U)$ smeared within $U > 0$ to obtain $|\lambda\rangle_i$ for any $\lambda \in \mathbb{R}$. This is equivalent to the Reeh-Schlieder theorem\footnote{See \cite{Witten:2018zxz} for a review.} \cite{ReehSchlieder1961}, which states that the algebra generated by polynomials of smeared fields with support on $U > 0$ generates a dense subspace of the Hilbert space when acting on the vacuum. To see this, define the smeared primary field
\begin{align}
    \varphi^+_i(\lambda) &:= \int_0^\infty \mathrm{d}U \, g^+_{h_i,\lambda}(U) \, \phi_i(U) \,, \label{eq:2D-varphiA-lambdatoU} \\
    g^+_{h,\lambda}(U) &:= U^{h-1 + i \lambda} e^{\frac{\pi \lambda}{2}} e^{-\frac{i\pi h}{2}} \frac{\sqrt{ 2\pi \Gamma(2 h) }}{\Gamma(h + i \lambda)} \,.
\end{align}
It follows that
\begin{equation}
    \varphi^+_i(\lambda) |\Omega\rangle = |\lambda\rangle_i, \qquad 
    \langle\Omega| \varphi^+_i(\lambda) = e^{-i \pi h} e^{\pi \lambda} \, {}_i\langle -\lambda| \,. \label{eq:2D-varphiA-lambda-vac}
\end{equation}
Likewise, we could also smear on the half-horizon $U < 0$. The analogous formulas are
\begin{align}
    \varphi^-_i(\lambda) &:= \int_{-\infty}^0 \mathrm{d}U \, g^-_{h_i,\lambda}(U) \, \phi_i(U) \,, \label{eq:2D-varphiB-lambdatoU} \\
    g^-_{h,\lambda}(U) &:= (-U)^{h-1 + i \lambda} e^{-\frac{\pi \lambda}{2}} e^{\frac{i\pi h}{2}} \frac{\sqrt{2\pi \Gamma(2 h)}}{\Gamma(h + i \lambda)} \,.
\end{align}
These satisfy
\begin{align}
    \varphi^-_i(\lambda) \ket{\Omega} &= \ket{\lambda}_i, \qquad 
    \langle \Omega\ \varphi^-_i(\lambda) = e^{ i \pi h} e^{- \pi \lambda}\, {}_i\langle -\lambda| \,. \label{eq:2D-varphiB-lambda-vac}
\end{align}
For $U \in \mathbb{R}$, we also have that 
\begin{align}
    \langle\Omega| \phi_i(U) | \lambda\rangle_j &= \delta_{ij} \, (U - i\epsilon)^{-h + i\lambda} \, e^{-\frac{i\pi h}{2} } e^{-\frac{\pi \lambda}{2}} \, \frac{\Gamma(h - i\lambda)}{\sqrt{ 2\pi \Gamma(2 h)}} \,, \\
    {}_j\langle \lambda|\phi_i(U)|\Omega\rangle &= \delta_{ij} \, (U + i\epsilon)^{-h - i\lambda}
    \, e^{+\frac{i \pi h}{2}} e^{-\frac{\pi \lambda}{2}} \, \frac{\Gamma(h + i\lambda)}{\sqrt{2\pi \Gamma(2 h)}} \,.
\end{align}
The three-point function is given by
\begin{align}
    \braket{ \phi_1(U_{1})  \phi_2 (U_2)    \phi_3 (U_{3}) }  =  C_{123} \frac{(-i)^{h_1 + h_2 + h_3 }}{(U_{12} - i \epsilon)^{h_1 + h_2 - h_3} (U_{13}- i \epsilon)^{h_1 + h_3 - h_2} (U_{23} - i \epsilon)^{h_2 + h_3 - h_1} } \,,
\label{eq:threepoint}
\end{align}
where the CFT structure constants obey
\begin{align}
    C_{ijk}^* = C_{kji}\,, \qquad\qquad C_{123} = C_{132} (-1)^{h_2 + h_3 + h_1} = C_{312} \,.
\end{align}

We emphasize that the field $\phi_i(U)$ is not an operator. Rather, it is an operator-valued distribution, and correlators such as \eqref{eq:threepoint} are distributions. The limit $\epsilon \rightarrow 0$ may be taken after smearing the fields against test functions.

\subsubsection{Early- and Late-Time CFTs} \label{sec:2D-AB-theory}

In the previous subsection, we reviewed states, operators, and correlation functions in chiral CFTs. We now consider two copies of a chiral CFT, which we will refer to as CFT A and CFT B. Following the conventions of Section \ref{sec:4dmodel}, we interpret CFT A as a theory of early-time modes, while CFT B describes late-time modes. We again refer to Figure \ref{fig:penrose} to establish our conventions. Primary fields in CFT A are denoted by $\phi_i^A(U)$, while primary fields in CFT B are denoted by $\phi_i^B(V)$. The algebra $\cala_{0,A}$ is defined to be generated by polynomials of the fields $\phi_i^A(U)$ smeared within $U > 0$, and $\cala_{0,B}$ is defined to be generated by polynomials of the fields $\phi_i^B(V)$ smeared within $V < 0$. In these definitions, we have invoked the null hyperspace formalism:
\begin{align}
    \cala_{0,A} \equiv \cala_{0,A}(H^+), \quad \quad     \cala_{0,B} \equiv \cala_{0,B}(H^-) \,.
\end{align}
The S-matrix is $e^{i \delta}$ where $\delta$ is given by \eqref{eq:2deikonal}. We will simply write
\begin{align}
    S = e^{i x P_A P_B} \,,
\end{align}
where $x < 0$ is a parameter, and $P_A$ and $P_B$ are the $P$ generators in CFTs A and B, respectively. Following \eqref{eq:tildeops}, we define
\begin{align}
    \tilde \phi^A(U) &:= S^{1/2} \phi^A(U) S^{-1/2}, \quad \quad U > 0 \,, \\
    \tilde \phi^B(V) &:= S^{-1/2} \phi^B(V) S^{1/2}, \quad \quad V < 0 \,.
\end{align}
The algebra $\cala$ is defined to be generated by polynomials of smeared $\tilde \phi^A(U)$ and $\tilde \phi^B(V)$ fields with support on $U > 0$ and $V  < 0$ respectively.

The Hamiltonians in \eqref{eq:aedef} and \eqref{eq:bldef} that define the time translation automorphisms on $\cala_{0,A}$ and $\cala_{0,B}$ are related to the boost operators $B_A$ and $B_B$ as follows:
\begin{align}
    H_A := -B_A \,, \quad \quad H_B := B_B \,.
\end{align}
One may explicitly check that $S$ commutes with the total Hamiltonian $H_A + H_B$, which implies that quantum gravitational correlation functions of operators in $\cala$ are invariant under time translations. We will sometimes find it convenient to consider operators labeled by time, instead of an affine coordinate on the horizon. We define
\begin{alignat}{2}
    \calo^A(t) &:= \left( e^{t}\right)^h  \phi^A( e^{t}) &&= e^{-i t B_A} \phi^A(1) e^{ i t B_A} \,, \label{eq:2D-A-conformal} \\
    \calo^B(t) &:= \left( e^{-t}\right)^h  \phi^B(- e^{-t}) &&= e^{i t B_B} \phi^B(-1) e^{- i t B_B} \,. \label{eq:2D-B-conformal}
\end{alignat}
The gravitationally dressed fields $\tilde \calo^A(t)$ and $\tilde \calo^B(t)$ are also defined using \eqref{eq:tildeops}. Frequency domain operators are defined through Fourier conjugation:
\begin{equation}
    \Phi^{A,B}(\lambda) := \int \mathrm{d}t \, e^{i\lambda t} \calo^{A,B}(t) \,. \label{eq:frequencydomainopdef}
\end{equation}
Their dressed counterparts $\tilde \Phi^{A,B}(\lambda)$ are again defined using \eqref{eq:tildeops}. The frequency domain operators are related to the earlier $\varphi^{A,+}(\lambda)$, $\varphi^{B,-}(\lambda)$ fields (\ref{eq:2D-varphiA-lambdatoU}), (\ref{eq:2D-varphiB-lambdatoU}) by the rescalings\footnote{Since the $A$ and $B$ theories are defined on the positive and negative half-lines, we drop the $\pm$'s for convenience.}
\begin{align}
    \Phi^A(\lambda)  &= \langle\lambda|\phi^A(1)|\Omega_A\rangle \, \varphi^A(\lambda), \qquad \Phi^B(\lambda) = \langle-\lambda|\phi^B(-1)|\Omega_B\rangle \, \varphi^B(-\lambda) \,. \label{eq:Phitovarphi}
\end{align}
Later in this paper, we will use the following identity:
\begin{equation}
    e^{i x P_A P_B} |\lambda_A\rangle_i |\lambda_B\rangle_j = \int_{-\infty}^\infty \dbar\lambda  \, 
    F_x( \lambda) |\lambda_A +  \lambda\rangle_i |\lambda_B +\lambda\rangle_j \,, \label{eq:2D-PAPB-lambda-basis}
\end{equation}
where the function $F_{x}(\lambda)$ is defined by
\begin{equation}
    F_x( \lambda)  := \Gamma\!\big(\epsilon +i\lambda \big)\, e^{-\frac{\pi\lambda}{2} \text{sign} (x)}\, |x|^{\, -i\lambda} \,. \label{eq:Fx-function}
\end{equation}
Its Fourier representation is given by\footnote{To derive this Fourier transform, start from
\begin{equation}
    \int_{-\infty}^\infty \dbar\lambda \, \Gamma(h - i \lambda) e^{i k \lambda} = e^{ h k} e^{-e^{k}} \,,
\end{equation}
which is convergent for $k$ in the strip $-\frac{\pi}{2} < \text{Im } k < \frac{\pi}{2}$, and then analytically continue $k$ to the boundary of the strip.}
\begin{equation}
    \int_{-\infty}^\infty \dbar\lambda \, e^{i k \lambda} \, F_x(\lambda) = e^{i x e^{-k}} \,.
\end{equation}


\subsubsection{Scrambling and Anti-Scrambling}

\label{sec:antiscrambling}

In our two-dimensional model of an observer in de Sitter space, the parameter $x$ is negative. For the case that $x$ is positive, the algebra $\cala$, which has been extensively studied by Penington and Tabor \cite{Penington:2025hrc}, may be interpreted as the algebra of an asymptotically AdS boundary in an eternal black hole spacetime. In this case, $x = G_N e^{\frac{2 \pi T}{\beta}}$. This algebra has a microscopic origin provided by AdS/CFT duality. It is the large-$N$ limit of the algebra of single-trace operators in the thermofield double state above the Hawking-Page temperature \cite{Leutheusser:2021qhd, Leutheusser:2021frk, Leutheusser:2022bgi}, where the time separation between any two operators can be either the scrambling time (up to additive order one constants) or of order one \cite{Chandrasekaran:2022eqq}. The correlation functions of operators in $\cala$ thus capture the physics of scrambling in the boundary quantum mechanics.

Motivated by proposals for worldline de Sitter holography \cite{Anninos:2011af, Anninos:2011zn, Nakayama:2011qh,Maldacena:2024spf,Chen:2025jqm,Witten:2023xze,Tietto:2025oxn,Narovlansky:2023lfz,Blommaert:2026ofx,Goto:2026ipq,Narovlansky:2025tpb}, one may speculate that $\cala$ also has a  quantum-mechanical dual for $x < 0$. However, the physics of a putative dual would substantially differ from the physics of quantum mechanical systems with AdS duals. In particular, while systems dual to AdS are known to display fast scrambling \cite{Sekino:2008he,Lashkari:2011yi,Shenker:2013pqa,Maldacena:2015waa}, a putative dual of dS would instead display ``anti-scrambling''. Scrambling is dual to gravitational time delays, while anti-scrambling is dual to gravitational time advances.

To be more precise, consider a precursor state in a quantum mechanical system, which may be constructed by perturbing a thermal state at time $t_0$ and evolving backwards by a scrambling time. The information content of this perturbation is scrambled throughout the system. Next, apply a new perturbation to the precursor and evolve forward in time. In a system with scrambling dynamics, the original perturbation will re-materialize after $t_0$, if at all. In contrast, ``anti-scrambling'' refers to the counterintuitive phenomenon where the original perturbation re-materializes earlier than $t_0$. The relation to gravitational time delay and advance may be seen in Figure \ref{fig:scrambling}.

\begin{figure}[ht]
    \centering
    \begin{minipage}[t]{0.45\linewidth}
        \centering
        \includegraphics[height=4cm]{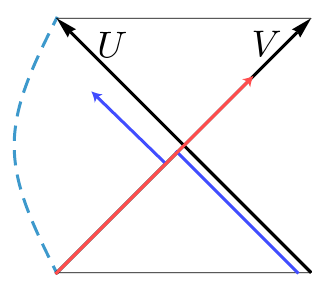}\\[4pt]
        {\small dS$_2$, $ \quad x < 0$
        \\
        ``anti-scrambling''}
    \end{minipage}\hfill
    \begin{minipage}[t]{0.45\linewidth}
        \centering
        \includegraphics[height=4cm]{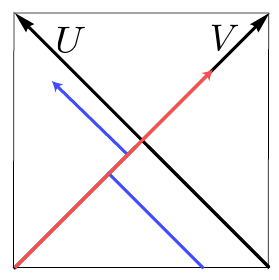}\\[4pt]
        {\small AdS$_2$, \, $  x > 0$ \\ ``scrambling''}
    \end{minipage}
    \caption{Left: Part of the Penrose diagram for $\mathrm{dS}_2$. The dashed line denotes the location of an observer boundary. The blue shock, which will eventually intercept the observer, represents a precursor state. Upon crossing the red shock, which represents an early-time perturbation from the observer, the blue shock experiences a time advance and reaches the observer at an earlier time. Right: The Penrose diagram of an eternal $\mathrm{AdS}_2$ black hole. The red shock causes the blue shock to arrive at the left AdS boundary at a later time.}
    \label{fig:scrambling}
\end{figure}

The presence of scrambling or anti-scrambling dynamics can be diagnosed from the sign of the leading correction to the OTOC during the onset of chaos. In our two-dimensional model, the OTOC in this regime is given by\footnote{A general ansatz for the early-time behavior of an OTOC is given in \cite{Kitaev:2017awl,Gu:2021xaj}.}
\begin{align}
    &\braket{\tilde \calo^A(t_1) \tilde \calo^B(t_2) \tilde \calo^A(t_3) \tilde \calo^B(t_4)} \nonumber \\
    &\hspace{5mm}= \frac{(-1)^{h_A+h_B}}
    {\left(2\sinh\frac{t_1-t_3}{2}\right)^{2h_A}
     \left(2\sinh\frac{t_2-t_4}{2}\right)^{2h_B}}
    \bigg[\,1
    -\,i\,x\,h_A h_B\,
    \frac{e^{\frac{1}{2}\left(t_2+t_4-t_1-t_3\right)}}
    {\sinh\frac{t_1-t_3}{2}\,\sinh\frac{t_2-t_4}{2}} + \calo(x^2)\,\bigg] \,,
    \label{eq:272}
\end{align}
where $x = c e^T$, and $c \in \mathbb{R}$ is small while $T$ is large. The time-ordered correlators are given by \eqref{eq:272} at $x = 0$. Assuming that these correlators are dual to correlators in a unitary quantum mechanical system at inverse temperature $2 \pi$, one may show that $c$ must be positive \cite{Maldacena:2015waa}. In particular, the duality implies that the OTOC must be analytic in the domain 
\begin{align}
    \mathrm{Im}\, t_4 > \mathrm{Im}\, t_3 > \mathrm{Im}\, t_2 > \mathrm{Im}\, t_1 > \mathrm{Im}\, t_4 - 2 \pi \,, \label{eq:OTOC-analyticity}
\end{align}
and obey the KMS condition. The authors of \cite{Maldacena:2015waa} set 
\begin{align}
    \label{eq:settings}
    t_1 = -t - i \pi \,, \qquad t_2 = t - \frac{i\pi}{2} \,, \qquad 
    t_3 = -t \,, \qquad t_4 = t + \frac{i\pi}{2} \,,
\end{align}
and studied the analytic continuation of the OTOC in $t$. The OTOC is analytic for $- \frac{\pi}{2} < \mathrm{Im}\, t < \frac{\pi}{2}$, and \cite{Maldacena:2015waa} showed that its magnitude is bounded in this strip by its value at $t = - \infty$, or $x = 0$.  Due to the limit \eqref{eq:limit} we used to compute the OTOC, this time is much earlier than the scrambling time and much later than the dissipation time.  After plugging in \eqref{eq:settings}, the bracketed term in \eqref{eq:272} becomes
\begin{align}
    1 - c\,e^{T}\,h_A h_B + \mathcal{O}\!\left(c^2 e^{2T}\right) \,.
\end{align} 
Because the leading correction can only decrease the magnitude of the OTOC, it follows that $c > 0$, or that anti-scrambling dynamics is ruled out.

If we do not impose any conditions on the behavior of correlation functions in the complex plane (such as the KMS condition), then there is no contradiction between anti-scrambling and unitary quantum mechanics. In fact, the traversable wormhole \cite{Gao:2016bin,Maldacena:2017axo} provides a concrete example of anti-scrambling. Consider two copies of a holographic CFT in the TFD state above the Hawking-Page temperature, so that the dual geometry is a two-sided AdS black hole. Create a precursor state by perturbing the right CFT at time $t = t_0 > 0$ and evolving backwards in time to $t = -t_0$. Assume that $t_0$ is of order the scrambling time. It was shown in \cite{Gao:2016bin,Maldacena:2017axo}  that by perturbing the precursor by a two-sided operator of the form $e^{i g O_L(0) O_R(0)}$ at $t = 0$, the original perturbation can experience a time advance.

The KMS condition is the key assumption of \cite{Maldacena:2015waa} that is violated by the dS OTOC. We discuss the KMS condition in detail in the next section. The implications for the observer's algebra are discussed in Section \ref{sec:properties}.


\section{Correlation Functions} \label{sec:OTOC}
In this section we specialize to the 2D theory, where correlators at any $x$ can be computed exactly. In Section \ref{sec:2D-OTOC-realt} we study OTOCs at real times and frequencies, in S-matrix perturbation theory and then exactly. To properly take the $x \rightarrow 0$ limit of the exact expressions, it is important to smear the operators with test functions. This procedure yields a well-defined perturbative expansion in $x$, agreeing to all orders with S-matrix perturbation theory. In Section \ref{sec:2D-OTOC-analc} we discuss the analytic continuation of OTOCs. Both the KMS condition and the detailed balance condition (its frequency-space analog) hold for $x>0$ and are violated for $x<0$. As we will explain, this is a consequence of the time advance which is induced by the gravitational S-matrix. A subtle point is that in the $x \rightarrow 0^-$ limit, the detailed balance condition holds, but the KMS condition may be violated.\footnote{To be clear, it is not possible to detect a violation of the KMS condition by expanding the gravitational S-matrix to a given power in $x$. However, after smearing the exact OTOC with test functions and performing the analytic continuation which is needed to verify the KMS condition, the $x \rightarrow 0^-$ limit does not necessarily agree with the result of applying the same analytic continuation to the smeared $x = 0$ OTOC. See Proposition \ref{eq:prop} below. This result does not contradict the $x \rightarrow 0^-$ behavior of the detailed balance condition because of a subtlety involving the convergence of the Fourier transform from the frequency-domain OTOC to the analytically continued time-domain OTOC. More details are provided beginning from Equation \eqref{eq:2D-G1-Fourier-analc}.} Finally, in Section \ref{sec:KMS-crossedprod} we apply the crossed product construction to $\cala$. When $x<0$, the Hartle-Hawking state is not a trace. However, the Hartle-Hawking state is a trace to all orders in a power series expansion in $x$.

\subsection{Correlators at Real Times} \label{sec:2D-OTOC-realt}
When $x = 0$, all gravitational and observer recoil effects are suppressed. We define the QFT correlation functions of our 2D theory:
\begin{align}
    G^{\mathrm{QFT}}_1(U_1, V_2, U_3, V_4) &:= \langle \phi^A(U_1) \phi^B(V_2) \phi^A(U_3) \phi^B(V_4) \rangle = \frac{1}{(2\pi)^2} \frac{(-1)^{h_A+h_B}}{(U_{13}-i\epsilon)^{2h_A} (V_{24}-i\epsilon)^{2h_B}} \,, \label{eq:2D-G1QFT-UV-def} \\
    G^{\mathrm{QFT}}_2(V_2, U_3, V_4, U_1) &:= \langle \phi^B(V_2) \phi^A(U_3) \phi^B(V_4) \phi^A(U_1) \rangle = \frac{1}{(2\pi)^2} \frac{(-1)^{h_A+h_B}}{(U_{13}+i\epsilon)^{2h_A} (V_{24}-i\epsilon)^{2h_B}} \,, \label{eq:2D-G2QFT-UV-def}
\end{align}
where $U_{13} := U_1 - U_3$ and $V_{24} := V_2 - V_4$. The OTOCs simply factorize into a product of two-point functions, since there is no coupling between the $A$ and $B$ theories.

\vspace{1em}

\noindent Now we consider finite $x$. The gravitational S-matrix acts trivially in any correlation function of three or fewer fields, since $P_A|\Omega\rangle = P_B|\Omega\rangle = 0$. The simplest non-trivial correlators are the four-point OTOCs: 
\begin{align}
    G_1(U_1, V_2, U_3, V_4) &:= \langle \tilde{\phi}^A(U_1) \tilde{\phi}^B(V_2) \tilde{\phi}^A(U_3) \tilde{\phi}^B(V_4) \rangle \label{eq:2D-G1-UV-def} \\
    &\,= \langle \phi^A(U_1) \phi^B(V_2) e^{ix P_AP_B} \phi^A(U_3) \phi^B(V_4) \rangle \,, \\
    G_2(V_2, U_3, V_4, U_1) &:= \langle \tilde{\phi}^B(V_2) \tilde{\phi}^A(U_3) \tilde{\phi}^B(V_4) \tilde{\phi}^A(U_1) \rangle \label{eq:2D-G2-UV-def} \\
    &\,= \langle \phi^B(V_2) \phi^A(U_3) e^{-ix P_A P_B} \phi^B(V_4) \phi^A(U_1) \rangle \,.
\end{align}
The simplest way to compute correlation functions is to take $x \ll 1$ and use S-matrix perturbation theory. This follows from (\ref{eq:2D-G1-UV-def}), (\ref{eq:2D-G2-UV-def}) by truncating $S = e^{ix P_A P_B}$ to some fixed order $x^N$, inserting a complete basis of $P_AP_B$ eigenstates, and computing the integrals using (\ref{eq:2D-PU-matrixel}). The $A$ and $B$ theories decouple at each order to yield the following perturbative expansion.
\begin{align}
    G_1^{\mathrm{pert},N}(U_1, V_2, U_3, V_4) &= G_1^{\mathrm{QFT}}(U_1, V_2, U_3, V_4) \label{eq:2D-G1-UV-Spert} \\
    &\hspace{5mm} \times \Big[ 1 + \sum_{n=1}^{N} \Big(\frac{-ix}{(U_{13}-i\epsilon)(V_{24}-i\epsilon)}\Big)^n \frac{(2h_A)_n (2h_B)_n}{\Gamma(1+n)} + \mathcal{O}(x^{N+1}) \Big] \,. \nonumber
\end{align}
There is a similar expression for $G_2^{\mathrm{pert},N}$ by replacing $U_{13}-i\epsilon \to -U_{13}-i\epsilon$ and $x\to -x$. It is tempting to take the $N\to \infty$ limit, but (\ref{eq:2D-G1-UV-Spert}) defines an asymptotic series which does not converge for any $x>0$ (these are known in the literature as ${}_2F_0$'s). That is, the all-orders series expansion of $S$ does not commute with the integrals over $\mathrm{Spec}\,P_A$, $\mathrm{Spec}\,P_B$. In Appendix \ref{sec:2D-Borel} we show the asymptotic series can be Borel resummed to recover the exact correlators of subsequent interest.


\subsubsection{Exact Correlators}

For the 2D model of Section \ref{sec:2D-model}, it is possible to compute the exact OTOCs in the time and frequency domains. We present the results here, leaving the derivations to Appendix \ref{sec:2Dcalcs-OTOCs-exact}. In position space it is most convenient to use Kruskal coordinates.\footnote{To obtain the OTOCs in static-patch time, we simply substitute $U_1=e^{t_1}$, $U_3 = e^{t_3}$, $V_2=-e^{-t_2}$, $V_4 = -e^{t_4}$, and multiply by the conformal factors in (\ref{eq:2D-A-conformal})--(\ref{eq:2D-B-conformal}).} Continuing with the conventions in Section \ref{sec:2D-AB-theory}, we find:
\begin{align}
    G_1(U_1, V_2, U_3, V_4) &= \langle \tilde{\phi}^A(U_1) \tilde{\phi}^B(V_2) \tilde{\phi}^A(U_3) \tilde{\phi}^B(V_4) \rangle \nonumber \\
    &= \frac{(-1)^{h_A}}{(2\pi)^2} x^{-2h_B} (U_{13})^{2h_B-2h_A} U^{\sigma_1} \Big(2h_B, 1-2h_A+2h_B; -\frac{i}{x}U_{13}V_{24}\Big)\,, \label{eq:2D-G1-UV-explicit} \\
    \sigma_1 &=
    \begin{cases}
        +\,, & x<0 \cap U_{13}<0 \cap V_{24}<0 \\
        0\,, & (x>0 \cap (U_{13}<0 \cup V_{24}<0)) \cup (x<0 \cap (U_{13}>0 \cup V_{24}>0)) \\
        -\,, & x>0 \cap U_{13}>0 \cap V_{24}>0
    \end{cases} \,, \nonumber \\
    \nonumber \\
    G_2(V_2, U_3, V_4, U_1) &= \langle \tilde{\phi}^B(V_2) \tilde{\phi}^A(U_3) \tilde{\phi}^B(V_4) \tilde{\phi}^A(U_1) \rangle \nonumber \\
    &= \frac{(-1)^{h_A}}{(2\pi)^2} x^{-2h_B} (U_{13})^{2h_B-2h_A} U^{\sigma_2}\Big(2h_B, 1-2h_A+2h_B; -\frac{i}{x} U_{13} V_{24}\Big) \,, \label{eq:2D-G2-UV-explicit} \\
    \sigma_2 &=
    \begin{cases}
        +\,, & x>0 \cap U_{13}>0 \cap V_{24}<0 \\
        0\,, & (x>0 \cap (U_{13}<0 \cup V_{24}>0)) \cup (x<0 \cap (U_{13}>0 \cup V_{24}<0)) \\
        -\,, & x<0 \cap U_{13}<0 \cap V_{24}>0
    \end{cases} \,. \nonumber
\end{align}
Above, we have assumed that $U_{13}$ and $V_{24}$ are non-zero so that we can set factors of $i\epsilon$ to zero. Here, $U(a,b;z)$ is the confluent hypergeometric-$U$ function, which has a branch cut on the negative real axis. We let $\sigma = (0,+, -)$ denote the main, upper, and lower sheets. The upper/lower sheets are reached by continuing across the branch cut in the upwards/downwards (equivalently, clockwise/anticlockwise) directions. We defer any comments concerning the analytic properties of $G_1$ and $G_2$ to Section \ref{sec:2D-OTOC-analc}.

Two interesting limits are $x\to 0$ and $x\to \infty$. The former is the focus of Section \ref{sec:2D-OTOC-realt-xto0}. For $x\to \pm\infty$, we note that the correlators decay as a model-dependent power:
\begin{align}
    \lim\limits_{x\to\pm\infty} G_1(U_1, V_2, U_3, V_4) &=
    \begin{dcases}
        \frac{(-1)^{h_B}}{(2\pi)^2} \frac{\Gamma(2h_B-2h_A)}{\Gamma(2h_B)} x^{-2h_A} (V_{24})^{2h_A-2h_B}, \quad & h_A < h_B \\
        \frac{(-1)^{h_A}}{(2\pi)^2} \frac{\Gamma(2h_A-2h_B)}{\Gamma(2h_A)} x^{-2h_B} (U_{13})^{2h_B-2h_A}, \quad & h_A > h_B \\
        \frac{(-1)^{h_A+1}}{(2\pi)^2\Gamma(2h_A)} x^{-2h_A} \log\Big(-\frac{i}{x} U_{13} V_{24}\Big), \quad &h_A=h_B
    \end{dcases} \,.
\end{align}
In Section \ref{sec:alg-free}, this decay will be associated with an emergent free product algebra.

\vspace{1em}

\noindent We can also compute the OTOCs in the frequency domain.
\begin{align}
    G_1(\lambda_1,\lambda_2,\lambda_3,\lambda_4) &:= \langle \tilde{\Phi}^A(\lambda_1) \tilde{\Phi}^B(\lambda_2) \tilde{\Phi}^A(\lambda_3) \tilde{\Phi}^B(\lambda_4) \rangle \nonumber \\
    &= 2\pi\delta(\lambda_1+\lambda_2+\lambda_3+\lambda_4) \, C_{\lambda_1\lambda_2\lambda_3\lambda_4} \, e^{\pi(\lambda_2+\lambda_1)} F_x(-\lambda_1-\lambda_3) \,, \label{eq:2D-G1-lambda-explicit} \\
    G_2(\lambda_2,\lambda_3,\lambda_4,\lambda_1) &:= \langle \tilde{\Phi}^B(\lambda_2) \tilde{\Phi}^A(\lambda_3) \tilde{\Phi}^B(\lambda_4) \tilde{\Phi}^A(\lambda_1) \rangle \nonumber \\
    &= 2\pi\delta(\lambda_1+\lambda_2+\lambda_3+\lambda_4) \, C_{\lambda_1\lambda_2\lambda_3\lambda_4} \, e^{\pi(\lambda_2+\lambda_3)}  F_{-x}(-\lambda_1-\lambda_3) \,, \label{eq:2D-G2-lambda-explicit} 
\end{align}
where $F_x(\lambda)$ was defined in (\ref{eq:Fx-function}), and $C_{\lambda_1\lambda_2\lambda_3\lambda_4}$ is given by
\begin{align}
    C_{\lambda_1\lambda_2\lambda_3\lambda_4} := \frac{e^{-\frac{\pi}{2}(\lambda_1+\lambda_2+\lambda_3+\lambda_4)}}{(2\pi)^2 \Gamma(2h_A)\Gamma(2h_B)}  \Gamma(h_A+i\lambda_1) \Gamma(h_B-i\lambda_2) \Gamma(h_A+i\lambda_3) \Gamma(h_B-i\lambda_4) \,.
\end{align}
In Section \ref{sec:2Dcalcs-OTOCs} we explicitly verify that the (suitably regulated) time and frequency domain representations for $G_1$ are related by the Fourier transform.


\subsubsection{The \texorpdfstring{$x \rightarrow 0$}{x to 0} Limit} \label{sec:2D-OTOC-realt-xto0}
The correlators (\ref{eq:2D-G1-UV-explicit})--(\ref{eq:2D-G2-UV-explicit}) and (\ref{eq:2D-G1-lambda-explicit})--(\ref{eq:2D-G2-lambda-explicit}) are presented for finite $x$. An important regime is the limit $x\to 0$, where observer recoil and gravitational backreaction are both suppressed. In this limit, we expect to recover the results of \cite{CLPW}.  

This limit must be taken with care. Since the discussion for $G_1$ and $G_2$ are nearly identical, we focus on the former. In position space, the hypergeometric function $U(a,b;z)$ permits an asymptotic expansion for $|z|\to\infty$ on its main sheet:
\begin{align}
    U(2h_B, 1-2h_A+2h_B; z) \sim z^{-{2h_B}} \sum_{n=0}^{\infty} \Big(-\frac{1}{z}\Big)^n \frac{(2h_A)_n (2h_B)_n}{\Gamma(1+n)} \,.
\end{align}
Carrying out this expansion for $G_1$ recovers the perturbative series $G_1^{\mathrm{pert}}$. However, on either the upper or lower sheet $G_1$ picks up a non-perturbative phase $e^{-i U_{13} V_{24}/x}$, which can be seen from the monodromy formula (\ref{eq:hypU-monodromy}). This oscillates wildly for $x\to 0$, leading to an ill-defined limit. It is therefore essential to consider correlation functions of fields smeared by test functions. In the (static patch) time domain we require the test functions $f_i(t_i)$ be of class $C^{\infty}_c(\mathbb{R})$, i.e. smooth with compact support. In Kruskal coordinates, this means that the smearing functions for the $A$- and $B$-theories are respectively supported in compact regions of $(0,\infty)$ and $(-\infty,0)$.

It is simpler to work in the frequency domain. Since the time and frequency domain correlators are related by a Fourier transform, we lose no generality in doing so. We start with the exact expression (\ref{eq:2D-G1-lambda-explicit}) for $G_1$. Once again this lacks a well-defined $x\to 0$ limit, and it is crucial to consider smeared fields. A natural class of functions are the Fourier conjugates of time domain smearing functions. If $f_i(t_i)$ is of compact support then $\tilde{f}_i(\lambda_i)$ extends to a $\lambda$-entire function. In particular, this means that frequency space test functions will not have compact support.

We now wish to evaluate the $x\to 0$ limit of $G_1(\lambda_i)$ smeared with frequency-space test functions. The variable $x$ only appears within the function $F_x(-\lambda_1-\lambda_3)$. We perform a change of variables, and order the integral involving $F_x(\lambda)$ last. We compute:
\begin{align}
    G_1(\tilde{f}) &= \int_{-\infty}^{\infty} \dbar\lambda_1\cdots \dbar\lambda_4 \, \tilde{f}_1(\lambda_1) \tilde{f}_2(\lambda_2) \tilde{f}_3(\lambda_3) \tilde{f}_4(\lambda_4) \, G_1(\lambda_1, \lambda_2, \lambda_3, \lambda_4) \\
    &= \int_{-\infty}^{\infty} \dbar\lambda_1\cdots \dbar\lambda_4 \mathrm{d}t_1 \cdots \mathrm{d}t_4 \, f_1(t_1) f_2(t_2) f_3(t_3) f_4(t_4) \\
    &\hspace{15mm} \times (e^{-i \lambda_1 (t_1+i\epsilon_1)} \cdots e^{-i \lambda_4 (t_4+i\epsilon_4)}) \, G_1(\lambda_1, \lambda_2, \lambda_3, \lambda_4) \nonumber \\
    &= \frac{(-1)^{h_A+h_B}}{(2\pi)^2\Gamma(2h_A)\Gamma(2h_B)} \int_{-\infty}^{\infty} \dbar\lambda \, F_x(\lambda) e^{\lambda(-\pi+(\epsilon_4-\epsilon_2))} \Gamma(2h_A-i\lambda) \Gamma(2h_B-i\lambda)\\
    &\hspace{5mm} \int_{-\infty}^{\infty} \mathrm{d}t_1 \cdots \mathrm{d}t_4 \, f_1(t_1) f_2(t_2) f_3(t_3) f_4(t_4) e^{h_A(t_1+t_3)} e^{-h_B(t_2+t_4)} \nonumber \\
    &\hspace{15mm} \times (e^{t_1} + e^{-\pi +(\epsilon_3-\epsilon_1)} e^{t_3})^{-2h_A+i\lambda} (-e^{-t_4} + e^{-\pi +(\epsilon_4-\epsilon_2)} e^{-t_2})^{-2h_B+i\lambda} \nonumber \\
    &= \frac{(-1)^{h_A+h_B}}{(2\pi)^2\Gamma(2h_A)\Gamma(2h_B)} \int_{-\infty}^{\infty} \dbar\lambda \, F_x(\lambda) e^{\lambda(-\pi+\epsilon)} \Gamma(2h_A-i\lambda) \Gamma(2h_B-i\lambda) \label{eq:2D-G1-lambda-smeared-step}\\
    &\hspace{5mm} \int_{-\infty}^{\infty} \mathrm{d}U_-\, g_1(U_-) (U_--i\epsilon)^{-2h_A+i\lambda} \int_{-\infty}^{\infty} \mathrm{d}V_- \, g_2(V_-)  (V_- -i\epsilon)^{-2h_B+i\lambda} \,. \nonumber
\end{align}
The first equality is by definition, and the second follows by writing the smearing functions using their Fourier representations for $\epsilon_4>\epsilon_3>\epsilon_2>\epsilon_1$. In the third equality we change variables from $\lambda_1$ to $\lambda := -\lambda_1-\lambda_3$, and perform the integrals over $\lambda_4$, $\lambda_3$, $\lambda_2$. More details are given in the real-time Fourier calculation of Appendix \ref{sec:2Dcalcs-OTOCs-exact}. In the fourth equality we switch to Kruskal coordinates, defining:
\begin{align}
    U_{\pm} &:= U_1 \pm U_3 = e^{t_1} \pm e^{t_3}\,, \qquad V_{\pm} := V_2 \pm  V_4 = (-e^{-t_2}) \pm (-e^{-t_4}) \,, \\
    g_1(U_-) &= 2\int_0^{\infty} \mathrm{d}U_+\, (U_1 U_3)^{h_A-1} f_1(t_1(U_+, U_-)) f_3(t_3(U_+, U_-)) \,, \nonumber \\
    g_2(V_-) &= 2\int_0^{\infty} \mathrm{d}V_+\, (V_2 V_4)^{h_B-1} f_2(t_2(V_+, V_-)) f_4(t_4(V_+, V_-)) \,, \label{eq:2D-g1g2-testfns}
\end{align}
where $f_1, f_3$ vanish identically when $U_1, U_3<0$, and $f_2, f_4$ vanish for $V_2, V_4>0$. Note that $g_1$ and $g_2$ are smooth functions of $U_-$ and $V_-$, with compact support in $\mathbb{R}$, and so are also valid test functions. Next we need the following result.

\begin{lemma} \label{lemma:Fxlambda-integral}
    Let $h(\lambda)$ be $\lambda$-analytic in the upper-half plane $\mathbb{H}_+$. Then, the following integral permits an asymptotic expansion in the $x\to 0$ limit:
    \begin{align}
        \int_{-\infty}^{\infty} \dbar\lambda\, F_x(\lambda) h(\lambda) = \sum_{n=0}^N \frac{(ix)^n}{\Gamma(1+n)} h(in) + \mathcal{O}(|x|^{N+1}), \qquad \forall\, N\in\mathbb{Z}_+ \,, \label{eq:Fx-xto0-identity-orderN}
    \end{align}
    provided that the integral of $F_x(\lambda)h(\lambda)$ on the vertical line segments $C^{\pm}_N(\Lambda) := \pm\Lambda + i[0,N]$ vanish in the $\Lambda\to\infty$ limit. 
    \begin{proof}
        We can use the $i\epsilon$-prescription in $F_x(\lambda)$ to write the integral as:
        \begin{align}
            \int_{-\infty}^{\infty} \dbar \lambda\, F_x(\lambda) h(\lambda) = \int_{-\infty-i\epsilon}^{\infty-i\epsilon} \dbar\lambda \, h(\lambda) \Gamma(i\lambda ) (-ix)^{-i\lambda} \,.
        \end{align}
        Since $h$ is analytic in $\mathbb{H}_+$, all poles of the integrand come from the $\Gamma$-function, at imaginary integers $\lambda=in$. The $\lambda$-contour $C := -i\epsilon + \mathbb{R}$ lies below all poles. We now employ a `contour-pulling' trick dragging the initial contour upwards in steps to $C_N := + i(N-\epsilon) + \mathbb{R}$ for $N\in\mathbb{Z}_+$, picking up residues at each pole. In the limit where $x\to 0$, the absolute value of the integrand on $C_{N+1}$ scales as $|x|^{N+1-\epsilon}$. There is no contribution from the vertical segments at $\mathrm{Re}(\lambda) \to \pm\infty$, since we have assumed the integral vanishes on $C^{\pm}_N(\Lambda\to\pm\infty)$. The residue calculation therefore yields:
        \begin{align}
            \int_{-\infty}^{\infty} \dbar\lambda\, F_x(\lambda) h(\lambda) &= \frac{2\pi i}{2\pi} \sum_{n=0}^N \mathop{\mathrm{Res}}\limits_{\lambda\to in} h(\lambda) \Gamma(i\lambda) (-ix)^{-i\lambda} + \mathcal{O}(|x|^{N+1}) \\
            &= \sum_{n=0}^N \frac{(ix)^{n}}{\Gamma(1+n)} h(in) + \mathcal{O}(|x|^{N+1}) \,.
        \end{align}
        The partial sums need not converge. For example, the partial sums diverge when $h(\lambda) = e^{-\lambda^2}$.
    \end{proof}
\end{lemma}

\noindent We now return to (\ref{eq:2D-G1-lambda-smeared-step}). In order to apply Lemma \ref{lemma:Fxlambda-integral}, we need that the function multiplying $F_x(\lambda)$ is analytic for $\lambda\in\mathbb{H}_+$, and the contribution on the vertical segments vanish at infinity. We justify the latter for a generic class of test functions in Appendix \ref{sec:G1xto0-detail}. To show the former, we use the statement in \textsection 3.6 of \cite{Gelfand:1964} that $(z -i\epsilon)^{i\lambda}$ is an entire distribution-valued function of $\lambda$, meaning it becomes an analytic function in $\lambda$ when smeared against a test function in $z$. Since $g_1,g_2 \in C^{\infty}_c$, the $U_-$ and $V_-$ integrals of (\ref{eq:2D-G1-lambda-smeared-step}) are $\lambda$-entire. Furthermore, the $\Gamma$-functions on the first line have only poles in $\mathbb{H}_-$, so the condition is satisfied. Applying (\ref{eq:Fx-xto0-identity-orderN}) yields:
\begin{align}
    G_1(\tilde{f}) &= \frac{(-1)^{h_A+h_B}}{(2\pi)^2} \sum_{n=0}^N (-ix)^n \frac{(2h_A)_n(2h_B)_n}{\Gamma(1+n)} \int_{-\infty}^{\infty} \frac{\mathrm{d}U_-\,g_1(U_-)}{(U_--i\epsilon)^{2h_A+n}} \int_{-\infty}^{\infty} \frac{\mathrm{d}V_-\, g_2(V_-)}{(V_--i\epsilon)^{2h_A+n}} + \mathcal{O}(x^{N+1})\\
    &= \int_{-\infty}^{\infty} \mathrm{d}t_1\cdots\mathrm{d}t_4\, f_1(t_1)f_2(t_2)f_3(t_3)f_4(t_4) \frac{(-1)^{h_A+h_B}}{(2\pi)^2} \frac{e^{h_A(t_1+t_3)}e^{h_B(t_2+t_4)}} {(U_{13}-i\epsilon)^{2h_A} (V_{24}-i\epsilon)^{2h_B}} \nonumber \\
    &\hspace{10mm} \times \Big[ \sum_{n=0}^N  \Big(\frac{-ix}{(U_{13}-i\epsilon)(V_{24}-i\epsilon)}\Big)^n \frac{(2h_A)_n (2h_B)_n}{\Gamma(1+n)} + \mathcal{O}(x^{N+1}) \Big] \\
    &= \int_{-\infty}^{\infty} \mathrm{d}t_1\cdots\mathrm{d}t_4\, f_1(t_1)f_2(t_2)f_3(t_3)f_4(t_4) \big[ G_1^{\mathrm{pert},N}(t_1,t_2,t_3,t_4) + \mathcal{O}(x^{N+1}) \big] \,.
\end{align}
In the second equality we have written the sum in terms of the initial smearing functions $f_i(t_i)$. In the third we use (\ref{eq:2D-G1-UV-Spert}), stripping off the conformal factors in going from Kruskal to static patch time correlators. Since this holds for any positive integer $N$, we conclude:
\begin{proposition} \label{prop:G1-smeared-xto0}
    The smeared exact correlators $G_1$ and $G_2$ permit a well-defined perturbative expansion in $x$, agreeing to all orders with $G_1^{\mathrm{pert}}$ and $G_2^{\mathrm{pert}}$. This holds for either sign of $x$. In the strict $x\to 0$ limit only the leading $x^0$ term survives, yielding the undressed $G_1^{\mathrm{QFT}}$ and $G_2^{\mathrm{QFT}}$.
\end{proposition}


\subsection{Analytic Properties of Correlators} \label{sec:2D-OTOC-analc}
We have previously restricted our study to correlation functions at real times with an appropriate $i\epsilon$-prescription, $t_i + i\epsilon_i$. Now, we analytically continue these coordinates into the complex plane, and devote this section to the study of the analytic properties of $G_1$ and $G_2$. In particular, we wish to see whether the KMS condition and detailed balance condition are satisfied. We first review them.
\begin{enumerate}
    \item The \textit{KMS condition} asserts that $G_1(t_1, t_2, t_3, t_4)$ is analytic in the `KMS strip' $\mathrm{Im}\,t_1 \in (-2\pi i, 0]$ when all other times are real (up to $i\epsilon$'s). Furthermore, its value on the lower boundary of the strip is given by $G_2$,
    \begin{align}
        \Delta_{\mathrm{KMS}}(t_1,t_2,t_3,t_4) := G_1(t_1 - 2\pi i, t_2, t_3, t_4) - G_2(t_2, t_3, t_4, t_1) = 0\,, \qquad \forall\, t_i\in\mathbb{R}\,. \label{eq:DeltaKMS-t-def}
    \end{align}
    In Kruskal coordinates, this means that $G_1(U_1,V_2,U_3,V_4)$ lifts to a $U_1$ analytic function in the region $\arg(U_1) \in (-2\pi i, 0)$, with
    \begin{align}
        \Delta_{\mathrm{KMS}}(U_1,V_2,U_3,V_4) &:= G_1(e^{-2\pi i}U_1,V_2,U_3,V_4) - G_2(V_2, U_3, V_4, U_1) \nonumber \\
        &= 0, \qquad\qquad \forall\, U_1, U_3>0\,, \quad V_2, V_4<0\,. \label{eq:DeltaKMS-UV-def}
    \end{align}
    \item The \textit{detailed balance condition} states:
    \begin{align}
        \Delta_{\mathrm{db.}}(\lambda_1,\lambda_2,\lambda_3,\lambda_4) &:= e^{-2\pi\lambda_1} G_1(\lambda_1, \lambda_2, \lambda_3, \lambda_4) - G_2(\lambda_2, \lambda_3, \lambda_4, \lambda_1) = 0\,,\quad \forall\, \lambda_i \in \mathbb{R}\,. \label{eq:Deltadb-def}
    \end{align}
\end{enumerate}
When the KMS condition holds, the detailed balance condition follows from the Fourier transform of (\ref{eq:DeltaKMS-t-def}). However, when $\Delta_{\mathrm{KMS}}\neq 0$ the relation between them is more subtle. To avoid confusion, in this paper $\Delta_{\mathrm{KMS}}$ is defined only for correlators in the time domain, and $\Delta_{\mathrm{db.}}$ for correlators in the frequency domain.

First we discuss exact results at finite $x$. Both the KMS and detailed balance conditions are satisfied for $x>0$ and are violated for $x<0$. We derive the result for $\Delta_{\mathrm{KMS}}$ in two ways: abstractly using Tomita-Takesaki theory, and directly by studying the analytic properties of explicit correlators. Next we take the $x\to 0^-$ limit, which requires smearing. Although detailed balance is satisfied to all perturbative orders in $x$, the KMS condition is violated at order $x^0$. The discrepency is explained by the fact that the Fourier representation of $G_1$ converges in only half of the KMS strip.


\subsubsection{KMS Violation from AQFT} \label{sec:KMS-aQFT}
The KMS violation is seen most straightforwardly using Tomita-Takesaki theory. A quick review of algebraic QFT and the tools used here is provided in Appendix \ref{sec:aQFT-review}. The first part of this section applies to a general model of gravitational dressing. Then, we specialize to 2D.

\vspace{1em}

\noindent \textbf{General Model} \\
We wish to consider a gravitationally-dressed algebra $\cala$ accessible to a de Sitter observer at very early and late times. As discussed in Section \ref{sec:4dmodel}, $\cala$ is generated by the dressed fields $S_{\mathrm{grav}}^{1/2} \,a\, S_{\mathrm{grav}}^{-1/2}$ and $S_{\mathrm{grav}}^{-1/2} \,b\, S_{\mathrm{grav}}^{1/2}$ with $a \in \cala_{0,A}$, $b\in \cala_{0,B}$.\footnote{To avoid confusion with the Tomita operator, in this section we denote the gravitational S-matrix by $S_{\mathrm{grav}}$.} The Hilbert space is built around the tensor product $|\Omega\rangle := |\Omega_A\rangle \otimes |\Omega_B\rangle$, where $|\Omega_{A}\rangle$ and $|\Omega_{B}\rangle$ are the vacua for the $A$ and $B$ theories, cyclic and separating with respect to $\cala_{0,A}$ and $\cala_{0,B}$. The modular Hamiltonian acts on each algebra as static patch time translation.

For any correlator of three or fewer fields, the S-matrix acts trivially and the KMS condition follows directly from the KMS condition on the individual $A$ and $B$ theories. We thus consider: 
\begin{align}
    G_1 &:= \langle \tilde{a}_1 \,\tilde{b}_2 \,\tilde{a}_3 \,\tilde{b}_4 \rangle \,, \qquad G_2 := \langle \tilde{b}_2 \,\tilde{a}_3 \,\tilde{b}_4 \,\tilde{a}_1 \rangle \,, \qquad a_1, a_3 \in \cala_{0,A} \,, \quad b_2,  b_4 \in \cala_{0,B}\,. \label{eq:general-G-def}
\end{align}
We now compute:
\begin{align}
    \langle \tilde{a}_1(s-\tfrac{i}{2}) \,\tilde{b}_2 \,\tilde{a}_3 \,\tilde{b}_4 \rangle &= \langle a_1(s) \, \Delta_{\Omega_A}^{-1/2} \, b_2 \, S_{\mathrm{grav}} \, a_3 \, b_4\rangle \\
    &= \langle J_{\Omega_A} \,a_1(s)\, J_{\Omega_A} J_{\Omega_A} \Delta_{\Omega_A}^{-1/2} \, b_2\,  S_{\mathrm{grav}} \,a_3 \, S_{\mathrm{grav}}^{-1} \, b_4\rangle^{\ast} \nonumber \\
    &= \langle J_{\Omega_A} \,a_1(s) \,J_{\Omega_A} S_{\Omega_A}^{\dagger} \,b_2\, S_{\mathrm{grav}} \,a_3\, S_{\mathrm{grav}}^{-1} \, b_4\rangle^{\ast} \nonumber \\
    &= \langle (J_{\Omega_A} \,a_1(s)\, J_{\Omega_A})^{\dagger} \,b_2\, S_{\mathrm{grav}} \,a_3 \, S_{\mathrm{grav}}^{-1} \, b_4 \rangle \nonumber \\
    &= \langle b_2^{\dagger}| J_{\Omega_A} \,a_1(s)^{\dagger}\, J_{\Omega_A} S_{\mathrm{grav}} \,a_3\, S_{\mathrm{grav}}^{-1} | b_4 \rangle \,, \\
    \langle \tilde{b}_2 \,\tilde{a}_3 \,\tilde{b}_4 \,\tilde{a}_1(s+\tfrac{i}{2}) \rangle &= \langle b_2 \, a_3 \, S_{\mathrm{grav}}^{-1} \, b_4 \,\Delta_{\Omega_A}^{-1/2} \, a_1(s) \rangle \\
    &= \langle b_2 \,S_{\mathrm{grav}} \,a_3\, S_{\mathrm{grav}}^{-1} \,b_4\, \Delta_{\Omega_A}^{-1/2} J_{\Omega_A}J_{\Omega_A} \,a_1(s)\, J_{\Omega_A} \rangle \nonumber \\
    &= \langle b_2 \,S_{\mathrm{grav}} \, a_3 \, S_{\mathrm{grav}}^{-1} \, b_4 \, S_{\Omega_A} J_{\Omega_A} \, a_1(s) \, J_{\Omega_A} \rangle \nonumber \\
    &= \langle b_2 \, S_{\mathrm{grav}} \,a_3\, S_{\mathrm{grav}}^{-1} \,b_4\, (J_{\Omega_A} \,a_1(s) \,J_{\Omega_A})^{\dagger}  \rangle \nonumber \\
    &= \langle b_2^{\dagger} | S_{\mathrm{grav}} \,a_3 \, S_{\mathrm{grav}}^{-1} J_{\Omega_A} \, a_1(s)^{\dagger} \, J_{\Omega_A}  |b_4 \rangle \,.
\end{align}
Here $s$ denotes modular time, which differs from static patch time by a factor of $2\pi$. In the second set of equalities we have used that $J_{\Omega_A}$ is antiunitary. The third set of equalities uses (\ref{eq:Tomita-polar}), and the fourth uses that $\langle \Omega_A|a \, S_{\Omega_A}^{\dagger}|\psi\rangle = \langle\psi |a|\Omega_A\rangle$. Taking the difference of the two correlators yields
\begin{align}
    \langle \tilde{a}_1(s-\tfrac{i}{2}) \,\tilde{b}_2 \,\tilde{a}_3 \,\tilde{b}_4 \rangle - \langle \tilde{b}_2 \,\tilde{a}_3 \,\tilde{b}_4 \,\tilde{a}_1(s+\tfrac{i}{2}) \rangle &= \langle b_2^{\dagger} | \big[ J_{\Omega_A} a_1(s)^{\dagger} J_{\Omega_A} \,,\, S_{\mathrm{grav}} a_3 S_{\mathrm{grav}}^{-1} \big] | b_4 \rangle \,. \label{eq:DeltaKMS-aQFT-general}
\end{align}
The KMS condition is satisfied if the commutator in (\ref{eq:DeltaKMS-aQFT-general}) vanishes:
\begin{align}
    \big[ J_{\Omega_A} a_1(s)^{\dagger} J_{\Omega_A} \,,\, S_{\mathrm{grav}} a_3 S_{\mathrm{grav}}^{-1} \big] &= 0\,. \label{eq:aQFT-commutator-general}
\end{align}
The first operator in (\ref{eq:aQFT-commutator-general}) belongs to the commutant $\cala_{0,A}'$. By performing a similar analysis on general higher-point OTOCs, we may conclude that the KMS condition is universally obeyed by vacuum correlators whenever S-matrix conjugation preserves the $A$- and $B$-theory algebras,
\begin{align}
  S_{\mathrm{grav}}^{1/2} \cala_{0,A} S_{\mathrm{grav}}^{-1/2} &\subseteq \cala_{0,A} \,, \label{eq:340} \\
  S_{\mathrm{grav}}^{-1/2} \cala_{0,B} S_{\mathrm{grav}}^{1/2}& \subseteq \cala_{0,B} \,. \label{eq:341}
\end{align}
If \eqref{eq:340} and \eqref{eq:341} are obeyed, then $\cala$ has a nontrivial commutant. We may define
\begin{equation}
\label{eq:tildecommutant}
\begin{aligned}
    \tilde{a}^\prime &:= S_{\mathrm{grav}}^{-1/2}\, a^\prime\, S_{\mathrm{grav}}^{1/2}, \quad & a^\prime &\in \cala_{0,A}^\prime \,, \\
    \tilde{b}^\prime &:= S_{\mathrm{grav}}^{1/2}\, b^\prime\, S_{\mathrm{grav}}^{-1/2}, \quad & b^\prime &\in \cala_{0,B}^\prime \,,
\end{aligned}
\end{equation}
In the AdS setting, $\tilde{b}^{\prime}$ and $\tilde{a}^{\prime}$ denote early- and late-time operators on the other AdS boundary and generate the commutant of $\cala$ \cite{Penington:2025hrc}. In the dS model of section \ref{sec:4dmodel}, $\tilde{b}^{\prime}$ and $\tilde{a}^{\prime}$ are not in the commutant. They may be associated with an auxiliary, antipodal observer.

Note that \eqref{eq:aQFT-commutator-general}, \eqref{eq:340}, and \eqref{eq:341} hold to all orders in S-matrix perturbation theory. This is because the ANEC operators in \eqref{eq:action} generate infinitesimal null translations on the horizon. An infinitesimal translation applied to a field that is supported away from the bifurcation surface will keep the field away from this surface. Thus, to a given order in perturbation theory, the vacuum $\ket{\Omega}$ is KMS, and the observer's algebra has a commutant generated by fields in the neighborhood of an auxiliary antipodal observer.

The gravitational S-matrix for a general model can be written in terms of the ANEC operators $P_A(\Omega)$ and $P_B(\Omega)$ (defined in \eqref{eq:ANECdef}), but the angular smearing kernel makes their algebraic properties difficult to use. We therefore specialize to 2D, where $S_{\mathrm{grav}} = e^{ix P_AP_B}$. Equation (\ref{eq:DeltaKMS-aQFT-general}) becomes:
\begin{align}
    \langle \tilde{a}_1(s-\tfrac{i}{2})\, \tilde{b}_2\, \tilde{a}_3\, \tilde{b}_4 \rangle - \langle \tilde{b}_2\, \tilde{a}_3 \,\tilde{b}_4 \,\tilde{a}_1(s+\tfrac{i}{2}) \rangle &= \langle b_2^{\dagger} | \big[ J_{\Omega_A} a_1(s)^{\dagger} J_{\Omega_A} \,,\, e^{ix P_AP_B} a_3 e^{-ixP_AP_B} \big] | b_4 \rangle \,. \label{eq:DeltaKMS-aQFT-2D}
\end{align}
The second operator in the commutator is now more transparent. It involves a half-sided translation by $-P_A$, which generates an automorphism of $\cala_{0,A}$ whenever the half-sided modular time is positive \cite{Wiesbrock:1993,Leutheusser:2021frk,Faulkner:2020}. The $x\gtrless 0$ cases can now be distinguished.
\begin{enumerate}
    \item Case 1: $x>0$ (AdS) \cite{Penington:2025hrc}. Since $P_B$ is negative-definite we have $e^{ix P_A P_B} a_3 e^{-ix P_A P_B} \in \cala_{0,A}$, and the KMS condition is satisfied. In particular, if $a_3$ is a field smeared on the half-horizon $U>0$, then the half-sided flow shifts its support in the positive direction. The algebras $\cala_{0,A}, \cala_{0,B}, \cala_{0,A}', \cala_{0,B}'$ are each mapped to themselves under the gravitational dressing defined in \eqref{eq:tildeops} and \eqref{eq:tildecommutant}, and the commutant $\cala'$ is generated by the operators in \eqref{eq:tildecommutant}. Quantum gravity preserves the causal separation of the two boundaries.
    \item Case 2: $x<0$ (dS). The half-sided flow now shifts the support of $a_3$ in the negative direction. Since $\mathrm{Spec}\,P_B = (-\infty,0]$ this may not stay in $\cala_{0,A}$, and the commutator (\ref{eq:DeltaKMS-aQFT-2D}) picks up a KMS-violating contact term when its support intersects that of $J_{\Omega_A} a_1(s)^{\dagger} J_{\Omega_A}$. We explicitly compute this contact term in the next section. The moral is that the gravitational S-matrix sends $\cala_{0,A}$ into causal contact with $\cala_{0,A}'$ (similarly, $\cala_{0,B}$ with $\cala_{0,B}'$). Quantum gravity thus sends antipodal observers into causal contact. It is now unclear which operators lie in $\cala'$. In Section \ref{sec:alg-type}, we provide evidence that $\cala'$ is trivial for our 2D model. In Section \ref{sec:alg-geometric} we will see the change in causal structure more explicitly in terms of the backreacted geometry. 
\end{enumerate}
Finally, we note that the derivations of (\ref{eq:DeltaKMS-aQFT-general}) and (\ref{eq:DeltaKMS-aQFT-2D}) use only the Tomita-Takesaki theory of the $\cala_{0,A}$ algebra about the state $|\Omega_A\rangle$, through the existence and properties of the operators $(S_{\Omega_A}, \Delta_{\Omega_A}, J_{\Omega_A})$. We have not assumed anything about the modular theory of the full algebra $\cala$. This is important, as in Section \ref{sec:alg-type} we will show that for $x<0$ the Tomita operator for $\cala$ does not exist.

\vspace{1em}

\noindent \textbf{2D Model} \\
We now consider exact calculations in our 2D model. It is easiest to work in Kruskal coordinates. To check if the KMS condition holds, we need to compare $G_2$ (\ref{eq:2D-G2-UV-def}) to $G_1$ (\ref{eq:2D-G1-UV-def}) analytically continued by $U_1 \to e^{-2\pi i}U_1$. We will carry the analytic continuation out in two steps, going from $U_1 \to e^{-i\pi}U_1 \to e^{-2\pi i}U_1$.

We make use of the fact that $\phi^A(U)|\Omega\rangle$ is analytic in the upper-half plane.\footnote{More precisely, this is a statement about correlators $\langle \Psi|\phi^A(U)|\Omega\rangle$ for a dense set of vectors $|\Psi\rangle$ generated by acting smeared fields acting on the vacuum. Since the other fields in $G_1$ and $G_2$ are understood to be smeared, this applies to our setting.} This is justified by writing $\langle \Psi | \phi^A(U) |\Omega\rangle = \langle \Psi | e^{-iP_U U} \phi^A(0) e^{iP_U U} |\Omega\rangle =  \langle \Psi | e^{-iP_U U} \phi^A(0) |\Omega\rangle$; since $\mathrm{Spec}\,P_U<0$, the right-hand side converges when we give $U$ a positive imaginary part. By the same argument, $\langle \Omega|\phi^A(U)$ is analytic in the lower-half plane. Since $U_1>0$, we can unambiguously make the analytic continuation $U_1 \to e^{-i\pi}U_1$ in $G_1$. This allows us to write:
\begin{align}
    G_1(e^{-i\pi}U_1, V_2, U_3, V_4) &= \langle \phi^A(U_1 e^{-i\pi}) \phi^B(V_2) e^{ix P_A P_B} \phi^A(U_3) \phi^B(V_4) \rangle \\
    &= \langle \phi^B(V_2) \phi^A(U_1 e^{-i\pi}) \phi^A(U_3-x P_B) \phi^B(V_4) \rangle \\
    &= \langle \phi^B(V_2) \phi^A(U_3-xP_B) \phi^B(V_4) \phi^A(e^{-i\pi}U_1) \rangle \nonumber \\
    &\hspace{5mm}+ \langle \phi^B(V_2) \big[\phi^A(U_1 e^{-i\pi})\,,\, \phi^A(U_3-x\hat{P}_B)\big] \phi^B(V_4)\rangle  \\
    &= \langle \phi^B(V_2) \phi^A(U_3) e^{-ix P_A P_B} \phi^B(V_4) \phi^A(-U_1) \rangle \nonumber \\
    &\hspace{5mm}+ \langle \phi^B(V_2) \big[\phi^A(U_1 e^{-i\pi}) \,,\, \phi^A(U_3-x P_B)\big] \phi^B(V_4)\rangle \,. \label{eq:2D-G1-ipianalc} 
\end{align}
The first term of (\ref{eq:2D-G1-ipianalc}) is an inner-product involving $\phi(-U_1)|\Omega\rangle$, and unambiguously permits the second continuation $-U_1 \to -U_1 e^{-i\pi} = U_1$ to yield $G_2$. If the second term can also be analytically continued this way, we have:
\begin{align}
    \Delta_{\mathrm{KMS}}(U_1, V_2, U_3, V_4) &:= G_1(e^{-2\pi i}U_1, V_2, U_3, V_4) - G_2(V_2, U_3, V_4, U_1) \label{eq:2D-DeltaKMS-UV-def} \\
    &= \langle \Omega| \phi^B(V_2) \big[\phi^A(-U_1)\,,\, \phi^A(U_3-x P_B)\big] \phi^B(V_4)|\Omega \rangle\,\big|_{U_1\to e^{-i\pi}U_1} \,. \label{eq:2D-DeltaKMS-UV-commutator}
\end{align}
To complete the derivation of $\Delta_{\mathrm{KMS}}$, we still need to show that the second part of the analytic continuation in (\ref{eq:2D-DeltaKMS-UV-commutator}) is well-defined. We demonstrate this by explicit evaluation, using the following distributional identity.
\begin{align}
    \langle \Omega_A| [\phi^A(U), \phi^A(U')] |\Omega_A\rangle &= (-1)^{h_A} \Big[\frac{1}{(U-U'-i\epsilon)^{2h_A}} - \frac{1}{(-U+U'-i\epsilon)^{2h_A}} \Big] \\
    &= \frac{-i (-1)^{h_A} }{\Gamma(2h_A)} \delta^{(2h_A-1)}(U-U') \,.\label{eq:field-1D-commutator-identity}
\end{align}
We substitute this into (\ref{eq:2D-DeltaKMS-UV-commutator}) prior to the continuation:
\begin{align}
    &\langle \Omega| \phi^B(V_2) \big[\phi^A(-U_1), \phi^A(U_3-x P_B)\big] \phi^B(V_4)|\Omega \rangle \nonumber \\
    &\hspace{5mm}= \frac{-i(-1)^{h_A}}{\Gamma(2h_A)} \int_{-\infty}^0 \frac{\dbar P_B}{-P_B} \delta^{(2h_A-1)}(-U_1-U_3+xP_B)  \langle \phi^B(V_2)|P_B\rangle \langle P_B| \phi^B(V_4)\rangle \\
    &\hspace{5mm}= \theta(-x) \frac{i (-1)^{h_A}}{2\pi \Gamma(2h_A)\Gamma(2h_B)} (-x)^{-2h_B} \partial_{U_3}^{2h_A-1} \Big[ (U_1+U_3)^{2h_B-1} e^{\frac{i}{x}(U_1+U_3+i\epsilon)(V_2-V_4-i\epsilon)} \Big] \\
    &\hspace{5mm}= \frac{\theta(-x)}{2\pi x} \frac{(-1)^{h_A+h_B}}{\Gamma(2h_A)\Gamma(2h_B)} \Big[ \partial_{U_3}^{2h_A-1} \partial_{V_2}^{2h_B-1} e^{\frac{i}{x}(U_1+U_3)(V_2-V_4-i\epsilon)} \Big] \,. \label{eq:2D-AtApt-commutator-calc}
\end{align}
Note that $P_B$ is an operator in the first line, but a dummy integration variable in the second after inserting a complete basis. The function (\ref{eq:2D-AtApt-commutator-calc}) is $U_1$-entire, permitting the second analytic continuation $U_1 \to e^{-i\pi}U_1$. Expanding the derivative using the Leibniz rule yields:
\begin{align}
    \Delta_{\mathrm{KMS}}(U_1,V_2,U_3,V_4) &= \frac{\theta(-x)}{2\pi x} \frac{(-1)^{h_A+h_B}}{\Gamma(2h_A)\Gamma(2h_B)} \Big[ \partial_{U_1}^{2h_A-1} \partial_{V_2}^{2h_B-1} e^{-\frac{i}{x}(U_{13}+i\epsilon)(V_{24}-i\epsilon)} \Big] \label{eq:2D-DeltaKMS-UV-explicit-0} \\
    &= \theta(-x) \frac{i (-1)^{h_A}}{2\pi \Gamma(2h_A)\Gamma(2h_B)} x^{-2h_B} \partial_{U_1}^{2h_A-1} \Big[ U_{13}^{2h_B-1} e^{-\frac{i}{x}(U_{13}+i\epsilon)(V_{24}-i\epsilon)} \Big] \label{eq:2D-DeltaKMS-UV-explicit-1} \\
    &= \theta(-x) \frac{i (-1)^{h_-}}{2\pi} \times
    \begin{dcases}
        (U_1-U_3)^{2h_+-2h_-}, & h_B\geq h_A \\
        (V_2-V_4)^{2h_+-2h_-}, & h_A\geq h_B
    \end{dcases} \label{eq:2D-DeltaKMS-UV-explicit-2} \\
    &\hspace{5mm} \times e^{-\frac{i}{x}(U_{13}+i\epsilon)(V_{24}-i\epsilon)} x^{-2h_+} \sum_{k=0}^{2h_--1} \frac{\big(-\frac{i}{x}(U_{13}V_{24})\,\big)^k}{\Gamma(1+k)\Gamma(2h_--k)\Gamma(1+2h_+-2h_-+k)} \,, \nonumber
\end{align}
where $h_+ := \max(h_A, h_B)$, and $h_- := \min(h_A, h_B)$.

We note that the leading exponential phase does not depend on the field scaling dimensions $(h_A, h_B)$, and is therefore model independent. In Section \ref{sec:3D-KMS} we find the same exponential phase for the 3D model (\ref{eq:3D-DeltaKMS}).

Aside from the $x\to 0^-$ limit, which we discuss in Section \ref{sec:2D-OTOC-analc-xto0}, we can also take $x\to -\infty$. This vanishes, with leading term
\begin{align}
    \lim\limits_{x\to-\infty}\Delta_{\mathrm{KMS}}(U_{13}, V_{24}) &= \frac{i(-1)^{h_-} x^{-2h_+}}{\Gamma(2h_-)\Gamma(1+2h_+-2h_-)} \times
    \begin{dcases}
        U_{13}^{2h_+-2h_-}, & h_B\geq h_A \\
        V_{24}^{2h_+-2h_-}, & h_A\geq h_B
    \end{dcases}\,. \label{eq:2D-DeltaKMS-UV-xtominfty}
\end{align}


\subsubsection{Exact Correlators}

\noindent \textbf{Position Space} \\
We now study the analytic properties of the exact correlation functions $G_1$ and $G_2$. For convenience we reproduce (\ref{eq:2D-G1-UV-explicit}) and (\ref{eq:2D-G2-UV-explicit}):
\begin{align}
    G_1(U_1, V_2, U_3, V_4) &= \langle \tilde{\phi}^A(U_1) \tilde{\phi}^B(V_2) \tilde{\phi}^A(U_3) \tilde{\phi}^B(V_4) \rangle \nonumber \\
    &= \frac{(-1)^{h_A}}{(2\pi)^2} x^{-2h_B} (U_{13})^{2h_B-2h_A} U^{\sigma_1} \Big(2h_B, 1-2h_A+2h_B; -\frac{i}{x} U_{13} V_{24}\Big)\,, \nonumber \\
    \sigma_1 &=
    \begin{cases}
        +\,, & x<0 \cap U_{13}<0 \cap V_{24}<0 \\
        0\,, & (x>0 \cap (U_{13}<0 \cup V_{24}<0)) \cup (x<0 \cap (U_{13}>0 \cup V_{24}>0)) \\
        -\,, & x>0 \cap U_{13}>0 \cap V_{24}>0
    \end{cases} \,, \nonumber
\end{align}
\begin{align}
    G_2(V_2, U_3, V_4, U_1) &= \langle \tilde{\phi}^B(V_2) \tilde{\phi}^A(U_3) \tilde{\phi}^B(V_4) \tilde{\phi}^A(U_1) \rangle \nonumber \\
    &= \frac{(-1)^{h_A}}{(2\pi)^2} x^{-2h_B} (U_{13})^{2h_B-2h_A} U^{\sigma_2}\Big(2h_B, 1-2h_A+2h_B; -\frac{i}{x} U_{13} V_{24} \Big) \,, \nonumber \\
    \sigma_2 &=
    \begin{cases}
        +\,, & x>0 \cap U_{13}>0 \cap V_{24}<0 \\
        0\,, & (x>0 \cap (U_{13}<0 \cup V_{24}>0)) \cup (x<0 \cap (U_{13}>0 \cup V_{24}<0)) \\
        -\,, & x<0 \cap U_{13}<0 \cap V_{24}>0
    \end{cases} \,. \nonumber
\end{align}
Although $G_1$ and $G_2$ were computed for $U_1,U_3>0$ and $V_2,V_4<0$, they may be analytically continued to arbitrary complex values in $\mathbb{C}^4$, provided that neither pair $(U_1, U_3)$, $(V_2, V_4)$ are coincident. The cost of doing so is that the correlators will be multivalued, due to the different sheets of the hypergeometric-$U$. Interestingly, the OTOCs for opposite signs $x>0$ and $x<0$ have the same domain of analyticity.

Now we check the KMS condition, for fixed $U_3>0$ and $V_2, V_4<0$. The multivalued OTOC can be thought of as a Riemann surface in $U_1$ with an infinite number of sheets. Analytically continuing $U_1\to e^{-2\pi i}U_1$ crosses the hypergeometric-$U$ branch cut if and only if $U_{13}>0$. For $x>0$ we continue to an upper sheet, and for $x<0$ we cross to a lower sheet. This yields:
\begin{align}
    G_1(e^{-2\pi i}U_1, V_2, U_3, V_4) &= \frac{(-1)^{h_A}}{(2\pi)^2} x^{-2h_B} (U_{13}+i\epsilon)^{2h_B-2h_A} U^{\sigma_1'} \Big(2h_B, 1-2h_A+2h_B; -\frac{i}{x}U_{13} V_{24}\Big)\,, \nonumber \\
    \sigma_1' &=
    \begin{cases}
        +\,, & (x>0 \cap U_{13}>0 \cap V_{24}<0) \cup (x<0 \cap (U_{13}>0 \cup V_{24}<0)) \\
        0\,, & (x>0 \cap (U_{13}<0 \cup V_{24}
        >0)) \cup (x<0 \cap U_{13}<0 \cap V_{24}>0))
    \end{cases} \,. \label{eq:2D-G1-UV-explicit-analc}
\end{align}
The KMS violation is just the difference $G_1(e^{-2\pi i}U_1) - G_2(U_1)$. For $x>0$ the correlators always live on the same sheet and agree, thus $\Delta_{\mathrm{KMS}} \equiv 0$. For $x<0$ the difference is exactly a monodromy, which may be computed using the identity:
\begin{align}
    U(a,b; z e^{2\pi im}) - e^{-2\pi i b m} U(a,b;z) &= \frac{2\pi i e^{-i \pi b m} \sin(\pi b m)}{\Gamma(1+a-b) \sin(\pi b)} {}_1\tilde{F}_1(a,b;z)\,, \qquad m\in\mathbb{Z} \,. \label{eq:hypU-monodromy}
\end{align}
Putting everything together yields:
\begin{align}
    \Delta_{\mathrm{KMS}} &= G_1(e^{-2\pi i}U_1, V_2, U_3, V_4) - G_2(V_2, U_3, V_4,U_1) \\
    &= \theta(-x) \frac{i(-1)^{h_A}}{2\pi\Gamma(2h_A)} x^{-2h_B} (U_{13}+i\epsilon)^{2h_B-2h_A} {}_1\tilde{F}_1\Big(2h_B, 1-2h_A+2h_B;  -\frac{i}{x}U_{13}V_{24}\Big) \label{eq:2D-DeltaKMS-UV-explicit-3} \\
    &= \theta(-x) \frac{i (-1)^{h_-}}{2\pi} \times
    \begin{dcases}
        (U_{13}+i\epsilon)^{2h_+-2h_-}, & h_B\geq h_A \\
        (V_{24}-i\epsilon)^{2h_+-2h_-}, & h_A\geq h_B
    \end{dcases} \nonumber \\
    &\hspace{5mm} \times e^{-\frac{i}{x}U_{13}V_{24}} x^{-2h_+} \sum_{k=0}^{2h_--1} \frac{\big(-\frac{i}{x}U_{13}V_{24}\big)^k}{\Gamma(1+k)\Gamma(2h_--k)\Gamma(1+2h_+-2h_-+k)}\,.
\end{align}
This precisely reproduces (\ref{eq:2D-DeltaKMS-UV-explicit-2}). The moral is that $G_1(e^{-2\pi i}U_1,V_2,U_3,V_4)$ and $G_2(V_2,U_3,V_4,U_1)$ take values on the same multivalued but meromorphic function. For $x>0$ the correlators live on the same sheet and are identical, but for $x<0$ the analytically continued $G_1$ lives one sheet higher. Finally, we observe that $\Delta_{\mathrm{KMS}}$ is $U_1$-entire, even when coincident with $U_3$.

\vspace{1em}

\noindent \textbf{Frequency Space} \\
In the frequency domain, we use (\ref{eq:2D-G1-lambda-explicit})--(\ref{eq:2D-G2-lambda-explicit}) to check the detailed balance condition:
\begin{align}
    \Delta_{\mathrm{db.}}(\lambda_i) &= e^{-2\pi\lambda_1}G_1(\lambda_1,\lambda_2,\lambda_3,\lambda_4) - G_2(\lambda_2,\lambda_3,\lambda_4,\lambda_1)  \\
    &= \theta(-x) \, 2\pi\delta(\lambda_1+\lambda_2+\lambda_3+\lambda_4) \, C_{\lambda_1\lambda_2\lambda_3\lambda_4}  e^{\pi(\lambda_2+\lambda_3)} 2\sinh(\pi(-\lambda_1-\lambda_3)) F_x(-\lambda_1-\lambda_3) \,, \label{eq:2D-deltaKMS-lambda-explicit} \\
    C_{\lambda_1\lambda_2\lambda_3\lambda_4} &= \frac{e^{-\frac{\pi}{2}(\lambda_1+\lambda_2+\lambda_3+\lambda_4)}}{(2\pi)^2 \Gamma(2h_A)\Gamma(2h_B)}  \Gamma(h_A+i\lambda_1) \Gamma(h_B-i\lambda_2) \Gamma(h_A+i\lambda_3) \Gamma(h_B-i\lambda_4) \,. \nonumber
\end{align}
This vanishes identically for $x>0$, and is non-zero for $x<0$. We note that the only $x$-dependence is contained in the function $F_x(-\lambda_1-\lambda_3)$, which was defined in \eqref{eq:Fx-function}.

\vspace{1em}

Next we turn our attention to the relation between the analytically continued correlators in the time and frequency domains. In particular, are they always Fourier conjugate?
\begin{align}
    \mathcal{F}^{-1}_{\lambda_i\to t_i+i\epsilon_i} \big[e^{-2\pi \lambda_1 y} G_1(\lambda_1,\lambda_2,\lambda_3,\lambda_4) \big] \stackrel{?}{=} G_1(t_1-2\pi i y, t_2, t_3, t_4) \,, \qquad y\in \mathbb{R}\,. \label{eq:2D-G1-Fourier-analc}
\end{align}
Here all $t_i$'s are all understood to be real, and $y$ parameterizes the imaginary part of the first time coordinate. The left-hand side is computed in Appendix \ref{sec:2Dcalcs-OTOCs-exact}. We find that (\ref{eq:2D-G1-Fourier-analc}) holds whenever all Fourier integrals are convergent, but the region of convergence in $y$ differs depending on the sign of $x$. We restrict the current discussion to the real-time $G_1$ prescription $\epsilon_4>\epsilon_3>\epsilon_2>\epsilon_1$ unless otherwise specified, but more general prescriptions are studied in the Appendix \ref{sec:2Dcalcs-OTOCs-exact}.
\begin{enumerate}
    \item Case 1: $x>0$. For $y\in [0,1)$, the $G_1$ Fourier integrals converge for all real $t_i$'s. Outside this region, the integrals diverge for any real $t_i$'s. The Fourier representation of $G_1$ is therefore well-behaved (i.e. always convergent) on the KMS strip $\mathrm{Im}(t_1-2\pi iy) \in (-2\pi,0]$, and this strip is maximal. Convergence at $y=1$ for all $t_i$'s is guaranteed if we use the $i\epsilon$-prescription for $G_2$.
    \item Case 2: $x<0$. For $y\in [0,1/2)$, the $G_1$ Fourier integrals converge for all real $t_i$'s. For $y \notin [0,1)$, the integrals diverge for any real $t_i$'s. When $V_{24}>0$, we additionally have convergence for $y\in [1/2,1)$. However, when $V_{24}<0$, the integrals diverge for any $y\in [1/2,1)$.  The maximal strip $y\in [0,1/2)$ where the Fourier representation of $G_1$ is well-behaved constitutes only half of the KMS strip. Convergence at $y=1/2$ for all $t_i$'s requires a different $i\epsilon$-prescription.
\end{enumerate}
There is a similar computation involving $G_2(t_2, t_3, t_4, t_1-2\pi i y)$, with prescription $\epsilon_1>\epsilon_4>\epsilon_3>\epsilon_2$. The Fourier representation is well-behaved on $y\in (-1, 0]$ for $x>0$, and $y\in (-1/2, 0]$ for $x<0$.

We are now in a position to comment on the relation between detailed balance and the KMS condition. When the inverse Fourier transform $\mathcal{F}^{-1}[e^{-2\pi y\lambda_1}G_1(\lambda_i)]$ is convergent for all real times at $y=1$ (for some regulator choice), then $\mathcal{F}^{-1}[\Delta_{\mathrm{db.}}(\lambda_i)$] is also convergent. In this case there are no subtleties, and we write
\begin{align}
    \mathcal{F}^{-1} \big[\Delta_{\mathrm{db.}}(\lambda_i)\big] = \Delta_{\mathrm{KMS}}(t_i)\,.
\end{align}
This is the case for $x>0$, where the equality becomes $0=0$. Next consider the $x<0$ case, where $\mathcal{F}^{-1}[e^{-2\pi y\lambda_1}G_1(\lambda_i)]$ is only convergent (for some regulator choice) on $y\in [0,1/2]$, and $\mathcal{F}^{-1}[e^{-2\pi y\lambda_1}G_2(\lambda_i)]$ on $y\in[-1/2,0]$. Here, we can only say
\begin{align}
    \mathcal{F}^{-1} \big[e^{\pi\lambda_1} \Delta_{\mathrm{db.}}(\lambda_i) \big] &= \mathcal{F}^{-1} \big[e^{-\pi\lambda_1} G_1(\lambda_i) - e^{\pi\lambda_1} G_2(\lambda_i)\big] \\
    &= G_1(t_1-i\pi, t_2,t_3, t_4) - G_2(t_2,t_3,t_4,t_1+i\pi) \,.
\end{align}
To compute the KMS violation, we would need to further analytically continue the right-hand side by $t_1 \to t_1 -i\pi$; this is precisely the final step (\ref{eq:2D-DeltaKMS-UV-explicit-0}) of the Tomita-Takesaki derivation of $\Delta_{\mathrm{KMS}}$. This analytic continuation does not commute with the $\lambda_1$-integral from the inverse Fourier transform, because analytically continuing first makes the integral diverge. Consequently, $\Delta_{\mathrm{db.}}$ and $\Delta_{\mathrm{KMS}}$ can have different properties. In particular, in Section \ref{sec:2D-OTOC-analc-xto0} we will see that the smeared $\Delta_{\mathrm{KMS}}$ can remain finite in the $x\to 0$ limit, while the smeared $\Delta_{\mathrm{db.}}$ is always non-perturbative.


\vspace{1em}

\noindent \textbf{Perturbative Correlators and Borel Resummation}

\noindent In contrast to the exact correlation functions, the correlators $G_{1,2}^{\mathrm{pert},N}(t_i)$ computed from S-matrix perturbation theory satisfy the KMS condition at each order $N$. This follows from analytically continuing (\ref{eq:2D-G1-UV-Spert}) and using the KMS condition of the undressed QFT correlators $G^{\mathrm{QFT}}_{1,2}$:
\begin{align}
    &G_1^{\mathrm{pert},N}(e^{-2\pi i}U_1, V_2, U_3, V_4) \\
    &\hspace{2mm} = G_1^{\mathrm{QFT}}(e^{-2\pi i}U_1, V_2, U_3, V_4) \Big[\sum_{n=0}^{N} \Big(\frac{-ix}{(e^{-2\pi i}U_1-U_3-i\epsilon)(V_{24}-i\epsilon)}\Big)^n \frac{(2h_A)_n (2h_B)_n}{\Gamma(1+n)} + \mathcal{O}(x^{N+1}) \Big] \nonumber \\
    &\hspace{2mm} = G_2^{\mathrm{QFT}}(V_2, U_3, V_4,U_1) \Big[\sum_{n=0}^{N} \Big(\frac{ix}{(-U_{13}-i\epsilon)(V_{24}-i\epsilon)}\Big)^n \frac{(2h_A)_n (2h_B)_n}{\Gamma(1+n)} + \mathcal{O}(x^{N+1}) \Big] \\
    &\hspace{2mm}= G_2^{\mathrm{pert},N}(V_2, U_3, V_4, U_1) \,.
\end{align}
Despite the limitations of S-matrix perturbation theory, it is possible to recover the exact correlators from $G_{1,2}^{\mathrm{pert}}$ through Borel resummation. Here we summarize the takeaways, leaving the calculations to Appendix \ref{sec:2D-Borel}. Denote the all-orders perturbatively computed correlators by
\begin{align}
    G_{i=1,2}^{\mathrm{pert}} = \sum_{n=0}^{\infty} c_{i,n}(U_1, V_2, U_3, V_4) \, x^n \,.
\end{align}
We then show:
\begin{align}
    G_{1,2} &= \mathcal{S}G_{1,2}^{\mathrm{pert}}\,, \\
    \mathcal{S}G_i^{\mathrm{pert.}} &:= \int_0^{\infty} \mathrm{d}s \, e^{-s} \mathcal{B} G_i^{\mathrm{pert}}(s x) = \int_0^{\infty} \mathrm{d}s \, e^{-s} \sum_{n=0}^{\infty} c_{i,n}(U_1, V_2, U_3, V_4) \frac{(sx)^n}{n!}\,. \label{eq:G1G2-Borel-integral}
\end{align}
The Laplace integrals in (\ref{eq:G1G2-Borel-integral}) are well-defined if the Borel transformed correlators $\mathcal{B}G_i^{\mathrm{pert.}}$ do not have singularities on the real axis. This condition is satisfied whenever $U_1 \neq U_3$ and $V_2\neq V_4$. When the pole of $\mathcal{B} G_i^{\mathrm{pert}}$ crosses the positive real axis (for instance during analytic continuation), the correlators pick up a contour contribution around the pole, which may be computed via residue.

The non-perturbative structure of the correlators is encoded in the $s$-analytic structure of the Borel transformed $\mathcal{B} G_i^{\mathrm{pert}}(s x)$, which scales as $1/x$. In particular, the KMS violation for $x<0$ can be computed as a contour integral around the singularity $s_{\ast}$ of $\mathcal{B} G_2^{\mathrm{pert.}}$:
\begin{align}
    \Delta_{\mathrm{KMS}}(U_1,V_2,U_3,V_4) &= \theta(-x) \int_{C_{s_{\ast}}} \mathrm{d}s\, e^{-s} \mathcal{B} G_2^{\mathrm{pert.}}(s x) \,.
\end{align}


\subsubsection{The \texorpdfstring{$x \rightarrow 0$}{x to 0} Limit} \label{sec:2D-OTOC-analc-xto0}

\noindent In the previous section, we derived expressions at finite $x$ for objects involving the analytic continuations of time- and frequency-domain correlators. These include (i) the matrix elements of a dressed $\cala_{0,A}$ field with a dressed $\cala_{0,A}'$ field (\ref{eq:2D-AtApt-commutator-calc}); (ii) the KMS violation $\Delta_{\mathrm{KMS}}$ (\ref{eq:2D-DeltaKMS-UV-explicit-0}); and (iii) the violation of detailed balance $\Delta_{\mathrm{db.}}$ (\ref{eq:2D-deltaKMS-lambda-explicit}). In this section we compute their $x \rightarrow 0$ limits. Since none of these exact expressions permit a Taylor expansion in $x$, smearing is again important. We will find that after smearing, (i) and (iii) vanish to all perturbative orders, but (ii) can stay finite due to contact terms between smeared fields.\footnote{There are now three notions of `non-perturbativeness' which should not be confused. The first means that computing an object to any fixed order in S-matrix perturbation theory yields zero. Another means that the exact result does not permit a power series expansion in $x\to 0^-$. Yet a third refers to the $x$-dependence of the smeared object. It should always be clear from context to which we refer.}

\vspace{1em}

\noindent \textbf{The Time Domain}

\noindent We first consider the smeared matrix elements of the commutator of a dressed $\cala_{0,A}$ field with a dressed $\cala_{0,A}'$ field.
\begin{align}
    \mathcal{C}^{LR}(f_1^-, f_3^+; f_2, f_4) &:= \int_{-\infty}^{\infty} \mathrm{d}U_1 \mathrm{d}U_3 \mathrm{d}V_2 \mathrm{d}V_4\, f_1^-(U_1) f_2(V_2) f_3^+(U_3) f_4(V_4) \label{eq:2D-AtApt-commutator-smeared-def} \\
    &\hspace{15mm} \times \langle \phi^B(V_2) | \big[S^{-\frac{1}{2}}\phi^A(U_1) S^{\frac{1}{2}}, S^{\frac{1}{2}}\phi^A(U_3) S^{-\frac{1}{2}} \big] | \phi^B(V_4) \rangle \,, \nonumber
\end{align}
where $f_1^-$, $f_3^+$, $f_2$, $f_4$ are smooth test functions with bounded derivatives to all orders, and supported within compact subsets of $(-\infty,0)$, $(0,\infty)$, $\mathbb{R}$, and $\mathbb{R}$ respectively.\footnote{The smeared $B$-operators belong to the the observer algebra $\cala$ only if $f_2$ and $f_4$ are supported on a compact subset of $(-\infty,0)$. However, the proof proceeds even if $f_2$ and $f_4$ have support on the positive axis. We think of (\ref{eq:2D-AtApt-commutator-smeared-def}) as arbitrary $B$-field matrix elements of a dressed $\cala_{0,A}$ field with a dressed $\cala_{0,A}'$ field.} We have:

\begin{proposition}
    The smeared matrix elements (\ref{eq:2D-AtApt-commutator-smeared-def}) of the commutator between a dressed $\cala_{0,A}$ field and dressed $\cala_{0,A}'$ field can be expressed as the following. In particular, this vanishes faster than any power of $x$ in the $x\to 0^-$ limit.
    \begin{align}
        \mathcal{C}^{LR}(f_1^-, f_3^+; f_2, f_4) &= \frac{\theta(-x)}{-2\pi x} \frac{(-1)^{h_A}}{\Gamma(2h_A)\Gamma(2h_B)} \label{eq:2D-AtApt-commutator-smeared} \\
        &\hspace{5mm} \int_{-\infty}^{\infty} \mathrm{d}U_1 \mathrm{d}U_3 \mathrm{d}V_2 \mathrm{d}V_4\, f_1^-(U_1) \partial^{2h_B-1}_{V_2}f_2(V_2) \partial_{U_3}^{2h_A-1} f_3^+(U_3) f_4(V_4) e^{\frac{i}{x} (-U_1+U_3)(V_2-V_4)} \nonumber \\
        &=0 + \mathcal{O}(|x|^n)\,, \qquad \forall \,n\in\mathbb{Z}_{\geq 0}\,.
    \end{align}
\end{proposition}
\begin{proof}
    We begin by substituting (\ref{eq:2D-AtApt-commutator-calc}) into (\ref{eq:2D-AtApt-commutator-smeared}) and integrate by parts. There are no boundary terms because $f_3^+$ has compact support in $(0,\infty)$, and we find
    \begin{align}
        \mathcal{C}^{LR}(f_1^-, f_3^+; f_2, f_4) &= \frac{\theta(-x)}{2\pi i} \frac{ (-1)^{h_A} x^{-2h_B} }{\Gamma(2h_A)\Gamma(2h_B)} \int_{-\infty}^{\infty} \mathrm{d}U_1 \mathrm{d}U_3 \mathrm{d}V_2 \mathrm{d}V_4\, f_1^-(U_1) f_2(V_2) \partial_{U_3}^{2h_A-1} f_3^+(U_3) f_4(V_4) \nonumber \\
        &\hspace{55mm} \times (-U_1+U_3)^{2h_B-1} e^{\frac{i}{x}(-U_1+U_3)(V_2-V_4)} \,.
    \end{align}
    Next we use the identity
    \begin{align}
        e^{\frac{i}{x} (-U_1+U_3)(V_2-V_4)} &= \Big(\frac{-ix}{-U_1+U_3}\Big)^m \partial^m_{V_2} e^{\frac{i}{x} (-U_1+U_3)(V_2-V_4)}\,. \qquad \forall \,m\in\mathbb{Z}_{\geq 0} \,.
    \end{align}
    Substituting this for $m = 2h_B-1+n$ into our above expression yields
    \begin{align}
        \mathcal{C}^{LR}(f_1^-, f_3^+; f_2, f_4) &:= \theta(-x) \frac{(-1)^{h_A} (-i)^n (-x)^{n-1} }{2\pi \Gamma(2h_A)\Gamma(2h_B)} \\
        &\hspace{5mm} \int_{-\infty}^{\infty} \mathrm{d}U_1 \mathrm{d}U_3 \mathrm{d}V_2 \mathrm{d}V_4\, g_n(U_1, V_2, U_3, V_4) e^{\frac{i}{x} (-U_1+U_3)(V_2-V_4)} \,, \nonumber \\
        g_n(U_1, V_2, U_3, V_4) &:= (-U_1+U_3)^{-n} f_1^-(U_1) \partial^{2h_B-1+n}_{V_2}f_2(V_2) \partial_{U_3}^{2h_A-1} f_3^+(U_3) f_4(V_4) \,.
    \end{align}
    Taking $n=0$ yields (\ref{eq:2D-AtApt-commutator-smeared}). We observe that the factor $(-U_1+U_3)^{-n}$ is bounded because $f_1^-$ and $f_3^+$ are supported in compact subsets of $(-\infty,0)$ and $(0,\infty)$. Since the test functions also have bounded derivatives to all orders, we can uniformly bound $|g_n(U_1, V_2, U_3, V_4)| < c_n$ for position-independent constants $c_n$. Finally we let $D\subset \mathbb{R}^4$ be any finite volume region containing the joint support of all four test functions. This lets us bound:
    \begin{align}
        |\mathcal{C}^{LR}(f_1^-, f_3^+; f_2, f_4)| &= \frac{\theta(-x) |x|^{n-1} }{2\pi \Gamma(2h_A)\Gamma(2h_B)} \Big| \int_{-\infty}^{\infty} \mathrm{d}U_1 \mathrm{d}U_3 \mathrm{d}V_2 \mathrm{d}V_4\, g_n(U_1, V_2, U_3, V_4) e^{\frac{i}{x} (-U_1+U_3)(V_2-V_4)} \Big| \nonumber \\
        &\leq \frac{\theta(-x) |x|^{n-1} }{2\pi \Gamma(2h_A)\Gamma(2h_B)} \int_{-\infty}^{\infty} \mathrm{d}U_1 \mathrm{d}U_3 \mathrm{d}V_2 \mathrm{d}V_4\, \Big| g_n(U_1, V_2, U_3, V_4) e^{\frac{i}{x} (-U_1+U_3)(V_2-V_4)} \Big| \nonumber \\
        &\leq \frac{c_n \mathrm{vol}(D)}{2\pi \Gamma(2h_A)\Gamma(2h_B)} \theta(-x) |x|^{n-1}  = 0 + \mathcal{O}(|x|^{n-1}) \,.
    \end{align}
    Since $n$ can be any positive integer, this completes our proof. 
\end{proof}

\noindent Next we turn our attention to the smeared KMS violation. Working in Kruskal coordinates, we have the following result.

\begin{proposition}
\label{eq:prop}
    Consider the KMS violation $\Delta_{\mathrm{KMS}}(U_1, V_2, U_3, V_4)$, smeared by test functions $f_{1,3}^+$, $f_{2,4}^-$ supported in compact subsets of $(0,\infty)$ and $(-\infty,0)$ respectively.\footnote{Our result is also valid if $f_2$ and $f_4$ have compact supports in $\mathbb{R}$, by the same proof. However, since $\Delta_{\mathrm{KMS}}$ is defined in terms of correlators $\cala$ operators, here we restrict $f_2$ and $f_4$ to have support in a compact subset of $(-\infty,0)$.}  This permits the following perturbative expansion,
    \begin{align}
    \Delta_{\mathrm{KMS}}(f) &:= \int_{-\infty}^{\infty} \mathrm{d}U_1 \mathrm{d}U_3 \mathrm{d}V_2 \mathrm{d}V_4\, f_1^+(U_1) f_2^-(V_2) f_3^+(U_3) f_4^-(V_4) \, \Delta_{\mathrm{KMS}}(U_1, V_2, U_3, V_4) \label{eq:2D-DeltaKMS-smeared-def} \\
    &= \theta(-x) \frac{(-1)^{1+h_A+h_B}}{\Gamma(2h_A)\Gamma(2h_B)} \nonumber \\
    &\hspace{5mm} \sum_{n\geq 0} \frac{(i|x|)^n}{\Gamma(1+n)} \int_0^{\infty}\mathrm{d}U (f_1^+)^{(2h_A-1+n)}(U) f_3^+(U) \int_{-\infty}^0 \mathrm{d}V (f_2^-)^{(2h_B-1+n)}(V) f_4^-(V) \label{eq:2D-smeared-DeltaKMS}
\end{align}
\end{proposition}
\begin{proof}
    We begin by substituting the equation (\ref{eq:2D-DeltaKMS-UV-explicit-0}) into the smeared correlator, perform the change of variables $U_{\pm} := U_1 \pm U_3$, $V_{\pm} := V_2 \pm V_4$, and use integration by parts. As in the commutator calculation there are no boundary terms, and we find
    \begin{align}
        \Delta_{\mathrm{KMS}}(f) &= \frac{\theta(-x)}{2\pi x} \frac{(-1)^{h_A+h_B}}{\Gamma(2h_A)\Gamma(2h_B)} \int_0^{\infty}\frac{\mathrm{d}U_+}{2}  \int_{-\infty}^0\frac{\mathrm{d}V_+}{2} \\
        &\hspace{5mm} \int_{\mathbb{R}^2}\mathrm{d}U_- \mathrm{d}V_-  \partial_{U_-}^{2h_A-1}\Big[ f_1^+(\tfrac{U_++U_-}{2}) f_3^+(\tfrac{U_+-U_-}{2})\Big]\partial_{V_-}^{2h_B-1} \Big[ f_2^-(\tfrac{V_++V_-}{2}) f_4^-(\tfrac{V_+-V_-}{2}) \Big] e^{-\frac{i}{x}U_-V_-} \,. \nonumber 
    \end{align}
    Now we use the method of stationary phase. The oscillatory exponential has one saddle point where $U_-=V_-=0$, and we make use of the following identity valid to all orders:
    \begin{align}
        &\int_{\mathbb{R}^2} \mathrm{d}U_- \mathrm{d}V_-\, g_1(U_-)g_2(V_-) e^{\frac{i}{|x|} U_-V_-} = 2\pi |x| \sum_{n\geq 0}\frac{(i|x|)^n}{\Gamma(1+n)} g_1^{(n)}(U_-) g_2^{(n)}(V_-) \,.
    \end{align}
    We thus have
    \begin{align}
        \Delta_{\mathrm{KMS}}(f) &= \frac{\theta(-x)}{2\pi x} \frac{(-1)^{h_A+h_B}}{ \Gamma(2h_A)\Gamma(2h_B)} 2\pi|x| \sum_{n\geq 0} \frac{(i|x|)^n}{\Gamma(1+n)} \\
        &\hspace{5mm} \partial_{U_-}^{2h_A-1+n}\Big[\int_0^{\infty}\frac{\mathrm{d}U_+}{2} f_1^+(\tfrac{U_++U_-}{2}) f_3^+(\tfrac{U_+-U_-}{2})\Big]\partial_{V_-}^{2h_B-1+n} \Big[\int_{-\infty}^0\frac{\mathrm{d}V_+}{2} f_2^-(\tfrac{V_++V_-}{2}) f_4^-(\tfrac{V_+-V_-}{2}) \Big] \nonumber \\
        &= \theta(-x) \frac{(-1)^{1+h_A+h_B}}{\Gamma(2h_A)\Gamma(2h_B)} \\
        &\hspace{5mm} \sum_{n\geq 0} \frac{(i|x|)^n}{\Gamma(1+n)} \int_0^{\infty}\mathrm{d}U (f_1^+)^{(2h_A-1+n)}(U) f_3^+(U) \int_{-\infty}^0 \mathrm{d}V (f_2^-)^{(2h_B-1+n)}(V) f_4^-(V) \,, \nonumber
    \end{align}
    which is our desired result.
\end{proof}
Note that the only contributions to $\Delta_{\mathrm{KMS}}(f)$ in the stationary phase approximation come from contact terms, or regions where the test functions $f_1^+$ and $f_3^+$ (also $f_2^-$ and $f_4^-$) have overlapping support. If either the $A$ or $B$ theory test functions are disjoint, then $\Delta_{\mathrm{KMS}}(f)$ vanishes faster than any power of $x$, and is truly non-perturbative. The calculation of the smeared commutator (\ref{eq:2D-AtApt-commutator-smeared-def}) differed in that $U_1$ was smeared over the negative line with support disjoint from $f_3^+$. On the other hand, in the presence of contact terms we have $\Delta_{\mathrm{KMS}}(f) \propto \theta(-x) x^0$ (except in the fine-tuned case where the $U$ or $V$ integrals in (\ref{eq:2D-smeared-DeltaKMS}) vanish). This is truly perturbative, and stays finite even in the strict $x\to 0^-$ limit.

Finally, recall from (\ref{eq:2D-DeltaKMS-UV-explicit-2}) that the unsmeared $\Delta_{\mathrm{KMS}}(U_1,V_2,U_3,V_4)$ can be written as an exponential $e^{-\frac{i}{x}U_{13}V_{24}}$ times a polynomial in $U_{13}V_{24}/x$. The aforementioned contact terms correspond to a saddle point of the universal phase $\Delta_{\mathrm{KMS}}(U_1, V_2, U_3, V_4) \propto e^{-\frac{i}{x}U_{13}V_{24}}$. In the stationary phase approximation $U_{13}$ and $V_{24}$ are integrated over characteristic lengths $\sim |x|^{1/2}$, so each term of the polynomial contributes to leading order; keeping only the leading order term in $x$ would have resulted in a Laurent series diverging as $x^{1-2h_+}$. All the polynomial coefficients are important in ensuring that the singular coefficients cancel.


\vspace{1em}

\noindent \textbf{The Frequency Domain}

\noindent It remains to smear $\Delta_{\mathrm{db.}}$ and take the $x\to 0$ limit. It is not immediately clear what the smearing functions should be. When we worked with $G_1(\lambda_i)$ in Section \ref{sec:2D-OTOC-realt-xto0}, the natural class of functions were Fourier conjugates $\tilde{f}(\lambda)$ of time-domain test functions $f(t)\in C^{\infty}_c(\mathbb{R})$. However, it follows from the calculation of $\mathcal{F}^{-1}[e^{-2\pi \lambda_1}G_1(\lambda_i)]$ in Appendix \ref{sec:2Dcalcs-OTOCs-exact} that the smearing integrals may diverge. To avoid this problem we will smear with fixed-width Gaussians centered at zero, of the form:\footnote{The following calculation can be generalized to Gaussians of variable width and mean.}
\begin{align}
    g_i(\lambda_i) = e^{-\frac{1}{2a} \lambda_i^2}\,, \qquad a>0\,. \label{eq:test-functions-lambda-Gaussian}
\end{align}
Their Fourier transforms yield position-space Gaussians which are $C^{\infty}$. Although not of compact support, they vanish superpolynomially in both static patch time and Kruskal time at either boundary, so there are no complications with boundary terms. We now compute:
\begin{align}
    \Delta_{\mathrm{db.}}(g) &:= \int_{-\infty}^{\infty} \dbar\lambda_1\cdots \dbar\lambda_4 \, g_1(\lambda_1) g_2(\lambda_2) g_3(\lambda_3) g_4(\lambda_4) \Delta_{\mathrm{db.}}(\lambda_1, \lambda_2, \lambda_3, \lambda_4) \\
    &= \frac{\theta(-x)}{2\pi i} \frac{1}{\Gamma(2h_A)\Gamma(2h_B)} \int \frac{\dbar\lambda\, F_x(\lambda)}{\Gamma(i\lambda+\epsilon)\Gamma(1-i\lambda-\epsilon)} \label{eq:Delta-db-smeared-xto0} \\
    &\hspace{5mm} \times \int \dbar\lambda_2\, e^{\pi\lambda_2} g_2(\lambda_2) g_4(\lambda-\lambda_2) \Gamma(h_B-i\lambda_2) \Gamma(h_B+i(\lambda_2-\lambda)) \nonumber \\
    &\hspace{5mm} \times \int \dbar\lambda_3\, e^{-\pi\lambda_3} g_3(-\lambda_3) g_1(-\lambda+\lambda_3) \Gamma(h_A-i\lambda_3) \Gamma(h_A+i(\lambda_3-\lambda)) \nonumber 
\end{align}
where we substituted (\ref{eq:Deltadb-def}), integrated over $\lambda_4$, and defined $\lambda=-\lambda_1-\lambda_3$. In order to apply Lemma \ref{lemma:Fxlambda-integral}, we will need that the function of $\lambda$ defined by the $\lambda_2$ and $\lambda_3$ integrals is $\lambda$-analytic in $\mathbb{H}_+$, and decays sufficiently rapidly for $\lambda$ on the segments $C^{\pm}_n(\Lambda)$ for $\Lambda\to\infty$. The first property follows from a standard theorem in complex analysis.
\begin{lemma}[Morera]
    Consider a function $f(\lambda)$ defined by a contour integral in the complex plane:
    \begin{align}
        f(\lambda) := \int_C \mathrm{d}z \, \Xi(\lambda,z) \,.
    \end{align}
    Then $f(\lambda)$ is analytic in an open region $D \subseteq \mathbb{C}$ if:
    \begin{enumerate}
        \item For each fixed $z\in C$, the function $\Xi(\lambda,z)$ is analytic on $D$.
        \item For every compact region $K \subseteq D$, there exists an absolutely integrable function $G_K(z)$ such that $|\Xi(\lambda,z)| \leq G_K(z)$, for all $\lambda\in K$ and $z\in C$.
    \end{enumerate}
\end{lemma}
\noindent We wish to apply this for $C = \mathbb{R}$, $D=\mathbb{H}_+$, and $\Xi_2(\lambda,\lambda_2)$, $\Xi_3(\lambda,\lambda_3)$ the $\lambda_2, \lambda_3$ integrands in (\ref{eq:Delta-db-smeared-xto0}). The first condition is immediate, since for $\lambda_{2,3}\in\mathbb{R}$ the $\Gamma$-functions only have poles for $\lambda\in\mathbb{H}_-$. The second condition follows after some algebra from the rapid decay of the Gaussian test functions as $\lambda\to\pm\infty$ along with the $\Gamma$-function vertical-line bound:
\begin{align}
    |\Gamma(z_1 + iz_2)| \leq c_R (1+|z_2|)^{M_R} e^{-\frac{\pi}{2}|z_2|}\,, \qquad \forall\, z_2\in\mathbb{R}\,, \quad z_1 \in R_{\mathrm{compact}} \subset (0,\infty) \,. \label{eq:Gamma-vertical-bound}
\end{align}
The intuition is that when $\lambda\in\mathbb{H}_+$ is bounded, each term contributes at most a growth of $e^{c|\lambda_2|}$, which at large $|\lambda_2|$ is killed by the Gaussian. The function multiplying $F_x(\lambda)$ in (\ref{eq:Delta-db-smeared-xto0}) is therefore $\lambda$-analytic in $\mathbb{H}_+$. We also need to show the second requirement of Lemma \ref{lemma:Fxlambda-integral}. Again using (\ref{eq:Gamma-vertical-bound}), one can show that for $\lambda \in C_N^{\pm}(\Lambda)$, the absolute value of the $\lambda_2$, $\lambda_3$ integrands is bounded by:
\begin{align}
    |\Xi_2(\lambda,\lambda_2)| \leq c_N (1+|\lambda_2|+|\Lambda|)^{m_N} e^{\frac{\pi}{2}|\lambda_2|} e^{-\frac{1}{2a}\lambda_2^2} e^{-\frac{1}{2a}(\lambda_2\pm\Lambda)^2}\,, \\
    |\Xi_3(\lambda,\lambda_3)| \leq c_N' (1+|\lambda_3|+|\Lambda|)^{m_N'} e^{\frac{\pi}{2}|\lambda_3|} e^{-\frac{1}{2a}\lambda_3^2} e^{-\frac{1}{2a}(\lambda_3\pm\Lambda)^2}\,.
\end{align}
The product of the $\lambda_2$, $\lambda_3$ integrals on $C_N^{\pm}(\Lambda)$ are then bounded uniformly by:
\begin{align}
    d_{1,N} (d_{2,N} + \Lambda)^{M_N} e^{-\frac{1}{2a} \Lambda^2 + d_{3,N} \Lambda} \,.
\end{align}
This decays uniformly as a Gaussian in $\Lambda$, therefore the integral (\ref{eq:Delta-db-smeared-xto0}) on the segments $C_N^{\pm}(\Lambda)$ vanishes as $\Lambda\to\infty$. We can therefore use Lemma \ref{lemma:Fxlambda-integral} and apply (\ref{eq:Fx-xto0-identity-orderN}) to (\ref{eq:Delta-db-smeared-xto0}). For fixed $N$, the residues at $\lambda=in$ all vanish due to the $1/\Gamma(i\lambda+\epsilon)$ factor in (\ref{eq:Delta-db-smeared-xto0}). We conclude:
\begin{align}
    \lim\limits_{x\to 0^-} \Delta_{\mathrm{db.}}(g) = 0 + \mathcal{O}(x^{N+1})\,, \qquad \forall\,N\in\mathbb{Z}_+\,.
\end{align}
That is, the $x\to 0$ limit of the smeared violation of the detailed balance condition vanishes at all perturbative orders.


\subsection{The Crossed Product} \label{sec:KMS-crossedprod}

Thus far, we have studied the algebra $\cala$ of dressed early- and late-time operators accessible to an observer. The observer's clock was treated classically and assumed to be perfect. We now ask how the algebra is modified when we quantize the clock. In particular, the clock is viewed as a quantum mechanical system with a given energy spectrum. For our purposes, a ``perfect clock'' has a continuous energy spectrum supported on the entire real line. The modified algebra $\cala_{\mathrm{cr}} = \cala\rtimes \mathbb{R}$ is obtained from the crossed-product construction \cite{Witten:2021unn, CLPW, Daele:1978, DeVuyst:2024khu,DeVuyst:2024fxc,DeVuyst:2024grw}.

The crossed product algebra is generated by the clock's Hamiltonian and QFT operators dressed to the clock. When the clock is not perfect, the matter operators are necessarily smeared in time. For our results to remain physically consistent, it is important that the smearing does not spoil the distinction between early- and late-time QFT operators. Through the energy-time uncertainty relation, this imposes constraints on the density of states of the clock.

We first consider an observer with a perfect clock. We construct the algebra and compute correlators in the Hartle-Hawking state. Next, we briefly discuss the case of an imperfect clock. Then, we show that for $x<0$, the violation of the detailed balance condition implies that the Hartle-Hawking state does not furnish a trace on $\cala_{\mathrm{cr}}$. Nevertheless, the Hartle-Hawking state defines a trace order-by-order in $x$.


\subsubsection{Constructing the Algebra}

Suppose the observer has a perfect clock with Hamiltonian $q$. Its canonical conjugate $p$ obeys  $[q, p] = i$. The operators $q$ and $p$ act on an auxiliary Hilbert space $\mathcal{H}_{\mathrm{obs}} = L^2(\mathbb{R},\mathrm{d}q)$. The total Hilbert space is $\mathcal{H}_{\mathrm{cr}} :=  \mathcal{H}_A \otimes \mathcal{H}_B \otimes \mathcal{H}_{\mathrm{obs}}$. To construct operators in $\cala_{\mathrm{cr}}$, we must impose the constraint $[\cala_{\mathrm{cr}}, H] = 0$, where $H := H_{\mathrm{QFT}}+ q = H_A + H_B + q$. A set of generators consists of  functions of $q$ and the observer-dressed fields $\tilde{\calo}(p) := e^{i H_{\mathrm{QFT}} p}\tilde{\calo}e^{-i H_{\mathrm{QFT}} p}$ for $\tilde{\calo}\in\cala$.\footnote{For example, $\tilde \calo$ may refer to either $\tilde \calo^A$ or $\tilde \calo^B$ in section \ref{sec:2D-AB-theory}.} That is:
\begin{align}
    \cala_{\mathrm{cr}} &= \big\{ \tilde{\calo}(p)\,,\, e^{ikq} \quad \mathrm{s.t.} \quad \tilde{\calo}\in\cala\,, k\in\mathbb{R} \big\}'' \\
    &= \big\{ S^{\frac{1}{2}} a(p) S^{-\frac{1}{2}}, S^{-\frac{1}{2}} b(p) S^{\frac{1}{2}}\,,\, e^{ikq} \quad \mathrm{s.t.} \quad a\in\cala_{0,A}\,, b\in\cala_{0,B}\,, k\in\mathbb{R} \big\}'' \,,
\end{align}
where in the second equality we have used that $[H_A + H_B, P_A P_B] = 0$. The operators $a(p)$ and $b(p)$ denote time-translated operators in $\cala_{0,A}$ and $\cala_{0,B}$.

\vspace{1em}

\noindent The Hartle-Hawking state on $\cala_{\mathrm{cr}}$ is defined by \cite{CLPW}
\begin{align}
    |\Psi_{\mathrm{HH}}\rangle &:= |\Omega\rangle \otimes e^{-\pi q}\,, \qquad p|\Psi_{\mathrm{HH}}\rangle = i\pi |\Psi_{\mathrm{HH}}\rangle \,. \label{crossed-prod-HH-state}
\end{align}
In this state, we compute
\begin{align}
    &\langle \Psi_{\mathrm{HH}} | f_0(q) \tilde{\calo}_1(p) f_1(q) \cdots \tilde{\calo}_n(p) f_n(q) |\Psi_{\mathrm{HH}} \rangle \nonumber \\
    &\hspace{5mm} = \int_{-\infty}^{\infty} \dbar\lambda_1\cdots\dbar\lambda_n\, \langle \Psi_{\mathrm{HH}} | f_0(q) e^{-ip\lambda_1} \tilde{\Phi}_1(\lambda_1) f_1(q) \cdots e^{-ip\lambda_n} \tilde{\Phi}_n(\lambda_n) f_n(q) |\Psi_{\mathrm{HH}} \rangle \\
    &\hspace{5mm} = \int_{-\infty}^{\infty} \dbar\lambda_1\cdots\dbar\lambda_n\, \langle\Omega| \tilde{\Phi}_1(\lambda_1)\cdots \tilde{\Phi}_n(\lambda_n)|\Omega\rangle \nonumber \\
    &\hspace{15mm}\times \langle e^{-\pi q}| f_0(q) f_1(q-\lambda_1) \cdots f_n(q-\lambda_1-\cdots-\lambda_n) |e^{-\pi (q-\lambda_1-\cdots-\lambda_n)} \rangle \\
    &\hspace{5mm}= \int_{-\infty}^{\infty} \mathrm{d}\omega \, e^{-2\pi \omega} f_0(\omega) \int_{-\infty}^{\infty}\dbar\lambda_1\cdots\dbar\lambda_n\, \big(f_1(\lambda_1)\cdots f_n(\lambda_n) \big) \nonumber \\
    &\hspace{55mm} \times \langle\Omega| \tilde{\Phi}_1(\omega-\lambda_1) \tilde{\Phi}_2(\lambda_1-\lambda_2)\cdots \tilde{\Phi}_n(\lambda_{n-1}-\lambda_n) |\Omega\rangle \,. \label{eq:crossed-prod-correlators-qbasis}
\end{align}
In the first equality we Fourier transform all the matter operators to the frequency domain (the transformed operators are denoted by $\tilde{\Phi}_i$'s). In the second equality, we act with the $p$'s to the right. In the third equality, we perform a change of variables. Finally, we use that the QFT correlator in the frequency domain is only supported when $\lambda_1+\cdots+\lambda_n=0$, which follows from time translation invariance of $|\Omega\rangle$. We arrive at an expression for any crossed-product correlator in terms of a smeared correlator of QFT operators in the frequency domain. The smearing can also be written in the time domain:
\begin{align}
    &\langle \Psi_{\mathrm{HH}} | f_0(q) \tilde{\calo}_1(p) f_1(q) \cdots \tilde{\calo}_n(p) f_n(q) |\Psi_{\mathrm{HH}} \rangle \nonumber \\
    &\hspace{5mm} = \int \mathrm{d}t_1\cdots \mathrm{d}t_n \, \langle\Omega| \tilde{\calo}_1(t_1) \cdots \tilde{\calo}_n(t_n) |\Omega\rangle \\
    &\hspace{10mm} \int \mathrm{d}\omega \, f_0(\omega) e^{i\omega(t_1+2\pi i)} \hspace{-1mm} \int \dbar\lambda_1\cdots \dbar\lambda_n \, f_1(\lambda_1) e^{i\lambda_1(t_2-t_1)} \cdots f_{n-1}(\lambda_{n-1}) e^{i\lambda_{n-1}(t_n-t_{n-1})} f_n(\lambda_n) e^{-i\lambda_n t_n} \nonumber \\
    &\hspace{5mm}= \int \mathrm{d}t_1\cdots\mathrm{d}t_n\, F(t_1,\ldots, t_n) \langle\Omega| \tilde{\calo}_1(t_1) \cdots \tilde{\calo}_n(t_n) |\Omega\rangle \,, \label{eq:crossed-prod-correlators-tbasis} \\
    &F(t_1,\ldots, t_n) :=  2\pi F_0(-2\pi i-t_1)F_1(t_1-t_2)\cdots F_{n-1}(t_{n-1}-t_n) F_n(t_n)\,,
\end{align}
where $F_i(t)$ denotes the Fourier transform of $f_i(\lambda)$:
\begin{align}
    F_i(t) = \int_{-\infty}^{\infty} \dbar \lambda \, f_i(\lambda)\, e^{-i\lambda t} \,,
    \qquad
    f_i(\lambda) = \int_{-\infty}^{\infty} \mathrm{d}t\, F_i(t)\, e^{+i\lambda t} \,. \label{eq:crossed-prod-test-fn-Fourier}
\end{align}
In order for $F_0(-2\pi i-t_1)$ to be well-defined, $F_0(z)$ should be analytic for $\mathrm{Im}(z) \in [-2\pi,0]$. Equivalently, the $\omega$-integral in the first equality must converge.


To see the (anti)-scrambling effect, it is important that the early- and late-time matter operators are separated by a large time $T$. As shown above, the dressed QFT operators are necessarily smeared in time. It is important that the functions $F_i(t)$ decay to zero quickly for $t \gtrsim T$, so that the distinction between an early- and late-time operator is not spoiled. Because we are working in a strict limit \eqref{eq:limit} where $T \rightarrow \infty$, there is no need to enforce this condition in practice. Below, we will specialize to the case where the $f_i(q)$ functions are $\delta$-functions,
\begin{equation}
    f_i(q) = \delta(E_i - q) \,.
\end{equation}
We may interpret each $\delta$-function as a Gaussian whose width $\Delta E$ respects the energy-time uncertainty $\Delta E \gtrsim \frac{1}{T}$. The limit $\Delta E \rightarrow 0$ is taken after the $T \rightarrow \infty$ limit \eqref{eq:limit}.

With this in mind, consider the correlator
\begin{equation}
\label{eq:deltafunctioncorr}
\langle \Psi_{\mathrm{HH}} | \tilde{\calo}_1(p) \delta(q-E_1) \cdots \tilde{\calo}_n(p) \delta(q-E_n) |\Psi_{\mathrm{HH}} \rangle \,.
\end{equation}
Each $\delta(q-E)$ may be interpreted as a ``projection'' onto the subspace of $\calh_{\mathrm{cr}}$ in which the clock has energy $E$.\footnote{This interpretation is not mathematically precise, but we will use it to motivate what we mean by a physical correlator.} When the clock is imperfect and has a spectrum given by $\{E_i\}_{i \in \mathbb{N}}$, a physically relevant projection is
\begin{equation}
    \rho_{\mathrm{obs}}(E) := \sum_i \delta(E - E_i) \,.
\end{equation}
The density of states of a physical clock should be bounded from below. For example, consider
\begin{equation}
\label{eq:heaviside}
    \rho_{\mathrm{obs}}(E) = \theta(E) \,.
\end{equation}
CLPW \cite{CLPW} conjugated all operators in the crossed product algebra by \eqref{eq:heaviside} to obtain a type II$_1$ algebra. We will not explicitly modify $\cala_{\mathrm{cr}}$ in this way. Instead, we will simply remark that the correlation function
\begin{equation}
\langle \Psi_{\mathrm{HH}} | \tilde{\calo}_1(p) \rho_{\mathrm{obs}}(q) \cdots \tilde{\calo}_n(p) \rho_{\mathrm{obs}}(q) |\Psi_{\mathrm{HH}} \rangle
\label{eq:physicalcorr}
\end{equation}
is of greater physical relevance than \eqref{eq:deltafunctioncorr}. Below, we will continue to work with \eqref{eq:deltafunctioncorr}. The following discussion may be trivially generalized to include other correlators such as \eqref{eq:physicalcorr}.

\subsubsection{The Hartle-Hawking State is not a Trace}

\noindent We now address the question of whether the Hartle-Hawking state is a trace. When $x>0$, the Hartle-Hawking state is a trace on the crossed product algebra, and in particular obeys the cyclic property ($\mathrm{Tr}\, ab = \mathrm{Tr}\, b a$) due to the KMS condition on $|\Omega\rangle$. To see this, we compute:
\begin{align}
    &\langle \Psi_{\mathrm{HH}} | f_1(q) \tilde{\calo}_1(p) \cdots f_n(q) \tilde{\calo}_n(p) |\Psi_{\mathrm{HH}}\rangle \\
    &\hspace{5mm} = 2\pi \int_{-\infty}^{\infty} \mathrm{d}t_1\cdots\mathrm{d}t_{n-1}\, F_1(-2\pi i-t_1)F_2(t_1-t_2)\cdots F_{n-1}(t_{n-2}-t_{n-1}) F_n(t_{n-1}) \\
    &\hspace{25mm} \times \langle\Omega| \tilde{\calo}_1(t_1) \cdots \tilde{\calo}_{n-1}(t_{n-1}) \tilde{\calo}_n(0) |\Omega\rangle \nonumber \nonumber \\
    &\hspace{5mm} = 2\pi \int_{-\infty}^{\infty} \mathrm{d}t_1\cdots\mathrm{d}t_{n-1}\, F_1(-t_1)F_2(-2\pi i+t_1-t_2)\cdots F_{n-1}(t_{n-2}-t_{n-1}) F_n(t_{n-1}) \\
    &\hspace{25mm} \times \langle\Omega| \tilde{\calo}_1(t_1-2\pi i) \tilde{\calo}_2(t_2) \cdots \tilde{\calo}_{n-1}(t_{n-1}) \tilde{\calo}_n(0) |\Omega\rangle \nonumber \\
    &\hspace{5mm} = 2\pi \int_{-\infty}^{\infty} \mathrm{d}t_1\cdots\mathrm{d}t_{n-1}\, F_1(-t_1)F_2(-2\pi i+t_1-t_2)\cdots F_{n-1}(t_{n-2}-t_{n-1}) F_n(t_{n-1}) \\
    &\hspace{25mm} \times \langle\Omega| \tilde{\calo}_2(t_2) \cdots \tilde{\calo}_{n-1}(t_{n-1}) \tilde{\calo}_n(0) \tilde{\calo}_1(t_1) |\Omega\rangle \nonumber \\
    &\hspace{5mm} = \langle\Psi_{\mathrm{HH}} | f_2(q) \tilde{\calo}_2(p) \cdots f_n(q) \tilde{\calo}_n(p) f_1(q) \tilde{\calo}_1(p) |\Psi_{\mathrm{HH}}\rangle
\end{align}
The first equality follows from (\ref{eq:crossed-prod-correlators-tbasis}). In the second we deform the $t_1$ contour downwards by $2\pi i$; this is allowed because $F_1(z)$ is analytic in $\mathrm{Im}(z) \in [-2\pi,0]$, and the QFT correlator analytic in the KMS strip. The third equality uses the KMS condition, and the fourth follows from a change of variable and the $t$-invariance of the QFT correlator.

\vspace{1em}

\noindent When $x<0$, the trace is no longer cyclic. It is simplest to work in the energy basis. From (\ref{eq:crossed-prod-correlators-qbasis}), we find that
\begin{align}
    &\langle \Psi_{\mathrm{HH}} | \tilde{\calo}_1(p) \delta(q-E_1) \cdots \tilde{\calo}_n(p) \delta(q-E_n) |\Psi_{\mathrm{HH}} \rangle \\
    &\hspace{5mm} = \int_{-\infty}^{\infty} \mathrm{d}\omega \, e^{-2\pi \omega} \langle\Omega| \tilde{\Phi}_1(\omega-E_1) \tilde{\Phi}_2(E_1-E_2) \cdots \tilde{\Phi}_n(E_{n-1} - E_n) |\Omega\rangle \,. \nonumber
\end{align}
To evaluate this for OTOCs, we simply substitute the frequency-space expressions (\ref{eq:2D-G1-lambda-explicit}) and (\ref{eq:2D-G2-lambda-explicit}). We obtain
\begin{align}
    G_1^{\mathrm{cr}}(E_1, E_2, E_3, E_4) &:= \langle \Psi_{\mathrm{HH}} | \tilde \calo^A(p) \delta(q-E_1) \tilde \calo^B(p) \delta(q-E_2) \tilde \calo^A(p) \delta(q-E_3) \tilde \calo^B(p) \delta(q-E_4) |\Psi_{\mathrm{HH}}\rangle \nonumber \\
    &= C_{E_1, E_2, E_3, E_4} e^{-\pi(E_2+E_4)} F_x(E_1-E_2+E_3-E_4) \,, \\
    &\, \nonumber \\
    G_2^{\mathrm{cr}}(E_2, E_3, E_4,E_1) &:= \langle \Psi_{\mathrm{HH}} | \tilde \calo^B(p) \delta(q-E_2) \tilde \calo^A(p) \delta(q-E_3) \tilde \calo^B(p) \delta(q-E_4) \tilde \calo^A(p) \delta(q-E_1) |\Psi_{\mathrm{HH}}\rangle \nonumber \\
    &= C_{E_1, E_2, E_3, E_4} e^{-\pi(E_1+E_3)}e^{\pi\,\mathrm{sgn}(x)\, (E_1-E_2+E_3-E_4)} F_x(E_1-E_2+E_3-E_4) \,,
\end{align}
\begin{align}
    C_{E_1, E_2, E_3, E_4} &:= \frac{\Gamma(h_B-i(E_1-E_2)) \Gamma(h_A+i(E_2-E_3)) \Gamma(h_B-i(E_3-E_4)) \Gamma(h_A+i(E_4-E_1))}{2\pi \Gamma(2h_A)\Gamma(2h_B)}  \,.
\end{align}
The non-cyclicity of the trace $\Delta_{\mathrm{Tr}} := G_1^{\mathrm{cr}} - G_2^{\mathrm{cr}}$ is given by:
\begin{align}
    \Delta_{\mathrm{Tr}} = \theta(-x) C_{E_1, E_2, E_3, E_4} e^{-\pi(E_2-E_4)} (1- e^{-2\pi(E_1-E_2+E_3-E_4)}) F_x(E_1-E_2+E_3-E_4) \,. \label{eq:2D-Deltatr-exact}
\end{align}
This is non-vanishing for $x<0$.

One may wonder if there is another state $|\Psi\rangle = |\Omega\rangle \otimes \rho_{\Psi}^{1/2}(q)$ which defines a cyclic trace for the algebra $\cala_{\mathrm{cr}}$ when $x<0$. We may again compute crossed-product correlators built from dressed out-of-time ordered matter fields with $\delta(q-E_i)$ insertions between; let us call these $G_{1,\rho_{\Psi}}^{\mathrm{cr}}$ and $G_{2,\rho_{\Psi}}^{\mathrm{cr}}$. We find:
\begin{align}
    \Delta_{\mathrm{Tr}}^{\rho_{\Psi}} = C_{E_1,E_2,E_3,E_4} \big(\rho_{\Psi}(E_4) e^{-\pi(E_2-E_4)} - \rho_{\Psi}(E_1) e^{\pi(E_1-E_3)} e^{\pi\,\mathrm{sgn}(x)\, (E_1-E_2+E_3-E_4)}\big) \,.
\end{align}
When $x>0$, this vanishes if and only if $\rho_{\Psi}(E_4) = e^{2\pi(E_1-E_4)}\rho_{\Psi}(E_1)$, which implies the Hartle-Hawking state $\rho_{\Psi}(q) = e^{-2\pi q}$. When $x<0$, cyclicity requires $\rho_{\Psi}(E_4) = e^{2\pi(E_2-E_3)} \rho_{\Psi}(E_1)$, which is clearly impossible since $(E_1, E_2, E_3, E_4)$ are all independent. We conclude that there is no choice of $\rho_{\Psi}$ that defines a trace on $\cala_{\mathrm{cr}}$ when $x<0$. 

\vspace{1em}

\noindent \textbf{The $x \rightarrow 0$ limit}

\noindent Finally, we take the $x\to 0^-$ limit. The expression (\ref{eq:2D-Deltatr-exact}) does not permit a na\"{i}ve expansion, so we need to smear with test functions. For simplicity, we will again use the Gaussians (\ref{eq:test-functions-lambda-Gaussian}) with fixed width and centered at zero, but again it is possible to generalize to Gaussians of differing widths and means. We now compute:
\begin{align}
    \Delta_{\mathrm{Tr}}(g) &:= \int_{-\infty}^{\infty} \dbar E_1\cdots \dbar E_4 \, g_1(E_1) g_2(E_2) g_3(E_3) g_4(E_4) \Delta_{\mathrm{Tr}}(E_1, E_2, E_3, E_4) \\
    &= \theta(-x) \int_{-\infty}^{\infty} \dbar E_1\cdots \dbar E_4 \, e^{-\frac{1}{2a}(E_1^2 + E_2^2 + E_3^2 + E_4^2)} C_{E_1, E_2, E_3, E_4} e^{-\pi(E_2-E_4)} \\
    &\hspace{25mm} \times (1- e^{-2\pi(E_1-E_2+E_3-E_4)}) F_x(E_1-E_2+E_3-E_4) \nonumber \\
    &= \frac{\theta(-x)}{8\pi i} \frac{1}{\Gamma(2h_A)\Gamma(2h_B)} \sqrt{\frac{\pi a}{2}} \int \dbar k\, F_x(k) \frac{e^{-\pi k}e^{-\frac{1}{8a}k^2}}{\Gamma(1-ik-\epsilon)\Gamma(ik+\epsilon)} \\
    &\hspace{10mm} \times \int \dbar k_1 \, e^{-\frac{1}{8a}k_1^2} e^{-\frac{\pi k_1}{2}} \Gamma(h_A - \tfrac{i}{2}(k-k_1)) \Gamma(h_A - \tfrac{i}{2}(k+k_1)) \nonumber \\
    &\hspace{10mm} \times \int \dbar k_2 \, e^{-\frac{1}{8a}k_2^2} e^{\frac{\pi k_2}{2}} \Gamma(h_B - \tfrac{i}{2}(k-k_2)) \Gamma(h_B - \tfrac{i}{2}(k+k_2)) \,, \nonumber 
\end{align}
In the final equality we have make a change of variables $k=\lambda_1-\lambda_2+\lambda_3-\lambda_4$, $k_1=\lambda_1+\lambda_2-\lambda_3-\lambda_4$, $k_2 = \lambda_1-\lambda_2-\lambda_3+\lambda_4$, $k'=\lambda_1$, and perform the integral over $k'$. Using the same techniques as (\ref{eq:Delta-db-smeared-xto0}), it can be shown using Morera's theorem that the $k_1$ and $k_2$ integrals define a function of $\lambda$ analytic in $\mathbb{H}_+$. Moreover, on the vertical segments $k \in C^{\pm}_N(\Lambda)$ the $k_1$ and $k_2$ integrals decay uniformly as a Gaussian in $\Lambda$, satisfying the requirements of Lemma \ref{lemma:Fxlambda-integral}. As with $\Delta_{\mathrm{db.}}$, we find that the $1/\Gamma(ik+\epsilon)$ factor kills all residues at $k=in$, therefore:
\begin{align}
    \mathrm{\Delta}_{\mathrm{Tr}}(g) = 0 + \mathcal{O}(x^{N+1})\,, \qquad \forall\, N\in\mathbb{Z}_{+}\,.
\end{align}

\noindent For more general test functions, the lesson is as follows. The all-orders power series expansion of the OTOC is given by
\begin{equation}
\braket{\tilde \calo^A(t_1) \tilde \calo^B(t_2) \tilde \calo^A(t_3) \tilde \calo^B(t_4)} = \sum_{n = 0}^\infty \frac{(-ix)^n}{n!}\,\frac{(-1)^{h_A+h_B}(2h_A)_n(2h_B)_n\;e^{-\frac{n}{2}(t_1+t_3)}\,e^{\frac{n}{2}(t_2+t_4)}}{\left(2\sinh\frac{t_1-t_3}{2}-i\epsilon\right)^{2h_A+n}\left(2\sinh\frac{t_2-t_4}{2}-i\epsilon\right)^{2h_B+n}} \,,
\end{equation}
which is understood to be an asymptotic series. We would like to apply the smearing formula \eqref{eq:crossed-prod-correlators-tbasis} to each term above individually, to verify the tracial property of the Hartle-Hawking state order-by-order. We find that
\begin{align}
\label{eq:G1eqan}
    &\langle \Psi_{\mathrm{HH}} | \tilde \calo^A(p) f_1(q) \tilde \calo^B(p) f_2(q) \tilde \calo^A(p) f_3(q) \tilde \calo^B(p) f_4(q) |\Psi_{\mathrm{HH}}\rangle \nonumber |_n
    \\
    &\hspace{5mm}= \pi\,\frac{(-ix)^n}{n!}(-1)^{h_A+h_B}(2h_A)_n(2h_B)_n \nonumber
    \\
    &\hspace{10mm} \times\int_{-\infty}^{\infty}\mathrm{d}u\,\mathrm{d}v\,\mathrm{d}s\;
    \frac{F_4\!\left(-2\pi i-\tfrac{v+u}{2}\right)F_1\!\left(\tfrac{v+u}{2}-s\right)F_2\!\left(s-\tfrac{v-u}{2}\right)F_3\!\left(\tfrac{v-u}{2}\right)\,e^{-\frac{n}{2}v}\,e^{\frac{n}{2}s}}
    {\left(2\sinh\frac{u}{2}-i\epsilon\right)^{2h_A+n}\left(2\sinh\frac{s}{2}-i\epsilon\right)^{2h_B+n}},
\end{align}
where a trivial change of variables has been performed. The $F_i$ functions refer to Fourier transforms (\ref{eq:crossed-prod-test-fn-Fourier}). We are primarily interested in the $v$ integral, because its convergence depends on the properties of the $F_i$ functions. To ensure that the integrals in \eqref{eq:G1eqan} converge for all $n \ge 0$, we should impose that either $F_4(k)$ or $F_2(k)$ decays to zero as $k \rightarrow \infty$ faster than any exponential, or that either $F_1(k)$ or $F_3(k)$ decays to zero as $k \rightarrow -\infty$ faster than any exponential. This means that either $f_2(q)$ or $f_4(q)$ should be analytic in the lower-half $q$ plane, or $f_1(q)$ or $f_3(q)$ should be analytic in the upper-half $q$ plane. With this in mind, we may compute \eqref{eq:G1eqan} using
\begin{equation}
    \int \prod_{i = 1}^4 \dbar E_i \, f_i(E_i) \, G_1^{\mathrm{cr}}(E_1, E_2, E_3, E_4) \,.
    \label{eq:G1int}
\end{equation}
For the sake of concreteness, assume that $f_1(q)$ is analytic in the upper-half $q$ plane. Then the integrand of \eqref{eq:G1int} is analytic in the upper-half $E_1$ plane, except for the poles of $F_x(E_1-E_2+E_3-E_4)$. We may produce an asymptotic series by closing the $E_1$ contour in the upper-half plane, or by using Equation \eqref{eq:Fx-xto0-identity-orderN}. The poles are located at $E_1-E_2+E_3-E_4 = i \mathbb{Z}_{\geq 0}$. As is clear from \eqref{eq:2D-Deltatr-exact}, this power series expansion will agree with the analogous series that one may obtain from 
\begin{equation}
    \int \prod_{i = 1}^4 \dbar E_i \, f_i(E_i) \, G_2^{\mathrm{cr}}(E_1, E_2, E_3, E_4) \,.
    \label{eq:G2int}
\end{equation}
We conclude that the Hartle-Hawking state is a trace on the crossed-product algebra at each order in perturbation theory in $x$ whenever this perturbative expansion is well-defined.


\section{Properties of the Algebra} \label{sec:properties}

We now study properties of the observer's algebra $\mathcal{A}$ in our 2D model. In Section \ref{sec:alg-type}, we show that for $x<0$, the vacuum $|\Omega\rangle$ is cyclic but not separating for $\mathcal{A}$, meaning that it is separating but not cyclic for its commutant $\mathcal{A}'$. This differs sharply from the $x>0$ case studied in \cite{Penington:2025hrc}. Then, we rule out a wide class of candidate operators in $\mathcal{A}'$, which supports our conjecture that $\mathcal{A}'$ is trivial for $x<0$. In Section \ref{sec:alg-free}, we show that a free product algebra emerges in the $x \rightarrow -\infty$ limit. The proof is identical to the $x \rightarrow +\infty$ proof provided in \cite{Penington:2025hrc}. In Section \ref{sec:alg-geometric}, we use dS JT gravity to depict spacetimes with multiple out-of-time-ordered shocks. These intuitively explain the Hilbert space structure and the emergence of a nontrivial commutant in the free product limit.


\subsection{The Type of the Algebra} \label{sec:alg-type}

For $x>0$, the authors of \cite{Penington:2025hrc} demonstrated that the gravitational algebra $\mathcal{A}$ and its commutant $\mathcal{A}'$ are Type $\mathrm{III}_1$ algebras, with cyclic and separating vector $|\Omega_A\rangle\otimes|\Omega_B\rangle$. The structure of the algebra for $x<0$ differs substantially. We first show:

\begin{proposition} \label{prop:non-separating}
    For $x<0$, the state $|\Omega\rangle = |\Omega_A\rangle\otimes|\Omega_B\rangle$ is cyclic but not separating for $\cala$. Consequently, $|\Omega\rangle$ is separating but not cyclic for $\cala'$.
\end{proposition}
\begin{proof}
    To see that $|\Omega\rangle$ is cyclic, we note that for any $a_i\in\cala_{0,A}$ and $b_i\in\cala_{0,B}$,
    \begin{align}
        \sum_i \tilde{a}_i \tilde{b}_i |\Omega\rangle = \sum_i S^{1/2} a_i|\Omega_A\rangle\otimes b_i|\Omega_B\rangle \,.
    \end{align}
    Since $|\Omega_A\rangle$, $|\Omega_B\rangle$ are cyclic for $\cala_{0,A}$, $\cala_{0,B}$ and $S^{1/2}$ is a unitary transformation, these states are dense in $\mathcal{H}_A\otimes\mathcal{H}_B$.
    
    Suppose now that $|\Omega\rangle$ is also separating for $\cala$. This allows us to use Theorem \ref{thm:Tomita-Takesaki} to define the Tomita operator, modular operator, and modular conjugation $S_{\Omega} = J_{\Omega}\Delta_{\Omega}^{1/2}$. Following \cite{Penington:2025hrc}, we show that $S_{\Omega}$ agrees with $S_{\Omega_A} \otimes S_{\Omega_B}$ on a dense set of states:
    \begin{align}
        (S_{\Omega_A}\otimes S_{\Omega_B}) \sum_i \tilde{a}_i \tilde{b}_i|\Omega\rangle &= (S_{\Omega_A}\otimes S_{\Omega_B}) e^{ix P_AP_B/2} \sum_i a_i b_i |\Omega\rangle \\
        &= (J_{\Omega_A} \otimes J_{\Omega_B}) (\Delta_{\Omega_A}^{1/2}\otimes  \Delta_{\Omega_B}^{1/2}) e^{ix P_AP_B/2} \sum_i a_i b_i |\Omega\rangle \nonumber \\
        &= (J_{\Omega_A} \otimes J_{\Omega_B}) e^{ix P_AP_B/2} (\Delta_{\Omega_A}^{1/2}\otimes  \Delta_{\Omega_B}^{1/2}) \sum_i a_i b_i |\Omega\rangle \nonumber \\
        &= e^{-ix P_AP_B/2} (J_{\Omega_A} \otimes J_{\Omega_B}) (\Delta_{\Omega_A}^{1/2} \otimes \Delta_{\Omega_B}^{1/2}) \sum_i a_i b_i |\Omega\rangle \nonumber \\
        &= e^{-ix P_AP_B/2} \sum_i S_{\Omega_A} a_i|\Omega_A\rangle\otimes S_{\Omega_B} b_i |\Omega_B\rangle \nonumber \\
        &= e^{-ix P_AP_B/2} \sum_i b_i^{\dagger} a_i^{\dagger} |\Omega\rangle \nonumber \\
        &= \sum_i \tilde{b}_i^{\dagger} \tilde{a}_i^{\dagger} |\Omega\rangle = S_{\Omega} \sum_i \tilde{a}_i \tilde{b}_i|\Omega\rangle \,.
    \end{align}
    In the third and fourth equalities we use the following identities \cite{Leutheusser:2021frk} for the half-sided Hamiltonians:
    \begin{align}
        &[K_{\Omega_A}, P_A] =  2\pi i P_A\,, \qquad [K_{\Omega_B}, P_B] = -2\pi i P_B\,, \qquad [J_{\Omega_A}, P_A] = [J_{\Omega_B}, P_B] = 0 \,, \\
        &[\Delta_{\Omega_A}^{1/2} \otimes \Delta_{\Omega_B}^{1/2}, P_A P_B] = [e^{-\frac{1}{2}(K_{\Omega_A} + K_{\Omega_B})}, P_A P_B] = 0 \,.
    \end{align}
    In particular, this means $\Delta_{\Omega}$ agrees with $\Delta_{\Omega_A}\otimes \Delta_{\Omega_B}$ on this set of states.

    \vspace{1em}
    
    \noindent Finally, we compute the following correlator in two ways:
    \begin{align}
        \langle \Omega| \tilde{\Phi}^A(\lambda_1) S_{\Omega} (\tilde{\Phi}^B(\lambda_2) \tilde{\Phi}^A(\lambda_3) \tilde{\Phi}^B(\lambda_4))^{\dagger}|\Omega\rangle &= \langle \Omega| \tilde{\Phi}^A(\lambda_1) S_{\Omega} (S_{\Omega}^{\dagger})^2 (\tilde{\Phi}^B(\lambda_2) \tilde{\Phi}^A(\lambda_3) \tilde{\Phi}^B(\lambda_4))^{\dagger} |\Omega\rangle \\
        &= \langle \Omega| \tilde{\Phi}^A(\lambda_1) \Delta_{\Omega} S_{\Omega}^{\dagger} (\tilde{\Phi}^B(\lambda_2) \tilde{\Phi}^A(\lambda_3) \tilde{\Phi}^B(\lambda_4))^{\dagger} |\Omega\rangle \\
        &= \langle \Omega| (\tilde{\Phi}^B(\lambda_2) \tilde{\Phi}^A(\lambda_3) \tilde{\Phi}^B(\lambda_4)) \Delta_{\Omega} \tilde{\Phi}^A(\lambda_1) |\Omega\rangle \\
        &= e^{2\pi\lambda_1} \langle \Omega| \tilde{\Phi}^B(\lambda_2) \tilde{\Phi}^A(\lambda_3) \tilde{\Phi}^B(\lambda_4) \tilde{\Phi}^A(\lambda_1) |\Omega\rangle \,, \label{eq:alg-Tomita-1} \\
        \langle \Omega| \tilde{\Phi}^A(\lambda_1) S_{\Omega} (\tilde{\Phi}^B(\lambda_2) \tilde{\Phi}^A(\lambda_3) \tilde{\Phi}^B(\lambda_4))^{\dagger}|\Omega\rangle &= \langle \Omega| \tilde{\Phi}^A(\lambda_1) \tilde{\Phi}^B(\lambda_2) \tilde{\Phi}^A(\lambda_3) \tilde{\Phi}^B(\lambda_4) |\Omega\rangle \,. \label{eq:alg-Tomita-2}
    \end{align}
    In the first equality we used that $(S_{\Omega}^{\dagger})^2=1$, in the second that $\Delta_{\Omega} = S_{\Omega}S_{\Omega}^{\dagger}$, and in the third that $\langle \Omega|\tilde{O}_1 S_{\Omega}^{\dagger}\tilde{O}_2|\Omega\rangle = \langle\Omega| \tilde{O}_2^{\dagger} \tilde{O}_1|\Omega\rangle$. In the fourth we have used that $\Delta_{\Omega}$ agrees with $\Delta_{\Omega_A}\otimes \Delta_{\Omega_B}$ on states $\tilde{a}|\Omega\rangle$, along with the KMS condition (\ref{eq:KMS-dS-lambda}) on the $\cala_{0,A}$ algebra:
    \begin{alignat}{2}
        \Delta_{\Omega}\tilde{\Phi}^A(\lambda)|\Omega\rangle &= (\Delta_{\Omega_A}\otimes\Delta_{\Omega_B}) \Phi^A(\lambda)|\Omega\rangle &&= (\Delta_{\Omega_A} \Phi^A(\lambda)|\Omega_A\rangle)\otimes |\Omega_B\rangle \\
        &= e^{2\pi\lambda}\Phi^A(\lambda)|\Omega\rangle &&= e^{2\pi\lambda}\tilde{\Phi}^A(\lambda)|\Omega\rangle \,.
    \end{alignat}
    Together, (\ref{eq:alg-Tomita-1}) and (\ref{eq:alg-Tomita-2}) imply the KMS condition for the frequency-space OTOC:
    \begin{align}
        G_1(\lambda_1,\lambda_2,\lambda_3,\lambda_4) &= \langle \tilde{\Phi}^A(\lambda_1) \tilde{\Phi}^B(\lambda_2) \tilde{\Phi}^A(\lambda_3) \tilde{\Phi}^B(\lambda_4) \rangle \\
        &= e^{2\pi\lambda_1} \langle \tilde{\Phi}^B(\lambda_2) \tilde{\Phi}^A(\lambda_3) \tilde{\Phi}^B(\lambda_4) \tilde{\Phi}^A(\lambda_1) \rangle = e^{2\pi\lambda_1} G_2(\lambda_1,\lambda_2,\lambda_3,\lambda_4) \,.
    \end{align}
    However, this contradicts the result of the explicit calculation (\ref{eq:2D-deltaKMS-lambda-explicit}), and we conclude that $|\Omega\rangle$ cannot be separating for $\cala$. 
\end{proof}
We can think of this result as a bound on the size of the commutant. Since $|\Omega\rangle$ is not cyclic for $\mathcal{A}'$, there are simply not enough operators for $\mathcal{A}'|\Omega\rangle$ to densely populate the Hilbert space $\mathcal{H}_A\otimes\mathcal{H}_B$. In fact, it is difficult to find non-trivial operators in $\mathcal{A}'$ at all. In comparison, for $x>0$ the commutant is generated by gravitationally dressing $\mathcal{A}_{0,A}'$ and $\mathcal{A}_{0,B}'$, but we have demonstrated in Section \ref{sec:KMS-aQFT} that for $x<0$ these operators do not commute with operators in $\mathcal{A}$. To further constrain $\mathcal{A}'$ systematically, we make use of the following result. 
\begin{lemma} \label{lemma:commutant-op-action}
    Suppose that $O\in \mathcal{A}'$. Then the action of $O$ on the Hilbert space $\mathcal{H} = \mathcal{H}_A\otimes\mathcal{H}_B$ is completely fixed by the action of $O$ on the vacuum, $O|\Omega\rangle$.
\end{lemma}
\begin{proof}
    We must specify the action of $O$ on a basis of $\mathcal{H}$ in terms of $O|\Omega\rangle$. For a chiral CFT with primaries $\phi_i$, a basis is furnished by states $\{|\Omega\rangle, |\lambda_A,\Omega_B\rangle_i, |\Omega_A,\lambda_B\rangle_j, |\lambda_A,\lambda_B\rangle_{ij}\}$. We recall the rescaled fields (\ref{eq:Phitovarphi}):
    \begin{alignat}{3}
        \varphi_i^A(\lambda) &:= \frac{1}{\langle\lambda|\phi_i^A(1)|\Omega_A\rangle} \Phi_i^A(\lambda)\,, \quad &&\varphi^A_i(\lambda) |\Omega\rangle = |\lambda_A,\Omega_B\rangle_i\,, \quad &&\langle\Omega|\varphi_i^A(\lambda) = e^{-i\pi(h_i+i\lambda)} {}_i\langle-\lambda,\Omega_B| \,, \nonumber \\
        \varphi_j^B(\lambda) &:= \frac{1}{\langle\lambda|\phi_j^B(-1)|\Omega_B\rangle} \Phi_j^B(-\lambda)\,, \quad &&\varphi^B_j(\lambda) |\Omega\rangle = |\Omega_A,\lambda_B\rangle_j\,, \quad &&\langle\Omega|\varphi_j^B(\lambda) = e^{+i\pi(h_j+i\lambda)} {}_j\langle\Omega_A,-\lambda| \,, \label{eq:2D-rescaled-fields}
    \end{alignat}
    and denote their dressed versions by $\tilde{\varphi}_i^A(\lambda)$ and $\tilde{\varphi}_j^B(\lambda)$. The action of $O$ on the states $|\lambda_A,\Omega_B\rangle_i$ and $|\Omega_A,\lambda_B\rangle_j$ is given by
    \begin{align}
        O |\lambda_A,\Omega_B\rangle_i = \varphi^A_i(\lambda_A) O |\Omega\rangle, \qquad O |\Omega_A,\lambda_B\rangle_j = \varphi^B_j(\lambda_B) O |\Omega\rangle \,.
    \end{align}
    It remains to compute $O|\lambda_1,\lambda_2\rangle_{ij}$. This can be done in two ways. We start with:
    \begin{align}
        \tilde{\varphi}_i^A(\lambda_1)\tilde{\varphi}_j^B(\lambda_2) O |\Omega\rangle &= O \tilde{\varphi}_i^A(\lambda_1)\tilde{\varphi}_j^B(\lambda_2) |\Omega\rangle = O e^{ix P_A P_B/2} |\lambda_1,\lambda_2\rangle_{ij} \nonumber \\
        &= \int \dbar\lambda\, F_{x/2}(\lambda) O |\lambda_1+\lambda\rangle_i |\lambda_2+\lambda\rangle_j \label{eq:alg-ABO-step} \\
        \tilde{\varphi}_j^B(\lambda_2)\tilde{\varphi}_i^A(\lambda_1) O |\Omega\rangle &= O \tilde{\varphi}_j^B(\lambda_2)\tilde{\varphi}_i^A(\lambda_1) |\Omega\rangle = O e^{-ix P_A P_B/2} |\lambda_1,\lambda_2\rangle_{ij} \nonumber \\
        &= \int \dbar\lambda\, F_{-x/2}(\lambda) O |\lambda_1+\lambda\rangle_i |\lambda_2+\lambda\rangle_j \label{eq:alg-BAO-step}
    \end{align}
    These equations can be inverted using the following identity:
    \begin{align}
        f_1(\lambda_1,\lambda_2) = \int\dbar\lambda \, F_x(\lambda) f_2(\lambda_1+\lambda, \lambda_2+\lambda) \ \Longleftrightarrow\  f_2(\lambda_1,\lambda_2) = \int \dbar\lambda' \, F_{-x}(\lambda') f_1(\lambda_1+\lambda', \lambda_2+\lambda') \,.
    \end{align}
    Applying this to (\ref{eq:alg-ABO-step})--(\ref{eq:alg-BAO-step}) yields two ways to express $O|\lambda_1,\lambda_2\rangle$ in terms of $O|\Omega\rangle$.
    \begin{alignat}{2}
        O|\lambda_1,\lambda_2\rangle_{ij} &= \int \dbar\lambda \, F_{-x/2}(\lambda) \, \tilde{\varphi}_i^A(\lambda_1+\lambda)\tilde{\varphi}_j^B(\lambda_2+\lambda) \, O |\Omega\rangle \label{eq:alg-KOI-step} \\
        &= \int \dbar\lambda \, F_{x/2}(\lambda) \, \tilde{\varphi}_j^B(\lambda_2+\lambda)\tilde{\varphi}_i^A(\lambda_1+\lambda) \, O |\Omega\rangle\,. \label{eq:alg-KOII-step}
    \end{alignat}
    This proves the lemma. Of course, the representations (\ref{eq:alg-KOI-step}) and (\ref{eq:alg-KOII-step}) must be equivalent, but we will see this is non-trivial. Matching them yields the following result.
\end{proof}
\begin{corollary}
    For any $O\in\mathcal{A}'$ we must have:
    \begin{align}
        \mathcal{K}_{\mathrm{I}, O}^{\lambda_1\lambda_2;ij} O|\Omega\rangle &= \mathcal{K}_{\mathrm{II}, O}^{\lambda_1\lambda_2;ij} O|\Omega\rangle \,. \label{eq:alg-KOI-KOII-constraint}
    \end{align}
    for all frequencies $\lambda_1$,$\lambda_2$ and primary labels $i$, $j$. Here we have defined:
    \begin{align}
        \mathcal{K}_{\mathrm{I}, O}^{\lambda_1\lambda_2;ij} &:= \int \dbar\lambda \, F_{-x/2}(\lambda) \, \tilde{\varphi}_i^A(\lambda_1+\lambda)\tilde{\varphi}_j^B(\lambda_2+\lambda) \,, \label{eq:alg-KOI-def} \\
        \mathcal{K}_{\mathrm{II}, O}^{\lambda_1\lambda_2;ij} &:= \int \dbar\lambda \, F_{x/2}(\lambda) \, \tilde{\varphi}_j^B(\lambda_2+\lambda)\tilde{\varphi}_i^A(\lambda_1+\lambda) \,, \label{eq:alg-KOII-def}  \\
        \Delta\mathcal{K}_O^{\lambda_1\lambda_2;ij} &:= \mathcal{K}_{\mathrm{I}, O}^{\lambda_1\lambda_2;ij} - \mathcal{K}_{\mathrm{II}, O}^{\lambda_1\lambda_2;ij} \,.
    \end{align}
\end{corollary}

\noindent To employ this relation, we first consider a complete basis of states $\mathcal{B} = \{|\alpha,\beta\rangle\}$ of $\mathcal{H}$. We say that $O|\Omega\rangle$ has `spectral support' in some subset $\mathcal{B}_0 \subseteq \mathcal{B}$ if $\langle \alpha', \beta'|O|\Omega\rangle=0$ for all $|\alpha',\beta'\rangle \in \mathcal{B}\backslash \mathcal{B}_0$. This is denoted by $\mathop{\mathrm{supp}}_{\mathcal{B}}(O|\Omega\rangle) \subseteq \mathcal{B}_0$. We now use (\ref{eq:alg-KOI-KOII-constraint}) to constrain the spectral support of $O|\Omega\rangle$ for any $O\in\mathcal{A}'$. This provides constraints on $\mathcal{A}'$ itself. For example, note that $\mathcal{A}'$ is trivial if and only if $O|\Omega\rangle$ has spectral support in $\{|\Omega\rangle\}$. It is simplest to work in the frequency basis, where $|\alpha\rangle\in \{|\Omega_A\rangle,|\lambda_A\rangle_i \}$ and $|\beta\rangle \in \{|\Omega_B\rangle, |\lambda_B\rangle_j\}$. We remark on other bases in Appendix \ref{sec:commutator-calcs}. 

\begin{proposition} \label{prop:commutant-AorB-constraint}
    Let $x<0$ and $O\in\mathcal{A}'$. If $O|\Omega\rangle$ can be written as a linear combination of states of the form $|\Omega\rangle$, $|\lambda_A,\Omega_B\rangle_k$, and $|\Omega_A,\lambda_B\rangle_k$, then $O$ is trivial.
\end{proposition}
\begin{proof}
    Suppose:
    \begin{align}
        O|\Omega\rangle = c_{\Omega}|\Omega\rangle + \sum_{k} \int\dbar\lambda_A \, c^A_{\lambda_A,k} |\lambda_A,\Omega_B\rangle_k + \sum_k \int\dbar\lambda_B \, c^B_{\lambda_B,k} |\Omega_A, \lambda_B\rangle_k \label{eq:alg-O-ansatz-AorB}
    \end{align}
    Without loss of generality we can assume $c_{\Omega}=0$, since we could apply the following proof to $O - c_{\Omega}\mathds{1}$ instead (by Lemma \ref{lemma:commutant-op-action} the coefficients in (\ref{eq:alg-O-ansatz-AorB}) uniquely define an operator). Our strategy is to compare the matrix elements on both sides of (\ref{eq:alg-KOI-KOII-constraint}) to constrain the coefficients $c^A_{\lambda_A,k}$ and $c^B_{\lambda_B,k}$. We will need the following:
    \begin{align}
        {}_{k'}\langle \Omega_A, -\lambda_B'| \mathcal{K}_{\mathrm{I}, O}^{\lambda_1\lambda_2;ij} |\Omega_A, \lambda_B\rangle_k &= {}_{k'}\langle \Omega_A, -\lambda_B'| \mathcal{K}_{\mathrm{II}, O}^{\lambda_1\lambda_2;ij} |\Omega_A, \lambda_B\rangle_{k} = 0\,, \\
        {}_{k'}\langle \Omega_A, -\lambda_B'| \mathcal{K}_{\mathrm{I}, O}^{\lambda_1\lambda_2;ij} |\lambda_A, \Omega_B\rangle_k  &= e^{-i\pi (h_{k'}+i\lambda_B')} \int \dbar\lambda \, F_{-\frac{x}{2}}(\lambda) \, \langle \tilde{\varphi}_{k'}^B(\lambda_B') \tilde{\varphi}_i^A(\lambda_1+\lambda)\tilde{\varphi}_j^B(\lambda_2+\lambda) \tilde{\varphi}_k^A(\lambda_A) \rangle \,, \\
        {}_{k'}\langle \Omega_A, -\lambda_B'| \mathcal{K}_{\mathrm{II}, O}^{\lambda_1\lambda_2;ij} |\lambda_A, \Omega_B\rangle_k  &= e^{-i\pi (h_{k'}+i\lambda_B')} \int \dbar\lambda \, F_{\frac{x}{2}}(\lambda) \, \langle \tilde{\varphi}_{k'}^B(\lambda_B') \tilde{\varphi}_j^B(\lambda_2+\lambda) \tilde{\varphi}_i^A(\lambda_1+\lambda) \tilde{\varphi}_k^A(\lambda_A) \rangle \,.
    \end{align}
    Both matrix elements in the first line can be written as correlators involving three $B$ fields and one $A$ field using (\ref{eq:alg-KOI-def})--(\ref{eq:alg-KOII-def}), and therefore vanish. The right-hand side of the second line is an OTOC we have already computed in (\ref{eq:2D-G2-lambda-explicit}). The correlator in the third line has no non-trivial S-matrix insertions and factorizes into a product of two-point functions, which in the $\lambda$-basis are just $\delta$-functions. Both the second and third lines vanish unless $h_i = h_k$ and $h_j = h_{k'}$. The remaining integrals are straightforwards to evaluate:
    \begin{align}
        &e^{i\pi (h_j+i\lambda_B')} \times {}_{k'}\langle \Omega_A, -\lambda_B'| \mathcal{K}_{\mathrm{I}, O}^{\lambda_1\lambda_2;ij} |\lambda_A, \Omega_B\rangle_k \nonumber \\
        &\hspace{5mm}= (-1)^{h_A+h_B} e^{\pi(\lambda_1-\lambda_B')} (2\pi\delta(\lambda_1+\lambda_A-\lambda_2-\lambda_B')) \int \dbar\lambda\, e^{\pi\lambda} F_{-x/2}(\lambda) F_{-x}(-\lambda-\lambda_1-\lambda_A) \\
        &\hspace{5mm}= 2\pi \delta(\lambda_1+\lambda_A-\lambda_2-\lambda_B') (-1)^{h_i+h_j} e^{-\pi(\lambda_A+\lambda_B')} e^{\pi (\lambda_1+\lambda_A)(1-\mathrm{sgn}(x))} F_{x/2}(-\lambda_1-\lambda_A) \,, \\
        &\, \nonumber \\
        &e^{i\pi (h_j+i\lambda_B')} \times {}_{k'}\langle \Omega_A, -\lambda_B'| \mathcal{K}_{\mathrm{II}, O}^{\lambda_1\lambda_2;ij} |\lambda_A, \Omega_B\rangle_k \nonumber \\
        &\hspace{5mm}= 2\pi \delta(\lambda_1+\lambda_A-\lambda_2-\lambda_B') (-1)^{h_i+h_j} e^{-\pi(\lambda_A+\lambda_B')} F_{x/2}(-\lambda_1-\lambda_A) \,,
    \end{align}
    where in the second equality we have used the Mellin-Barnes identity (\ref{eq:2D-MB-identity1}). These match for $x>0$, but not for $x<0$:
    \begin{align}
        &{}_{k'}\langle \Omega_A, -\lambda_B'| \Delta\mathcal{K}_O^{\lambda_1\lambda_2;ij} |\lambda_A, \Omega_B\rangle_k \nonumber \\
        &\hspace{5mm}= \theta(-x)\, 2\pi(\lambda_1+\lambda_A-\lambda_2-\lambda_B') \delta_{ik}\delta_{jk'} (-1)^{h_i} e^{\pi\lambda_1} 2\sinh(\pi(\lambda_1+\lambda_A)) F_{x/2}(-\lambda_1-\lambda_A) \,. \label{eq:DeltaKO-BA-matrixel}
    \end{align}
    \noindent Finally, we use these results together with (\ref{eq:alg-O-ansatz-AorB}) to compute:
    \begin{align}
        &{}_{k'}\langle \Omega_A, -\lambda_B'| \Delta\mathcal{K}_O^{\lambda_1\lambda_2;ij} O|\Omega\rangle \\
        &\hspace{5mm}= \sum_{k} \int\dbar\lambda_A \, c^A_{\lambda_A,k}\times  {}_{k'}\langle \Omega_A, -\lambda_B'|  \Delta\mathcal{K}_O^{\lambda_1\lambda_2;ij} |\lambda_A,\Omega_B\rangle_k + \sum_{k} \int\dbar\lambda_B \, c^B_{\lambda_B,k}\,  \times 0 \nonumber \\
        &\hspace{5mm} = \theta(-x) \delta_{jk'} (-1)^{h_i} e^{\pi\lambda_1} 2\sinh(\pi(\lambda_2+\lambda_B')) F_{x/2}(-\lambda_2-\lambda_B') \, c^A_{-\lambda_1+\lambda_2+\lambda_B',i} \,. \nonumber
    \end{align}
    Since $i,j,k',\lambda_1,\lambda_2,\lambda_B'$ are all arbitrary, the constraint (\ref{eq:alg-KOI-KOII-constraint}) demands that for $x<0$, all coefficients $c^A_{\lambda_A,k}$ must vanish. A very similar calculation using (\ref{eq:2D-G1-lambda-explicit}) yields:
    \begin{align}
        &{}_{k'}\langle -\lambda_A', \Omega_B| \Delta\mathcal{K}_O^{\lambda_1\lambda_2;ij} |\Omega_A,\lambda_B\rangle_k \nonumber \\
        &\hspace{5mm}= \theta(-x)\, 2\pi(\lambda_1+\lambda_A-\lambda_2-\lambda_B') \delta_{ik'}\delta_{jk} (-1)^{h_j} e^{-\pi\lambda_2} 2\sinh(\pi(\lambda_2+\lambda_B)) F_{-\frac{x}{2}}(-\lambda_2-\lambda_B) \,. \label{eq:DeltaKO-AB-matrixel}
    \end{align}
    Evaluating ${}_{k'}\langle -\lambda_A',\Omega_B| \Delta\mathcal{K}_O^{\lambda_1\lambda_2;ij} O|\Omega\rangle=0$ with the ansatz (\ref{eq:alg-O-ansatz-AorB}) now forces all coefficients $c^B_{\lambda_B,k}$ to vanish. We conclude that any $O\in\mathcal{A}'$ of the form (\ref{eq:alg-O-ansatz-AorB}) must be trivial.
\end{proof}

\noindent We have worked with a restricted ansatz. By Lemma \ref{lemma:commutant-op-action}, the most general operator acting on $\mathcal{H}$ is (uniquely) specified by:
\begin{alignat}{2}
    O|\Omega\rangle &= c_{\Omega}|\Omega\rangle &&+ \sum_{k} \int\dbar\lambda_A \, c^A_{\lambda_A,k} |\lambda_A,\Omega_B\rangle_k + \sum_l \int\dbar\lambda_B \, c^B_{\lambda_B,l} |\Omega_A, \lambda_B\rangle_l \label{eq:alg-O-ansatz-general} \\
    & &&+ \sum_{k,l} \int\dbar\lambda_A\dbar\lambda_B \, c^{AB}_{\lambda_A,\lambda_B;k,l} |\lambda_A,\lambda_B\rangle_{kl} \,. \nonumber
\end{alignat}
A proof that $\mathcal{A}'$ is trivial would require showing that $O\in\mathcal{A}'$ implies the coefficients $c^A_{\lambda_A,k}$, $c^B_{\lambda_B,l}$, and $c^{AB}_{\lambda_A,\lambda_B;k,l}$ are all vanishing. We believe this is possible using the same strategy as Proposition \ref{prop:commutant-AorB-constraint}--that is, by comparing matrix elements on both sides of (\ref{eq:alg-KOI-KOII-constraint}) with the ansatz (\ref{eq:alg-O-ansatz-general}) to derive constraints on the coefficients. This is technically challenging, since the matrix elements involve integrals of correlators with five and six fields, and so contain the frequency three-point function (\ref{eq:threepoint-lambda}). In Appendix \ref{sec:commutator-calcs} we show that in cases where the computation is tractable, matrix elements of $\mathcal{K}_{\mathrm{I}, O}^{\lambda_1\lambda_2;ij}$ and $\mathcal{K}_{\mathrm{II}, O}^{\lambda_1\lambda_2;ij}$ are equal when $x>0$ and different when $x<0$. This allows us to generalize Proposition \ref{prop:commutant-AorB-constraint} to the following.
\begin{restatable}{proposition}{commutantthreefieldconstraint}
    \label{prop:commutant-3field-constraint}
    For $x<0$ and any non-trivial $O\in\mathcal{A}'$, $O|\Omega\rangle$ cannot have spectral support in any of the following sets:
    \begin{align}
        \mathcal{S}_1 \cup \{|\lambda_A^{\ast},\lambda_B\rangle_{k^{\ast},l} \} \,, \quad \mathcal{S}_1 \cup \{|\lambda_A,\lambda_B^{\ast}\rangle_{k^{\ast},l} \} \,, \quad \mathcal{S}_1 \cup \{|\lambda_A^{\ast} \,,\lambda_B\rangle_{k,l^{\ast}} \}\,, \quad \mathcal{S}_1 \cup \{|\lambda_A,\lambda_B^{\ast}\rangle_{k,l^{\ast}} \}  \,, \label{eq:commutant-3field-constraint}
    \end{align}
    where $\mathcal{S}_1 := \{|\Omega\rangle, |\lambda_A,\Omega_B\rangle_k, |\Omega_A,\lambda_B\rangle_l\}$. Here the starred frequencies and primary indices can each take any one value, but are not summed over.
\end{restatable}

\noindent Propositions \ref{prop:non-separating}--\ref{prop:commutant-3field-constraint} together provide support for the following conjecture.
\begin{conjecture}
    For $x<0$, the commutant is trivial, $\cala' = \{\mathds{1}\}$. The algebra $\cala$ thus consists of all bounded operators $\mathcal{B}(\mathcal{H})$, and is of Type $\mathrm{I}$. 
\end{conjecture}

\subsection{The Free Product Limit} \label{sec:alg-free}

Until now, we have argued that for finite $x < 0$, the observer's algebra $\cala$ has a trivial commutant, and is thus type I. We now consider the limit $x \rightarrow - \infty$, in which the time separation between early- and late-time operators is much larger than the scrambling time. We will show that a type III$_1$ algebra emerges. This algebra is the free product of the QFT algebras $\cala_{0,A}$ and $\cala_{0,B}$, which we denote as $\cala_{0,A} * \cala_{0,B}$. We will first discuss the mathematical result, and then we will interpret it using Penrose diagrams of de Sitter space with multiple out-of-time-ordered shocks.

The analogous results in AdS were already obtained by \cite{Penington:2025hrc}, following the work of \cite{Chandrasekaran:2022eqq}. In particular, their proof of the emergent free product algebra when $x \rightarrow +\infty$ directly applies to the $x \rightarrow - \infty$ case as well. We now summarize the key steps of the proof, and refer the reader to \cite{Penington:2025hrc} for more details.

We would like to study the $x \rightarrow -\infty$ limit of correlators of early- and late-time operators. Recall that an early-time operator is given by $\tilde a$ (defined in \eqref{eq:tildeops}) for $a \in \cala_{0,A}$, while a late-time operator is given by $\tilde b$ for $b \in \cala_{0,B}$. In the proof, we restrict our attention to $a \in \hat \cala_{0,A}$ and $b \in \hat \cala_{0,B}$, where $\hat \cala_{0,A}$ is defined to be the sub-algebra of $\cala_{0,A}$ that only contains operators with compact support on the horizon, and $\hat \cala_{0,B}$ is defined analogously. For example, $\hat \cala_{0,A}$ contains bounded functions of the smeared field $\int \mathrm{d} U \, f(U) \, \tilde{\phi}^A(U)$, where the smearing function $f(U)$ has compact support on $U > 0$. The sub-algebra $\hat \cala_{0,A}$ (resp. $\hat \cala_{0,B}$) is dense in $\cala_{0,A}$ (resp. $\cala_{0,B}$) in the strong operator topology. The key technical advantage of restricting to these sub-algebras is that operators that are separated by a large modular translation exhibit clustering. That is, consider a set of operators $\{ a_i \}_{i \in \mathbb{N}} \subset \hat \cala_{0,A}$. Define $a[s] := e^{-i s P_A} \,a\, e^{i s P_A}$ to represent a modular translation of $a$, where $s \in \mathbb{R}$. Let $\{ s_i(\lambda) \}_{i \in \mathbb{N}}$ be a family of functions of $\lambda \in \mathbb{R}$ with the property that $\lim_{\lambda \rightarrow \infty} (s_i(\lambda) - s_j(\lambda)) \in \{-\infty,\infty \} \cup \mathbb{R}$. Due to clustering, we have that
\footnote{
    We have introduced new notation on the right hand side. Given a partition $\pi$ of $\{1,\cdots,n\}$, define
    \begin{align}
        \langle  \prod_{i = 1}^n a_i \rangle_\pi := \prod_{I} \braket{\Omega | \prod_{i \in I} a_i | \Omega},
    \end{align}
    where $I$ runs over the subsets of this partition. The products above are ordered so that the $i$ index increases from left-to-right.
}
\begin{equation}
\lim_{\lambda \rightarrow \infty}    \braket{\prod_{i = 1}^n a_i[s_i(\lambda)]  } = \lim_{\lambda \rightarrow \infty} \prod_{I} \braket{\Omega | \prod_{i \in I} a_i[s_i(\lambda)] | \Omega} := \lim_{\lambda \rightarrow \infty} \langle  \prod_{i = 1}^n a_i[s_i(\lambda)] \rangle_\pi \,,
\label{eq:clustering}
\end{equation}
The partition $\pi$ is defined so that $i$ and $j$ are in the same subset iff $\lim_{\lambda \rightarrow \infty} (s_i(\lambda) - s_j(\lambda)) \in \mathbb{R}$.  The intuition behind \eqref{eq:clustering} is that the horizon behaves like a Cauchy surface for the QFT, and correlations should vanish at large distances. See \cite{Penington:2025hrc} for a proof.

 Consider $\{ a_i \}_{i \in \mathbb{N}} \subset \hat \cala_{0,A}$ and $\{ b_i \}_{i \in \mathbb{N}} \subset \hat \cala_{0,B}$. As shown in \cite{Penington:2025hrc}, the correlator $\braket{ \tilde a_1 \tilde b_1 \cdots \tilde a_n \tilde b_n}$ may be written as:\footnote{Note that $\alpha$ in \cite{Penington:2025hrc} refers to $x$ here.}
\begin{align}
    \braket{ \tilde a_1 \tilde b_1 \cdots \tilde a_n \tilde b_n} &= \int_{-\infty}^\infty \prod_{j = 1}^n \left[ \frac{\mathrm{d}y_j \,\mathrm{d}x_j}{2 \pi |x|} \right] e^{\frac{i}{x} \sum_{j = 1}^n (x_{j+1} - x_j) y_j} \braket{ a_1[x_1]  \cdots a_n[x_n]   } \braket{b_1[ y_1] \cdots b_n[y_n]} \,,
\end{align}
where $x_{n + 1} := 0$. The integrand contains a product of correlators in the $A$ and $B$ QFTs. We may write each correlator in terms of its connected components. In particular,
\begin{align}
    \braket{ a_1  \cdots a_n   } &= \sum_{\pi}     \braket{ a_1  \cdots a_n   }_{\pi,c}, 
        \label{eq:ccdef}
    \\
    \braket{ a_1  \cdots a_n   }_{\pi,c} &:= \prod_{I} \langle \prod_{i \in I} a_i \rangle_c \,,
\end{align}
where the sum in \eqref{eq:ccdef} runs over all partitions of $\{1,\cdots,n\}$.  Note that $\braket{\cdots}_{\pi,c}$ is given by a product of connected correlation functions, each of which corresponds to a subset $I$ of $\pi$. Equation \eqref{eq:ccdef} serves as a recursive definition of the connected correlation functions.\footnote{A connected correlation function of two or more operators vanishes if any of those operators is the identity.} From \eqref{eq:ccdef}, it follows that
\begin{align}
    \braket{\tilde a_1 \tilde b_1 \cdots \tilde a_n \tilde b_n}
    &= \int_{-\infty}^\infty \prod_{j=1}^n \left[ \frac{\mathrm{d}y_j\,  \mathrm{d}x_j}{2\pi |x|} \right]
    e^{\frac{i}{x} \sum_{j=1}^n (x_{j+1} - x_j) y_j} \hspace{-2mm}
    \sum_{\pi_A,\pi_B} \hspace{-1mm} \braket{a_1[x_1] \cdots a_n[x_n]}_{\pi_{A} , c }
    \braket{b_1[y_1] \cdots b_n[y_n]}_{\pi_B , c}.
\end{align}
Penington and Tabor \cite{Penington:2025hrc} first show that when $\pi_A$ and $\pi_B$ are mutually crossing, then the corresponding term in the sum above makes a vanishing contribution to the left hand side as $x \rightarrow \infty$. That is, let us consider
\begin{align}
    \int_{-\infty}^\infty \prod_{j=1}^n \left[ \frac{\mathrm{d}y_j\, \mathrm{d}x_j}{2\pi |x|} \right]
    e^{\frac{i}{x} \sum_{j=1}^n (x_{j+1} - x_j) y_j}
     \braket{a_1[x_1] \cdots a_n[x_n]}_{\pi_{A} , c }
    \braket{b_1[y_1] \cdots b_n[y_n]}_{\pi_B , c},
    \label{eq:526}
\end{align}
where $\pi_A$ and $\pi_B$ are mutually crossing. Penington and Tabor \cite{Penington:2025hrc} make a change of variables from $\{ x_i , y_i \}$ to $\{X_a, x_i^\prime, Y_b, y_i^\prime \}$, where $a$ indexes a subset $I_a$ of $\pi_A$, $b$ indexes a subset $I_b$ of $\pi_B$, $i \in \{1,\cdots,n\}$, and the variables $x_i^\prime$ and $y_i^\prime$ are subject to the constraints
\begin{equation}
    \sum_{i \in I_a} x_i^\prime =    \sum_{i \in I_b} y_i^\prime = 0,
\end{equation}
for all $a$ and $b$. The change of variables is given by
\begin{equation}
    x_i = X_{a_i} + x^\prime_{i}, \quad \quad y_i  = Y_{b_i} + y_{i}^\prime,
\end{equation}
where $a_i$ refers to the subset of $\pi_A$ that contains element $i$, and $b_i$ refers to the subset of $\pi_B$ that contains element $i$. That is, $X_{a}$ is the average of the $x_i$ elements within the subset $I_a$, and $x_i^\prime$ represents the fluctuations around the average. This change of variables is advantageous because the QFT correlation functions in \eqref{eq:526} are independent of  $X_a$ and $Y_b$, due to translation invariance. Thus, the $X_a$ and $Y_b$ integrals may be performed explicitly. Because the exponent in \eqref{eq:526} is linear in each $x_i$ and $y_i$ variable, these integrals are straightforward $\delta$-function integrals. Penington and Tabor carefully evaluate these integrals and keep track of the power of $|x|$ that is produced. Given that $\pi_A$ and $\pi_B$ are mutually crossing, they show that the power of $|x|$ that is produced is strictly less than $n$. This means that after integrating over the $X_a$ and $Y_b$ variables, the result is a negative power of $|x|$ times an integral over the $x_i^\prime$ and $y_i^\prime$ variables. The integrand is given by
\begin{equation}
    e^{\frac{i}{x} \vec{x}^{\prime \top} M \vec{y}^\prime}
    \braket{a_1[x_1^\prime] \cdots a_n[x_n^\prime]}_{\pi_{A} , c }
    \braket{b_1[y_1^\prime] \cdots b_n[y_n^\prime]}_{\pi_B , c} \prod_{a} \delta\left(\sum_{i \in I_a} x_i^\prime\right)  \prod_{b} \delta\left(\sum_{i \in I_b} y_i^\prime\right) \,,
\end{equation}
times other possible $\delta$-functions involving $x_i^\prime$ or $y_i^\prime$ that have been produced from the integrals over $X_a$ and $Y_b$. Above, $M$ refers to a matrix of numerical factors. To take the $x \rightarrow \infty$ limit, Penington and Tabor drop the $e^{\frac{i}{x} \vec{x}^{\prime\top} M \vec{y}^{\prime}}$ term above, and argue that the integral over $x_i^\prime$ and $y_i^\prime$ is finite because the connected correlation functions decay to zero for large separations, due to clustering. Thus, the final result of \eqref{eq:526} in the $x \rightarrow \infty$ limit is zero due to the prefactor of $|x|$ to a negative power. This argument is indifferent to the choice of $x \rightarrow +\infty$ or $x \rightarrow - \infty$, and hence applies to our de Sitter setup.

Penington and Tabor's treatment of the case where $\pi_A$ and $\pi_B$ are mutually noncrossing is also indifferent to the sign of $x$. The final conclusion is that for $a \in \hat \cala_{0,A}$ and $b \in \hat \cala_{0,B}$,
\begin{align}
    \lim_{x \rightarrow - \infty}  \braket{ \tilde a_1 \tilde b_1 \cdots \tilde a_n \tilde b_n} &= \sum_{\substack{\pi_A,\pi_B \\ \text{mutually noncrossing}}} \braket{a_1 \cdots a_n}_{\pi_{A,c}} \braket{b_1 \cdots b_n}_{\pi_{B,c}}
    \\
    &:= \braket{  a_1  b_1 \cdots  a_n  b_n}_F \,.
    \label{eq:Fdef}
\end{align}
This indicates the emergence of a free product algebra in the $x \rightarrow -\infty$ limit. We now provide a concise review of free product algebras (see \cite{Chandrasekaran:2022eqq} for more details). The algebra $\cala_{0,A} * \cala_{0,B}$ is generated by monomials of the form
\begin{equation}
    a_1 b_1 \cdots a_n b_n \,,
\end{equation}
where $a_i \in \cala_{0,A}$ and $b_i \in \cala_{0,B}$.\footnote{We have dropped the hats on $\cala_{0,A}$ and $\cala_{0,B}$ because we are ultimately interested in the von Neumann algebra generated by the free product of $\hat \cala_{0,A}$ and $\hat \cala_{0,B}$. Because $\hat \cala_{0,A}$ and $\hat \cala_{0,B}$ are dense in $ \cala_{0,A}$ and $\cala_{0,B}$, this is equivalent to the free product of $\cala_{0,A}$ and $\cala_{0,B}$.} The vacuum state is defined by the expectation values in \eqref{eq:Fdef}, and a Hilbert space may be constructed using the GNS construction. Consider $\{ A_i \}_{i \in \mathbb{N}} \subset \cala_{0,A}$ and $\{ B_i \}_{i \in \mathbb{N}} \subset \cala_{0,B}$ such that
\begin{equation}
    \braket{A_i} = \braket{B_i} = 0, \quad \forall \, i \,.
\end{equation}
We will work with a convention where operators with vanishing expectation values are denoted by uppercase letters, while we continue to use lowercase letters to denote general operators. The correlation function
\begin{equation}
    \braket{A_1 B_1 \cdots A_n B_n}_F
    \label{eq:ABcorr} = 0 
\end{equation}
necessarily vanishes because given any mutually noncrossing partitions $\pi_A$ and $\pi_B$, there must be at least one singleton subset. A general expectation value such as \eqref{eq:Fdef} can be computed iteratively by substituting
\begin{equation}
    a_i \rightarrow A_i + \braket{a_i}, \quad \quad     b_i \rightarrow B_i + \braket{b_i} \,,
    \label{eq:sub}
\end{equation}
and using \eqref{eq:ABcorr} to simplify the result. The substitution \eqref{eq:sub} also implies that $\cala_{0,A} * \cala_{0,B}$ may be generated by operators in $\cala_{0,A}$ and $\cala_{0,B}$ with vanishing expectation values, together with the identity.

We now consider the GNS construction of the Hilbert space $\calh_F$. States in $\calh_F$ are spanned by $\ket{X}$, for $X \in \cala_{0,A} * \cala_{0,B}$. We let $\ket{\Omega_F}$ denote the GNS vacuum, which corresponds to the identity operator. Other states are obtained by acting on $\ket{\Omega_F}$ with operators in $\cala_{0,A}$ and $\cala_{0,B}$ with vanishing expectation values. For example,
\begin{equation}
 A_1 B_1 \cdots A_n B_n  \ket{\Omega_F} =  A_1 B_1 \cdots A_n   \ket{B_n} = A_1 B_1 \cdots    \ket{A_n B_n} =   \ket{ A_1 B_1 \cdots A_n B_n} \,.
\end{equation}
Given two such states,\footnote{Here and elsewhere, we have only considered monomials that begin with an $A$ operator and end with a $B$ operator. To complete the discussion, we should also include the analogous formulas for more general monomials. These formulas are implicit in the discussion, and are left to the reader to derive as an exercise.} their inner product is
\begin{equation}
\label{eq:innerproduct}
    \braket{ A_{n + 1} B_{n + 1} \cdots A_{n + m} B_{n + m} | A_1 B_1 \cdots A_n B_n } = \delta_{nm} \braket{A_{n+1}^\dagger A_1} \braket{B_{n+1}^\dagger B_1} \cdots \braket{A_{2 n}^\dagger A_n} \braket{B_{2n}^\dagger B_n} \,.
\end{equation}
This indicates that $\calh_F$ has the following structure \cite{Chandrasekaran:2022eqq},
\begin{equation}
\label{eq:calhf}
    \calh_F = \ket{\Omega_F} \oplus \calh_A^* \oplus \calh_B^* \oplus \left(\calh_A^* \otimes \calh_B^*\right) \oplus \left(\calh_B^* \otimes \calh_A^* \right) \oplus \left(\calh_A^* \otimes \calh_B^* \otimes \calh_A^*\right) \oplus \cdots \,,
\end{equation}
where $\calh_A^*$ is the orthogonal complement of $\ket{\Omega_A}$ in $\calh_A$, and $\calh_B^*$ is the orthogonal complement of $\ket{\Omega_B}$ in $\calh_B$. 

As reviewed in \cite{Chandrasekaran:2022eqq}, $\cala_{0,A} * \cala_{0,B}$ is a type III$_1$ von Neumann algebra. In particular, this implies that it has a nontrvial commutant. We now briefly discuss how the commutant $\left(\cala_{0,A} * \cala_{0,B}\right)^\prime$ may be constructed. Just as $\cala_{0,A} * \cala_{0,B}$ may be generated by operators in $\cala_{0,A}$ and $\cala_{0,B}$ with vanishing expectation values, the commutant may be generated by operators in $\cala_{0,A}^\prime$ and $\cala_{0,B}^\prime$ with vanishing expectation values, where the primed algebras are the commutants of $\cala_{0,A}$ and $\cala_{0,B}$ in $\calh_A$ and $\calh_B$. In particular, let $A^\prime \in \cala_{0,A}^\prime$, and assume that $\braket{A^\prime} = 0$. Define $\ket{A^\prime} \in \calh_A$ to be 
\begin{equation}
    \ket{A^\prime} := A^\prime \ket{\Omega_A} \,.
\end{equation}
This state belongs to $\calh_A^*$, and thus may be canonically identified as a state $\ket{A^\prime}_F \in \calh_F$ using \eqref{eq:calhf}. That is, it is defined to belong to the first appearance of $\calh_A^*$ in \eqref{eq:calhf}. Next, we define $A^\prime$ as an operator on $\calh_F$ as follows,
\begin{equation}
    A^\prime \ket{A_1 B_1 \cdots A_n B_n} :=      A_1 B_1 \cdots A_n B_n \ket{A^\prime}_F \,.
\end{equation}
This defines $A^\prime$ on all of $\calh_F$ because $\calh_F$ is spanned by states of the form $\ket{X}$, where $X$ is a monomial of $A_i$ and $B_i$ operators. Given $B^\prime \in \cala_{0,B}$, we may define $B^\prime$ as an operator on $\cala_F$ in an analogous way. Intuitively, $A^\prime$ and $B^\prime$ act ``on the right.'' This definition is consistent because $\ket{\Omega_F}$ is cyclic and separating for $\cala_{0,A} * \cala_{0,B}$.

A key point is that each of the summands in \eqref{eq:calhf} captures a semiclassical geometry in which out-of-time-ordered shocks backreact on de Sitter space. In the next subsection, we will draw Penrose diagrams of these geometries and physically interpret the commutant.

\subsection{Geometric Interpretation} \label{sec:alg-geometric}

Having established the emergence of a type III$_1$ free product algebra as $x \rightarrow -\infty$, we now discuss how this algebra can be interpreted using de Sitter geometries with multiple out-of-time-ordered shocks. The interpretation of the AdS free product algebra ($x \rightarrow +\infty$) was provided in \cite{Chandrasekaran:2022eqq}. The AdS geometry with multiple shocks has a long wormhole, and was studied in \cite{Shenker:2013yza}. To discuss the analogous geometry in de Sitter, we will use dS JT gravity \cite{Maldacena:2019cbz, Cotler:2019nbi} as a toy model.

We consider dS JT gravity with matter on a spatial interval times time. On the timelike boundaries, we fix the induced metric and dilaton. In particular, we set the dilaton to zero on the boundaries. The action is
\begin{equation}
    S[g,\phi] = \frac{1}{16\pi G_N} \int_M \mathrm{d}^2x\,\sqrt{-g}\;\phi\left(R-2\right) + S_{\text{matter}}[g] \,,
\end{equation}
where $\phi$ is the dilaton. The GHY boundary term vanishes due to the $\phi = 0$ boundary condition. The dilaton's dynamics is governed by the following equation of motion,
\begin{equation}
\left(g_{\mu\nu} \nabla^2 - \nabla_\mu\nabla_\nu + g_{\mu\nu}\right)\phi = 8\pi G_N\,T^{\rm mat}_{\mu\nu} \,.
\label{eq:dilaton}
\end{equation}
The dilaton and matter fields live in an ambient $\mathrm{dS}_2$ spacetime. A physical solution may be constructed by ending the spacetime at two timelike loci where $\phi = 0$. Such loci must be geodesics. We interpret one of these geodesics as the worldline of an observer, and the other geodesic may be interpreted as an auxiliary observer. This is completely analogous to the discussion in Section \ref{sec:4dmodel}, where an observer is modeled as a timelike boundary at $r = r_{\mathrm{Obs}}$. 

It is convenient to describe the ambient $\mathrm{dS}_2$ spacetime in the embedding space formalism, where the spacetime is modeled as a hyperboloid in an ambient Minkowski space,
\begin{equation}
-(X^0)^2 + (X^1)^2 + (X^2)^2 = X_\mu X^\mu = 1 \,.
\end{equation}
A vacuum solution to \eqref{eq:dilaton} is given by
\begin{equation}
\label{eq:dilatonsoln}
    \phi = X^\mu Y_\mu \,,
\end{equation}
for arbitrary $Y_\mu$. To ensure the existence of timelike geodesic boundaries, we focus on solutions for which $Y_\mu Y^\mu > 0$. The dS$_2$ spacetime has three Killing vectors, whose components in the ambient Minkowski space are given by
\begin{equation}
\xi_i^\mu := \epsilon\indices{_i^\mu_\nu} X^\nu \,,
\end{equation}
where $\epsilon_{\mu \nu \rho}$ is a totally-antisymmetric Levi-Civita symbol that obeys $\epsilon_{012} = 1$. The $\mu$ index in $\epsilon\indices{_i^\mu_\nu}$ is raised using the mostly-plus Minkowski metric. The $i$ index is valued in $\{0,1,2\}$ and labels a choice of Killing vector, while the $\mu$ index denotes the vector components.

We are primarily interested in the case where the matter consists of null shockwaves. Each shockwave carries a set of charges which may be computed as follows,
\begin{equation}
    Q_i := \int_{-\infty}^\infty \mathrm{d}U \, T_{UU} \xi^U_i \,.
\end{equation}
This depends on the stress tensor of the shock as well as a choice of Killing vector. Above, $U$ is an affine coordinate along a null hypersurface that intersects the shock at one point, and $\xi^U_i$ is the Killing vector component along this hypersurface. Equation \eqref{eq:dilaton} implies that across a shock, the change of the $Y_\mu$ coefficients in \eqref{eq:dilatonsoln} is governed by
\begin{equation}
\label{eq:deltaY}
    \Delta Y_\mu = 8 \pi G_N \, Q_\mu \,.
\end{equation}
Note that the dilaton profile is continuous across a shock. Furthermore, the shockwave charges obey $Q_\mu Q^\mu = 0$.

\begin{figure}[htbp]
    \centering
    \includegraphics[width=\textwidth]{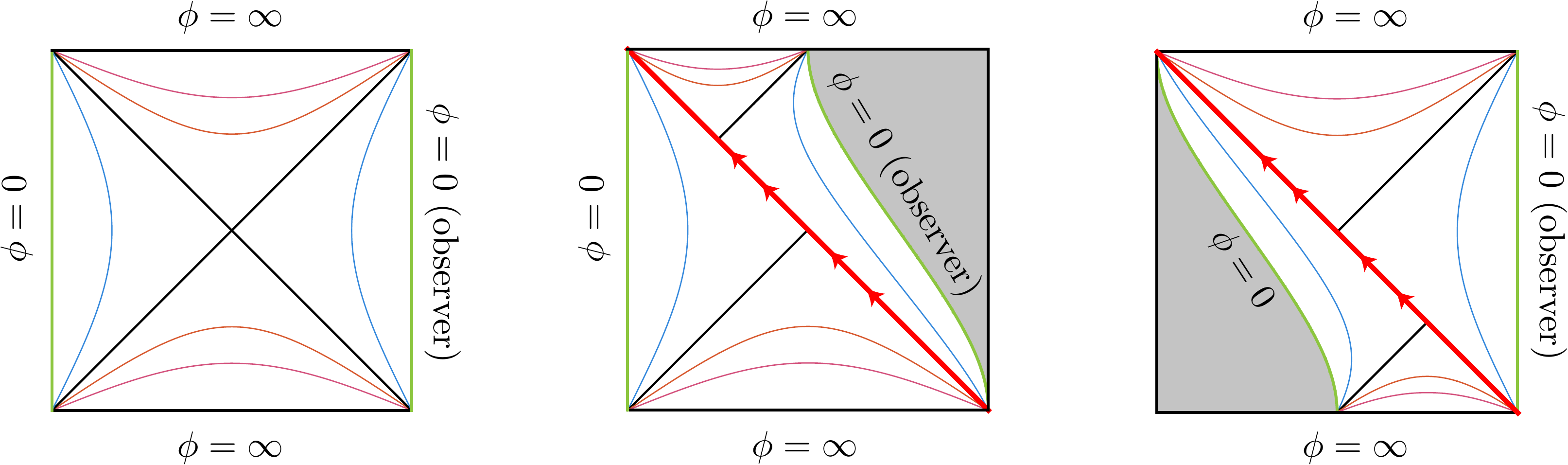}
    \caption{Left: a Penrose diagram for dS$_2$. Level curves of a vacuum solution to \eqref{eq:dilaton} are drawn. The physical spacetime ends at the timelike geodesic boundaries where $\phi = 0$. Center: the observer emits an early-time shockwave, shown in red with arrows. Given that the dilaton profile below the shock is unchanged, the profile above the shock is determined by \eqref{eq:deltaY}. The $\phi = 0$ locus, and hence the physical end of the spacetime, shifts inwards. The shaded gray region is part of the ambient dS$_2$ spacetime. Right: we apply a passive transformation given by an isometry of the ambient dS$_2$ spacetime that moves the observer geodesic back to its original location. This sets the stage for the next step, in which a late-time shockwave is emitted from the upper-right corner of the diagram. The resulting geometry is determined by evolving from the future to the past by using \eqref{eq:deltaY} to track how the dilaton changes across a shock.
    }
    \label{fig:shockplot}
\end{figure}

We may use \eqref{eq:dilatonsoln} to construct de Sitter shockwave geometries. See Figure \ref{fig:shockplot}. Our starting point is dS$_2$ without matter. At very early times, the observer emits a shock which travels along the past cosmological horizon. The profile of the dilaton to the past of the shock is unchanged, while the dilaton to the future of the shock is determined by \eqref{eq:deltaY}. With this prescription, we may say that the shock is dressed to the observer. The observer's geodesic trajectory becomes the new $\phi = 0$ locus. As shown in Figure \ref{fig:shockplot}, we may act with an isometry of the ambient dS$_2$ spacetime on both the matter and dilaton to move the observer $\phi = 0$ locus to its original location. Then, we may consider a second shockwave that is emitted from the observer's worldline at late times. This shockwave emanates from the upper-right corner of the rightmost diagram in Figure \ref{fig:shockplot}. To construct the two-shock geometry, we evolve this shockwave back in time, and take the dilaton profile above this shock to remain unchanged, while the profile below the shock is governed by \eqref{eq:deltaY}. Importantly, the second shockwave will reflect off the left timelike boundary, and the first shockwave will reflect off the observer's new geodesic trajectory. This means that \eqref{eq:deltaY} must be applied multiple times to compute the full evolution of the geometry. The analogous AdS geometries with multiple shockwaves \cite{Shenker:2013yza} are much simpler to study because a shock emitted from a boundary never reaches the other boundary.

To greatly simplify the de Sitter multi-shock geometries, we will work in a limit where much of the shockwave backreaction may be ignored. We let $G_N \rightarrow 0$, but assume that the time separation $T$ between early- and late-time shocks on the observer's worldline is large, with $G_N e^T$ held fixed. This is the same limit as \eqref{eq:limit}. As indicated in Figure \ref{fig:shockplot}, every time a shock is emitted from the observer's worldline, we perform a passive transformation in the ambient dS$_2$ geometry that moves their new worldline back onto their original worldline. This worldline is preserved by a boost isometry. We can use this boost isometry to work in a reference frame where the late-time shocks have null energies (and hence charges $Q_i$) proportional to $e^T$, while the early-time shocks have order one energies. In this reference frame, we may ignore the backreaction from the early-time shocks on the dilaton. Likewise, we may also work in a reference frame where we may ignore the backreaction of the late-time shocks. With these simplifications, some multi-shock geometries are depicted in Figure \ref{fig:columnplot}.

\begin{figure}[!htbp]
    \centering
    \includegraphics[width=0.9\textwidth]{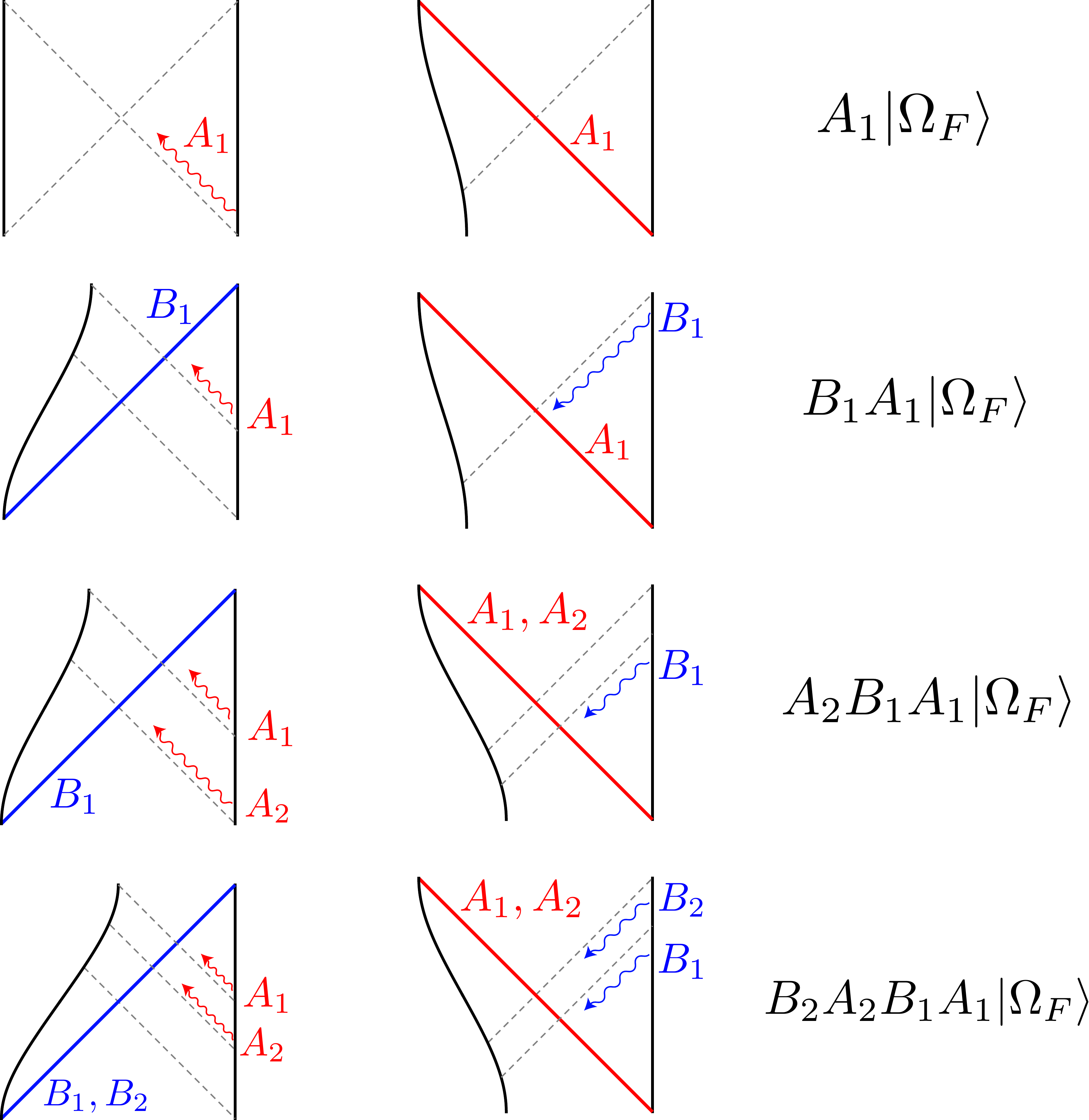}
    \caption{Each row depicts the geometry that corresponds to the state in the rightmost column. The left column corresponds to a reference frame in which the early-time shocks do not backreact, while the middle column 
    corresponds to a reference frame in which the late-time shocks do not backreact. To go from one row to the next, we add an early- or late-time shock and evolve forwards or backwards in time, following the procedure outlined in Figure \ref{fig:shockplot}.}
    \label{fig:columnplot}
\end{figure}

As shown in Figure \ref{fig:columnplot}, the late-time shocks separate the early-time modes along a null direction, and the early-time shocks likewise separate the late-time modes. The free product limit corresponds to sending $G_N e^T \rightarrow \infty$. In this limit, the correlation between the $A_1$ and $A_2$ modes vanishes. This explains why the inner product \eqref{eq:innerproduct} factorizes, and why the Hilbert space has the structure \eqref{eq:calhf}. Each row in Figure \ref{fig:columnplot} depicts the geometry of a summand in \eqref{eq:calhf}.

From the discussion in section \ref{sec:alg-free}, the commutant should be generated by early- and late-time operators on the other observer boundary (on the left of the diagrams in Figure \ref{fig:columnplot}). To illustrate this point, we have reproduced one such diagram in Figure \ref{fig:columnplotcommutant} and drawn a mode that corresponds to $A^\prime$, a late-time operator on the left boundary. Away from the free product limit, $A_1^\prime$ is not in the commutant, because it may not commute with $A_2$. In the free product limit, the $A_1^\prime$ and $A_2$ modes become infinitely separated, and thus commute.

\begin{figure}[!htbp]
    \centering
    \includegraphics[width=0.4\textwidth]{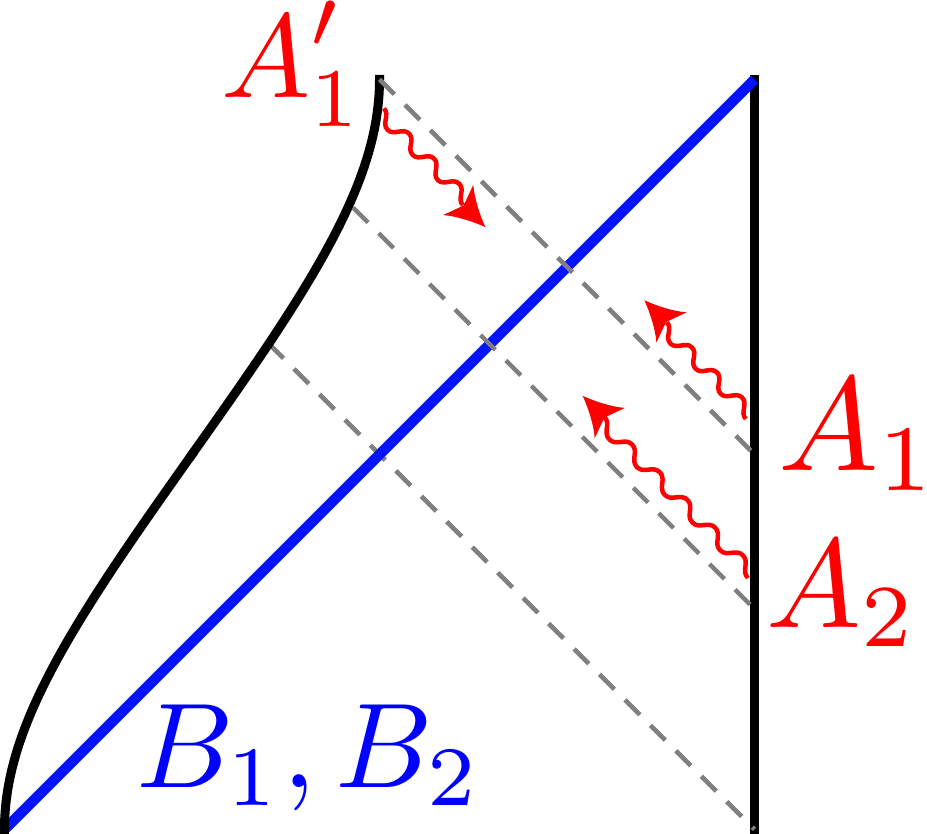}
    \caption{The geometry that corresponds to a state $A_1^\prime B_2 A_2 B_1 A_1 \ket{\Omega_F} \in \calh_F$. The $A_1^\prime$ mode commutes with the $A_1$ mode, but does not commute with $A_2$ unless we work in the free product limit.}
    \label{fig:columnplotcommutant}
\end{figure}

In our original construction of the observer's algebra in Section \ref{sec:4dmodel}, we only considered two types of operators: early- and late-time operators. We did not consider operators at intermediate times (such as order one times) because they do not create shocks that scatter near the horizon. That is, any intermediate time operators would factor out of correlation functions. In the free product limit, intermediate time operators do not necessarily factorize, because the scrambling time is now considered to be an intermediate timescale compared to the very late timescale at which the early- and late-time operators are inserted. Thus, it would not be correct to say that the observer's algebra transitions from a type I algebra to a type III$_1$ algebra in the free product limit. Rather, in this limit, one may define an emergent free product algebra by only considering early- and late-time operators that are inserted at very large times.


\section{Discussion}

In this paper, we studied certain $G_N$ corrections to an observer's algebra $\cala$ in dS. When operators are time-separated by a scrambling time, the gravitational interactions in the $G_N \rightarrow 0$ limit are entirely captured by an eikonal S-matrix. In a general observer model, the eikonal phase can have either sign. We developed a two-dimensional model in which the eikonal phase is negative. Operators associated to an antipodal observer are no longer in the commutant of $\cala$, and we conjecture that the commutant is trivial. After including a clock Hamiltonian in the observer's algebra, their algebra becomes a crossed-product algebra $\cala_{\mathrm{cr}}$. We found that the Hartle-Hawking state is not a trace on $\cala_{\mathrm{cr}}$, because the cyclic property ($\mathrm{Tr}\, ab = \mathrm{Tr}\, ba$) is violated. Equivalently, we find that the canonical vacuum state on $\cala$ is not KMS with respect to time translations along the observer's worldline. An important point is that the physical effects discussed in this paper are only relevant for out-of-time-ordered correlators of early- and late-time operators. The time-ordered correlators of operators in $\cala$ agree with those of QFT on a dS background.

Physically, the KMS condition expresses the fact that a static patch observer experiences a thermal environment at the dS temperature. Although we find that the vacuum is not KMS with respect to static patch time translations, this statement about the observer's experience still holds. Assume a simple coupling between the observer and their environment,
\begin{equation}
    H_{\mathrm{int}}(t) = X(t) \phi(t) \,,
\end{equation}
where $X(t)$ acts on the observer's Hilbert space, and $\phi(t)$ is a quantum field in the observer's environment. In the interaction picture, the time evolution of the joint observer-environment system is governed by the propagator
\begin{equation}
    \mathcal{T} \, e^{-i \int \mathrm{d}t \, H_{\mathrm{int}}(t)} \,.
\end{equation}
The time-ordering symbol above implies that when the environment is initialized in the Bunch--Davies vacuum, the evolution of the observer's density matrix is insensitive to the OTOCs discussed in this paper.\footnote{When the environment is initialized in a precursor state, OTOCs can become relevant. We consider this example to be fine-tuned.} Because an observer cannot reverse the arrow of time, they cannot observe the breakdown of the KMS condition.

A similar algebraic study of dS was conducted in \cite{Kolchmeyer:2024fly}. The observer was treated as a quantum particle that propagates on a fixed dS background. The observer's algebra was argued to have a trivial commutant, because their wavefunction can become delocalized across the entire spacetime. It was also found in \cite{Kolchmeyer:2024fly} that the Hartle-Hawking state is not tracial. Correlation functions were computed by Wick-contracting free fields. However, \cite{Kolchmeyer:2024fly} found that the Hartle-Hawking state behaves like a trace whenever contributions from crossed Wick contractions can be ignored. As shown in this paper, these contributions can be ignored in the free product limit. Our work extends the main physical conclusions of \cite{Kolchmeyer:2024fly} to settings that incorporate gravitational backreaction.

We have purely worked in Lorentzian signature. An alternative approach is to compute correlation functions of operators along an observer's worldline using the Euclidean path integral. For example, in the setup of CLPW \cite{CLPW}, the Hartle-Hawking state (up to normalization) may be identified as the unique trace on the observer's algebra. This does not invoke the Euclidean path integral. The wavefunction for this state (Equation \eqref{crossed-prod-HH-state}) corresponds to an exponential probability distribution for the observer's energy, $e^{- 2 \pi E}$. This Boltzmann factor can also be recovered from a saddle of the Euclidean path integral with a careful prescription for treating the imaginary prefactor \cite{Maldacena:2024spf, Chen:2025jqm}.\footnote{A similar discussion of the Euclidean path integral is contained in \cite{Anninos:2017hhn}.} It would be interesting to extend the results of \cite{Maldacena:2024spf,Chen:2025jqm} to include matter operators, so that a direct comparison can be made with the correlators of the Lorentzian crossed-product algebra. The Boltzmann distribution captures the observer's gravitational backreaction on their cosmic horizon to leading order. Using the Euclidean path integral, corrections were recently obtained in \cite{Blommaert:2026ofx}. It would be interesting to find a Lorentzian interpretation of these corrections.

Our results pose a challenge for observer-centric static patch holography. CLPW \cite{CLPW} constructed an algebra of observables for a static patch observer and identified a maximum entropy state, which we have referred to as $|\Psi_{\mathrm{HH}}\rangle$. Static patch holography proposes that a static patch is dual to a finite-dimensional Hilbert space, which necessarily has a maximum entropy (or tracial) state. The semiclassical state $|\Psi_{\mathrm{HH}}\rangle$ should be dual to this tracial state. However, we have found that when matter operators in CLPW's algebra are time-evolved by a scrambling time, $|\Psi_{\mathrm{HH}}\rangle$ no longer behaves like a trace because the cyclic property ($\mathrm{Tr}\, ab = \mathrm{Tr}\, ba$) is violated.

An important question is how the failure of the KMS condition, or the failure of the tracial property of the Hartle-Hawking state, constrains a putative quantum-mechanical dual. We have attributed these properties to anti-scrambling (or gravitational time advance). As we discussed in Section \ref{sec:antiscrambling}, anti-scrambling can be realized in the traversable wormhole, which involves a coupling between two systems. This suggests that anti-scrambling models may be constructed by coupling two different chaotic quantum systems. This approach has led to success in reproducing free field correlators on a timelike geodesic in dS \cite{Goto:2026ipq, Narovlansky:2023lfz}. A natural next step is to explore OTOCs in doubled systems.\footnote{A concrete example was considered in \cite{Narovlansky:2025tpb}.}

\vspace{1em}

\section*{Acknowledgments}

\begin{singlespace}

We thank Chris Akers, Shoaib Akhtar, Ahmed Almheiri, Elba Alonso-Monsalve, Dionysios Anninos, Tarek Anous, Anna Biggs, Andreas Blommaert, Chang-Han Chen, Yiming Chen, Xi Dong, Anatoly Dymarsky, Daniel Harlow, Temple He, Aidan Herderschee, Mikhail Ivanov, Victor Ivo, Daniel Jafferis, Marc Klinger, Jonah Kudler-Flam, Henry Lin, Hong Liu, Juan Maldacena, Ohad Mamroud, Tommaso Marini, Alexey Milekhin, Beatrix M\"{u}hlmann, Vladimir Narovlansky, Xiao-Liang Qi, Gautam Satishchandran, Allic Sivaramakrishnan, Antony Speranza, Manu Srivastava, Douglas Stanford, Zimo Sun, Leonard Susskind, Haifeng Tang, Erez Urbach, Jan Pieter van der Schaar, Erik Verlinde, Herman Verlinde, Edward Witten, Jiuci Xu, Zhenbin Yang, Daiming Zhang, Ying Zhao, and Kathryn Zurek for helpful discussions, comments, and questions. We acknowledge fruitful discussions during the workshop ``Observers, wormholes and complex saddles in cosmology," organized at the Bernoulli Center for Fundamental Studies (EPFL, Lausanne) from 18--22 May 2026. DKK is supported by the Paul Dirac Fund and the Fund for Memberships in Natural Sciences at the Institute for Advanced Study and the National Science Foundation under Grant No. PHY-2514611. WC is supported by the U.S. Department of Energy, Office of Science, Office of High Energy Physics of U.S. Department of Energy under grant Contract Number DE-SC0012567.

\end{singlespace}


\newpage
\appendix

\section{AQFT Review} \label{sec:aQFT-review}

To set the stage for our study of the gravitational algebras constructed in Section \ref{sec:obsalg}, we review some aspects of algebraic quantum field theory (aQFT). For more details, some excellent expositions are \cite{Witten:2018zxz, Sorce:2023, Haag:1992}. One starts by considering the set $\mathcal{B}(\mathcal{H})$ of bounded operators on a Hilbert space. A von Neumann algebra is a subset\footnote{Von Neumann algebras may also be defined abstractly without reference to a Hilbert space. Given a state (positive linear functional) $\omega$ on the algebra $\mathfrak{A}$, one can then construct a representation $(\pi_{\omega}(\mathfrak{A}), \mathcal{H}_{\omega})$ using the GNS construction. Roughly speaking, this involves identifying operators with vectors $a \to |a\rangle = a|\omega\rangle$ equipped with inner product $\langle a|b\rangle = \omega(a^{\dagger}b)$, quotienting by null states, and taking the Hilbert space completion.} of operators $\mathfrak{A} \subseteq \mathcal{B}(\mathcal{H})$ closed under linear combinations, multiplication, the adjoint operation, and limits in the strong operator topology.\footnote{The strong operator topology is defined by pointwise norm convergence on all vectors, $O_n|\psi\rangle \to O |\psi\rangle$ for all $|\psi\rangle\in\mathcal{H}$. Von Neumann algebras can be equivalently defined using a host of other topologies, including the weak operator topology (convergence of all matrix elements).} Its commutant $\mathfrak{A}'$ consists of all operators in $\mathcal{B}(\mathcal{H})$ commuting with $\mathfrak{A}$. The double commutant theorem states that $\mathfrak{A}'' = \mathfrak{A}$ for any von Neumann algebra $\mathfrak{A}$. In aQFT, the basic algebras of interest are quantum fields smeared by smooth test functions with support in some spacetime region $R$. 

Given an algebra $\mathfrak{A}$ acting on Hilbert space $\mathcal{H}$, we say a vector $|\omega\rangle \in \mathcal{H}$ is  \textit{cyclic} for $\mathfrak{A}$ if the vectors $\{a|\omega\rangle, a\in\mathfrak{A}\}$ are dense in $\mathcal{H}$. Furthermore, $|\omega\rangle$ is \textit{separating} for $\mathfrak{A}$ if the condition $\{a|\omega\rangle = 0, \ \forall \, a\in\mathfrak{A}\}$ implies that $a= 0$. Any state that is cyclic for either $\mathfrak{A}$ or $\mathfrak{A}'$ is separating for the other. The existence of a cyclic separating vector allows us to use the following powerful theorem \cite{Takesaki:2001}.

\begin{theorem}[Tomita-Takesaki] \label{thm:Tomita-Takesaki}
    Let $\mathfrak{A}$ be a von Neumann algebra, with cyclic separating vector $|\Omega\rangle$, and define the Tomita operator as the antilinear operator on $\mathcal{H}$ through the action
    \begin{align}
        S_{\Omega} a|\Omega\rangle = a^{\dagger}|\Omega\rangle\,, \qquad a\in\mathfrak{A}\,. \label{eq:Tomita-op}
    \end{align}
    The operator $S$ permits a polar decomposition
    \begin{align}
        S_{\Omega} = J_{\Omega}\Delta_{\Omega}^{\frac{1}{2}} = \Delta_{\Omega}^{-\frac{1}{2}} J_{\Omega}, \qquad S_{\Omega}^{\dagger} = J_{\Omega}\Delta_{\Omega}^{-\frac{1}{2}} = \Delta_{\Omega}^{\frac{1}{2}} J_{\Omega} \label{eq:Tomita-polar}
    \end{align}
    where $\Delta_{\Omega} = e^{-K_{\Omega}}$ is a positive operator, and $J_{\Omega} = J^{-1}_{\Omega} = J^{\dagger}_{\Omega}$ is an antilinear isometry. We call $\Delta_{\Omega}$ the modular operator, $K_{\Omega}$ the modular Hamiltonian, and $J_{\Omega}$ the modular conjugation. These satisfy $\Delta_{\Omega}|\Omega\rangle = J_{\Omega}|\Omega\rangle = |\Omega\rangle$. Furthermore,
    \begin{enumerate}
        \item The modular operator defines a one-parameter family of automorphisms on $\mathfrak{A}$ known as modular flow:
        \begin{align}
            a(s) := \Delta_{\Omega}^{is} \, a\, \Delta_{\Omega}^{-is} \in \mathfrak{A}\,, \qquad \forall \, a\in\mathfrak{A},\quad s\in\mathbb{R} \,. \label{eq:modular-flow}
        \end{align}
        \item Modular conjugation exchanges the algebras $\mathfrak{A}$ and $\mathfrak{A}'$:
        \begin{align}
            J_{\Omega} \,a\, J_{\Omega} \in \mathfrak{A}' \,, \qquad \forall\,a\in\mathfrak{A} \,. \label{eq:modular-conj}
        \end{align}
        \item The state $|\Omega\rangle$ satisfies the KMS condition for the algebra $\mathfrak{A}$ under modular flow. That is, for any $a,b\in\mathfrak{A}$ there exists a function $F$ analytic in the strip $-1< \mathrm{Im}\,s \leq 0$, such that
        \begin{align}
             F(s) = \langle\Omega | a(s) \, b|\Omega\rangle\,, \qquad F(s-i) = \langle\Omega| b\, a(s)|\Omega\rangle \,, \qquad \forall\, s\in\mathbb{R}\,. \label{eq:KMS}
        \end{align}
    \end{enumerate}
\end{theorem}
\noindent The KMS condition for $|\Omega\rangle$ generalizes the notion of a thermal state for the observables $\mathfrak{A}$, in the absence of Hilbert space factorization and well-defined density operators. Here $K_{\Omega}$ is understood to be the Hamiltonian, and we have normalized the temperature as $\beta=1$. Indeed, when density matrices and traces are well-defined we can use the cyclicity of the trace to verify:
\begin{align}
    \langle a(s) \,b \rangle_{\beta} &= \mathrm{Tr} \big[ e^{-\beta H} e^{iH s} \,a \, e^{-iH s} b \big] = \mathrm{Tr} \big[b \,e^{-\beta H} e^{iH s} \,a\, e^{-iH s} \big] \nonumber \\
    &= \mathrm{Tr} \big[ e^{-\beta H} \,b\, e^{iH(s+i\beta)} \,a\, e^{-iH(s+i\beta)} \big] = \langle b \, a(s+i\beta) \rangle_{\beta} \,.
\end{align}

\noindent Now we specialize to de Sitter. The simplest case is the setting of CLPW \cite{CLPW} where we first take $G_N=0$, and consider a classical observer with $M_{\mathrm{obs}}\to\infty$ sitting at the north pole. We define $\mathfrak{A}$ to be the von Neumann algebra of fields accessible to the observer; by the timelike tube theorem this is equivalently the algebra of fields smeared in the de Sitter static patch. Let $|\Omega\rangle$ denote the Bunch-Davies state, which is cyclic and separating for $\mathfrak{A}$. We can then identify $K_{\Omega} = 2\pi i \partial_t$ as static patch translation\footnote{We work in the convention where modular time and static patch time differ by a factor of $2\pi$.} and $J_{\Omega}$ as $\mathcal{CPT}$. The commutant $\mathfrak{A}'$ is generated by smeared fields in the anti-static patch, or by operators accessible to an observer at the south pole. The KMS condition and its Fourier transform (the detailed balance condition) can be written as:
\begin{alignat}{2}
    &\langle \Omega| a(t_1-2\pi i) b(t_2) |\Omega\rangle &&= \langle \Omega| b(t_2) a(t_1) |\Omega\rangle \label{eq:KMS-dS-t} \,, \\
    & e^{-2\pi\lambda_1}\langle \Omega| a(\lambda_1) b(\lambda_2)|\Omega\rangle &&= \langle \Omega| b(\lambda_2) a(\lambda_1)|\Omega\rangle \,, \qquad \forall\, a, b\in\mathfrak{A}\,. \label{eq:KMS-dS-lambda}
\end{alignat}
The thermality of the Bunch-Davies vacuum can be understood from entropic grounds. In the vacuum the field degrees of freedom between the left and right static patches are in a maximally entangled state. Since the left patch is causally separated from the north pole, the restriction to the observer's algebra corresponds to tracing out the left degrees of freedom and leads to a thermal ensemble.

This paper generalizes the CLPW story to the gravitationally dressed algebra $\cala$ accessible to a de Sitter observer at very early and late times. As discussed in Section \ref{sec:4dmodel}, $\cala$ is generated by the dressed fields $e^{ix P_A P_B/2} \,a\, e^{-ix P_A P_B/2}$ and $e^{-ix P_A P_B}/2 \,b\, e^{ix P_A P_B/2}$, with $a \in \cala_{0,A}$ and $b\in \cala_{0,B}$. The Hilbert space is built around the tensor product $|\Omega\rangle := |\Omega_A\rangle \otimes |\Omega_B\rangle$, where $|\Omega_{A}\rangle$ and $|\Omega_{B}\rangle$ are the Bunch-Davies vacua for the $A$ and $B$ theories, cyclic and separating with respect to $\cala_{0,A}$ and $\cala_{0,B}$. The modular Hamiltonian acts on each algebra as static patch time translation.


\section{4D Model: Eikonal Phase from Recoil}

\label{sec:recoileffect}

In this appendix, we show how a pointlike observer in $\mathrm{dS}_4$ recoils off of a shockwave. We use embedding space coordinates $X^A$, $A = 0,1,2,3,4$, with the de Sitter hyperboloid defined by:
\begin{equation}
    \eta_{AB}X^A X^B = 1 \,, \qquad \eta_{AB} = \mathrm{diag}(-1,+1,+1,+1,+1) \,.
\end{equation}
We introduce Kruskal coordinates $(U,V,\theta,\phi)$ via:
\begin{alignat}{2}
    &X_0 = \frac{2(V+U)}{4-UV}\,, \qquad &&X_4 = \frac{2(V-U)}{4-UV}\,, \nonumber \\
    &X_1 = \frac{4+UV}{4-UV}\sin\theta\cos\phi\,, \qquad &&X_2 = \frac{4+UV}{4-UV}\sin\theta\sin\phi\,, \qquad X_3 = \frac{4+UV}{4-UV}\cos\theta \,.
\end{alignat}
The induced metric on the hyperboloid in these coordinates is:
\begin{equation}
    \mathrm{d}s^2 = -\frac{16}{(4-UV)^2}\, \mathrm{d}U\mathrm{d}V + \frac{(4+UV)^2}{(4-UV)^2} \, \mathrm{d}\Omega_2^2 \,,
\end{equation}
which at $U=V=0$ reduces to:
\begin{equation}
    \mathrm{d}s^2\big|_{U=V=0} = -\mathrm{d}U \mathrm{d}V + \mathrm{d}\Omega_2^2 \,.
\end{equation}
The North Pole worldline corresponds to $V = -2e^{-t}$, $U = 2e^{t}$, where $t$ is the proper time of a geodesic observer. The future cosmological horizon of the North Pole observer is the surface $V=0$, and $\mathscr{I}^\pm$ corresponds to $UV = -4$.

\vspace{1em}

\noindent The observer lives on the North Pole. Their worldline is given by
\begin{equation}
    X_{\mathrm{Obs}}^A(\tau) = \begin{pmatrix}\sinh(M\tau)\\0\\0\\0\\\cosh(M\tau)\end{pmatrix} \,,
\end{equation}
where $\tau$ is their proper time. Suppose that at the point \begin{equation}
    X_{\mathrm{emit}}^A = \begin{pmatrix}\sinh t_0\\0\\0\\0\\\cosh t_0\end{pmatrix}
\end{equation}
the observer emits a future-directed localized shockwave, which may be parameterized by
\begin{equation}
    X_{\mathrm{shock}}^A(\lambda) = \begin{pmatrix}\sinh t_0\\0\\0\\0\\\cosh t_0\end{pmatrix} + e^{t_0}P^V  \lambda\begin{pmatrix}\cosh t_0\\0\\0\\1\\\sinh t_0\end{pmatrix} \,,
\end{equation}
where $\lambda$ is an affine parameter, chosen such that $\partial_\lambda X_{\mathrm{shock}}^A = P^V \partial_V$ at $V = 0$.

The observer recoils off of the shockwave. Their new worldline is related to $X^A_{\mathrm{Obs}}(\tau)$ by an isometry that fixes $X^A_{\mathrm{emit}}$ and acts trivially in the $(X_1, X_2)$ plane. An infinitesimal isometry that obeys these conditions is given by
\begin{align}
    \label{eq:A9}
    \delta X_0 &= \epsilon X_3 \,, \\
    \delta X_1 &= 0 \,, \\
    \delta X_2 &= 0 \,, \\
    \delta X_3 &= \epsilon(X_0 - \tanh t_0 \, X_4) \,, \\
    \delta X_4 &= \epsilon \tanh t_0 \, X_3 \,,
    \label{eq:A13}
\end{align}
which we may represent using the matrix $Y$,
\begin{equation}
    Y\indices{^A_B} :=
    \begin{pmatrix}
    0 & 0 & 0 & 1 & 0 \\
    0 & 0 & 0 & 0 & 0 \\
    0 & 0 & 0 & 0 & 0 \\
    1 & 0 & 0 & 0 & -\tanh t_0 \\
    0 & 0 & 0 & \tanh t_0 & 0
    \end{pmatrix}\,.
\end{equation}
The isometry that relates the old and new observer worldlines must take the form $e^{\alpha Y}$. By acting on $X_{\mathrm{Obs}}^A(\tau)$ with $e^{\alpha Y}$, we can deduce that the tangent vector of the new worldline at $X_{\mathrm{emit}}^A$ must take the form
\begin{equation}
T^A :=    M^\prime \begin{pmatrix}\cosh t_0\cosh(\alpha\operatorname{sech} t_0)\\0\\0\\\sinh(\alpha\operatorname{sech} t_0)\\\cosh(\alpha\operatorname{sech} t_0)\sinh t_0\end{pmatrix}
\end{equation}
for some suitable choice of $\alpha$ and $M^\prime$. To solve for these parameters, we compute the charges of the old and new observer worldlines and the shock with respect to the $\mathrm{dS}_4$ isometries. That is, let $J_i^A$ refer to the Killing vectors of $\mathrm{dS}_4$, where $i$ labels a Killing vector. We then compute
\begin{equation}
    Q^{\mathrm{Obs}}_i := J_{A\, i}\, \partial_\tau X^A_{\mathrm{Obs}}(\tau)\,, \qquad Q^{\mathrm{Obs}^\prime}_i := J_{A\, i}\, T^A \,, \qquad Q^{\mathrm{shock}}_i := J_{A \, i}\,  \partial_\lambda X^A_{\mathrm{shock}}(\lambda) \,,
\end{equation}
and require that
\begin{equation}
    Q^{\mathrm{Obs}}_i =  Q^{\mathrm{Obs}^\prime}_i + Q^{\mathrm{shock}}_i \,.
\end{equation}
Working in the limit $t_0 \rightarrow - \infty$, we find that
\begin{equation}
    \alpha = -\frac{P^V}{2 M}\,, \qquad M^\prime = M \,.
\end{equation}
Working in the $t_0 \rightarrow -\infty$ limit, the infinitesimal isometry in \eqref{eq:A9}--\eqref{eq:A13} acts on $U = 0$ as $\delta V = 2 \cos \theta$. Hence, the isometry $e^{- \alpha Y}$ acts at $U = 0$ as
\begin{equation}
    V \rightarrow V +  \frac{P^V}{ M} \cos \theta \,.
\end{equation}
This corresponds to an eikonal phase of
\begin{equation}
   I =  -\frac{P_V P^V}{ M} \cos \theta = \frac{2 P_V P_U}{M} \cos \theta \,,
\end{equation}
in agreement with \eqref{eq:recoilphase}. In this appendix, we have worked in units where $\ell_{\mathrm{dS}} = 1$. The quantities $I$ and $\delta$, defined in \eqref{eq:action} and \eqref{eq:eikonalphase},  simply differ by $e^{\frac{2 \pi T}{\beta_c}}$.


\section{3D Model}

\label{sec:app3D}

In this appendix, we discuss an observer in dS$_3$. The observer is pointlike and sources a conical defect. We take the matter theory to be a massive scalar field. We first discuss the quantization of this field on dS$_3$ with a pair of conical defects. Then, we compute the eikonal phase. We will show that the near-Nariai limit of Section \ref{sec:nearnariai} effectively reduces the number of spacetime dimensions to two. In Section \ref{sec:3D-KMS} we compute OTOCs and discuss their analytic properties. The KMS condition is preserved to all orders in S-matrix perturbation theory, but for $x<0$ it is violated non-perturbatively. In the $x\to 0^-$ limit, we reproduce the leading $x$-term of $\Delta_{\mathrm{KMS}}$ (\ref{eq:2D-DeltaKMS-UV-explicit-2}) with $h_A=h_B=1/2$.

\subsection{Scalar Quantization in \texorpdfstring{$\mathrm{dS}_3$}{dS3}} 
In global coordinates $(\rho, \theta, \varphi)$, the geometry of dS$_3$ with an observer is given by
\begin{align}
    \mathrm{d}s^2 = -\mathrm{d}\rho^2 + \cosh^2\rho\,( \mathrm{d}\theta^2 + \cos^2\theta \mathrm{d}\varphi^2)\,, \qquad \rho\in\mathbb{R}\,, \qquad \theta \in  \left(- \frac{\pi}{2},\frac{\pi}{2} \right), \qquad \varphi \equiv \varphi + \alpha, \label{eq:dS3-metric}
\end{align}
where $\alpha = 2\pi - \Omega$ and $\Omega = 8\pi G_N M \in (0,2\pi)$ is the deficit angle. The observer has mass $M$ and sits at the north pole, $\theta=\pi/2$. In three dimensions, the term ``near-Nariai limit'' refers to the limit $\Omega \rightarrow 2 \pi$. The conical defect at the south pole can be interpreted as an auxiliary observer that will not play a role in the present discussion. It is useful to consider other coordinates in place of $(\rho,\theta)$. Three choices are
\begin{enumerate}
    \item Static patch coordinates $(t,\vartheta)$:
    \begin{align}
        \mathrm{d}s^2 &= -\sin^2\vartheta\mathrm{d}t + \mathrm{d}\vartheta^2 + \cos^2\vartheta \mathrm{d}\varphi^2\,, \qquad t\in\mathbb{R}\,, \qquad \vartheta \in \left( - \frac{\pi}{2}, \frac{\pi}{2} \right)  \,, \label{eq:dS3-static-patch-coords} \\
        \cos\vartheta &= \cos\theta \cosh\rho\,, \quad \tanh t = \frac{\tanh\rho}{\sin\theta}\,, \qquad \sinh\rho = \sin\theta\sinh t\,, \quad \tan\theta = \cosh\eta \tan\vartheta \,. \nonumber
    \end{align}
    This covers only the causal diamonds of the north pole ($\vartheta=\pi/2$) and south pole ($\vartheta=-\pi/2$). We call the region $\vartheta\in (0,\pi/2)$ the static patch, and $\vartheta\in(-\pi/2,0)$ the anti-static patch.
    \item Kruskal coordinates $(U,V)$:
    \begin{align}
        \mathrm{d}s^2 &= \frac{-4 \mathrm{d}U\mathrm{d}V}{(1- UV)^2} + \Big( \frac{1 + UV}{1 - UV}\Big)^2 \mathrm{d}\varphi^2\,, \qquad
        (|U|,|V|) = (e^{-t}\tan\tfrac{\vartheta}{2}, e^t\tan\tfrac{\vartheta}{2})\,.\label{eq:dS3-Kruskal-coords}
    \end{align}
    The static patch is given by $\{U<0, V>0\}$, and the antistatic patch by $\{U>0, V<0\}$.
    \item Poincar\'{e} coordinates $(\zeta_-,\xi_-)$:
    \begin{align}
        \mathrm{d}s^2 &= \frac{1}{\zeta_-^2}(-\mathrm{d}\zeta_-^2 + \mathrm{d}\xi_-^2 + \xi_-^2\mathrm{d}\varphi^2)\,, \label{eq:dS3-Poincare-coords} \\
        (\zeta_-,\xi_-) &= \Big(\tfrac{UV-1}{2V}, \tfrac{UV+1}{2V}\Big)\,, \qquad (U,V) = \Big(\zeta_-+\xi_-, -\tfrac{1}{\zeta_--\xi_-}\Big)\,, \nonumber
    \end{align}
    where $\zeta_-\in(-\infty,0)$ and $\xi_-\in(0,\infty)$. This covers $V>0$. The north pole is $\xi_-=0$. One can also define an analogous set of coordinates $(\zeta_+,\xi_+)$ which covers $U<0$ by exchanging $-U$ and $V$.
\end{enumerate}

\noindent To quantize a free scalar field, it is convenient to pick a null Cauchy surface, which has the property that a solution of the Klein-Gordon equation is determined by its restriction to the surface. For QFT A, which describes operators inserted near the north pole at early times, we should take this null Cauchy surface to be $U = 0$. With this choice, it is natural to solve the KG equation in the $(\zeta_+,\xi_+)$ Poincare coordinates that cover $U < 0$. For QFT B, we should choose the $V = 0$ Cauchy surface and $(\zeta_-,\xi_-)$ Poincare coordinates. 

For the time being, we work in Poincar\'{e} coordinates $(\zeta_-, \xi_-, \varphi)$. The matter action and equation of motion are given by
\begin{align}
    S &= -\int \frac{1}{2} \mathrm{d}^3x \sqrt{-g} \big(g^{ab}\partial_a\phi\partial_b\phi + m^2\phi^2\big)\,, \qquad\qquad m^2 = 1 + s^2, \quad s > 0\,, \label{eq:dS3-scalar-action} \\
    0 &= (-\nabla^2 + m^2)\phi \label{eq:dS3-eom} \\
    &= \Big( \zeta_-^2\partial_{\zeta_-}^2 - \zeta_-\partial_{\zeta_-} - \zeta_-^2\partial_{\xi_-}^2 - \frac{\zeta_-^2}{\xi_-}\partial_{\xi_-} - \frac{\zeta_-^2}{\xi_-^2}\partial_{\varphi}^2 + m^2\Big)\phi(\zeta_-,\xi_-,\varphi). \label{eq:dS3-eom-Poincare}
\end{align}
The Klein-Gordon equation is separable, and we can expand in Fourier modes along the $S^1$. On the $V = 0$ horizon, each mode is a plane wave, characterized by null momentum $k_U < 0$. The mode expansion into positive and negative frequencies takes the form
\begin{align}
    \phi(X) &= \sum_{n\in\frac{2\pi}{\alpha} \mathbb{Z}} \int_{-\infty}^0 \dbar k_U \Big[ v^+_{n,k_U}(X) a_{n,k_U} + v^-_{n,k_U}(X) a_{n,k_U}^{\dagger} \Big]\,, \qquad v_{n,k_U}^-(X) = v^+_{n,k_U}(X)^{\ast}\,, \label{eq:dS3-kU-mode-exp} \\
    v^+_{n,k_U}(X) &= \sqrt{\frac{\pi}{-k_U \alpha}} e^{in\varphi} J_{|n|}(-k_U\xi_-) \frac{e^{-\frac{i\pi |n|}{2}} k_U\zeta_-}{2\sinh\pi s} \Big[e^{\frac{\pi s}{2}} J_{is}(k_U\zeta_-) - e^{-\frac{\pi s}{2}} J_{-is}(k_U\zeta_-) \Big] \label{eq:dS3-kU-mode-explicit-1} \\
    &= \frac{e^{in\varphi} }{\sqrt{-k_U \pi \alpha}} e^{-\frac{i\pi |n|}{2}} (-i k_U\zeta_-) K_{is}(-ik_U\zeta_-)J_{|n|}(-k_U\xi_-)\,, \label{eq:dS3-kU-mode-explicit-2}
\end{align}
where $J_{\nu}(z)$ and $K_{\nu}(z)$ are Bessel functions. To solve for these modefunctions, one must impose that they do not diverge at the north pole. In particular, note that $J_{|n|}(-k_U\xi_-) \sim \xi_-^{|n|}$ as $\xi_- \rightarrow 0$. These modes have the expected behavior on $V=0$,\footnote{The $\cdots$ in \eqref{eq:dS3-kU-mode-BC} contains a $e^{- i k_U/V}$ term, which we disregard because it vanishes in the $V \rightarrow 0$ limit after smearing $k_U$.}
\begin{equation}
    \lim\limits_{V\to 0} v^+_{n,k_U}(X) = \frac{e^{in\varphi}}{\sqrt{4 \pi (-k_U) \alpha}}\left(e^{ik_U U} + \cdots \right) \,,\label{eq:dS3-kU-mode-BC}
\end{equation}
and are orthonormal with respect to the Klein-Gordon inner product,
\begin{equation}
    \langle v_{n,k_U}^a | v_{n',k_U'}^b\rangle_{\mathrm{KG}} = i\int_{\Sigma}\mathrm{d}\Sigma^{\mu}\, v^{a\,\ast}_{n,k_U} \overset\leftrightarrow{\partial}_{\mu} v^b_{n',\,k_U'} = \mathrm{sgn}(a)\, \delta_{ab}\, \delta_{nn'}\, \delta(k_U-k_U') \,. \label{eq:dS3-kU-mode-norm}
\end{equation}
The Bunch-Davies vacuum state $|\Omega\rangle$ is annihilated by all the positive frequency modes, $a_{n,k}|\Omega\rangle = 0$ for all $k<0$. The oscillators obey the usual commutation relations,
\begin{align}
    [a_{n_1,k_1}, a_{n_2,k_2}^{\dagger}] &= \delta_{n_1,n_2}\delta(k_1-k_2)\,. \label{eq:dS3-mode-commutators}
\end{align}
This implies the canonical commutation relations between the field and its conjugate momentum,
\begin{align}
    [\phi(\zeta_-,\xi_-,\varphi), \Pi(\zeta_-,\xi_-',\varphi')] = i\delta(\xi_- -\xi_-')\delta(\varphi-\varphi')\,, \qquad \Pi(X) = -\frac{\xi_-}{\zeta_-} \partial_{\zeta_-}\phi(X)\,. \label{eq:dS3-CCR}
\end{align}
We now show:
\begin{proposition}
    In the near-Nariai limit $\alpha\to 0$, the oberver's algebra is generated by fields smeared on the half-horizon $V=0$, $U<0$.
\end{proposition}
\begin{proof}
    In the limit where $\alpha \rightarrow 0$ we may perform a dimensional reduction on the shrinking $S^1$. We restrict the sum in \eqref{eq:dS3-kU-mode-exp} to $n = 0$. We define
    \begin{equation}
       \phi_{\mathrm{2D}}(X) := \sqrt{\alpha} \int_{-\infty}^0 \dbar k_U \Big[ v^+_{0,k_U}(X) a_{0,k_U} + v^-_{0,k_U}(X) a_{0,k_U}^{\dagger} \Big] \,.
       \label{eq:phi2d}
    \end{equation}
    The observer's algebra is defined to be the algebra generated by polynomials of $\phi_{\mathrm{2D}}$ smeared within the static patch. We will show that this algebra is the same as the algebra generated by polynomials of $\phi_{\mathrm{2D}}$ smeared on the half-horizon $V = 0$, $U < 0$. In particular, we will show that $\phi_{\mathrm{2D}}$ anywhere in the static patch can be written as $\phi_{\mathrm{2D}}$ smeared on the half-horizon. First, we evaluate \eqref{eq:phi2d} on the $V = 0$ horizon,
    \begin{equation}
        \left. \phi_{\mathrm{2D}}(U) \right|_{V=0} = \int_{-\infty}^0 \dbar k_U \Big[ 
       \frac{1}{\sqrt{4 \pi (-k_U) }} e^{ik_U U}  a_{0,k_U} + \text{h.c.} \Big] \,.
       \label{eq:phi2dV0}
    \end{equation}
    which implies that
    \begin{equation}
        a_{0,k_U} = \sqrt{4 \pi (-k_U)} \int_{-\infty}^{\infty} \mathrm{d}U \, \left. \phi_{\mathrm{2D}}(U) \right|_{V=0} \, e^{-i k_U U} \,,
    \end{equation}
    and we use the convention that $a_{0,-k_U} = a_{0,k_U}^\dagger$. Substituting into \eqref{eq:phi2d}, we find that
    \begin{equation}
        \phi_{\mathrm{2D}}(X) = \int_{-\infty}^{\infty} \mathrm{d}U \, \left. \phi_{\mathrm{2D}}(U) \right|_{V=0} \int_{-\infty}^0 \dbar k_U \Big[ 
        (-2 i  k_U\zeta_-) K_{is}(-ik_U\zeta_-)J_{0}(-k_U\xi_-)  \,  \, e^{-i k_U U} 
        + \text{h.c.} \Big] \,.
        \label{eq:phi2dV0kernel}
    \end{equation}
    Note that $J_{0}(-k_U\xi_-)$ is analytic in the $k_U$ plane, while $K_{is}(-ik_U\zeta_-)$ has a branch cut that is conventionally taken to lie along the positive imaginary $k_U$ axis. We may rewrite \eqref{eq:phi2dV0kernel} as follows, 
    \begin{equation}
        \phi_{\mathrm{2D}}(X) = \int_{-\infty}^{\infty} \mathrm{d}U \, \left. \phi_{\mathrm{2D}}(U) \right|_{V=0} \int_{-\infty}^{\infty} \dbar k_U \Big[ 
        (-2 i  k_U\zeta_-) K_{is}(-ik_U\zeta_-)J_{0}(-k_U\xi_-)  \,  \, e^{-i k_U U}  \Big] \,,
        \label{eq:phi2dV0kernelkuinf}
    \end{equation}
    where the $k_U$ contour above is chosen to avoid the branch cut that emanates from $k_U = 0$. In the lower half plane, where the integrand is analytic, the integrand behaves as
    \begin{equation}
        (-2 i  k_U\zeta_-) K_{is}(-ik_U\zeta_-)J_{0}(-k_U\xi_-) \rightarrow  \sqrt{\frac{-\zeta_-}{\xi_-}}\; e^{\,i(\zeta_- + \xi_-) k_U}
    \end{equation}
    for asymptotically large $k_U$. Thus, we may close the contour in the lower half plane and conclude that the integral vanishes when $U >   \zeta_- + \xi_- $. Because $\zeta_- + \xi_-$ is at most zero for $\phi_{\mathrm{2D}}(X)$ in the static patch, it follows that any $\phi_{\mathrm{2D}}(X)$ in the static patch can be expressed in terms of $\phi_{\mathrm{2D}}(X)$ smeared on the $V = 0$, $U < 0$ half-horizon. Moreover, we can use integration by parts to rewrite \eqref{eq:phi2dV0kernelkuinf} as
    \begin{equation}
        \phi_{\mathrm{2D}}(X) := \int_{-\infty}^{\infty} \mathrm{d}U \, \left. \partial_U \phi_{\mathrm{2D}}(U) \right|_{V=0} \int_{-\infty}^{\infty} \dbar k_U \Big[ 
        (-2 \zeta_-) K_{is}(-ik_U\zeta_-)J_{0}(-k_U\xi_-)  \,  \,  e^{-i k_U U}  \Big] \,,
    \end{equation}
    The integrand of the $k_U$ integral oscillates wildly near $k_U = 0$, but is still integrable there. Hence, we conclude that $\phi_{\mathrm{2D}}$ evaluated anywhere within the static patch can be written as $\partial_U \phi_{\mathrm{2D}}$ smeared on the half-horizon $V = 0$, $U < 0$.
\end{proof}

So far, we have established that the algebra of a massive scalar field in the static patch becomes, in the near-Nariari limit, the algebra generated by $\partial_U \phi_{\mathrm{2D}}$ on the half-horizon $V = 0$, $U < 0$. Before we discuss the computation of the eikonal phase, we summarize the formulas for the modefunctions for the two cases where we take either $V = 0$ or $U = 0$ to be the null Cauchy surface used for quantization. Below, we evaluate the modefunctions on either the north or south poles.
\begin{alignat}{2}
    v^+_{k_U}(t,\varphi)|_{\mathrm{N}/\mathrm{S}} &= \frac{-ik_U U}{\sqrt{-k_U \pi \alpha}} K_{is}(-ik_U U) &&= \frac{\pm ik_U e^{-t}}{\sqrt{-k_U \pi \alpha}} K_{is}(\pm ik_U e^{-t}) \,, \nonumber \\
    v^+_{k_V}(t,\varphi)|_{\mathrm{N}/\mathrm{S}} &= \frac{-ik_V V}{\sqrt{-k_V \pi \alpha}} K_{is}(-ik_VV) &&= \frac{\mp ik_V e^t}{\sqrt{-k_V \pi \alpha}} K_{is}(\mp ik_V e^t) \,. \label{eq:dS3-NS-modes-explicit}
\end{alignat}
The choice of $\mathrm{N}$ or $\mathrm{S}$ poles corresponds to picking the upper or lower sets of signs.


\subsection{Eikonal Phase}

Next we discuss the computation of the eikonal phase for $\mathrm{dS}_3$, working with \eqref{eq:dS3-metric} with general $\alpha$. We consider a shockwave with energy $e^{T/2} k_V$ localized at angle $\varphi_V$ on the internal circle, which scatters off of a shockwave with energy $e^{T/2} k_U$ localized at angle $\varphi_U$. These source the stress energies:
\begin{align}
    \left. T_{UU}(U,\varphi) \right|_{V = 0} &= e^{\frac{T}{2}}k_U \delta(U)\delta(\varphi-\varphi_U)\,, \qquad \left. T_{VV}(V,\varphi) \right|_{U = 0} = e^{\frac{T}{2}}k_V \delta(V) \delta(\varphi-\varphi_V)\,. \label{eq:dS3-T-ansatz}
\end{align}
The ANEC operators are defined by
\begin{alignat}{2}
    P_U(\varphi) &:= -\int dU \, T_{UU}(U,V = 0, \varphi)\,, \qquad\qquad &&P_U := \int_0^\alpha d\varphi \, P_U(\varphi) \,, \\
    P_V(\varphi) &:= -\int  dV \, T_{VV}(U = 0,V, \varphi)\,, \qquad\qquad  &&P_V := \int_0^\alpha d\varphi \, P_V(\varphi) \,.
\end{alignat}
In the linearized gravity approximation, the backreaction leads to a double shockwave geometry, for which we compute the action to quadratic order:
\begin{align}
    \mathrm{d}s^2 &= \mathrm{d}s^2_{\mathrm{dS}_3} + h_{UU}\,\mathrm{d}U^2 + h_{VV}\,\mathrm{d}V^2, \label{eq:dS3-backreaction-ansatz} \\
    I &= \frac{1}{2} \int \mathrm{d}U \mathrm{d}V\mathrm{d}\varphi\, \sqrt{-g_{\mathrm{d}S_3}}  \Big[h^{UU} \mathcal{D}_- h_{UU} + h^{VV} \mathcal{D}_+ h_{VV} + T^{UU}h_{UU} + T^{VV} h_{VV} \Big] \,. \label{eq:dS3-eikonal-linearized} \\
    \mathcal{D}_{\pm} &:= -\frac{1-UV}{1+UV} (\mp U\partial_U \pm V\partial_V + 2) - \frac{(1-UV)^2}{(1+UV)^2} \partial_{\varphi}^2 \,, \nonumber
\end{align}
The linearized Einstein's equations are solved by:
\begin{align}
    h_{UU} &= -16\pi G_N e^{\frac{T}{2}} k_U \delta(U) f(\varphi-\varphi_U), \qquad h_{VV} = -16\pi G_N e^{\frac{T}{2}} k_V \delta(V) f(\varphi-\varphi_V) \label{eq:dS3-shockwave}
\end{align}
where $f$ is the Greens function for the angular equation
\begin{align}
    (\partial_{\varphi}^2 + 1) f(\varphi-\varphi') &= \delta(\varphi-\varphi'), \qquad\qquad \varphi \equiv \varphi + \alpha \label{eq:dS3-EE-angular} \\
    f(\varphi-\varphi') &= \frac{\cos(|\varphi-\varphi'|-\alpha/2)}{2\sin(\alpha/2)} .\label{eq:dS3-EE-angular-GF}
\end{align}
The zero-mode of \eqref{eq:dS3-EE-angular} is regulated by the observer backreaction. Substituting everything back, the on-shell action evaluates to
\begin{align}
    I &= x \,k_U k_V, \qquad x := -4\pi G_N e^T f(\varphi_U - \varphi_V) \label{eq:dS3-eikonal-phase}
\end{align}
In the limit where $G_N\to 0$ and $G_N e^T$ is of order 1, we note that the metric perturbations scale as $G_N^{1/2}$ while the action remains finite, so the linearized approximation remains valid. Compared to the eikonal phase for $\mathrm{AdS}$, a new feature is that $I$ can take both signs. As discussed in Section \ref{sec:antiscrambling}, this can be interpreted as a Shapiro time advance or delay for our scattering process signaling scrambling or anti-scrambling. When the deficit angle is not too big (so that $\alpha>\pi$), we note the sign is always positive for sufficiently small angular separation. This had to be true because the spacetime looks flat on very small spatial scales, for which there is only a Shapiro time delay. The total S-matrix is
\begin{equation}
    S = \exp\Big[-i 4\pi G_N e^T \int_0^\alpha \mathrm{d}\varphi_U \mathrm{d}\varphi_V \, P_U(\varphi_U) P_V(\varphi_V) \frac{\cos(|\varphi_U-\varphi_V|-\alpha/2 )}{2\sin(  \alpha/2)} \Big] \,.
\end{equation}
In particular, we can take the near-Nariai limit by defining $S_c = \alpha/(4G_N)$, and fixing $S_c$ while taking $\alpha \rightarrow 0$. The result is\footnote{The Kruskal coordinates here differ from those used in the main text by $(U_{\text{here}},V_{\text{here}}) = (\frac{1}{2}U_{\text{there}},\frac{1}{2}V_{\text{there}})$.  Accounting for this difference, \eqref{eq:3Dsmatrix} agrees with \eqref{eq:2deikonal}. We should identify $\Delta S_c$ with $S_c$, as the cosmic horizon has no entropy when $\alpha = 0$. Note that $\Lambda = \ell_{\mathrm{dS}}^{-2}$ in three dimensions, and we have been working in units where $\ell_{\mathrm{dS}} = 1$.}
\begin{equation}
\label{eq:3Dsmatrix}
    S = \exp\left(-i  \frac{\pi}{ S_c} e^T P_U P_V \right)\,, \qquad S_c :=    \frac{\alpha}{4 G_N} \,.
\end{equation}
We see that in the near-Nariai limit, the observer's static patch algebra is equivalent to the algebra generated by a massless chiral scalar on either the future or past boundary of the static patch, and the S-matrix that governs shockwave scattering is given by \eqref{eq:3Dsmatrix}.


\subsection{KMS Violation} \label{sec:3D-KMS}
Finally, we compute out-of-time-ordered correlation functions of fields along the observer worldline, and study their analytic properties. We define:
\begin{align}
    G_1(X_1, X_2, X_3, X_4) &= \langle \phi(-\tfrac{T}{2}+\eta_1,\varphi_1) \phi(\tfrac{T}{2}+\eta_2,\varphi_2) \phi(-\tfrac{T}{2}+\eta_3,\varphi_3) \phi(\tfrac{T}{2}+\eta_4,\varphi_4) \rangle \,, \label{eq:3D-G1} \\
    G_2(X_2, X_3, X_4, X_1) &= \langle \phi(\tfrac{T}{2}+\eta_2,\varphi_2) \phi(-\tfrac{T}{2}+\eta_3,\varphi_3) \phi(\tfrac{T}{2}+\eta_4,\varphi_4) \phi(-\tfrac{T}{2}+\eta_1,\varphi_1) \rangle \,, \label{eq:3D-G2}
\end{align}
where $\eta_i$ denotes static patch time. Using the mode decomposition (\ref{eq:dS3-kU-mode-explicit-1}) and (\ref{eq:dS3-NS-modes-explicit}), each of these may be computed as an overlap between an asymptotic in-state and out-state, with their respective Hilbert spaces related by S-matrix conjugation. Switching to Kruskal coordinates with $U_1,U_3>0$ and $V_2,V_4<0$, we compute:
\begin{align}
    G_1 &= \int_{-\infty}^0 \dbar k_U\dbar k_V \dbar k_U' \dbar k_V'\, e^{i I(k_U k_V, \varphi_3-\varphi_4)} \langle k_U', \varphi_1 |k_U, \varphi_3 \rangle_{\mathrm{in}} \langle k_V', \varphi_2 |k_V, \varphi_4\rangle_{\mathrm{in}} \\
    &\hspace{14mm} \times v_{k_U'}(U_1, \varphi_1)^{\ast} v_{k_V'}(V_2, \varphi_2)^{\ast} v_{k_U}(U_3, \varphi_3) v_{k_V}(V_4, \varphi_4) \nonumber \\
    &= \frac{U_1 U_3 V_2 V_4}{(\pi\alpha)^2}  \int_0^{\infty} \dbar k_U \dbar k_V \, k_Uk_V e^{ix k_Uk_V} K_{is}(i k_U U_1) K_{is}(-i k_U U_3)  K_{is}(i k_V V_2) K_{is}(-i k_V V_4) \,,  \label{eq:3D-G1-integrals} \\
    G_2 &= \frac{U_1 U_3 V_2 V_4}{(\pi\alpha)^2}  \int_0^{\infty} \dbar k_U \dbar k_V \, k_Uk_V e^{-ix k_Uk_V} K_{is}(-i k_U U_1) K_{is}(i k_U U_3)  K_{is}(i k_V V_2) K_{is}(-i k_V V_4) \,, \label{eq:3D-G2-integrals}
\end{align}
where $x = -4\pi G_N e^T f(\varphi_3 - \varphi_4)$. Note that the only angular-dependence is contained in $x$.\footnote{We assume that the observer is equipped with a local frame so that the angular orientation of a shockwave can be fixed. Thus, we do not integrate over $\varphi_3$ and $\varphi_4$.} We now wish to check the KMS condition (\ref{eq:DeltaKMS-t-def}).

\vspace{1em}

\noindent In S-matrix perturbation theory, the $k_U$ and $k_V$ integrals factor at each order:
\begin{align}
    G_1^{\mathrm{pert.}} &= \frac{U_1U_3 V_2V_4}{(\pi\alpha)^2} \sum_{n\geq 0} \frac{(ix)^n}{\Gamma(1+n)} \mathcal{I}_n(iU_1, -iU_3) \, \mathcal{I}_n(iV_2, -iV_4) \,, \\
    G_2^{\mathrm{pert.}} &= \frac{U_1U_3 V_2V_4}{(\pi\alpha)^2} \sum_{n\geq 0} \frac{(-ix)^n}{\Gamma(1+n)} \mathcal{I}_n(-iU_1, iU_3) \, \mathcal{I}_n(iV_2, -iV_4) \\
    \mathcal{I}_n(\alpha,\beta) &= \int_0^{\infty} \mathrm{d}z\, z^{n+1} K_{is}(\alpha z) K_{is}(\beta z) \,, \nonumber \\
    &= 2^{n-1}\alpha^{-2-n-is}\beta^{is} \Gamma(1+\tfrac{n}{2})^2 \Gamma(1+\tfrac{n}{2}-is)\Gamma(1+\tfrac{n}{2}+is) \nonumber \\
    &\hspace{10mm} \times{}_2\tilde{F}_1(1+\tfrac{n}{2}, 1+\tfrac{n}{2}+is, 2+n; 1- \tfrac{\beta^2}{\alpha^2})\,, \qquad\qquad \mathrm{Re}(\alpha+\beta)>0 \,, \label{eq:3D-perturbative-KMS-identity}
\end{align}
where in the last equality we have used the identity \cite{Gradshteyn:2014} (6.576). The integrals for $G_1^{\mathrm{pert.}}$ converge for the prescriptions $U_{13}-i\epsilon$ and $V_{24}-i\epsilon$, while for $G_2^{\mathrm{pert.}}$ we need $U_{13}+i\epsilon$ and $V_{24}-i\epsilon$. To check the KMS condition at each order, it is simpler to work with the closed form integral $\mathcal{I}_n$. The Bessel function $K_{\nu}(z)$ is analytic in $z$ away from a branch cut from $(-\infty,0)$. We are thus able to analytically continue $\mathcal{I}_n(iU_1, -iU_3)$ as $U_1\to e^{-i\pi}U_1$ and $\mathcal{I}_n(-iU_1, iU_3)$ as $U_1\to e^{+i\pi}U_1$. These two integrals may be related by deforming the contour in the opposite way. We find:
\begin{align}
    \mathcal{I}_n(ie^{-i\pi} U_1, -i U_3) &= \int_0^{\infty} \dbar k_U\, k_U^{n+1} K_{is}(ik_U e^{-i\pi} U_1) K_{is}(-ik_U U_3) \\
    &= \int_{-\infty}^0 \dbar k_U\, (e^{i\pi} k_U)^{n+1} K_{is}(ik_U U_1) K_{is}(-i e^{i\pi} k_U U_3) = (-1)^n \mathcal{I}_n(-ie^{i\pi} U_1, i U_3) \,, \nonumber
\end{align}
where in the second equality we have rotate the contour contourclockwise as $k_U \to e^{i\pi}k_U$. The $(-1)^n$ factor compensates for the relative sign in front of $x$ between OTOCs, and at each order we have $G_1^{\mathrm{pert.}}(e^{-i\pi} U_1) = G_2^{\mathrm{pert.}}(e^{i\pi} U_1)$. The result of the integral (\ref{eq:3D-perturbative-KMS-identity}) permits the second continuation $U_1\to e^{-i\pi}U_1$, yielding the KMS condition to all orders in S-matrix perturbation theory. We remark that the perturbative correlators can again be Borel resummed, but the analytic properties of $\mathcal{B}G_i^{\mathrm{pert.}}$ are difficult to study due to the product of ${}_2F_1$'s in the summand.

\vspace{1em}

\noindent We now turn to the full expressions (\ref{eq:3D-G1-integrals})--(\ref{eq:3D-G2-integrals}). We still know that $G_1$ is well-defined under $U_1\to e^{-i\pi}U_1$, and $G_2$ under $U_1\to e^{i\pi}U_1$. For $x>0$ there is a slick way of seeing that $\Delta_{\mathrm{KMS}} = 0$. The important observation is that after the analytic continuation $U_1\to e^{-i\pi}U_1$, we are still able to deform the $k_U$ by $k_U \to e^{i\pi} k_U$. This is justified since it only makes the integral more convergent, as for $x>0$ the exponential $e^{ixk_Uk_V}$ becomes dampened. We do not run into any poles or branch cuts, and we find:
\begin{align}
    &G_1(e^{-i\pi}U_1, V_2, U_3, V_4) \\
    &\hspace{5mm}= \frac{-U_1 U_3 V_2 V_4}{(\pi\alpha)^2}  \int_0^{\infty} \dbar k_U \dbar k_V \, k_Uk_V e^{ix k_Uk_V} K_{is}(i k_U e^{-i\pi}U_1) K_{is}(-i k_U U_3)  K_{is}(i k_V V_2) K_{is}(-i k_V V_4) \nonumber \\
    &\hspace{5mm}= \frac{-U_1 U_3 V_2 V_4}{(\pi\alpha)^2}  \int_{-\infty}^0 \dbar k_U \int_0^{\infty} \dbar k_V \, e^{i\pi} k_Uk_V e^{ix e^{i\pi}k_Uk_V} \nonumber \\
    &\hspace{40mm} \times K_{is}(i e^{i\pi}k_U e^{-i\pi}U_1) K_{is}(-i e^{i\pi}k_U U_3)  K_{is}(i k_V V_2) K_{is}(-i k_V V_4) \nonumber \\
    &\hspace{5mm}= \frac{-U_1 U_3 V_2 V_4}{(\pi\alpha)^2}  \int_0^{\infty} \dbar k_U \dbar k_V \, k_Uk_V e^{-ix k_Uk_V} K_{is}(i k_U U_1) K_{is}(i k_U U_3) K_{is}(i k_V V_2) K_{is}(-i k_V V_4) \nonumber \\
    &\hspace{5mm}= G_2(e^{i\pi}U_1, V_2, U_3, V_4) \,.
\end{align}
The KMS condition is thus satisfied. For $x<0$ we can use a similar trick, except now in order to keep the integrals convergent, we must deform the contour in the other direction $k_U \to e^{-i\pi}k_U$ instead. There are two novel features compared to the $x>0$ case. The first is that we cross the branch cut of the Bessel functions during this procedure. The second is that there is a contribution from the arc at infinity. 
Denoting the total contour by $C_{\mathrm{branch}} + C_{\infty}$, we compute them separately. We first consider:
\begin{align}
    &G_1(e^{-i\pi}U_1, V_2, U_3, V_4)\big|_{C_{\mathrm{branch}}} \nonumber \\
    &\hspace{5mm}= \frac{-U_1 U_3 V_2 V_4}{(\pi\alpha)^2}  \int_{-\infty}^0 \dbar k_U \int_0^{\infty} \dbar k_V \, e^{-i\pi} k_Uk_V e^{ix e^{-i\pi}k_Uk_V} \\
    &\hspace{40mm} \times K_{is}(i e^{-i\pi}k_U e^{-i\pi}U_1) K_{is}(-i e^{-i\pi}k_U U_3)  K_{is}(i k_V V_2) K_{is}(-i k_V V_4) \nonumber \\
    &G_1(e^{-i\pi}U_1, V_2, U_3, V_4)\big|_{C_{\mathrm{branch}}} - G_2(e^{i\pi}U_1, V_2, U_3, V_4) \nonumber \\
    &\hspace{5mm}= \frac{-U_1 U_3 V_2 V_4}{(\pi\alpha)^2}  \int_0^{\infty} \dbar k_U \dbar k_V \, k_Uk_V e^{-ix k_Uk_V}  K_{is}(i k_V V_2) K_{is}(-i k_V V_4) \\
    &\hspace{40mm} \times \Big[ K_{is}(e^{-2\pi i}i k_U U_1) K_{is}(e^{-2\pi i} i k_U U_3) - K_{is}(i k_U U_1) K_{is}(i k_U U_3) \Big]  \nonumber \\
    &\hspace{5mm}= \frac{i |x|}{4\pi \alpha^2} \Big(s^2 + \frac{1}{4}\Big) \frac{V_2 V_4}{\sqrt{U_1 U_3}}  \int_0^{\infty} \dbar k_V \, k_V^2 K_{is}(ik_V V_2) K_{is}(-i k_V V_4) \\
    &\hspace{55mm} \times \,{}_2F_1\Big(\frac{3}{2}-is, \frac{3}{2}+is, 2; \frac{(U_1+U_3)^2 - (x k_V)^2}{4 U_1 U_3} - i\epsilon\Big) \,. \nonumber 
\end{align}
The last step is non-trivial. We have used the Bessel monodromy \cite{DLMF} (10.34.4), written the $k_U$ integrand in terms of Bessel $J$'s, used the (derivative of the) integral identity \cite{Gradshteyn:2014} (6.612.3), and converted the Legendre $Q$'s to a ${}_2F_1$ using \cite{DLMF} (14.3.7). 

\vspace{1em}

\noindent It remains to compute the contribution from the arc at infinity. We use the Bessel asymptotics \cite{DLMF} (10.40.2) to find:
\begin{align}
    e^{i x k_U k_V} k_U K_{is}(-ik_U U_1) K_{is}(-ik_U U_3) \to \frac{i\pi}{2\sqrt{U_1 U_3}} e^{i(U_1+U_3 + x k_V)k_U} \Big[1 + \mathcal{O}(|k_U|^{-1}) \Big] \,.
\end{align}
Since $|k_U|\to \infty$, the $k_U$-integral yields a $\delta$-function, fixing $k_V = (U_1+U_3)/|x|$ in the remaining $k_V$ integral. Note that this contribution can also be written for $x>0$, but since $U_1,U_3>0$ the $\delta$-function always vanishes. This appears to be the manifestation of the contact term in the 2D Tomita-Takesaki derivation (\ref{eq:2D-AtApt-commutator-calc}). We thus have
\begin{align}
    G_1(e^{-i\pi}U_1, V_2, U_3, V_4)\big|_{C_{\infty}} &= \frac{-i\sqrt{U_1U_3}V_2 V_4}{(2\pi\alpha)^2} \frac{U_1+U_3}{x^2} K_{is}\Big(\frac{i}{|x|} V_2(U_1+U_3) \Big) K_{is}\Big(\frac{-i}{|x|}V_4(U_1+U_3)\Big)
\end{align}
Finally, we put everything together and make the second analytic continuation $U_1\to e^{-i\pi}U_1$. Our result is:
\begin{align}
    \Delta_{\mathrm{KMS}} &= G_1(e^{-2i\pi}U_1, V_2, U_3, V_4) - G_2(V_2, U_3, V_4, U_1) \\
    &= \frac{\theta(-x)}{(2\pi \alpha)^2x^2} \bigg[ \big(\sqrt{U_1 U_3} V_2 V_4 \, U_{13}\big) \, K_{is}\Big(-\frac{i}{|x|} V_2(U_{13}+i\epsilon) \Big) K_{is}\Big(\frac{i}{|x|}V_4(U_{13}+i\epsilon)\Big) \nonumber \\
    &\hspace{25mm} +\Big(s^2 + \frac{1}{4}\Big) \frac{-V_2 V_4}{2\sqrt{U_1 U_3}} \int_0^{\infty} \mathrm{d}z \, z^2 K_{is}\Big(\frac{iV_2}{|x|} z\Big) K_{is}\Big(\hspace{-1mm}-\frac{iV_4 }{|x|}z\Big) \nonumber \\
    &\hspace{70mm} \times \,{}_2F_1\Big(\frac{3}{2}-is, \frac{3}{2}+is, 2; \frac{z^2 - (U_{13}^2+i\epsilon)}{4 U_1 U_3} \Big) \bigg] \label{eq:3D-DeltaKMS}
\end{align}
The first term is $s$-independent, while the second term is not. Compared to the 2D case the second term is new, arising from analytic continuation past a branch cut of the two-point function. Note that this piece vanishes when $s=i/2$, for which the scalar is conformally coupled. Finally, we can take the $x\to 0^-$ limit, when $U_1$ and $U_3$ (also $V_2$ and $V_4$) are distinct so there is no need for smearing. The second term oscillates rapidly. One can show using the stationary phase approximation with endpoints that it is subleading to the first. In total we find:
\begin{align}
    \lim_{x\to 0}\Delta_{\mathrm{KMS}}(U_1, V_2, U_3, V_4) &= (U_1 U_3 V_2 V_4)^{\frac{1}{2}}\frac{\theta(-x)}{8\pi\alpha^2 |x|} e^{-\frac{i}{x}U_{13}V_{24}} \Big(1 + \mathcal{O}(x^{-1})\Big) \,.
\end{align}
Up to a numerical prefactor, this matches the 2D result\footnote{once the 2D conformal factors are restored.}  (\ref{eq:2D-DeltaKMS-UV-explicit-0}), with $h_A=h_B=1/2$.


\newpage

\section{2D Model} \label{sec:2Dcalcs}

\subsection{Correlation Functions} \label{sec:2Dcalcs-OTOCs}

\subsubsection{Exact Correlators} \label{sec:2Dcalcs-OTOCs-exact}

Here we derive explicit formulas for the OTOCs in position space (\ref{eq:2D-G1-UV-explicit})--(\ref{eq:2D-G2-UV-explicit}) and frequency space (\ref{eq:2D-G1-lambda-explicit})--(\ref{eq:2D-G2-lambda-explicit}). For real times with an appropriate $i\epsilon$-prescription we check that they are Fourier conjugates. For analytically continued times, the Fourier relation remains valid for $x>0$ so long as we are in the KMS strip. However, for $x<0$ the Fourier integrals are only convergent on half the KMS strip. The consequences for $\Delta_{\mathrm{KMS}}$ and $\Delta_{\mathrm{db.}}$ are discussed in Section \ref{sec:2D-OTOC-analc}.

\vspace{1em}

\noindent \textbf{Correlators at Real Times}

\noindent In the following we take $U_1, U_3>0$ and $V_2, V_4<0$. We compute:
\begin{align}
    G_1(U_1, V_2, U_3, V_4) &:= \langle \phi^A(U_1) \phi^B(V_2) e^{ix\hat{P}_A\hat{P}_B} \phi^A(U_3) \phi^B(V_4) \rangle \\
    &= \int_{-\infty}^0 \frac{\dbar P_A\dbar P_B }{\Gamma(2h_A)\Gamma(2h_B)} |P_A|^{2h_A-1} |P_B|^{2h_B-1} e^{iP_A(U_{13}-i\epsilon)} e^{iP_B(V_{24}-i\epsilon)} e^{ixP_AP_B} \label{eq:2D-G1-PAPB-integral} \\
    &= (-1)^{h_A} \int_{-\infty}^0 \frac{\dbar P_B}{\Gamma(2h_B)} \, e^{iP_B(V_{24}-i\epsilon)} |P_B|^{2h_B-1} (U_{13}-i\epsilon+xP_B)^{-2h_A} \,, \\
    G_2(V_2, U_3, V_4, U_1) &:= \langle \phi^B(V_2) \phi^A(U_3) e^{-ix\hat{P}_A\hat{P}_B} \phi^B(V_4) \phi^A(U_1) \rangle \\
    &= \int_{-\infty}^0 \frac{\dbar P_A\dbar P_B }{\Gamma(2h_A)\Gamma(2h_B)} |P_A|^{2h_A-1} |P_B|^{2h_B-1} e^{-iP_A(U_{13}+i\epsilon)} e^{iP_B(V_{24}-i\epsilon)} e^{-ixP_AP_B} \label{eq:2D-G2-PAPB-integral} \\
    &= (-1)^{h_A} \int_{-\infty}^0 \frac{\dbar P_B}{\Gamma(2h_B)} \, e^{iP_B(V_{24}-i\epsilon)} |P_B|^{2h_B-1} (U_{13}+i\epsilon-xP_B)^{-2h_A} \,,
\end{align}
where in the second pair of equalities we have inserted a resolution of the identity in the $P_{A,B}$-basis and used (\ref{eq:2D-Plambda-matrixel}). After computing the $P_A$ integral, the remaining $P_B$ integral can be evaluated using the identity \cite{DLMF} (13.4.4):
\begin{align}
    \int_0^{\infty}\mathrm{d}s \,e^{-zs} s^{a-1} (s+1)^{-a+b-1} = \Gamma(a) U(a,b;z)\,, \qquad \mathrm{Re}(a), \mathrm{Re}(z) > 0 \,. \label{eq:hypergeometricU-integral-identity}
\end{align}
Let us focus on $G_1$. The cases $x>0$ and $x<0$ must be treated separately. For $x>0$ the identity (\ref{eq:hypergeometricU-integral-identity}) applies for $U_{13}<0$ and $V_{24}<0$; the OTOC for other parameter regions is obtained through analytic continuation. The function $U(a,b;z)$ is $z$-analytic aside from a branch cut from $(-\infty,0)$, which is crossed downwards (counterclockwise) onto the lower sheet only when $U_{13}$ and $V_{24}$ are both positive due to the $i\epsilon$-prescription. For $x<0$, (\ref{eq:hypergeometricU-integral-identity}) is valid for $U_{13}>0$ and $V_{24}>0$. We cross the branch cut upwards (clockwise) for $U_{13},V_{24}<0$, continuing on to the upper sheet. Denoting the main, upper, and lower sheets by $U^0$, $U^+$, and $U^-$, we have altogether (\ref{eq:2D-G1-UV-explicit}):
\begin{align}
    G_1(U_1, V_2, U_3, V_4) &= \frac{(-1)^{h_A}}{(2\pi)^2} x^{-2h_B} (U_{13}-i\epsilon)^{2h_B-2h_A} U^{\sigma_1} \Big(2h_B, 1-2h_A+2h_B; -\frac{i}{x}U_{13}V_{24}\Big)\,, \\
    \sigma_1 &=
    \begin{cases}
        +\,, & x<0 \cap U_{13}<0 \cap V_{24}<0\\
        0\,, & (x>0 \cap (U_{13}<0 \cup V_{24}<0)) \cup (x<0 \cap (U_{13}>0 \cup V_{24}>0)) \\
        -\,, & x>0 \cap U_{13}>0 \cap V_{24}>0
    \end{cases} \,, \nonumber
\end{align}
where the $i\epsilon$'s in the hypergeometric function can now be dropped. The computation for the alternate ordered OTOC is very similar, and yields (\ref{eq:2D-G2-UV-explicit}).

\vspace{1em}

\noindent We may also compute the correlation functions in frequency space.
\begin{align}
    G_1(\lambda_1,\lambda_2,\lambda_3,\lambda_4)  &= \langle \Phi^A(\lambda_1) \Phi^B(\lambda_2) e^{ix\hat{P}_A\hat{P}_B} \Phi^A(\lambda_3) \Phi^B(\lambda_4) \rangle \\
    &= \int \dbar\lambda \, F_x(\lambda) \int \dbar \lambda_A\dbar\lambda_B \nonumber \\
    &\hspace{10mm} \times \langle\Omega|\Phi^A(\lambda_1)\Phi^B(\lambda_2)|\lambda_A+\lambda,\lambda_B+\lambda\rangle \langle \lambda_A,\lambda_B|\Phi^A(\lambda_3)\Phi^B(\lambda_4) |\Omega\rangle \nonumber \\
    &= C_{\lambda_1\lambda_2\lambda_3\lambda_4} e^{-i\pi(h_A+i\lambda_1)} e^{i\pi(h_B-i\lambda_2)} \int\dbar\lambda\, F_x(\lambda) \, (2\pi)^2 \delta(\lambda_1+\lambda_3+\lambda) \delta(\lambda_2+\lambda_4-\lambda) \nonumber \\
    &= 2\pi\delta(\lambda_1+\lambda_2+\lambda_3+\lambda_4) \, C_{\lambda_1\lambda_2\lambda_3\lambda_4} \, e^{\pi(\lambda_2+\lambda_1)} F_x(-\lambda_1-\lambda_3) \,,
\end{align}
where in the second equality we have inserted a complete basis of boost-eigenstates $|\lambda_A,\lambda_B\rangle$ and used (\ref{eq:2D-PAPB-lambda-basis}). In the third equality we use (\ref{eq:Phitovarphi}), (\ref{eq:2D-varphiA-lambda-vac}), (\ref{eq:2D-varphiB-lambda-vac}), and have defined:
\begin{align}
    C_{\lambda_1\lambda_2\lambda_3\lambda_4} &:= (-1)^{h_A+h_B} \langle\lambda_1|\phi^A(1)|\Omega\rangle\langle\lambda_3|\phi^A(1)|\Omega\rangle \langle -\lambda_2|\phi^B(-1) |\Omega\rangle \langle -\lambda_4|\phi^B(-1) |\Omega\rangle \\
    &= \frac{e^{-\frac{\pi}{2}(\lambda_1+\lambda_2+\lambda_3+\lambda_4)}}{(2\pi)^2 \Gamma(2h_A)\Gamma(2h_B)}  \Gamma(h_A+i\lambda_1) \Gamma(h_B-i\lambda_2) \Gamma(h_A+i\lambda_3) \Gamma(h_B-i\lambda_4) \,.
\end{align}
A similar calculation yields
\begin{align}
    G_2(\lambda_2,\lambda_3,\lambda_4,\lambda_1) &= \langle \Phi^B(\lambda_2) \Phi^A(\lambda_3) e^{-ix\hat{P}_A\hat{P}_B} \Phi^B(\lambda_4) \Phi^A(\lambda_1) \rangle \\
    &= C_{\lambda_1\lambda_2\lambda_3\lambda_4} e^{-i\pi(h_A+i\lambda_3)} e^{i\pi(h_B+i\lambda_2)} \int\dbar\lambda\, F_{-x}(\lambda) \,(2\pi)^2 \delta(\lambda_1+\lambda_3+\lambda) \delta(\lambda_2+\lambda_4-\lambda) \nonumber \\
    &= 2\pi\delta(\lambda_1+\lambda_2+\lambda_3+\lambda_4) \, C_{\lambda_1\lambda_2\lambda_3\lambda_4} \, e^{\pi(\lambda_2+\lambda_3)}  F_{-x}(-\lambda_1-\lambda_3) \,.
\end{align}
Together, these reproduce (\ref{eq:2D-G1-lambda-explicit}) and (\ref{eq:2D-G2-lambda-explicit}).

\newpage

\noindent Next, we check that $G_1(\lambda_1,\lambda_2,\lambda_3,\lambda_4)$ is the Fourier transform of $G_1(t_1,t_2,t_3,t_4)$. Each $t_i$ is given an $i\epsilon$-prescription which a priori is not fixed.
\begin{align}
    \mathcal{F}^{-1}_{\lambda_i\to t_i+i\epsilon} G_1 &= \int \dbar\lambda_1\cdots\dbar\lambda_4 \, e^{-i\sum_{i=1}^4 \lambda_i (t_i + i\epsilon_i)} G_1(\lambda_1,\lambda_2,\lambda_3,\lambda_4) \label{2D-G1-Fourier-realt-calc} \\
    &= \frac{1}{(2\pi)^2 \Gamma(2h_A)\Gamma(2h_B)} \int \dbar\lambda \, F_x(\lambda) e^{i\lambda(t_1-t_4 + i(\epsilon_1+\epsilon_4))} e^{-\pi\lambda} \\
    &\hspace{10mm} \times \int \dbar\lambda_2\,  e^{-i\lambda_2(t_2-t_4 + i(\epsilon_2-\epsilon_4))} e^{\pi\lambda_2} \Gamma(h_B-i\lambda_2) \Gamma(h_B-i(\lambda-\lambda_2)) \nonumber \\
    &\hspace{10mm} \times \int \dbar\lambda_3\,  e^{i\lambda_3(t_1-t_3 + i(\epsilon_1-\epsilon_3))} e^{-\pi\lambda_3} \Gamma(h_A+i\lambda_3) \Gamma(h_A-i(\lambda+\lambda_3)) \nonumber \\
    &= \frac{(-1)^{h_A+h_B}}{(2\pi)^2 \Gamma(2h_A)\Gamma(2h_B)} \frac{e^{h_A(t_1+t_3)}}{(U_1 + e^{i(-\pi + \epsilon_3-\epsilon_1)}U_3)^{2h_A}} \frac{e^{-h_B(t_2+t_4)}}{(-V_4 + e^{i(-\pi + \epsilon_4-\epsilon_2)}(-V_2) )^{2h_B}} \\
    &\hspace{5mm} \int \dbar \lambda\, F_x(\lambda) e^{\lambda(-\pi +\epsilon_4 - \epsilon_1)} \Gamma(2h_A-i\lambda) \Gamma(2h_B-i\lambda) \nonumber \\
    &\hspace{10mm} \times (U_1 + e^{i(-\pi + \epsilon_3-\epsilon_1)}U_3)^{i\lambda} (-V_4 + e^{i(-\pi + \epsilon_4-\epsilon_2)}(-V_2))^{i\lambda} \,. \nonumber
\end{align}
In the first equality we use the energy-conserving $\delta$-function to perform the $\lambda_4$ integral, and make the change of variable $\lambda = -\lambda_1-\lambda_3$. The $\lambda_2$ and $\lambda_3$ integrals immediately factorize. They respectively converge if $\epsilon_4-\epsilon_2>0$ and $\epsilon_3-\epsilon_1>0$, after which we evaluate using the Mellin-Barnes identity \cite{Gradshteyn:2014} (6.422.3):
\begin{align}
    \frac{1}{2\pi i}\int_{\gamma-i\infty}^{\gamma+i\infty} \mathrm{d}s\, \Gamma(-s) \Gamma(\beta+s) t^s = \Gamma(\beta) (1+t)^{-\beta}, \qquad \mathrm{Re}(1-\beta) < \gamma <0, \quad |\arg(t)|<\pi \,.  \label{eq:2D-MB-identity1}
\end{align}
Rewriting the simplified expression in Kruskal coordinates yields the third equality. The remaining $\lambda$-integral converges for both signs of $x$ provided $\epsilon_4-\epsilon_1>0$ and $\epsilon_3-\epsilon_2>0$. We note that all integrals converge for $0<\epsilon_1 < \epsilon_2 < \epsilon_3 < \epsilon_4$, which is the order (\ref{eq:OTOC-analyticity}) for which the position-space correlator $G_1(t_1,t_2,t_3,t_4)$ is well-defined. In Kruskal coordinates this corresponds to replacing $U_{13} \to U_{13}-i\epsilon$ and $V_{24}\to V_{24}-i\epsilon$, in agreement with (\ref{eq:2D-G1-UV-explicit}). For brevity we omit all $i\epsilon$'s for the rest of this calculation, with the above prescription taken to be implicit. The final integral can be evaluated using \cite{Gradshteyn:2014} (6.422.1):
\begin{align}
    &\hspace{-5mm} \frac{1}{2\pi i} \int_{-i\infty}^{i\infty} \mathrm{d}s\, \Gamma(s-k-\lambda)\Gamma(\lambda+\mu-s+\tfrac{1}{2}) \Gamma(\lambda-\mu-s+\tfrac{1}{2}) z^s \nonumber \\
    &= \Gamma(\tfrac{1}{2}-k-\mu) \Gamma(\tfrac{1}{2}-k+\mu) z^{\frac{1}{2}+\lambda+\mu} \, U^{\sigma}(\tfrac{1}{2}+\mu-k, 1+2\mu; z) \,, \\
    \sigma &= \begin{cases}
        0, & \arg z \in [-\pi, \pi] \\
        +, & \arg z \in (-\tfrac{3\pi}{2}, -\pi) \\
        -, & \arg z \in (\pi, \tfrac{3\pi}{2})
    \end{cases}, \qquad \mathrm{Re}(k+\lambda) <0, \qquad \mathrm{Re}(\lambda) > |\mathrm{Re} \, \mu| - \tfrac{1}{2} ,\qquad |\arg z| < \tfrac{3\pi}{2} \,. \label{eq:2D-MB-identity2} \nonumber 
\end{align}
Substituting this, we find:
\begin{align}
    \mathcal{F}^{-1}_{\lambda_i\to t_i} G_1 &= e^{h_A(t_1+t_3)} e^{-h_B(t_2+t_4)} \frac{(-1)^{h_A}}{(2\pi)^2} x^{-2h_B} U_{13}^{2h_B-2h_A}U^{\sigma_1}(2h_B, 1+2h_B-2h_A; -\tfrac{i}{x} U_{13} V_{24}) \nonumber \\
    &= e^{h_A(t_1+t_3)} e^{-h_B(t_2+t_4)} G_1(U_1, V_2, U_3, V_4) = G_1(t_1, t_2, t_3, t_4) \,,
\end{align}
where $\sigma_1$ is as defined in (\ref{eq:2D-G1-UV-explicit}). The explicit $t$-dependence in the final line precisely match the conformal factors (\ref{eq:2D-A-conformal}), (\ref{eq:2D-B-conformal}) in transforming from Kruskal coordinates to static patch time, yielding our desired result.

\vspace{1em}

\noindent \textbf{Fourier Transform at Analytically Continued Times} \\
Finally, we check if and when the analytically continued time and frequency domain correlators are still related by analytic continuation (\ref{eq:2D-G1-Fourier-analc}):
\begin{align}
    \mathcal{F}^{-1}_{\lambda_i\to t_i+i\epsilon_i} \big[e^{-2\pi y\lambda_1} G_1(\lambda_1,\lambda_2,\lambda_3,\lambda_4) \big] \stackrel{?}{=} G_1(t_1-2\pi i y, t_2, t_3, t_4), \qquad y\in \mathbb{R}\,.
\end{align}
For simplicity, only the first time is analytically continued off the real axis. We parameterize this by an imaginary shift $-2\pi i y$ for $y\in\mathbb{R}$, while keeping all $t_i$'s real (up to $i\epsilon$'s). The calculation is very similar to the real-time Fourier transform calculation (\ref{2D-G1-Fourier-realt-calc}). Again we change variables to $\lambda=-\lambda_1-\lambda_3$, and perform the integrals in the order (right-to-left) $\dbar\lambda \, (\dbar\lambda_2\dbar\lambda_3) \, \dbar\lambda_4$:
\begin{align}
    &\mathcal{F}^{-1}_{\lambda_i\to t_i+i\epsilon} \big[e^{-2\pi y \lambda_1}G_1\big] = \int \dbar\lambda_1\cdots\dbar\lambda_4 \, e^{-i\lambda_1(t_1+i\epsilon_1-2\pi iy)} e^{-i\sum_{i=1}^3 \lambda_i (t_i + i\epsilon_i)} G_1(\lambda_1,\lambda_2,\lambda_3,\lambda_4) \label{2D-G1-Fourier-analc-calc} \\
    &\hspace{5mm}= \frac{1}{(2\pi)^2 \Gamma(2h_A)\Gamma(2h_B)} \int \dbar\lambda \, F_x(\lambda) e^{i\lambda(t_1-t_4 + i(\epsilon_1+\epsilon_4))} e^{-\pi\lambda(1-2y)} \\
    &\hspace{15mm} \times \int \dbar\lambda_2\,  e^{-i\lambda_2(t_2-t_4 + i(\epsilon_2-\epsilon_4))} e^{\pi\lambda_2} \Gamma(h_B-i\lambda_2) \Gamma(h_B-i(\lambda-\lambda_2)) \nonumber \\
    &\hspace{15mm} \times \int \dbar\lambda_3\,  e^{i\lambda_3(t_1-t_3 + i(\epsilon_1-\epsilon_3))} e^{-\pi\lambda_3(1-2y)} \Gamma(h_A+i\lambda_3) \Gamma(h_A-i(\lambda+\lambda_3)) \nonumber \\
    &\hspace{5mm}= \frac{(-1)^{h_A+h_B}}{(2\pi)^2 \Gamma(2h_A)\Gamma(2h_B)} \frac{e^{h_A(t_1+t_3)}e^{2\pi ih_Ay}}{(U_1 + e^{i(-\pi(1-2y) + \epsilon_3-\epsilon_1)}U_3)^{2h_A}} \frac{e^{-h_B(t_2+t_4)}}{(-V_4 + e^{i(-\pi + \epsilon_4-\epsilon_2)}(-V_2) )^{2h_B}} \\
    &\hspace{10mm} \int \dbar \lambda\, F_x(\lambda) e^{-\lambda(\pi(1-2y) + \epsilon_1-\epsilon_4)} \Gamma(2h_A-i\lambda) \Gamma(2h_B-i\lambda) \nonumber \\
    &\hspace{15mm} \times (U_1 + e^{i(-\pi(1-2y) + \epsilon_3-\epsilon_1)}U_3)^{i\lambda} (-V_4 + e^{i(-\pi + \epsilon_4-\epsilon_2)}(-V_2))^{i\lambda} \,. \nonumber
\end{align}
In the third equality we go to Kruskal coordinates, where $U_1,U_3>0$ and $V_2,V_4<0$. The $\lambda_2$-integral converges if and only if $\epsilon_4>\epsilon_2$. The $\lambda_3$-integral can only converge when $y\in [0,1]$. This is regulator independent in the interior of the interval, but at $y=0$ and $y=1$ we need (respectively) $\epsilon_3>\epsilon_1$ and $\epsilon_3<\epsilon_1$. Finally there is the integral over $\lambda$, whose study we need only restrict to $y\in [0,1]$. The $x>0$ and $x<0$ cases are different. First consider $x>0$. For $y\in(0,1)$, the $\lambda$-integral converges without additional regulator conditions for all real $U_{13}$ and $V_{24}$. At $y=0$ we further need $\epsilon_4>\epsilon_1$. At $y=1$, we need $\epsilon_1 > \epsilon_2$. Next we study $x<0$. For $y\in (0,1/2)$, the $\lambda$-integral converges without additional regulator conditions for all real $U_{13}$ and $V_{24}$. Convergence at $y=0$ requires $\epsilon_3>\epsilon_2$, while for $y=1/2$ we need $\epsilon_1>\epsilon_2$. For $y\in (1/2, 1)$ the $\lambda$-integral diverges for $V_{24}<0$ and converges for $V_{24}>0$, both independently of regulator prescriptions. 

\vspace{1em}

\noindent When all integrals converge, we can again use \cite{Gradshteyn:2014} (6.422.1) to show that the inverse Fourier transform (\ref{2D-G1-Fourier-analc-calc}) yields $G_1(t_1-2\pi i y, t_2, t_3, t_4)$. We collect our results:
\begin{enumerate}
    \item Case 1: $x>0$. For $y\in [0,1)$, the Fourier integrals converge for all real times if $\epsilon_4>\epsilon_3>\epsilon_2>\epsilon_1$. For $y \in (0,1]$, all integrals converge for all times if $\epsilon_1>\epsilon_4>\epsilon_3>\epsilon_2$. These are the real-time regulator prescriptions for $G_1$ and $G_2$, respectively. Furthermore, there is no $i\epsilon$-prescription for which the integrals converge for all times if $y$ is outside of $[0,1]$. We see the Fourier representation of $G_1$ is always convergent on the KMS strip $\mathrm{Im}(t_1-2\pi iy) \in (-2\pi,0]$, and this strip is maximal.
    \item Case 2: $x<0$. For $y\in [0,1/2)$, all Fourier integrals converge for all real times given the $i\epsilon$-prescription of $G_1$. When $V_{24}>0$ with this prescription we additionally have convergence for $y\in [1/2,1)$. However, for any $y\in(1/2,1)$ and $V_{24}<0$, the Fourier integrals diverge regardless of the regulator. Convergence for all real times at the midpoint $y=1/2$ is satisfied for the $G_2$ regulator prescription. The maximal strip $y\in [0,1/2)$ where the Fourier representation of $G_1$ is always convergent constitutes only half of the KMS strip.
\end{enumerate}
A similar computation shows that the Fourier representation of $G_2(t_2, t_3, t_4, t_1-2\pi i y)$ with $\epsilon_1>\epsilon_4>\epsilon_3>\epsilon_2$ is always well-behaved on $y\in (-1, 0]$ for $x>0$, and $y\in (-1/2, 0]$ for $x<0$.


\subsubsection{Perturbative Correlators and Borel Resummation} \label{sec:2D-Borel} 
In this section we show how the exact correlation functions and the KMS violation can be recovered from the perturbative series through Borel resummation. In particular, the KMS violating term is given by the contour integral about a singularity in the Borel plane.

The OTOCs are computed in perturbation theory by expanding the integrands (\ref{eq:2D-G1-PAPB-integral}), (\ref{eq:2D-G2-PAPB-integral}) as a power series in $x$, and exchanging the sum with the $P_A$, $P_B$ integrals, which subsequently factorize. Strictly speaking these limits do not commute, which manifests through an asymptotic series with zero radius of convergence:
\begin{align}
    G_1^{\mathrm{pert.}}(U_1, V_2, U_3, V_4) &= \frac{1}{(2\pi)^2} \frac{(-1)^{h_A+h_B}}{(U_{13}-i\epsilon)^{2h_A} (V_{24}-i\epsilon)^{2h_B}} \sum_{n\geq 0} \Big(\frac{-ix}{(U_{13}-i\epsilon)(V_{24}-i\epsilon)}\Big)^n \frac{(2h_A)_n (2h_B)_n}{\Gamma(1+n)}\,.
\end{align}
The Borel transformed series is defined by dividing each term by $n!$. This new sum converges with finite radius of convergence to a hypergeometric ${}_2F_1$:
\begin{align}
    \mathcal{B}G_1^{\mathrm{pert.}} &= \frac{1}{(2\pi)^2} \frac{(-1)^{h_A+h_B}}{(U_{13}-i\epsilon)^{2h_A} (V_{24}-i\epsilon)^{2h_B}} {}_2F_1 \Big(2h_A, 2h_B, 1; \frac{-ix}{(U_{13}-i\epsilon)(V_{24}-i\epsilon)}\Big) \,, \nonumber \\
    \mathcal{B}G_2^{\mathrm{pert.}} &= \frac{1}{(2\pi)^2} \frac{(-1)^{h_A+h_B}}{(U_{13}+i\epsilon)^{2h_A} (V_{24}-i\epsilon)^{2h_B}} {}_2F_1 \Big(2h_A, 2h_B, 1; \frac{-ix}{(U_{13}+i\epsilon)(V_{24}-i\epsilon)}\Big) \,,
\end{align}
Since $h_A$ and $h_B$ are positive integers, the hypergeometric function ${}_2F_1(z)$ is holomorphic away from a pole at $z=1$. The prescription to go around the pole is specified by the $i\epsilon$'s. For $\mathcal{B}G_1^{\mathrm{pert.}}$, we stay clear of this pole in the region where $\sigma_1(x,U_{13},V_{24})=0$, where $\sigma_1$ was defined in (\ref{eq:2D-G1-UV-explicit}). The same applies to $\mathcal{B}G_2^{\mathrm{pert.}}$ with $\sigma_2(x,U_{13},V_{24})=0$.

\vspace{1em}

\noindent The Borel resummed functions are obtained from taking the Laplace transforms:
\begin{align}
    \mathcal{S}G_1^{\mathrm{pert.}} &= \frac{1}{(2\pi)^2}  \frac{(-1)^{h_A+h_B}}{(U_{13}-i\epsilon)^{2h_A} (V_{24}-i\epsilon)^{2h_B}} \int_0^{\infty} \mathrm{d}s \, e^{-s} {}_2F_1\Big(2h_A, 2h_B, 1; \frac{-ix s}{(U_{13}-i\epsilon)(V_{24}-i\epsilon)}\Big) \,, \nonumber \\
    \mathcal{S}G_2^{\mathrm{pert.}}  &= \frac{1}{(2\pi)^2} \frac{(-1)^{h_A+h_B}}{(U_{13}+i\epsilon)^{2h_A} (V_{24}-i\epsilon)^{2h_B}} \int_0^{\infty} \mathrm{d}s \, e^{-s} {}_2F_1\Big(2h_A, 2h_B, 1; \frac{-ix s}{(U_{13}+i\epsilon)(V_{24}-i\epsilon)}\Big) \,. \label{eq:2D-G-Borel-integral-2F1}
\end{align}
Let us focus on $\mathcal{S}G_1^{\mathrm{pert.}}$. The integral is well-defined if there are no singularities on the integration contour $(0,\infty)$. The only singularity in the Borel plane is at $s_1^{\ast} := (i/x)(U_{13}-i\epsilon)(V_{24}-i\epsilon)$, which crosses the contour when $U_{13}V_{24} \to 0^+$ with $\mathrm{sgn}(x) = \mathrm{sgn}(U_{13}) = \mathrm{sgn}(V_{24})$. When $U_{13}$ and $V_{24}$ are not in this problematic quadrant, we use the identity
\begin{align}
    \int_0^\infty \mathrm{d}s\, e^{-s} {}_2F_1(a,b;1; s z) = (-z)^{-b} U(a,1+b-a, -1/z)\,.
\end{align}
This recovers the exact correlator $G_1$ in the region where $\sigma_1(x,U_{13},V_{24})=0$. To derive the correlator in the final quadrant, note that when $s_1^{\ast}$ crosses the Laplace contour, the expression $\mathcal{S}G_1^{\mathrm{pert.}}$ picks up a contour around this pole, which can be computed by its residue. By Cauchy's theorem, the residue has a plus sign in front if the real axis is crossed downwards, and a minus sign if crossed from below. There is a similar calculation with $\mathcal{S}G_2^{\mathrm{pert.}}$, where the integrand has a singularity at $s_2^{\ast} := (i/x)(U_{13}+i\epsilon)(V_{24}-i\epsilon)$. The residues perfectly match the hypergeometric-$U$ monodromy between sheets, and we recover (\ref{eq:2D-G1-UV-explicit}) and (\ref{eq:2D-G2-UV-explicit}) at finite $x$ for all $U_{13}, V_{24}\neq 0$. The residue contributions are summarized in columns 4 and 6 of Table \ref{table:Borel-plane-residues}.

\vspace{1em}

\noindent Lastly, let us see how the KMS violating contribution arises from the Borel calculation. We wish to analytically continue $\mathcal{S}G_1^{\mathrm{pert.}}$ as $U_1\to e^{-2\pi i}U_1$. This yields:
\begin{align}
    &\mathcal{S}G_1^{\mathrm{pert.}}(e^{-2\pi i}U_1) \\
    &\hspace{5mm}= \frac{1}{(2\pi)^2} \frac{(-1)^{h_A+h_B}}{(U_{13}+i\epsilon)^{2h_A} (V_{24}-i\epsilon)^{2h_B}} \int_0^{\infty} \mathrm{d}s \, {}_2F_1\Big(2h_A, 2h_B, 1; \frac{-ix s}{(U_{13}+i\epsilon)(V_{24}-i\epsilon)}\Big) + \text{residues} \,. \nonumber \\
    &\hspace{5mm}= \mathcal{S}G_2^{\mathrm{pert.}}(U_1) + \text{residues} \,,
\end{align}
where we have worked modulo residues from singularities crossing the positive real axis. Let us now check that these residues restore the expected results. During the analytic continuation, the Borel pole crosses the real axis in the downward direction if and only if $U_{13}>0$, in which case the expression picks up a residue with a plus sign. Note that this occurs for both signs of $x$. When $x>0$, these add together with the residues in the original $\mathcal{S}G_1^{\mathrm{pert.}}$ to match those in $\mathcal{S}G_2^{\mathrm{pert.}}$ and $\Delta_{\mathrm{KMS}}=0$. However, when $x<0$, $\mathcal{S}G_1^{\mathrm{pert.}}(e^{-2\pi i}U_1)$ always differs by a residue. These results are summarized in the following table.
\begin{table}[h!]
    \centering
    \begin{tabular}{@{}ccccr@{\hspace{0.4em}${}+{}$\hspace{0.4em}}lc@{}}
        \toprule
        &
        $\mathrm{sgn}\,U_{13}$ &
        $\mathrm{sgn}\,V_{24}$ &
        $\mathcal{S}G_1^{\mathrm{pert.}}(U_1)$ &
        \multicolumn{2}{c}{$\mathcal{S}G_1^{\mathrm{pert.}}(e^{-2\pi i}U_1)$} &
        $\mathcal{S}G_2^{\mathrm{pert.}}(U_1)$
        \\ \midrule
    
        \multirow{4}{*}{$x>0$}
        & $+$ & $+$
        & $-\mathrm{res.}$
        & $-\mathrm{res.}$ & $\mathrm{res.}$
        & $0$
        \\
        & $+$ & $-$
        & $0$
        & $0$ & $\mathrm{res.}$
        & $+\mathrm{res.}$
        \\
        & $-$ & $+$
        & $0$
        & $0$ & $0$
        & $0$
        \\
        & $-$ & $-$
        & $0$
        & $0$ & $0$
        & $0$
        \\ \midrule
    
        \multirow{4}{*}{$x<0$}
        & $+$ & $+$
        & $0$
        & $0$ & $\mathrm{res.}$
        & $0$
        \\
        & $+$ & $-$
        & $0$
        & $0$ & $\mathrm{res.}$
        & $0$
        \\
        & $-$ & $+$
        & $0$
        & $0$ & $0$
        & $-\mathrm{res.}$
        \\
        & $-$ & $-$
        & $+\mathrm{res.}$
        & $+\mathrm{res.}$ & $0$
        & $0$
        \\ \bottomrule
    \end{tabular}
    \caption{Residue contributions from the Borel singularity crossing the positive real axis for the correlators $\mathcal{S}G_1^{\mathrm{pert.}}(U_1)$, $\mathcal{S}G_2^{\mathrm{pert.}}(U_1)$, and $\mathcal{S}G_1^{\mathrm{pert.}}(e^{-2\pi i}U_1)$. The last two columns agree for $x>0$, and differ by a residue for $x<0$.}
    \label{table:Borel-plane-residues}
\end{table}

\noindent For $x<0$, we compute the KMS violation by using that the ${}_2F_1$ can be written as a terminating polynomial:
\begin{align}
    &\mathcal{S}G_1^{\mathrm{pert.}}(e^{-2\pi i}U_1) - \mathcal{S}G_2^{\mathrm{pert.}}(U_1) \nonumber\\
    &\hspace{5mm}= \theta(-x) 2\pi i \mathop{\mathrm{Res}}_{s\to s_2^{\ast}} e^{-s} \mathcal{B} G_2^{\mathrm{pert.}}(U_{13}, V_{24}; s x) \\
    &\hspace{5mm}=\frac{2\pi i\, \theta(-x)(-1)^{h_A+h_B}}{(2\pi)^2 (U_{13}+i\epsilon)^{2h_A} (V_{24}-i\epsilon)^{2h_B}} \sum_{k=0}^{2h_--1} \hspace{-1mm} \frac{(1-2h_A)_k (1-2h_B)_k}{\Gamma(1+k)^2} \mathop{\mathrm{Res}}_{s\to s_2^{\ast}} \Big[ e^{-s} \big(\tfrac{s}{s_2^{\ast}}\big)^k \big(1-\tfrac{s}{s_2^{\ast}}\big)^{1-2h_A-2h_B} \Big] \nonumber \\
    &\hspace{5mm}= \theta(-x) \frac{i (-1)^{h_-}}{2\pi} \times
    \begin{dcases}
        (U_{13}+i\epsilon)^{2h_+-2h_-}, & h_B\geq h_A \\
        (V_{24}-i\epsilon)^{2h_+-2h_-}, & h_A\geq h_B
    \end{dcases} \\
    &\hspace{10mm} \times e^{-\frac{i}{x}U_{13}V_{24}} x^{-2h_+} \sum_{k=0}^{2h_--1} \frac{\big(-\frac{i}{x}U_{13}V_{24}\,\big)^k}{\Gamma(1+k)\Gamma(2h_--k)\Gamma(1+2h_+-2h_-+k)} \,. \nonumber
\end{align}
This is precisely $\Delta_{\mathrm{KMS}}$ (\ref{eq:2D-DeltaKMS-UV-explicit-2}). We conclude:
\begin{align}
    \Delta_{\mathrm{KMS}}(U_1, V_2, U_3, V_4) &= \theta(-x) \int_{C_{s_2^{\ast}}} \mathrm{d}s\, e^{-s} \mathcal{B} G_2^{\mathrm{pert.}}(s x) \,, \label{eq:2D-DeltaKMS-Borel}
\end{align}
where the contour circles $s_2^{\ast}$ counterclockwise. We say that the contact term in the Tomita-Takesaki derivation of $\Delta_{\mathrm{KMS}}$ is encoded in the analytic structure of the Borel transformed perturbative correlator. This scales as $1/x$, and is therefore non-perturbative. It would be interesting to see if an analog of (\ref{eq:2D-DeltaKMS-Borel}) continues to hold for more general theories, where the analytic structure of $\mathcal{B} G_2^{\mathrm{pert.}}$ may be more complex.


\subsubsection{A Detail in Proposition \ref{prop:G1-smeared-xto0}} \label{sec:G1xto0-detail}
In this section we address a subtlety in applying Lemma \ref{lemma:Fxlambda-integral} to (\ref{eq:2D-G1-lambda-smeared-step}), which is used to extract the perturbative expansion of the smeared correlator $G_1$. For the contour pulling trick to work, the integral on the vertical line segments $C^{\pm}_N(\Lambda) := \pm \Lambda + i[0,N]$ must vanish when we take $\Lambda\to\infty$ for any $n\in\mathbb{Z}_+$. That is,
\begin{align}
    \lim\limits_{\Lambda\to 0}I_{\pm,N}(\Lambda) &:= \lim\limits_{\Lambda\to\infty} \int_{C^{\pm}_N(\Lambda)} \dbar\lambda \, F_x(\lambda) h(\lambda) = 0\,, \\
    h(\lambda) &:= e^{\lambda(-\pi+\epsilon)} \Gamma(2h_A-i\lambda) \Gamma(2h_B-i\lambda) \,, \label{eq:Fxlambda-edges-h} \\
    &\hspace{5mm}\times \int_{-\infty}^{\infty} \mathrm{d}U_-\, g_1(U_-) (U_--i\epsilon)^{-2h_A+i\lambda} \int_{-\infty}^{\infty} \mathrm{d}V_- \, g_2(V_-)  (V_- -i\epsilon)^{-2h_B+i\lambda} \,. \nonumber
\end{align}
We will show this is satisfied for a generic class of test functions. Recall that we require test functions $f_i(t_i)$ in the time domain to be of class $C_c^{\infty}(\mathbb{R})$, smooth with bounded support. It follows that the functions $g_1(U_-)$ and $g_2(V_-)$ defined in (\ref{eq:2D-g1g2-testfns}) are also of class $C^{\infty}_c(\mathbb{R})$, as the $U_+,V_+$ integrands are supported away from $0$. An example of such functions is given by bump functions of the form:
\begin{align}
    B_R(z) := e^{-\frac{(z/R)^2}{1-(z/R)^2}} \mathds{1}(|z|<R) = e^{-\frac{(z/R)^2}{1-(z/R)^2}} \times 
    \begin{dcases}
        1, & |z| < R \\
        0, & |x| \geq R
    \end{dcases} \,.\label{eq:bump-fn}
\end{align}
We take $g_1(U_-) = B_R(U_-)$ and $g_2(V_-) = B_R(V_-)$. The function $h(\lambda)$ can be explicitly evaluated by making the change of variables $y=(z/R)^2(1-(z/R)^2)^{-1}$ and using (\ref{eq:hypergeometricU-integral-identity}) for each the $z>0$ and $z<0$ regions:
\begin{align}
    \int_{-\infty}^{\infty} \mathrm{d}z\, B_R(z) (z-i\epsilon)^{-2h+i\lambda} &= (-1)^h \pi R^{1-2h+i\lambda} \frac{e^{\frac{\pi\lambda}{2}} }{\Gamma(\frac{1}{2}+h-\frac{i\lambda}{2})} U\big(\tfrac{1}{2}-h_A+\tfrac{i\lambda}{2},0; 1\big) \,,
\end{align}
Substituting this into (\ref{eq:Fxlambda-edges-h}) yields:
\begin{align}
    h(\lambda) &= c \, e^{i\lambda \log((R/2)^2)} \Gamma(h_A-\tfrac{i\lambda}{2}) \Gamma(h_B-\tfrac{i\lambda}{2}) U\big(\tfrac{1}{2}-h_A+\tfrac{i\lambda}{2},0; 1) \, U\big(\tfrac{1}{2}-h_B+\tfrac{i\lambda}{2},0; 1) \,,
\end{align}
where $c$ is a constant. Finally, we need the asymptotics of $h(\lambda)$ on $C^{\pm}_N(\Lambda)$. For any fixed $n$, we may take $\Lambda$ sufficiently large so that the first argument of each hypergeometric function is arbitrarily close to $\pm\pi/2$. Using the $\Gamma$-vertical line bound (\ref{eq:Gamma-vertical-bound}) with with \cite{DLMF} (13.8.11) yields:
\begin{align}
    \lim\limits_{\Lambda\to\infty} |h(\lambda)| \leq c_N' \Lambda^{M_N} e^{-2\Lambda^{1/2}}, \qquad \forall\, \lambda \in C^{\pm}_N(\Lambda) \,,
\end{align}
where $c_N'$ and $M_N$ are $N$-dependent constants. The exponential decay in $\Lambda^{1/2}$ overpowers $F_x(\lambda)$, which for $x<1$ scales as $\mathcal{O}(\Lambda^{-1/2})$ on $C^{\pm}_N(\Lambda)$. When $g_1$ and $g_2$ are bump functions (\ref{eq:bump-fn}), we thus have:
\begin{align}
    \lim\limits_{\Lambda\to 0}I_{\pm,N}(\Lambda) &:= \lim\limits_{\Lambda\to\infty} \int_{C^{\pm}_N(\Lambda)} \dbar\lambda \, F_x(\lambda) h(\lambda) = 0 \,.
\end{align}
If $g_1(U_-)$ and $g_2(V_-)$ are bump functions centered away from zero, we cannot derive explicit formulas for $h(\lambda)$, but we expect the same analysis to hold--the only worrysome region is near $U_-=0$ (similarly $V_-=0$), since this is where the $(U_--i\epsilon)^{-2h_A+i\lambda}$ factor of (\ref{eq:Fxlambda-edges-h}) is not smooth. Since any $C_c^{\infty}$ function may be well approximated by a finite sum of shifted bump functions, we conclude that we may apply Lemma \ref{lemma:Fxlambda-integral} to (\ref{eq:2D-G1-lambda-smeared-step}) for generic test functions $f_i(t_i)$.


\subsection{Commutant Calculations} \label{sec:commutator-calcs}
In this appendix we provide further support for our conjecture that the commutant algebra $\mathcal{A}'$ is trivial by comparing matrix elements of $\mathcal{K}_{\mathrm{I}, O}^{\lambda_1\lambda_2;ij}$ and $\mathcal{K}_{\mathrm{II}, O}^{\lambda_1\lambda_2;ij}$ in the ket $|\lambda_A,\lambda_B\rangle_{kl}$. When the computation is tractable, we will see that the matrix elements agree for $x>0$ but disagree for $x<0$. This will lead to the result of Proposition \ref{prop:commutant-3field-constraint}. We begin with:
\begin{align}
    &{}_{k'}\langle -\lambda_A', \Omega_B| \mathcal{K}_{\mathrm{I}, O}^{\lambda_1\lambda_2;ij} |\lambda_A, \lambda_B\rangle_{kl} \nonumber \\
    &\hspace{2mm} = \int \dbar\lambda \, F_{-\frac{x}{2}}(\lambda) \times {}_{k'}\langle -\lambda_A', \Omega_B| e^{\frac{i}{2}x\hat{P}_A\hat{P}_B} \varphi_i^A(\lambda_1+\lambda) e^{-ix\hat{P}_A\hat{P}_B} \varphi_j^B(\lambda_2+\lambda) e^{\frac{i}{2}x\hat{P}_A\hat{P}_B} |\lambda_A, \lambda_B\rangle_{kl} \\
    &\hspace{2mm}= \int \dbar\lambda \dbar\lambda'\, F_{-\frac{x}{2}}(\lambda) F_{\frac{x}{2}}(\lambda') \times {}_{k'}\langle-\lambda_A'|\varphi_i^A(\lambda_1+\lambda)|\lambda_A+\lambda'\rangle_k \langle\Omega_B |\varphi_j^B(\lambda_2+\lambda)|\lambda_B+\lambda'\rangle_l \\
    &\hspace{2mm}= 2\pi\delta(\lambda_1+\lambda_A+\lambda_A'-\lambda_2-\lambda_B) e^{i\pi(h_j+i\lambda_2)} \delta_{jl} \nonumber \\
    &\hspace{2mm}\hspace{5mm} \int\dbar\lambda\, e^{-\pi\lambda}F_{-\frac{x}{2}}(\lambda) F_{\frac{x}{2}}(-\lambda-\lambda_2-\lambda_B) \times {}_{k'}\langle -\lambda_A'| \varphi_i^A(\lambda+\lambda_1) | -\lambda-\lambda_1-\lambda_A' \rangle_k \,, \label{eq:KOI-matrixel-AAB} \\
    &\, \nonumber \\
    &{}_{k'}\langle -\lambda_A', \Omega_B| \mathcal{K}_{\mathrm{II}, O}^{\lambda_1\lambda_2;ij} |\lambda_A, \lambda_B\rangle_{kl} \nonumber \\
    &\hspace{2mm}= \int \dbar\lambda \, F_{\frac{x}{2}}(\lambda) \times {}_{k'}\langle -\lambda_A', \Omega_B| e^{-\frac{i}{2}x\hat{P}_A\hat{P}_B} \varphi_j^B(\lambda_2+\lambda) e^{ix\hat{P}_A\hat{P}_B} \varphi_i^A(\lambda_1+\lambda) e^{-\frac{i}{2}x\hat{P}_A\hat{P}_B} |\lambda_A, \lambda_B\rangle_{kl} \\
    &\hspace{2mm}=\int \dbar\lambda \dbar\lambda'\dbar\lambda''\, F_{\frac{x}{2}}(\lambda) F_{-\frac{x}{2}}(\lambda') F_{x}(\lambda'') \nonumber \\
    &\hspace{2mm}\hspace{10mm}\times {}_{k'}\langle-\lambda_A'-\lambda''|\varphi_i^A(\lambda_1+\lambda)|\lambda_A+\lambda'\rangle_k \langle\Omega_B|\varphi_j^B(\lambda_2+\lambda)|\lambda'+\lambda''+\lambda_B\rangle_l \\
    &\hspace{2mm}= 2\pi\delta(\lambda_1+\lambda_A+\lambda_A'-\lambda_2-\lambda_B) e^{i\pi(h_j+i\lambda_2)} \delta_{jl} \nonumber \\
    &\hspace{2mm}\hspace{5mm} \int\dbar\lambda\dbar\lambda\, e^{-\pi\lambda}F_{\frac{x}{2}}(\lambda) F_{-\frac{x}{2}}(\lambda') F_{x}(-\lambda-\lambda'-\lambda_1-\lambda_A-\lambda_A') \nonumber \\
    &\hspace{2mm}\hspace{10mm} \times {}_{k'}\langle -\lambda-\lambda'-\lambda_1-\lambda_A | \varphi_i^A(\lambda+\lambda_1) | \lambda'+\lambda_A \rangle_k \,. \label{eq:KOII-matrixel-AAB}
\end{align}
The first set of equalities use (\ref{eq:alg-KOI-def}) and (\ref{eq:alg-KOII-def}). Some instances of the S-matrix can be moved to the far left, where they are annihilated by $\langle\Omega_B|$. The non-trivial S-matrices are expressed in terms of $F_x(\lambda)$'s using the identity (\ref{eq:2D-PAPB-lambda-basis}). Both matrix elements now take the form of an integral involving a two-point function and a three-point function. In the final lines we evaluate the two-point function, which contribute overall $\delta$-functions.

Next we need the 3-point functions in frequency space. For our chiral CFTs, these are given by a combination of ${}_3F_2$'s:
\begin{align}
     {}_i\langle -\lambda_1| \varphi^A_j(\lambda_2) | \lambda_3\rangle_k &= \sqrt{(2\pi)^3\Gamma(2h_i)\Gamma(2h_j)\Gamma(2h_k)} \, e^{i\pi(h_i+i\lambda_1)} e^{-\frac{i\pi}{2}(h_i+h_j+h_k)} C_{ijk} \mathcal{F}(-\lambda_1, \lambda_2, \lambda_3) \,, \nonumber \\
     {}_i\langle-\lambda_1| \varphi^B_j(\lambda_2) | \lambda_3\rangle_k &= \sqrt{(2\pi)^3\Gamma(2h_i)\Gamma(2h_j)\Gamma(2h_k)} \, e^{-i\pi(h_i+i\lambda_1)} e^{\frac{i\pi}{2}(h_i+h_j+h_k)} C_{ijk} \mathcal{F}(-\lambda_1, \lambda_2, \lambda_3) \,, \label{eq:threepoint-lambda}\\
     \mathcal{F}(-\lambda_1, \lambda_2, \lambda_3) &:= \tfrac{\Gamma(1+\Delta_i)\Gamma(1+\Delta_j)}{\Gamma(h_k+i\lambda_3)}
     \Big[
        \tfrac{\Gamma(h_i-i\lambda_1)}{\Gamma(h_i+i\lambda_1)} {}_3\tilde{F}_2\Big(
        \begin{smallmatrix}
            h_i-i\lambda_1, \ h_j+i\lambda_2, \ -\Delta_k \\ 1+h_j-h_k-i\lambda_1, \ 1+h_i-h_k+i\lambda_2
        \end{smallmatrix};1\Big) \nonumber \\
        &\hspace{32mm}
        + \tfrac{\Gamma(h_j-i\lambda_2)}{\Gamma(h_j+i\lambda_2)} {}_3\tilde{F}_2\Big(
        \begin{smallmatrix}
            h_i+i\lambda_1, \ h_j-i\lambda_2, \ -\Delta_k \\ 1+h_j-h_k+i\lambda_1, \ 1+h_i-h_k-i\lambda_2
        \end{smallmatrix};1\Big) (-1)^{h_i+h_j+h_k}
    \Big] \nonumber \\
    &\hspace{5mm}+ \tfrac{\Gamma(1+\Delta_i)\Gamma(1+\Delta_k)}{\Gamma(h_k+i\lambda_2)}
     \Big[
        \tfrac{\Gamma(h_k-i\lambda_3)}{\Gamma(h_k+i\lambda_3)} {}_3\tilde{F}_2\Big(
        \begin{smallmatrix}
            h_i+i\lambda_1, \ h_k-i\lambda_3, \ -\Delta_j \\ 1-h_j+h_k+i\lambda_1, \ 1+h_i-h_j-i\lambda_3
        \end{smallmatrix};1\Big) \nonumber \\
        &\hspace{32mm}
        + \tfrac{\Gamma(h_i-i\lambda_1)}{\Gamma(h_i+i\lambda_1)} {}_3\tilde{F}_2\Big(
        \begin{smallmatrix}
            h_i-i\lambda_1, \ h_k+i\lambda_3, \ -\Delta_j \\ 1-h_j+h_k-i\lambda_1, \ 1+h_i-h_j+i\lambda_3
        \end{smallmatrix};1\Big) (-1)^{h_i+h_j+h_k} 
    \Big] \nonumber \\
    &\hspace{5mm}+ \tfrac{\Gamma(1+\Delta_j)\Gamma(1+\Delta_k)}{\Gamma(h_i+i\lambda_1)}
     \Big[
        \tfrac{\Gamma(h_j-i\lambda_2)}{\Gamma(h_j+i\lambda_2)} {}_3\tilde{F}_2\Big(
        \begin{smallmatrix}
            h_j-i\lambda_2, \ h_k+i\lambda_3, \ -\Delta_i \\ 1-h_i+h_k-i\lambda_2, \ 1-h_i+h_j+i\lambda_3
        \end{smallmatrix};1\Big) \nonumber \\
        &\hspace{32mm}
        + \tfrac{\Gamma(h_k-i\lambda_3)}{\Gamma(h_k+i\lambda_3)} {}_3\tilde{F}_2\Big(
        \begin{smallmatrix}
            h_j+i\lambda_2, \ h_k-i\lambda_3, \ -\Delta_i \\ 1-h_i+h_k+i\lambda_2, \ 1-h_i+h_j-i\lambda_3
        \end{smallmatrix};1\Big) (-1)^{h_i+h_j+h_k}
    \Big] \,, \label{eq:threepoint-lambda-3F2}
\end{align}
where $\Delta_i = h_i-h_j-h_k$, $\Delta_j= h_j-h_i-h_k$, and $\Delta_k=h_k-h_i-h_j$. This can be derived from the position-space three-point function (\ref{eq:threepoint}) using the representations (\ref{eq:2D-varphiA-lambdatoU}) and (\ref{eq:2D-varphiB-lambdatoU}). The six ${}_3F_2$'s come from the integral over the three positions $(U_1, U_2, U_3)$, where each sector $U_i < U_j < U_k$ is treated separately. For general conformal weights the expression (\ref{eq:threepoint-lambda-3F2}) does not simplify nicely, and substituting (\ref{eq:threepoint-lambda}) into (\ref{eq:KOI-matrixel-AAB})--(\ref{eq:KOII-matrixel-AAB}) produces a sum of multivariate Mellin-Barnes integrals which are difficult to evaluate. However, the three-point function simplifies when two conformal weights add to the third. For the $A$-theory, we have:
\begin{align}
    {}_i\langle -\lambda_1| \varphi^A_j(\lambda_2) | \lambda_3\rangle_k &= C_{ijk} \, 2\pi \delta(\lambda_1+\lambda_2+\lambda_3) \\
    &\hspace{5mm}\times \begin{dcases}
        \sqrt{\frac{(2\pi)^3 \Gamma(2h_i)}{\Gamma(2h_j) \Gamma(2h_k)}} \frac{\Gamma(h_j-i\lambda_2) \Gamma(h_k-i\lambda_3)}{\Gamma(h_i+i\lambda_1)}, & h_i = h_j+h_k \\
        \sqrt{\frac{(2\pi)^3 \Gamma(2h_j)}{\Gamma(2h_i) \Gamma(2h_k)}} (-1)^{h_k} e^{-\pi\lambda_3} \frac{\Gamma(h_i-i\lambda_1) \Gamma(h_k-i\lambda_3)}{\Gamma(h_j+i\lambda_2)}, & h_j = h_i+h_k \\
        \sqrt{\frac{(2\pi)^3 \Gamma(2h_k)}{\Gamma(2h_i) \Gamma(2h_j)}} (-1)^{h_j} e^{\pi\lambda_2} \frac{\Gamma(h_i-i\lambda_1) \Gamma(h_j-i\lambda_2)}{\Gamma(h_k+i\lambda_3)}, & h_k= h_i+h_j
    \end{dcases} \,. \nonumber
\end{align}
We now return to (\ref{eq:KOI-matrixel-AAB}) and (\ref{eq:KOII-matrixel-AAB}), and suppose that $h_{k'} = h_i + h_k$. The first matrix element becomes:
\begin{align}
    {}_{k'}\langle -\lambda_A', \Omega_B| \mathcal{K}_{\mathrm{I}, O}^{\lambda_1\lambda_2;ij} |\lambda_A, \lambda_B\rangle_{kl} &= 2\pi\delta(\lambda_1+\lambda_A+\lambda_A'-\lambda_2-\lambda_B) \, e^{i\pi(h_j+i\lambda_2)} \delta_{jl}\\
    &\hspace{5mm} \Big|\frac{x}{2}\Big|^{i(\lambda_2+\lambda_B)} e^{\frac{\pi}{2}(\lambda_2+\lambda_B)\,\mathrm{sgn}(x)} \sqrt{\frac{(2\pi)^3\Gamma(2h_{k'})}{\Gamma(2h_i)\Gamma(2h_k)}} C_{k'ik}\, \mathcal{I}^{\mathrm{I}}_{k'ik}(x) \,, \nonumber \\
    \mathcal{I}^{\mathrm{I}}_{k'ik}(x) &:= \frac{1}{\Gamma(h_{k'}+i\lambda_A')} \int\dbar\lambda\, e^{-2\pi\lambda\, \theta(-x)} \Gamma(i\lambda+\epsilon) \Gamma(i\lambda + h_k + i(\lambda_1+\lambda_A')) \nonumber \\
    &\hspace{10mm} \times \Gamma(-i\lambda+h_i-i\lambda_1) \Gamma(-i\lambda+\epsilon -i(\lambda_1+\lambda_A+\lambda_A')) \,.
\end{align}
The integral $\mathcal{I}^{\mathrm{I}}$ is always convergent. When $x>0$, it can be evaluated using the ${}_2F_1$ Mellin-Barnes representation \cite{DLMF} (15.6.7), and the ${}_2F_1$ simplifies to $\Gamma$-functions.
\begin{align}
    {}_2\tilde{F}_1(a,b,c;z) &= \frac{1}{\Gamma(a)\Gamma(b)\Gamma(c-a)\Gamma(c-b)} \int_{i\mathbb{R}} \frac{\mathrm{d}s}{2\pi i} \Gamma(-s) \Gamma(a+s)\Gamma(b+s)\Gamma(c-a-b-s) (1-z)^s \,; \nonumber \\
    &\hspace{5mm} |\mathrm{arg}(1-z)|<\pi, \quad a,b,c-a,c-b \not\in \mathbb{Z}_{\leq 0} \,.
\end{align}
The $x<0$ case is related by analytic continuation. We define $1-z=e^{-\tau}$ so $(1-z)^s = e^{i\tau\lambda}$, and continue the $x>0$ result as $\tau\to \tau+2\pi i$. The ${}_2F_1$ winds around $z=1$ once clockwise in a circle of radius $>1$, and the result is given by monodromy:
\begin{align}
    &{}_2F_1(a,b,c;e^{-2\pi i}z) - {}_2F_1(a,b,c;z) \nonumber \\
    &\hspace{5mm} = (e^{-2\pi i(c-a-b)} - 1) \tfrac{\Gamma(c)\Gamma(a+b-c)}{\Gamma(a)\Gamma(b)} (1-z)^{c-a-b} {}_2F_1(c-a,c-b,c-a-b+1;1-z)\,.
\end{align}
Altogether, we find:
\begin{align}
    \mathcal{I}^{\mathrm{I}}_{k'ik}(x) &:= \frac{\Gamma(h_i-i\lambda_1)\Gamma(h_k-i\lambda_A)\Gamma(-i(\lambda_1+\lambda_A+\lambda_A'))}{\Gamma(h_{k'}-i(\lambda_1+\lambda_A))} \nonumber \\
    &\hspace{5mm} \times\Big[1 + \theta(-x) 2e^{\pi(\lambda_1+\lambda_A')} \frac{\sinh(\pi\lambda_1) \sinh(\pi(\lambda_1+\lambda_A+\lambda_A'))}{\sinh(\pi(\lambda_1+\lambda_A))} \Big] \,.
\end{align}
Next we turn to the second matrix element, which involves two integrals:
\begin{align}
    {}_{k'}\langle -\lambda_A', \Omega_B| \mathcal{K}_{\mathrm{II}, O}^{\lambda_1\lambda_2;ij} |\lambda_A, \lambda_B\rangle_{kl} &= 2\pi\delta(\lambda_1+\lambda_A+\lambda_A'-\lambda_2-\lambda_B) \, e^{i\pi(h_j+i\lambda_2)} \delta_{jl}\\
    &\hspace{5mm} |\tfrac{x}{2}|^{i(\lambda_2+\lambda_B)} e^{\frac{\pi}{2}(\lambda_2+\lambda_B)\mathrm{sgn}(x)} \sqrt{\frac{(2\pi)^3\Gamma(2h_{k'})}{\Gamma(2h_i)\Gamma(2h_k)}} C_{k'ik}\, \mathcal{I}^{\mathrm{II}}_{k'ik}(x) \,, \nonumber
\end{align}
\begin{align}
    \mathcal{I}^{\mathrm{II}}_{k'ik}(x) &:= \int \dbar\lambda\dbar\lambda' \, 2^{i(\lambda_1+\lambda_A+\lambda_A'+\lambda+\lambda')} e^{-\pi\lambda} e^{\pi\mathrm{sgn}(x)\,\lambda'} \Gamma(\epsilon+i\lambda) \Gamma(\epsilon+i\lambda') \\
    &\hspace{10mm} \times \frac{\Gamma(h_i-i(\lambda+\lambda_1)) \Gamma(h_k-i(\lambda'+\lambda_A)) \Gamma(\epsilon-i(\lambda+\lambda'+\lambda_1+\lambda_A+\lambda_A'))}{\Gamma(h_{k'}-i(\lambda+\lambda'+\lambda_1+\lambda_A))} \nonumber \\
    &= \frac{\Gamma(h_i-i\lambda_1) \Gamma(h_k-i\lambda_A) \Gamma(-i(\lambda_1+\lambda_A+\lambda_A'))}{\Gamma(h_{k'}-i(\lambda_1+\lambda_A))} \label{eq:KOII-matrixel-kpikcase-step} \\
    &\hspace{5mm} \times 2^{i(\lambda_1+\lambda_A+\lambda_A')} F_1(\epsilon-i(\lambda_1+\lambda_A+\lambda_A'), h_i-i\lambda_1, h_k-i\lambda_A, h_{k'}-i(\lambda_1-\lambda_A); \tfrac{1}{2}, \tfrac{1}{2}) \,, \nonumber
\end{align}
where in the final equality we use the Mellin-Barnes representation for the Appell $F_1$ function \cite{Exton:1978} (5.2.4.26), valid for both signs of $x$:
\begin{align}
    F_1(a,b,b',c;z_1,z_2) &= \frac{\Gamma(c)}{\Gamma(a)\Gamma(b)\Gamma(b')} \hspace{-1mm} \int_{i\mathbb{R}\times i\mathbb{R}} \frac{\mathrm{d}s\mathrm{d}t}{(2\pi i)^2} \frac{\Gamma(a+s+t)\Gamma(b+s)\Gamma(b'+t)\Gamma(-s)\Gamma(-t)}{\Gamma(c+s+t)} \nonumber \\
    &\hspace{42mm} \times (-z_1)^s (-z_2)^t \,; \qquad  |\arg(-z_1)|, |\arg(-z_2)| < \pi\,.
\end{align}
The entire second line of (\ref{eq:KOII-matrixel-kpikcase-step}) evaluates to 1, and we find that $\mathcal{I}^{\mathrm{II}}_{k'ik}(x)$ is the same as $\mathcal{I}^{\mathrm{I}}_{k'ik}(x)$ without the monodromy term. Putting everything together, we have:
\begin{align}
    &{}_{k'}\langle -\lambda_A', \Omega_B| \Delta\mathcal{K}_O^{\lambda_1\lambda_2;ij}|\lambda_A, \lambda_B\rangle_{kl} \big|_{h_{k'} = h_i + h_k} \label{eq:DeltaKO-AAB-matrixel} \\
    &\hspace{5mm} = 2\pi i \, \theta(-x)\,  2\pi \delta(\lambda_1+\lambda_A+\lambda_A'-\lambda_2-\lambda_B) \Big[ \delta_{jl} (-1)^{h_j} e^{-\pi\lambda_2} 2\sinh(\pi(\lambda_2+\lambda_B)) F_{-\frac{x}{2}}(-\lambda_2-\lambda_B) \Big] \nonumber \\
    &\hspace{10mm} \times \Big[ \sqrt{\frac{(2\pi)^3 \Gamma(2h_{k'})}{\Gamma(2h_i)\Gamma(2h_k)}} C_{k'ik} \frac{e^{-\pi\lambda_A}}{2\sinh\pi\lambda_A} \frac{\Gamma(1-h_{k'}+i(\lambda_1+\lambda_A))}{\Gamma(1-h_i+i\lambda_1)\Gamma(1-h_k+i\lambda_A)}\Big] \,. \nonumber
\end{align}
We see matrix elements $\mathcal{K}^{\lambda_1,\lambda_2;ij}_{\mathrm{I},O}$ and $\mathcal{K}^{\lambda_1,\lambda_2;ij}_{\mathrm{II},O}$ are equal for $x>0$, but not for $x<0$. Note that from the point of view of the integrands this highly non-trivial, and the equality for $x>0$ is rather remarkable. For $x<0$ the matrix element (\ref{eq:DeltaKO-AAB-matrixel}) factorizes into pieces depending only on the $B$-field frequencies, the $A$-field frequencies, and an overall $\delta$-function coupling them. In particular, the $B$-dependent terms are identical to those in (\ref{eq:DeltaKO-AB-matrixel}). Again the only $x$-dependence is contained in $\theta(-x) F_{-x/2}(-\lambda_2-\lambda_B)$.

\vspace{1em}

\noindent This calculation (and its analog ${}_{l'}\langle \Omega_A, -\lambda_B'| \Delta\mathcal{K}_O^{\lambda_1\lambda_2;ij}|\lambda_A, \lambda_B\rangle_{kl}$) allows us to generalize Proposition \ref{prop:commutant-AorB-constraint} to the following.
\commutantthreefieldconstraint*
\begin{proof}
    These are all similar, so let us consider the first. Suppose that:
    \begin{align}
        O|\Omega\rangle &= \sum_{k} \int\dbar\lambda_A \, c^A_{\lambda_A,k} |\lambda_A,\Omega_B\rangle_k + \sum_l \int\dbar\lambda_B \, \Big[ c^B_{\lambda_B,l} |\Omega_A, \lambda_B\rangle_l + \, c^{AB}_{\lambda_A^{\ast},\lambda_B;k^{\ast},l} |\lambda_A^{\ast},\lambda_B\rangle_{k^{\ast},l} \Big] \,. \nonumber
    \end{align}
    We now take $h_{k'} = h_i + h_k$, for a primary such that the CFT structure constant $C_{k'ik}\neq 0$. This is guaranteed to exist, since it must show up in the OPE of $\phi_i^A \phi_k^A$. We compute:
    \begin{align}
        {}_{k'}\langle -\lambda_A', \Omega_B| \Delta\mathcal{K}_O^{\lambda_1\lambda_2;ij} O |\Omega\rangle &=  0 + \sum_l \int\dbar\lambda_B\, \Big[ c^B_{\lambda_B,l} \times {}_{k'}\langle -\lambda_A', \Omega_B| \Delta\mathcal{K}_O^{\lambda_1\lambda_2;ij} |\Omega_A,\lambda_B\rangle_l \\
        &\hspace{35mm} + c^{AB}_{\lambda_A^{\ast}, \lambda_B; k^{\ast},l} \times {}_{k'}\langle -\lambda_A', \Omega_B| \Delta\mathcal{K}_O^{\lambda_1\lambda_2;ij} |\lambda_A^{\ast},\lambda_B\rangle_{k^{\ast},l} \Big] \nonumber 
    \end{align}
    It can then be shown that $c^{AB}_{\lambda_A^{\ast},\lambda_B;k^{\ast},l} \equiv 0$ for all $\lambda_B,l$ by substituting (\ref{eq:DeltaKO-AB-matrixel}) and (\ref{eq:DeltaKO-AAB-matrixel}), using the $\delta$-functions to remove the sum and integral, and using that the result holds for any $i,j,\lambda_1,\lambda_2,\lambda_A'$. The result follows now directly from Proposition \ref{prop:commutant-AorB-constraint}. 
\end{proof}

Alternatively we could have studied the spectral support of $O|\Omega\rangle$ in a basis diagonalizing the $P_{A,B}$ operators, again computing matrix elements of $\Delta \mathcal{K}_O^{\lambda_1\lambda_2;ij}$ in this basis. This circumvents the need to use the three-point function, but the problem is that the $\lambda$ integrals from $\mathcal{K}_{\mathrm{I},O}^{\lambda_1\lambda_2;ij}$ and $\mathcal{K}_{\mathrm{II},O}^{\lambda_1\lambda_2;ij}$ do not always converge. The integrals can always be evaluated by analytic continuation from a region where they are convergent, but this leads to no constraints--even for $x<0$ we always have $\Delta\mathcal{K}_O^{\lambda_1\lambda_2;ij} = 0$ at the level of matrix-elements.

Let us therefore return to the $\lambda$-basis. Generally, the matrix elements of $\Delta\mathcal{K}_O^{\lambda_1\lambda_2;ij}|\lambda_A, \lambda_B\rangle_{kl}$ with the bras ${}_{k'}\langle -\lambda_A', \Omega_B|$, ${}_{l'}\langle \Omega_A,-\lambda_B'|$, and ${}_{k'l'}\langle -\lambda_A', -\lambda_B'|$ can be expressed in terms of multivariate Mellin-Barnes integrals, but their evaluation is difficult. Nevertheless, the expressions (\ref{eq:DeltaKO-AB-matrixel}), (\ref{eq:DeltaKO-BA-matrixel}), and (\ref{eq:DeltaKO-AAB-matrixel}), suggest that for $x>0$ these all vanish, while for $x<0$ they are non-trivial. Together, we believe that these constraints should rule out any non-vacuum states in the spectral support of $O|\Omega\rangle$, implying that $\mathcal{A}'$ is trivial.


\newpage

\bibliographystyle{JHEP}
\nocite{}
\bibliography{thebibliography}

@article{CLPW,
    author = "Chandrasekaran, Venkatesa and Longo, Roberto and Penington, Geoff and Witten, Edward",
    title = "{An algebra of observables for de Sitter space}",
    eprint = "2206.10780",
    archivePrefix = "arXiv",
    primaryClass = "hep-th",
    doi = "10.1007/JHEP02(2023)082",
    journal = "JHEP",
    volume = "02",
    pages = "082",
    year = "2023"
}

@article{Penington:2025hrc,
    author = "Penington, Geoff and Tabor, Elisa",
    title = "{The algebraic structure of gravitational scrambling}",
    eprint = "2508.21062",
    archivePrefix = "arXiv",
    primaryClass = "hep-th",
    month = "8",
    year = "2025"
}

@article{Mack1977convergence,
  author  = {Mack, G.},
  title   = {Convergence of Operator Product Expansions on the Vacuum in Conformal Invariant Quantum Field Theory},
  journal = {Communications in Mathematical Physics},
  volume  = {53},
  number  = {2},
  pages   = {155--184},
  year    = {1977},
  doi     = {10.1007/BF01609130},
}

@article{Witten:2021unn,
    author = "Witten, Edward",
    title = "{Gravity and the crossed product}",
    eprint = "2112.12828",
    archivePrefix = "arXiv",
    primaryClass = "hep-th",
    doi = "10.1007/JHEP10(2022)008",
    journal = "JHEP",
    volume = "10",
    pages = "008",
    year = "2022"
}

@article{Chandrasekaran:2022eqq,
    author = "Chandrasekaran, Venkatesa and Penington, Geoff and Witten, Edward",
    title = "{Large N algebras and generalized entropy}",
    eprint = "2209.10454",
    archivePrefix = "arXiv",
    primaryClass = "hep-th",
    doi = "10.1007/JHEP04(2023)009",
    journal = "JHEP",
    volume = "04",
    pages = "009",
    year = "2023"
}

@article{Jensen:2023yxy,
    author = "Jensen, Kristan and Sorce, Jonathan and Speranza, Antony J.",
    title = "{Generalized entropy for general subregions in quantum gravity}",
    eprint = "2306.01837",
    archivePrefix = "arXiv",
    primaryClass = "hep-th",
    doi = "10.1007/JHEP12(2023)020",
    journal = "JHEP",
    volume = "12",
    pages = "020",
    year = "2023"
}

@article{Kudler-Flam:2023qfl,
    author = "Kudler-Flam, Jonah and Leutheusser, Samuel and Satishchandran, Gautam",
    title = "{Generalized black hole entropy is von Neumann entropy}",
    eprint = "2309.15897",
    archivePrefix = "arXiv",
    primaryClass = "hep-th",
    doi = "10.1103/PhysRevD.111.025013",
    journal = "Phys. Rev. D",
    volume = "111",
    number = "2",
    pages = "025013",
    year = "2025"
}

@article{Kudler-Flam:2024psh,
    author = "Kudler-Flam, Jonah and Leutheusser, Samuel and Satishchandran, Gautam",
    title = "{Algebraic Observational Cosmology}",
    eprint = "2406.01669",
    archivePrefix = "arXiv",
    primaryClass = "hep-th",
    month = "6",
    year = "2024"
}

@article{DeVuyst:2024fxc,
    author = "De Vuyst, Julian and Eccles, Stefan and Hoehn, Philipp A. and Kirklin, Josh",
    title = "{Crossed products and quantum reference frames: on the observer-dependence of gravitational entropy}",
    eprint = "2412.15502",
    archivePrefix = "arXiv",
    primaryClass = "hep-th",
    doi = "10.1007/JHEP07(2025)063",
    journal = "JHEP",
    volume = "07",
    pages = "063",
    year = "2025",
    note = "[Erratum: JHEP 10, 234 (2025)]"
}

@article{DeVuyst:2024khu,
    author = "De Vuyst, Julian and Eccles, Stefan and Hoehn, Philipp A. and Kirklin, Josh",
    title = "{Gravitational entropy is observer-dependent}",
    eprint = "2405.00114",
    archivePrefix = "arXiv",
    primaryClass = "hep-th",
    doi = "10.1007/JHEP07(2025)146",
    journal = "JHEP",
    volume = "07",
    pages = "146",
    year = "2025"
}

@article{Chen:2025tbh,
    author = "Chen, Bin and Xu, Jie",
    title = "{An algebra for covariant observers in de Sitter space}",
    eprint = "2511.00622",
    archivePrefix = "arXiv",
    primaryClass = "hep-th",
    month = "11",
    year = "2025"
}

@article{Strohmaier:2023opz,
    author = "Strohmaier, Alexander and Witten, Edward",
    title = "{The Timelike Tube Theorem in Curved Spacetime}",
    eprint = "2303.16380",
    archivePrefix = "arXiv",
    primaryClass = "hep-th",
    doi = "10.1007/s00220-024-05009-3",
    journal = "Commun. Math. Phys.",
    volume = "405",
    number = "7",
    pages = "153",
    year = "2024"
}

@article{Kudler-Flam:2025pol,
    author = "Kudler-Flam, Jonah and Prabhu, Kartik and Satishchandran, Gautam",
    title = "{Vacua and infrared radiation in de Sitter quantum field theory}",
    eprint = "2503.19957",
    archivePrefix = "arXiv",
    primaryClass = "hep-th",
    month = "3",
    year = "2025"
}

@article{Cordova:2018ygx,
    author = "C{\'o}rdova, Clay and Shao, Shu-Heng",
    title = "{Light-ray Operators and the BMS Algebra}",
    eprint = "1810.05706",
    archivePrefix = "arXiv",
    primaryClass = "hep-th",
    doi = "10.1103/PhysRevD.98.125015",
    journal = "Phys. Rev. D",
    volume = "98",
    number = "12",
    pages = "125015",
    year = "2018"
}

@article{Kolchmeyer:2024fly,
    author = "Kolchmeyer, David K. and Liu, Hong",
    title = "{Chaos and the Emergence of the Cosmological Horizon}",
    eprint = "2411.08090",
    archivePrefix = "arXiv",
    primaryClass = "hep-th",
    reportNumber = "MIT-CTP/5805",
    month = "11",
    year = "2024"
}

@article{Maldacena:2019cbz,
    author = "Maldacena, Juan and Turiaci, Gustavo J. and Yang, Zhenbin",
    title = "{Two dimensional Nearly de Sitter gravity}",
    eprint = "1904.01911",
    archivePrefix = "arXiv",
    primaryClass = "hep-th",
    doi = "10.1007/JHEP01(2021)139",
    journal = "JHEP",
    volume = "01",
    pages = "139",
    year = "2021"
}

@article{Cotler:2019nbi,
    author = "Cotler, Jordan and Jensen, Kristan and Maloney, Alexander",
    title = "{Low-dimensional de Sitter quantum gravity}",
    eprint = "1905.03780",
    archivePrefix = "arXiv",
    primaryClass = "hep-th",
    doi = "10.1007/JHEP06(2020)048",
    journal = "JHEP",
    volume = "06",
    pages = "048",
    year = "2020"
}

@article{Maldacena:2016upp,
    author = "Maldacena, Juan and Stanford, Douglas and Yang, Zhenbin",
    title = "{Conformal symmetry and its breaking in two dimensional Nearly Anti-de-Sitter space}",
    eprint = "1606.01857",
    archivePrefix = "arXiv",
    primaryClass = "hep-th",
    doi = "10.1093/ptep/ptw124",
    journal = "PTEP",
    volume = "2016",
    number = "12",
    pages = "12C104",
    year = "2016"
}

@article{Wall:2011hj,
    author = "Wall, Aron C.",
    title = "{A proof of the generalized second law for rapidly changing fields and arbitrary horizon slices}",
    eprint = "1105.3445",
    archivePrefix = "arXiv",
    primaryClass = "gr-qc",
    doi = "10.1103/PhysRevD.85.104049",
    journal = "Phys. Rev. D",
    volume = "85",
    pages = "104049",
    year = "2012",
    note = "[Erratum: Phys.Rev.D 87, 069904 (2013)]"
}

@article{ReehSchlieder1961,
  author  = {Reeh, H. and Schlieder, S.},
  title   = {Bemerkungen zur Unit\"ar\"aquivalenz von Lorentzinvarianten Feldern},
  journal = {Il Nuovo Cimento},
  volume  = {22},
  number  = {5},
  pages   = {1051--1068},
  year    = {1961},
  doi     = {10.1007/BF02787889}
}

@article{Witten:2018zxz,
    author = "Witten, Edward",
    title = "{APS Medal for Exceptional Achievement in Research: Invited article on entanglement properties of quantum field theory}",
    eprint = "1803.04993",
    archivePrefix = "arXiv",
    primaryClass = "hep-th",
    doi = "10.1103/RevModPhys.90.045003",
    journal = "Rev. Mod. Phys.",
    volume = "90",
    number = "4",
    pages = "045003",
    year = "2018"
}

@article{Sorce:2023,
   title={Notes on the type classification of von Neumann algebras},
   volume={36},
   ISSN={1793-6659},
   url={http://dx.doi.org/10.1142/S0129055X24300024},
   DOI={10.1142/s0129055x24300024},
   number={02},
   journal={Reviews in Mathematical Physics},
   publisher={World Scientific Pub Co Pte Ltd},
   author={Sorce, Jonathan},
   year={2023},
   month=Dec
}

@book{Haag:1992,
  title={Local Quantum Physics: Fields, Particles, Algebras},
  author={Haag, R.},
  isbn={9783540536109},
  lccn={92045286},
  series={R.Balian, W.Beiglbock, H.Grosse},
  url={https://books.google.com/books?id=Op_vAAAAMAAJ},
  year={1992},
  publisher={Springer-Verlag}
}

@book{Takesaki:2001,
  title     = {Theory of Operator Algebras Volumes I--III},
  author    = {Takesaki, M.},
  series    = {Encyclopaedia of Mathematical Sciences},
  publisher = {Springer Berlin Heidelberg},
  address   = {Berlin},
  year      = {2001--2013}
}

@article{Wiesbrock:1993,
    author = {Hans-Werner Wiesbrock},
    title = {{Half-sided modular inclusions of von-Neumann-algebras}},
    volume = {157},
    journal = {Communications in Mathematical Physics},
    number = {1},
    publisher = {Springer},
    pages = {83 -- 92},
    year = {1993},
}

@inbook{Daele:1978,
    place={Cambridge},
    series={London Mathematical Society Lecture Note Series},
    title={Crossed products of von Neumann algebras},
    booktitle={Continuous Crossed Products and Type III Von Neumann Algebras},
    publisher={Cambridge University Press},
    author={Daele, A. van},
    year={1978},
    pages={1–33},
    collection={London Mathematical Society Lecture Note Series}
}

@article{Leutheusser:2021qhd,
    author = "Leutheusser, Samuel and Liu, Hong",
    title = "{Causal connectability between quantum systems and the black hole interior in holographic duality}",
    eprint = "2110.05497",
    archivePrefix = "arXiv",
    primaryClass = "hep-th",
    reportNumber = "MIT-CTP/5335",
    doi = "10.1103/PhysRevD.108.086019",
    journal = "Phys. Rev. D",
    volume = "108",
    number = "8",
    pages = "086019",
    year = "2023"
}

@article{Leutheusser:2021frk,
    author = "Leutheusser, Samuel Aaron Wehlau and Liu, Hong",
    title = "{Emergent Times in Holographic Duality}",
    eprint = "2112.12156",
    archivePrefix = "arXiv",
    primaryClass = "hep-th",
    reportNumber = "MIT-CTP/5382",
    doi = "10.1103/PhysRevD.108.086020",
    journal = "Phys. Rev. D",
    volume = "108",
    number = "8",
    pages = "086020",
    year = "2023"
}

@article{Leutheusser:2022bgi,
    author = "Leutheusser, Sam and Liu, Hong",
    title = "{Subregion-subalgebra duality: Emergence of space and time in holography}",
    eprint = "2212.13266",
    archivePrefix = "arXiv",
    primaryClass = "hep-th",
    doi = "10.1103/PhysRevD.111.066021",
    journal = "Phys. Rev. D",
    volume = "111",
    number = "6",
    pages = "066021",
    year = "2025"
}

@article{Faulkner:2020,
    title = {Recovering the QNEC from the ANEC},
    author = {Ceyhan, Fikret and Faulkner, Thomas},
    doi = {10.1007/s00220-020-03751-y},
    journal = {Communications in Mathematical Physics},
    number = 2,
    volume = 377,
    place = {United States},
    year = {2020},
    month = {5}
}

@article{Anninos:2011af,
    author = "Anninos, Dionysios and Hartnoll, Sean A. and Hofman, Diego M.",
    title = "{Static Patch Solipsism: Conformal Symmetry of the de Sitter Worldline}",
    eprint = "1109.4942",
    archivePrefix = "arXiv",
    primaryClass = "hep-th",
    doi = "10.1088/0264-9381/29/7/075002",
    journal = "Class. Quant. Grav.",
    volume = "29",
    pages = "075002",
    year = "2012"
}

@article{Nakayama:2011qh,
    author = "Nakayama, Ryuichi",
    title = "{The World-Line Quantum Mechanics Model at Finite Temperature which is Dual to the Static Patch Observer in de Sitter Space}",
    eprint = "1112.1267",
    archivePrefix = "arXiv",
    primaryClass = "hep-th",
    reportNumber = "EPHOU-11-009",
    doi = "10.1143/PTP.127.393",
    journal = "Prog. Theor. Phys.",
    volume = "127",
    pages = "393--408",
    year = "2012"
}

@article{Maldacena:2024spf,
    author = "Maldacena, Juan",
    title = "{Real observers solving imaginary problems}",
    eprint = "2412.14014",
    archivePrefix = "arXiv",
    primaryClass = "hep-th",
    month = "12",
    year = "2024"
}

@article{Chen:2025jqm,
    author = "Chen, Yiming and Stanford, Douglas and Tang, Haifeng and Yang, Zhenbin",
    title = "{On the phase of the de Sitter density of states}",
    eprint = "2511.01400",
    archivePrefix = "arXiv",
    primaryClass = "hep-th",
    doi = "10.1007/JHEP05(2026)068",
    journal = "JHEP",
    volume = "05",
    pages = "068",
    year = "2026"
}

@article{Witten:2023xze,
    author = "Witten, Edward",
    title = "{A background-independent algebra in quantum gravity}",
    eprint = "2308.03663",
    archivePrefix = "arXiv",
    primaryClass = "hep-th",
    doi = "10.1007/JHEP03(2024)077",
    journal = "JHEP",
    volume = "03",
    pages = "077",
    year = "2024"
}

@article{Tietto:2025oxn,
    author = "Tietto, Damiano and Verlinde, Herman",
    title = "{A microscopic model of de Sitter spacetime with an observer}",
    eprint = "2502.03869",
    archivePrefix = "arXiv",
    primaryClass = "hep-th",
    month = "2",
    year = "2025"
}

@article{Narovlansky:2023lfz,
    author = "Narovlansky, Vladimir and Verlinde, Herman",
    title = "{Double-scaled SYK and de Sitter holography}",
    eprint = "2310.16994",
    archivePrefix = "arXiv",
    primaryClass = "hep-th",
    doi = "10.1007/JHEP05(2025)032",
    journal = "JHEP",
    volume = "05",
    pages = "032",
    year = "2025"
}

@article{Blommaert:2026ofx,
    author = "Blommaert, Andreas and Tietto, Damiano and Verlinde, Herman",
    title = "{An observer's quantization of 3d de Sitter}",
    eprint = "2606.26241",
    archivePrefix = "arXiv",
    primaryClass = "hep-th",
    month = "6",
    year = "2026"
}

@article{Goto:2026ipq,
    author = "Goto, Kanato and Milekhin, Alexey and Verlinde, Herman and Xu, Jiuci",
    title = "{Generalized Free Fields in de Sitter from 1D CFT}",
    eprint = "2605.03037",
    archivePrefix = "arXiv",
    primaryClass = "hep-th",
    reportNumber = "OU-HET-1309, RIKEN-iTHEMS-Report-26",
    month = "5",
    year = "2026"
}

@article{Narovlansky:2025tpb,
    author = "Narovlansky, Vladimir",
    title = "{Towards a microscopic description of de Sitter dynamics}",
    eprint = "2506.02109",
    archivePrefix = "arXiv",
    primaryClass = "hep-th",
    month = "6",
    year = "2025"
}

@article{Sekino:2008he,
    author = "Sekino, Yasuhiro and Susskind, Leonard",
    title = "{Fast Scramblers}",
    eprint = "0808.2096",
    archivePrefix = "arXiv",
    primaryClass = "hep-th",
    reportNumber = "SU-ITP-08-18, OIQP-08-08, SU-ITP-08/18, OIQP-08-08",
    doi = "10.1088/1126-6708/2008/10/065",
    journal = "JHEP",
    volume = "10",
    pages = "065",
    year = "2008"
}

@article{Lashkari:2011yi,
    author = "Lashkari, Nima and Stanford, Douglas and Hastings, Matthew and Osborne, Tobias and Hayden, Patrick",
    title = "{Towards the Fast Scrambling Conjecture}",
    eprint = "1111.6580",
    archivePrefix = "arXiv",
    primaryClass = "hep-th",
    doi = "10.1007/JHEP04(2013)022",
    journal = "JHEP",
    volume = "04",
    pages = "022",
    year = "2013"
}

@article{Shenker:2013pqa,
    author = "Shenker, Stephen H. and Stanford, Douglas",
    title = "{Black holes and the butterfly effect}",
    eprint = "1306.0622",
    archivePrefix = "arXiv",
    primaryClass = "hep-th",
    reportNumber = "SU-ITP-13-08",
    doi = "10.1007/JHEP03(2014)067",
    journal = "JHEP",
    volume = "03",
    pages = "067",
    year = "2014"
}

@article{Maldacena:2015waa,
    author = "Maldacena, Juan and Shenker, Stephen H. and Stanford, Douglas",
    title = "{A bound on chaos}",
    eprint = "1503.01409",
    archivePrefix = "arXiv",
    primaryClass = "hep-th",
    doi = "10.1007/JHEP08(2016)106",
    journal = "JHEP",
    volume = "08",
    pages = "106",
    year = "2016"
}

@article{Kitaev:2017awl,
    author = "Kitaev, Alexei and Suh, S. Josephine",
    title = "{The soft mode in the Sachdev-Ye-Kitaev model and its gravity dual}",
    eprint = "1711.08467",
    archivePrefix = "arXiv",
    primaryClass = "hep-th",
    doi = "10.1007/JHEP05(2018)183",
    journal = "JHEP",
    volume = "05",
    pages = "183",
    year = "2018"
}

@article{Gu:2021xaj,
    author = "Gu, Yingfei and Kitaev, Alexei and Zhang, Pengfei",
    title = "{A two-way approach to out-of-time-order correlators}",
    eprint = "2111.12007",
    archivePrefix = "arXiv",
    primaryClass = "hep-th",
    doi = "10.1007/JHEP03(2022)133",
    journal = "JHEP",
    volume = "03",
    pages = "133",
    year = "2022"
}

@article{Gao:2016bin,
    author = "Gao, Ping and Jafferis, Daniel Louis and Wall, Aron C.",
    title = "{Traversable Wormholes via a Double Trace Deformation}",
    eprint = "1608.05687",
    archivePrefix = "arXiv",
    primaryClass = "hep-th",
    doi = "10.1007/JHEP12(2017)151",
    journal = "JHEP",
    volume = "12",
    pages = "151",
    year = "2017"
}

@article{Maldacena:2017axo,
    author = "Maldacena, Juan and Stanford, Douglas and Yang, Zhenbin",
    title = "{Diving into traversable wormholes}",
    eprint = "1704.05333",
    archivePrefix = "arXiv",
    primaryClass = "hep-th",
    doi = "10.1002/prop.201700034",
    journal = "Fortsch. Phys.",
    volume = "65",
    number = "5",
    pages = "1700034",
    year = "2017"
}

@book{Gradshteyn:2014,
  title={Table of Integrals, Series, and Products},
  author={Gradshteyn, I.S. and Ryzhik, I.M.},
  isbn={9781483265643},
  url={https://books.google.com/books?id=F7jiBQAAQBAJ},
  year={2014},
  publisher={Academic Press}
}

@misc{DLMF,
    key = "{\relax DLMF}",
    title = "{\it NIST Digital Library of Mathematical Functions}",
    howpublished = "\url{https://dlmf.nist.gov/}, Release 1.2.7 of 2026-06-15",
    url = "https://dlmf.nist.gov/",
    note =  "F.~W.~J. Olver, A.~B. {Olde Daalhuis}, D.~W. Lozier, B.~I. Schneider,
            R.~F. Boisvert, C.~W. Clark, B.~R. Miller, B.~V. Saunders,
            H.~S. Cohl, and M.~A. McClain, eds."
}

@book{Exton:1978,
  author    = {Exton, Harold},
  title     = {Handbook of Hypergeometric Integrals: Theory, Applications,
               Tables, Computer Programs},
  year      = {1978},
  publisher = {Ellis Horwood; Halsted Press, a division of J. Wiley},
  address   = {Chichester, England; New York},
  isbn      = {0853121222},
  url       = {https://search.worldcat.org/title/4135088}
}

@book{Gelfand:1964,
    language = {eng},
    publisher = {Academic Press},
    title = {Generalized functions . Volume I . Properties and operations / I. M. Gel'fand and G. E. Shilov,... ; translated by Eugen Saletan,...},
    year = {1964},
    author = {Gelfand,Izrail Moiseevich and Šilov,Georgij Evgenʹevič and Saletan,Eugène Jérôme},
    address = {New York London},
    booktitle = {Generalized functions . Volume I . Properties and operations},
    isbn = {0-12-279501-6},
}

@article{Shenker:2013yza,
    author = "Shenker, Stephen H. and Stanford, Douglas",
    title = "{Multiple Shocks}",
    eprint = "1312.3296",
    archivePrefix = "arXiv",
    primaryClass = "hep-th",
    reportNumber = "SU-ITP-13-24",
    doi = "10.1007/JHEP12(2014)046",
    journal = "JHEP",
    volume = "12",
    pages = "046",
    year = "2014"
}

@article{Maldacena:2001kr,
    author = "Maldacena, Juan Martin",
    title = "{Eternal black holes in anti-de Sitter}",
    eprint = "hep-th/0106112",
    archivePrefix = "arXiv",
    reportNumber = "NSF-ITP-01-59",
    doi = "10.1088/1126-6708/2003/04/021",
    journal = "JHEP",
    volume = "04",
    pages = "021",
    year = "2003"
}

@article{Boruch:2024kvv,
    author = "Boruch, Jan and Iliesiu, Luca V. and Lin, Guanda and Yan, Cynthia",
    title = "{How the Hilbert space of two-sided black holes factorises}",
    eprint = "2406.04396",
    archivePrefix = "arXiv",
    primaryClass = "hep-th",
    doi = "10.1007/JHEP06(2025)092",
    journal = "JHEP",
    volume = "06",
    pages = "092",
    year = "2025"
}

@article{Kolchmeyer:2023gwa,
    author = "Kolchmeyer, David K.",
    title = "{von Neumann algebras in JT gravity}",
    eprint = "2303.04701",
    archivePrefix = "arXiv",
    primaryClass = "hep-th",
    doi = "10.1007/JHEP06(2023)067",
    journal = "JHEP",
    volume = "06",
    pages = "067",
    year = "2023"
}

@article{Penington:2023dql,
    author = "Penington, Geoff and Witten, Edward",
    title = "{Algebras and States in JT Gravity}",
    eprint = "2301.07257",
    archivePrefix = "arXiv",
    primaryClass = "hep-th",
    month = "1",
    year = "2023"
}

@article{Gao:2000ga,
    author = "Gao, Sijie and Wald, Robert M.",
    title = "{Theorems on gravitational time delay and related issues}",
    eprint = "gr-qc/0007021",
    archivePrefix = "arXiv",
    doi = "10.1088/0264-9381/17/24/305",
    journal = "Class. Quant. Grav.",
    volume = "17",
    pages = "4999--5008",
    year = "2000"
}

@article{Stanford:2021bhl,
    author = "Stanford, Douglas and Yang, Zhenbin and Yao, Shunyu",
    title = "{Subleading Weingartens}",
    eprint = "2107.10252",
    archivePrefix = "arXiv",
    primaryClass = "hep-th",
    doi = "10.1007/JHEP02(2022)200",
    journal = "JHEP",
    volume = "02",
    pages = "200",
    year = "2022"
}

@article{Blommaert:2025bgd,
    author = "Blommaert, Andreas and Kudler-Flam, Jonah and Urbach, Erez Y.",
    title = "{Absolute entropy and the observer{\textquoteright}s no-boundary state}",
    eprint = "2505.14771",
    archivePrefix = "arXiv",
    primaryClass = "hep-th",
    doi = "10.1007/JHEP11(2025)113",
    journal = "JHEP",
    volume = "11",
    pages = "113",
    year = "2025"
}

@article{Anninos:2024wpy,
    author = "Anninos, Dionysios and Galante, Dami{\'a}n A. and Maneerat, Chawakorn",
    title = "{Cosmological observatories}",
    eprint = "2402.04305",
    archivePrefix = "arXiv",
    primaryClass = "hep-th",
    doi = "10.1088/1361-6382/ad5824",
    journal = "Class. Quant. Grav.",
    volume = "41",
    number = "16",
    pages = "165009",
    year = "2024"
}

@article{An:2021fcq,
    author = "An, Zhongshan and Anderson, Michael T.",
    title = "{The initial boundary value problem and quasi-local Hamiltonians in general relativity}",
    eprint = "2103.15673",
    archivePrefix = "arXiv",
    primaryClass = "gr-qc",
    doi = "10.1088/1361-6382/ac0a86",
    journal = "Class. Quant. Grav.",
    volume = "38",
    number = "15",
    pages = "154001",
    year = "2021"
}

@article{Banihashemi:2022jys,
    author = "Banihashemi, Batoul and Jacobson, Ted",
    title = "{Thermodynamic ensembles with cosmological horizons}",
    eprint = "2204.05324",
    archivePrefix = "arXiv",
    primaryClass = "hep-th",
    doi = "10.1007/JHEP07(2022)042",
    journal = "JHEP",
    volume = "07",
    pages = "042",
    year = "2022"
}

@article{Batra:2024kjl,
    author = "Batra, Gauri and De Luca, G. Bruno and Silverstein, Eva and Torroba, Gonzalo and Yang, Sungyeon",
    title = "{Bulk-local dS$_{3}$ holography: the matter with $ T\overline{T} $ + {\ensuremath{\Lambda}}$_{2}$}",
    eprint = "2403.01040",
    archivePrefix = "arXiv",
    primaryClass = "hep-th",
    doi = "10.1007/JHEP10(2024)072",
    journal = "JHEP",
    volume = "10",
    pages = "072",
    year = "2024"
}

@article{Anninos:2011zn,
    author = "Anninos, Dionysios and Anous, Tarek and Bredberg, Irene and Ng, Gim Seng",
    title = "{Incompressible Fluids of the de Sitter Horizon and Beyond}",
    eprint = "1110.3792",
    archivePrefix = "arXiv",
    primaryClass = "hep-th",
    doi = "10.1007/JHEP05(2012)107",
    journal = "JHEP",
    volume = "05",
    pages = "107",
    year = "2012"
}

@article{Anninos:2017hhn,
    author = "Anninos, Dionysios and Hofman, Diego M.",
    title = "{Infrared realization of dS$_{2}$ in AdS$_{2}$}",
    eprint = "1703.04622",
    archivePrefix = "arXiv",
    primaryClass = "hep-th",
    doi = "10.1088/1361-6382/aab143",
    journal = "Class. Quant. Grav.",
    volume = "35",
    number = "8",
    pages = "085003",
    year = "2018"
}

@article{Anninos:2018svg,
    author = "Anninos, Dionysios and Galante, Dami{\'a}n A. and Hofman, Diego M.",
    title = "{De Sitter horizons {\&} holographic liquids}",
    eprint = "1811.08153",
    archivePrefix = "arXiv",
    primaryClass = "hep-th",
    doi = "10.1007/JHEP07(2019)038",
    journal = "JHEP",
    volume = "07",
    pages = "038",
    year = "2019"
}

@article{Aalsma:2020aib,
    author = "Aalsma, Lars and Shiu, Gary",
    title = "{Chaos and complementarity in de Sitter space}",
    eprint = "2002.01326",
    archivePrefix = "arXiv",
    primaryClass = "hep-th",
    reportNumber = "MAD-TH-20-01",
    doi = "10.1007/JHEP05(2020)152",
    journal = "JHEP",
    volume = "05",
    pages = "152",
    year = "2020"
}

@article{Dray:1984ha,
    author        = "Dray, Tevian and 't Hooft, Gerard",
    title         = "{The gravitational shock wave of a massless particle}",
    journal       = "Nucl. Phys. B",
    volume        = "253",
    pages         = "173--188",
    year          = "1985",
    doi           = "10.1016/0550-3213(85)90525-5"
}

@article{Aichelburg:1970dh,
    author        = "Aichelburg, Peter C. and Sexl, Roman U.",
    title         = "{On the gravitational field of a massless particle}",
    journal       = "Gen. Rel. Grav.",
    volume        = "2",
    pages         = "303--312",
    year          = "1971",
    doi           = "10.1007/BF00758149"
}

@incollection{Penrose:1972xrn,
    author        = "Penrose, Roger",
    title         = "{The geometry of impulsive gravitational waves}",
    booktitle     = "{General Relativity: Papers in Honour of J. L. Synge}",
    editor        = "O'Raifeartaigh, L.",
    publisher     = "Clarendon Press",
    address       = "Oxford",
    pages         = "101--115",
    year          = "1972"
}

@article{tHooft:1987vrq,
    author        = "'t Hooft, Gerard",
    title         = "{Graviton dominance in ultra-high-energy scattering}",
    journal       = "Phys. Lett. B",
    volume        = "198",
    pages         = "61--63",
    year          = "1987",
    doi           = "10.1016/0370-2693(87)90159-6"
}

@article{Sfetsos:1994xa,
    author        = "Sfetsos, Konstadinos",
    title         = "{On gravitational shock waves in curved space-times}",
    journal       = "Nucl. Phys. B",
    volume        = "436",
    pages         = "721--745",
    year          = "1995",
    eprint        = "hep-th/9408169",
    archivePrefix = "arXiv",
    doi           = "10.1016/0550-3213(94)00573-W"
}

@article{Marini:2026zjk,
    author = "Marini, Tommaso and Qi, Xiao-Liang and Verlinde, Herman",
    title = "{3D near-de Sitter gravity and the soft mode of DSSYK}",
    eprint = "2604.21014",
    archivePrefix = "arXiv",
    primaryClass = "hep-th",
    month = "4",
    year = "2026"
}

@article{Verlinde:2024znh,
    author = "Verlinde, Herman",
    title = "{Double-scaled SYK, chords and de Sitter gravity}",
    eprint = "2402.00635",
    archivePrefix = "arXiv",
    primaryClass = "hep-th",
    doi = "10.1007/JHEP03(2025)076",
    journal = "JHEP",
    volume = "03",
    pages = "076",
    year = "2025"
}

@article{Verlinde:2024zrh,
    author = "Verlinde, Herman and Zhang, Mengyang",
    title = "{SYK correlators from 2D Liouville-de Sitter gravity}",
    eprint = "2402.02584",
    archivePrefix = "arXiv",
    primaryClass = "hep-th",
    doi = "10.1007/JHEP05(2025)053",
    journal = "JHEP",
    volume = "05",
    pages = "053",
    year = "2025"
}

@article{Rahman:2024iiu,
    author = "Rahman, Adel A. and Susskind, Leonard",
    title = "{$p$-Chords, Wee-Chords, and de Sitter Space}",
    eprint = "2407.12988",
    archivePrefix = "arXiv",
    primaryClass = "hep-th",
    month = "7",
    year = "2024"
}

@article{Miyashita:2026zyl,
    author = "Miyashita, Shoichiro and Sekino, Yasuhiro and Susskind, Leonard",
    title = "{Holograms and Standard Models}",
    eprint = "2607.05678",
    archivePrefix = "arXiv",
    primaryClass = "hep-th",
    month = "7",
    year = "2026"
}

@article{Anninos:2011ui,
    author = "Anninos, Dionysios and Hartman, Thomas and Strominger, Andrew",
    title = "{Higher Spin Realization of the dS/CFT Correspondence}",
    eprint = "1108.5735",
    archivePrefix = "arXiv",
    primaryClass = "hep-th",
    doi = "10.1088/1361-6382/34/1/015009",
    journal = "Class. Quant. Grav.",
    volume = "34",
    number = "1",
    pages = "015009",
    year = "2017"
}

@article{Anninos:2017eib,
    author = "Anninos, Dionysios and Denef, Frederik and Monten, Ruben and Sun, Zimo",
    title = "{Higher Spin de Sitter Hilbert Space}",
    eprint = "1711.10037",
    archivePrefix = "arXiv",
    primaryClass = "hep-th",
    doi = "10.1007/JHEP10(2019)071",
    journal = "JHEP",
    volume = "10",
    pages = "071",
    year = "2019",
    note = "[Erratum: JHEP 06, 085 (2024)]"
}

@article{Anninos:2026hia,
    author = {Anninos, Dionysios and Baracco, Chiara and Letsios, Vasileios A. and M{\"u}hlmann, Beatrix},
    title = "{dS$^4$ Metamorphosis}",
    eprint = "2602.19812",
    archivePrefix = "arXiv",
    primaryClass = "hep-th",
    month = "2",
    year = "2026"
}

@article{Silverstein:2024xnr,
    author = "Silverstein, Eva and Torroba, Gonzalo",
    title = "{Timelike-bounded dS$_{4}$ holography from a solvable sector of the T$^{2}$ deformation}",
    eprint = "2409.08709",
    archivePrefix = "arXiv",
    primaryClass = "hep-th",
    doi = "10.1007/JHEP03(2025)156",
    journal = "JHEP",
    volume = "03",
    pages = "156",
    year = "2025"
}

@article{Banihashemi:2026mje,
    author = "Banihashemi, Batoul and Batra, Gauri and Law, Y. T. Albert and Silverstein, Eva and Torroba, Gonzalo",
    title = "{The yes boundaries wavefunctions of the universe}",
    eprint = "2604.10267",
    archivePrefix = "arXiv",
    primaryClass = "hep-th",
    month = "4",
    year = "2026"
}

@article{McAllister:2024lnt,
    author = "McAllister, Liam and Moritz, Jakob and Nally, Richard and Schachner, Andreas",
    title = "{Candidate de Sitter vacua}",
    eprint = "2406.13751",
    archivePrefix = "arXiv",
    primaryClass = "hep-th",
    reportNumber = "CERN-TH-2024-090",
    doi = "10.1103/PhysRevD.111.086015",
    journal = "Phys. Rev. D",
    volume = "111",
    number = "8",
    pages = "086015",
    year = "2025"
}

@article{Susskind:2022bia,
    author = "Susskind, Leonard",
    title = "{De Sitter Space, Double-Scaled SYK, and the Separation of Scales in the Semiclassical Limit}",
    eprint = "2209.09999",
    archivePrefix = "arXiv",
    primaryClass = "hep-th",
    doi = "10.22128/jhap.2024.920.1103",
    journal = "JHAP",
    volume = "5",
    number = "1",
    pages = "1--30",
    year = "2025"
}

@article{Susskind:2021esx,
    author = "Susskind, Leonard",
    title = "{Entanglement and Chaos in De Sitter Space Holography: An SYK Example}",
    eprint = "2109.14104",
    archivePrefix = "arXiv",
    primaryClass = "hep-th",
    doi = "10.22128/jhap.2021.455.1005",
    journal = "JHAP",
    volume = "1",
    number = "1",
    pages = "1--22",
    year = "2021"
}

@article{Banks:2000fe,
    author = "Banks, Tom",
    editor = "Duff, Michael J. and Liu, J. T. and Lu, J.",
    title = "{Cosmological breaking of supersymmetry?}",
    eprint = "hep-th/0007146",
    archivePrefix = "arXiv",
    reportNumber = "RUNHETC-2000-24, SCIPP-00-23",
    doi = "10.1142/S0217751X01003998",
    journal = "Int. J. Mod. Phys. A",
    volume = "16",
    pages = "910--921",
    year = "2001"
}

@article{Sivaramakrishnan:2024ydy,
    author = "Sivaramakrishnan, Allic",
    title = "{Correlators of Worldline Proper Length}",
    eprint = "2406.17205",
    archivePrefix = "arXiv",
    primaryClass = "hep-th",
    month = "6",
    year = "2024"
}

@article{Cheung:2026euf,
    author = "Cheung, Clifford and Sivaramakrishnan, Allic and Wilson-Gerow, Jordan and Zhou, Lihang",
    title = "{On Perturbatively Dressed Observables}",
    eprint = "2605.26077",
    archivePrefix = "arXiv",
    primaryClass = "hep-th",
    month = "5",
    year = "2026"
}

@article{Usatyuk:2024isz,
    author = "Usatyuk, Mykhaylo and Zhao, Ying",
    title = "{Closed universes, factorization, and ensemble averaging}",
    eprint = "2403.13047",
    archivePrefix = "arXiv",
    primaryClass = "hep-th",
    doi = "10.1007/JHEP02(2025)052",
    journal = "JHEP",
    volume = "02",
    pages = "052",
    year = "2025"
}

@article{Abdalla:2025gzn,
    author = "Abdalla, Ahmed I. and Antonini, Stefano and Iliesiu, Luca V. and Levine, Adam",
    title = "{The gravitational path integral from an observer{\textquoteright}s point of view}",
    eprint = "2501.02632",
    archivePrefix = "arXiv",
    primaryClass = "hep-th",
    doi = "10.1007/JHEP05(2025)059",
    journal = "JHEP",
    volume = "05",
    pages = "059",
    year = "2025"
}

@article{Harlow:2025pvj,
    author = "Harlow, Daniel and Usatyuk, Mykhaylo and Zhao, Ying",
    title = "{Quantum mechanics and observers for gravity in a closed universe}",
    eprint = "2501.02359",
    archivePrefix = "arXiv",
    primaryClass = "hep-th",
    reportNumber = "MIT-CTP/5824",
    doi = "10.1007/JHEP02(2026)108",
    journal = "JHEP",
    volume = "02",
    pages = "108",
    year = "2026"
}

@article{Usatyuk:2024mzs,
    author = "Usatyuk, Mykhaylo and Wang, Zi-Yue and Zhao, Ying",
    title = "{Closed universes in two dimensional gravity}",
    eprint = "2402.00098",
    archivePrefix = "arXiv",
    primaryClass = "hep-th",
    doi = "10.21468/SciPostPhys.17.2.051",
    journal = "SciPost Phys.",
    volume = "17",
    number = "2",
    pages = "051",
    year = "2024"
}

@article{Akers:2025ahe,
    author = "Akers, Chris and Bueller, Gracemarie and DeWolfe, Oliver and Higginbotham, Kenneth and Reinking, Johannes and Rodriguez, Rudolph",
    title = "{On observers in holographic maps}",
    eprint = "2503.09681",
    archivePrefix = "arXiv",
    primaryClass = "hep-th",
    doi = "10.1007/JHEP05(2025)201",
    journal = "JHEP",
    volume = "05",
    pages = "201",
    year = "2025"
}

@article{Engelhardt:2025azi,
    author = "Engelhardt, Netta and Gesteau, Elliott and Harlow, Daniel",
    title = "{Observer complementarity for black holes and holography}",
    eprint = "2507.06046",
    archivePrefix = "arXiv",
    primaryClass = "hep-th",
    reportNumber = "MIT-CTP/5884",
    month = "7",
    year = "2025"
}

@article{Zhao:2026mpl,
    author = "Zhao, Ying",
    title = "{''It from Bit'': The Hartle-Hawking state and quantum mechanics for de Sitter observers}",
    eprint = "2602.05939",
    archivePrefix = "arXiv",
    primaryClass = "hep-th",
    reportNumber = "MIT-CTP/6000",
    month = "2",
    year = "2026"
}

@article{Harlow:2026hky,
    author = "Harlow, Daniel",
    title = "{Observers, $\alpha$-parameters, and the Hartle-Hawking state}",
    eprint = "2602.03835",
    archivePrefix = "arXiv",
    primaryClass = "hep-th",
    reportNumber = "MIT-CTP/5900",
    month = "2",
    year = "2026"
}

@article{Alonso-Monsalve:2024oii,
    author = "Alonso-Monsalve, Elba and Harlow, Daniel and Jefferson, Patrick",
    title = "{Phase space of Jackiw-Teitelboim gravity with positive cosmological constant}",
    eprint = "2409.12943",
    archivePrefix = "arXiv",
    primaryClass = "hep-th",
    reportNumber = "MIT-CTP/5769",
    doi = "10.1007/JHEP03(2026)008",
    journal = "JHEP",
    volume = "03",
    pages = "008",
    year = "2026"
}

@article{DeVuyst:2024grw,
    author = "De Vuyst, Julian and Eccles, Stefan and Hoehn, Philipp A. and Kirklin, Josh",
    title = "{Linearization (in)stabilities and crossed products}",
    eprint = "2411.19931",
    archivePrefix = "arXiv",
    primaryClass = "hep-th",
    doi = "10.1007/JHEP05(2025)211",
    journal = "JHEP",
    volume = "05",
    pages = "211",
    year = "2025"
}

@article{Hoback:2026yqj,
    author = "Hoback, Sarah and Jafferis, Daniel L. and Wei, Zixia",
    title = "{Monitoring a de Sitter universe through an anti-de Sitter window}",
    eprint = "2606.31705",
    archivePrefix = "arXiv",
    primaryClass = "hep-th",
    month = "6",
    year = "2026"
}

@article{Wei:2025guh,
    author = "Wei, Zixia",
    title = "{Observers and Timekeepers: From the Page-Wootters Mechanism to the Gravitational Path Integral}",
    eprint = "2506.21489",
    archivePrefix = "arXiv",
    primaryClass = "hep-th",
    month = "6",
    year = "2025"
}

@article{Kabat:1992tb,
    author = "Kabat, Daniel N. and Ortiz, Miguel",
    title = "{Eikonal quantum gravity and Planckian scattering}",
    eprint = "hep-th/9203082",
    archivePrefix = "arXiv",
    reportNumber = "MIT-CTP-2069",
    doi = "10.1016/0550-3213(92)90627-N",
    journal = "Nucl. Phys. B",
    volume = "388",
    pages = "570--592",
    year = "1992"
}

\end{document}